\documentclass[article,english,12pt]{article}
\usepackage{amscd,amsfonts,amsmath,amssymb,amsthm,bbm,bm,latexsym,mathrsfs}
\usepackage{epsfig,graphics,graphicx,subfigure,longtable,float}
\usepackage[]{natbib}
\usepackage[english]{babel}
\usepackage[usenames]{color}
\usepackage{bookmark}
\usepackage{sgame}
\allowdisplaybreaks
\usepackage[top=1.2in, bottom=1.6in, left=1.3in, right=1.3in]{geometry}
\usepackage{adjustbox}
\usepackage{gensymb}
\usepackage{epstopdf}
\usepackage{pdfpages}
\usepackage{tikz}
\usepackage{booktabs}
\usepackage{multirow}
\usepackage{ctable}
\usetikzlibrary{shapes,arrows}	
\usepackage{tkz-berge}
\usetikzlibrary{bayesnet}
\usepackage[toc,page]{appendix}
\usepackage{chapterbib}
\usepackage[normalem]{ulem}
\usepackage{cancel}

\usepackage{enumitem}
\makeatother

\usepackage{xcolor}

\usepackage{caption}
\usepackage{graphicx}

\theoremstyle{definition} 
\newtheorem{defn}{Definition}
\theoremstyle{assumption} 

\theoremstyle{proposition} 
\newtheorem{prop}{Proposition}
\theoremstyle{lemma} 
\newtheorem{lemma}{Lemma}
\theoremstyle{proposition} 
\newtheorem{rem}{Remark}
\theoremstyle{definition} 
\newtheorem{example}{Example}

\newtheorem{theorem}{Theorem}

\makeatletter
\@addtoreset{proofpart}{theorem}
\makeatother

\usepackage{rotating}

\usepackage{url}
\defcitealias{Hackweb}{Hackmageddon}

\begin{document}

\title{\vspace{-50pt}Generalized Poisson Difference Autoregressive Processes\thanks{Authors' research used the SCSCF and HPC multiprocessor cluster systems at University Ca' Foscari. This work was funded in part by the Universit\'e Franco-Italienne "Visiting Professor Grant" UIF 2017-VP17\_12, by the French government under management of Agence Nationale de la Recherche as part of the "Investissements d'avenir" program, reference ANR-19-P3IA-0001 (PRAIRIE 3IA Institute) and by the Institut Universitaire de France. 
}}
\author{\hspace{-20pt} Giulia Carallo$^{\dag}$\setcounter{footnote}{3}\footnote{Corresponding author: \href{mailto:giulia.carallo@unive.it}{giulia.carallo@unive.it} (G. Carallo). Other contacts:  \href{mailto:r.casarin@unive.it}{r.casarin@unive.it} (R. Casarin). \href{mailto:Christian.Robert@ceremade.dauphine.fr}{Christian.Robert@ceremade.dauphine.fr} (C. P. Robert).} \hspace{15pt}
Roberto Casarin$^{\dag}$ \hspace{15pt} Christian P. Robert$^{\ddag\sharp}$ 
        \\ 
        \vspace{-7pt}
        \\
        {\centering {\small{$^\dag$Ca' Foscari University of Venice, Italy}}}\\
        {\centering {\small{$^\ddag$Universit\'e Paris-Dauphine, France}}}\\
          {\centering {\small{$^\sharp$University of Warwick, United Kindom}}}        
}
\vspace{-5pt}
\date{}
\maketitle

\vspace{-10pt}

\textbf{Abstract.}  
This paper introduces a new stochastic process with values in the set $\mathbb{Z}$ of integers with sign. The increments of process are Poisson differences and the dynamics has an autoregressive structure. We study the properties of the process and exploit the thinning representation to derive stationarity conditions and the stationary distribution of the process. We provide a Bayesian inference method and an efficient posterior approximation procedure based on Monte Carlo. Numerical illustrations on both simulated and real data show the effectiveness of the proposed inference.\\

\medskip


\medskip 

\textbf{MSC2010 subject classifications.} Primary 62G05, 62F15, 62M10, 62M20.\\

\section{Introduction}

In many real-world applications, time series of counts are commonly observed given the discrete nature of the variables of interest. Integer-valued variables 
 appear very frequently in many fields, such as medicine (see \cite{CarRoyLam1999}), epidemiology (see \cite{Zeg1988} and \cite{DavDunWan1999}), finance (see \cite{LieMolPoh2006} and \cite{RydShe2003}), economics (see \cite{Fre1998} and \cite{FreCab2004}), in social sciences (see \cite{PedKar2011}), sports (see \cite{ShaMoy2016}) and oceanography (see \cite{CunVasBou2018}). In this paper, we build on Poisson models, which is one of the most used model for counts data and propose a new model for integer-valued data with sign based on the generalized Poisson difference (GPD) distribution. An advantage in using this distribution relies on the possibility to account for overdispersed data with more flexibility, with respect to the standard Poisson difference distribution, a.k.a. Skellam distribution. Despite of its flexibility, GPD models have not been investigated and applied to many fields, yet. \cite{ShaMoy2014} proposed a GPD distribution obtained as the difference of two underling generalized Poisson (GP) distributions with different intensity parameters. They showed that this distribution is a special case of the GPD by \cite{Con1986} and studied its properties. They provided a Bayesian framework for inference on GPD and a zero-inflated version of the distribution to deal with the excess of zeros in the data. \cite{ShaMoy2016} showed empirically that GPD can perform better than the Skellam model.
 
As regards to the construction method, two main classes of models can be identified in the literature: parameter driven and observation driven. In parameter-driven models the parameters are functions of an unobserved stochastic process, and the observations are independent conditionally on the latent variable. In the observation-driven models the parameter dynamics is a function of the past observations.  Since this paper focuses on the observation-driven approach, we refer the reader to \cite{MacZuc1997} for a review of parameter-driven models.

Thinning operators are a key ingridient for the analysis of observation-driven models. The mostly used thinning operator is the binomial thinning, introduced by \cite{SteVHa1979} for the definition of self-decomposable distribution for positive integer-valued random variables. In mathematical biology, the binomial thinning can be interpreted as natural selection or reproduction, and in probability it is widely applied to study integer-valued processes. The binomial thinning has been generalized along different directions. \cite{Lat1998} proposed a generalized binomial thinning where individuals can reproduce more than once. \cite{KimPar2008} introduced the signed binomial thinning, in order to allows for negative values. \cite{Joe1996} and \cite{ZheBasDat2007} introduced the random coefficient thinning to account for external factors that may affect the coefficient of the thinning operation, such as unobservable environmental factors or states of the economy. When the coefficient follows a beta distribution one obtain the beta-binomial thinning (\cite{McK1985}, \cite{McK1986} and \cite{Joe1996}). \cite{AlOAly1992}, proposed the iterated thinning, which can be used when the process has negative-binomial marginals. \cite{AlzAlO1993} introduced the quasi-binomial thinning, that is more suitable for generalized Poisson processes. \cite{ZhaWanZhu2010} introduced the signed generalized power series thinning operator, as a generalization of \cite{KimPar2008} signed binomial thinning. Thinning operation can be combined linearly to define new operations such as the binomial thinning difference (\cite{Fre2010}) and the quasi-binomial thinning difference (\cite{CunVasBou2018}). For a detailed review of the thinning operations and their properties different surveys can be consulted: \cite{MacZuc1997}, \cite{KedFok2005}, \cite{McK2003}, \cite{Wei2008}, \cite{ScoWeiGou2015}. In this paper, we apply the quasi-binomial thinning difference.

In the integer-valued autoregressive process literature, thinning operations have been used either to define a process, such as in the literature on integer-valued autoregressive-moving average models (INARMA), or to study the properties of a process, such as in the literature on integer-valued GARCH (INGARCH). INARMA have been firstly introduced by \cite{McK1986} and \cite{AlOAlz1987} by using the binomial thinning operator. \cite{DuLi1991} extended to the higher order $p$ the first-order INAR model of \cite{AlOAlz1987}. \cite{KimPar2008} introduced an integer-valued autoregressive process with signed binomial thinning operator, INARS($p$), able for time series defined on $\mathbb{Z}$. \cite{AndKar2014}introduced SINARS, that is a special case of INARS model with Skellam innovations. In order to allow for negative integers, \cite{Fre2010} proposed a true integer-valued autoregressive model (TINAR(1)), that can be seen as the difference between two independent Poisson INAR(1) process. \cite{AlzAlO1993} have studied an integer-valued ARMA process with Generalized Poisson marginals while \cite{AlzOma2014} proposed a Poisson difference INAR(1) model. \cite{CunVasBou2018} firstly applied the GPD distribution to build a stochastic process. The authors proposed an INAR with GPD marginals and provided the properties of the process, such as mean, variance, kurtosis and conditional properties. 

\cite{RydShe2000} introduced heteroskedastic integer-valued processes with Poisson marginals. Later on, \cite{Hei2003} introduced an autoregressive conditional Poisson model and \cite{Fer2006} proposed the INGARCH process.  Both models have Poisson margins. \cite{Zhu2012} defined a INGARCH process to model overdispersed and underdispersed count data with GP margins and \cite{AloAlzOma2018} proposed a Skellam model with GARCH dynamics for the variance of the process. \cite{KooLitLuc2014} proposed a Generalized Autoregressive Score (GAS) Skellam model. In this paper, we extend \cite{Fer2006} and \cite{Zhu2012} by assuming GPD marginals for the INGARCH model, and use the quasi-binomial thinning difference to study the properties of the new process.

%

Another contribution of the paper regards the inference approach. In the literature, maximum likelihood estimation has been widely investigated for integer-valued processes, whereas a very few papers discuss Bayesian inference procedures. \cite{CheLee2016} introduced Bayesian zero-inflated GP-INGARCH, with structural breaks. \cite{ZhuLi2009} proposed a Bayesian Poisson INGARCH(1,1) and \cite{Chen2016} a Bayesian Autoregressive Conditional Negative Binomial model. In this paper, we develop a Bayesian inference procedure for the proposed GPD-INGARCH process and a Markov chain Monte Carlo (MCMC) procedure for posterior approximation. One of the advantages of the Bayesian approach is that extra-sample information on the parameters value can be included in the estimation process through the prior distributions. Moreover, it can be easily combined with a data augmentation strategy to make the likelihood function more tractable.

We apply our model to a cyber-threat dataset and contribute to cyber-risk literature providing evidence of temporal patterns in the mean and variance of the threats, which can be used to predict threat arrivals. Cyber threats are increasingly considered as a top global risk for the financial and insurance sectors and for the economy as a whole \citep[e.g.][]{EIOPA2}. As pointed out in \cite{Hass16}, the frequency of cyber events substantially increased in the past few years and cyber-attacks occur on a daily basis. Understanding cyber-threats dynamics and their impact is critical to ensure effective controls and risk mitigation tools. Despite these evidences and the relevance of the topic, the research on the analysis of cyber threats is scarce and scattered in different research areas such as cyber security \citep{Agra18}, criminology \cite{Bren04}, economics \cite{Anderson610} and sociology. In statistics there are a few works on modelling and forecasting cyber-attacks. \cite{Xu2017} introduced a copula model to predict effectiveness of cyber-security. \cite{Wer17} used an autoregressive integrated moving average model to forecast the number of daily cyber-attacks. \cite{Edwards2015HypeAH} apply Bayesian Poisson and negative binomial models to analyse data breaches and find evidence of over-dispersion and absence of time trends in the number of breaches. See \cite{Hus2018} for a review on modelling cyber threats.

The paper is organized as follows. In Section 2 we introduce the parametrization used for the GPD and define the GPD-INGARCH process. Section 3 aims at studying the properties of the process. Section 4 presents a Bayesian inference procedure. Section 5 and 6 provide some illustration on simulated and real data, respectively. Section 7 concludes.

\section{Generalized Poisson Difference INGARCH}\label{Model}
A random variable $X$ follows a Generalized Poisson (GP) distribution if and only if its probability mass function (pmf) is
\begin{equation}
P_{x}(\theta, \lambda)=\frac{\theta(\theta +x\lambda)^{x-1}}{x!} e^{-\theta-x\lambda}\;\;\;\;\;\;\;\;\; x=0, 1, 2, \ldots
\end{equation}
with parameters $\theta >0$ and $0\leq \lambda < 1$ \citep[see][]{Con1986}. We denote this distribution with $GP(\theta,\lambda)$. Let $X\sim GP(\theta_{1}, \lambda)$ and $Y\sim GP(\theta_{2}, \lambda)$ be two independent GP random variables, \cite{Con1986} showed that the probability distribution of $Z=(X-Y)$ follows a Generalized Poisson Difference distribution (GDP) with pmf:
\begin{equation}\label{Eq31_maintext}
P_z(\theta_1,\theta_2,\lambda)=e^{-\theta_{1}-\theta_{2}-z\lambda} \sum_{y=0}^{\infty} \frac{\theta_{2}(\theta_{2}+y\lambda)^{y-1}}{y!} \frac{\theta_{1}(\theta_{1}+(y+z)\lambda)^{y+z-1}}{(y+z)!} e^{-2y\lambda}
\end{equation}
where $z$ takes integer values in the interval $(-\infty, +\infty)$ and $0<\lambda <1$ and $\theta_{1},\theta_{2} >0$ are the parameters of the distribution. See Appendix \ref{App:GPD} for a more general definition of the GPD with possibly negative $\lambda$.

In the following Lemma we state the convolution property of the GPD distribution since which will be used in this paper. Appendix \ref{Proof:convolution} provides an original proof of this result.
\begin{lemma}[\textbf{Convolution Property}]\label{Lemma:convo}
The sum of two independent random GPD variates, $X+Y$, with parameters $(\theta_{1},\theta_{2},\lambda)$ and $(\theta_{3},\theta_{4},\lambda)$ is a GPD variate with parameters $(\theta_{1}+\theta_{3},\theta_{2}+\theta_{4},\lambda)$. The difference of two independent random GPD variates, $X-Y$, with parameters $(\theta_{1},\theta_{2},\lambda)$ and $(\theta_{3},\theta_{4},\lambda)$ is a GPD variate with parameters $(\theta_{1}+\theta_{4},\theta_{2}+\theta_{3},\lambda)$.
\end{lemma}

We use an equivalent pmf and a re-parametrization of the GPD, which are better suited for the definition of a INGARCH model. A random variable $Z$ follows a GPD if and only if its probability distribution is
\begin{small}
\begin{equation}\label{e214rmaintext}
P_z(\mu,\sigma^2,\lambda)= e^{-\sigma^{2}-z\lambda} \hspace{-18pt}\sum_{s=\max(0,-z)}^{+\infty}\!\! \frac{1}{4}\frac{\sigma^{4}+\mu^{2}}{s!(s+z)!}\!\!\left[ \frac{\sigma^{2}+\mu}{2}+(s+z)\lambda\right]^{s+z-1}\!\! \left[ \frac{\sigma^{2}-\mu}{2}+s\lambda\right]^{s-1}\!\! e^{-2\lambda s}
\end{equation}
\end{small}
We denote this distribution with $GPD(\mu,\sigma^2,\lambda)$. 
\begin{rem}\label{RemEquiv}
The probability distribution in Eq. \ref{e214rmaintext} is equivalent to the one in Eq. \ref{Eq31_maintext} up to the reparametrization $\mu = \theta_{1}-\theta_{2}$ and $\sigma^{2}=\theta_{1}+\theta_{2}$. See Appendix \ref{App:proofs} for a proof.
\end{rem}

The mean, variance, skewness and kurtosis of a GDP random variable can be obtained in close form by exploiting the representation of the GDP as difference between independent GP random variables.
\begin{rem}\label{remMom}
Let $Z\sim GPD(\mu,\sigma^2,\lambda)$, then mean and variance are:
\begin{equation}
E(Z)=\frac{\mu}{1-\lambda}, \qquad V(Z)=\frac{\sigma^{2}}{(1-\lambda)^{3}}
\end{equation}
and the Pearson skewness and kurtosis are:
\begin{equation}
S(Z)= \frac{\mu}{\sigma^{3}} \frac{(1+2\lambda)}{\sqrt{1-\lambda}}, \qquad K(Z)=3+\frac{1+8\lambda+6\lambda^{2}}{\sigma^{2}(1-\lambda)}
\end{equation}
\end{rem}
See Appendix \ref{App:proofs} for a proof.\\
\begin{figure}[t]
\centering
\begin{tabular}{cc}
(a) GPD$(\mu,\sigma^2,\lambda)$ distribution for $\sigma^{2}=10$ &(b) GPD$(\mu,\sigma^2,\lambda)$ distribution for $\mu=2$\\
\includegraphics[scale=0.4]{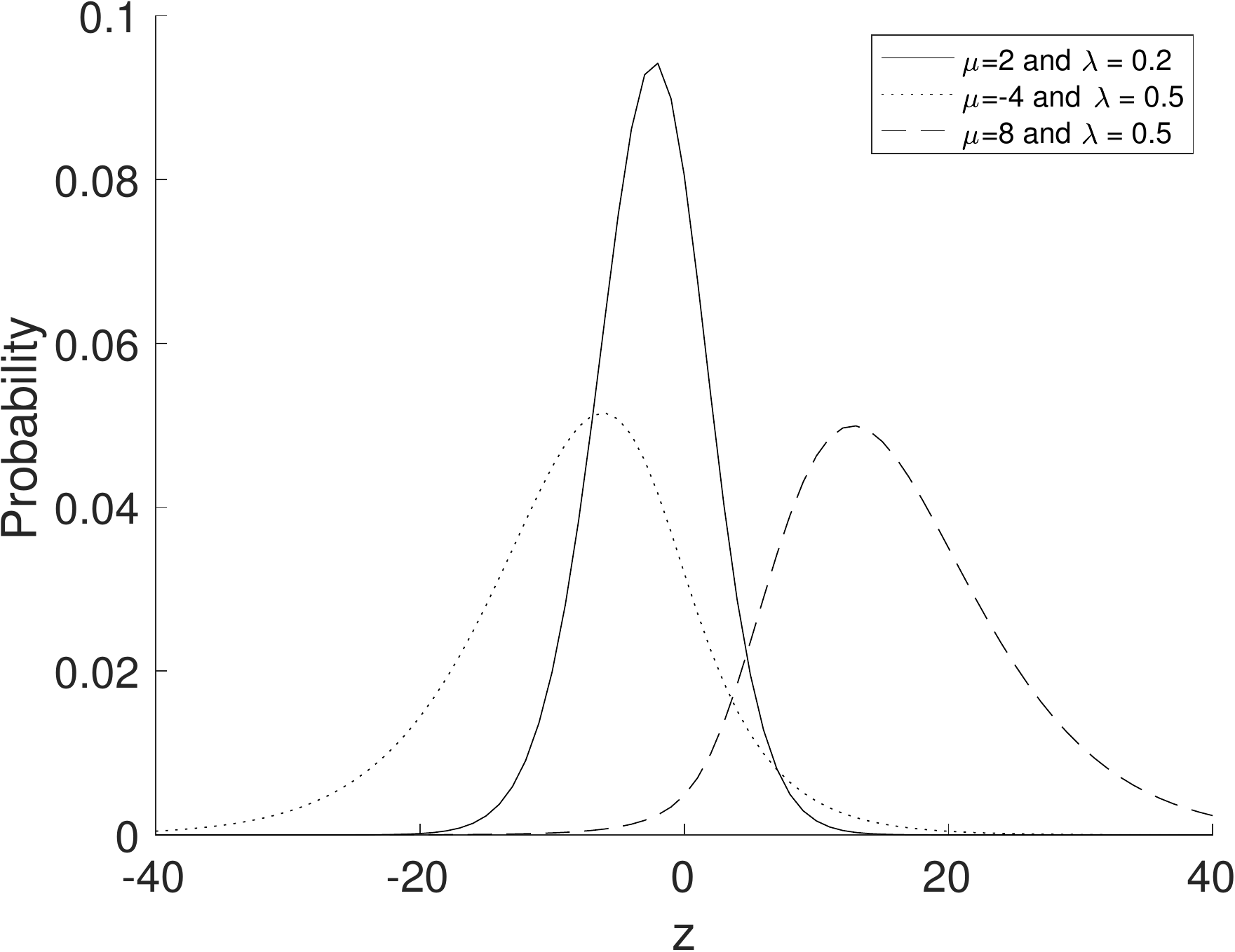}&
\includegraphics[scale=0.4]{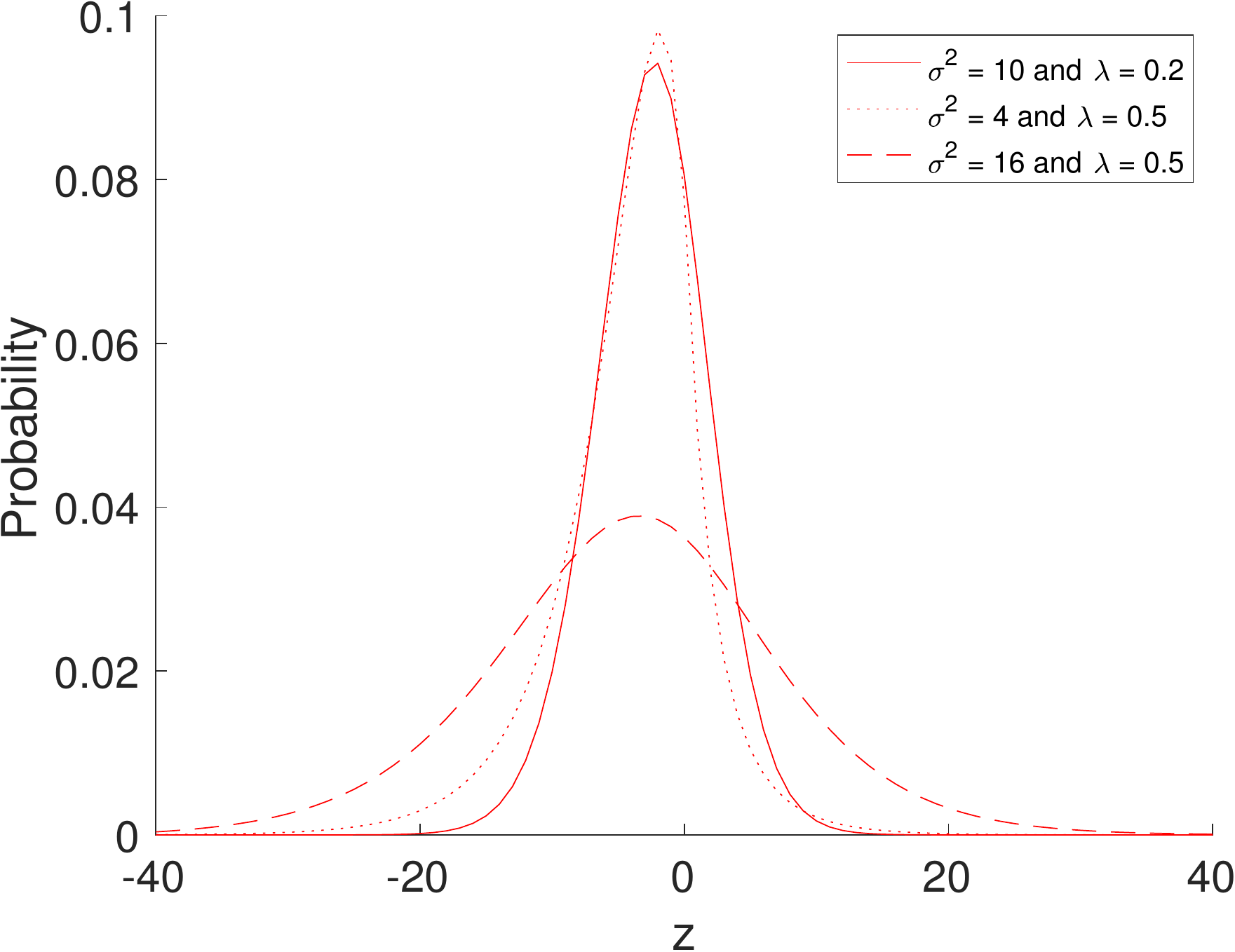}\\
\end{tabular}
\caption{Generalized Poisson difference distribution GPD$(\mu,\sigma^2,\lambda)$ for some values of $\lambda$, $\mu$ and $\sigma^{2}$. The distribution with $\lambda=0.2$, $\mu=2$ and $\sigma^{2}=10$ (solid line) is taken as baseline in both panels.}
\label{fig:1maintext}
\end{figure}

Figure \ref{fig:1maintext} shows the sensitivity of the probability distribution with respect to the location parameter $\mu$ (panel a), the scale parameter $\sigma^{2}$ (panel b) and the skewness parameter $\lambda$ (different lines in each plot). For given values of $\lambda$ and $\mu$, when $\sigma^{2}$ decreases the dispersion of the GPD increases (dotted and dashed lines in the right plot). For given values of $\lambda$ and $\sigma^{2}$, the distribution is right-skewed for $\mu=8$, which corresponds to $S(Z)=0.7155$, and left-skewed for $\mu=-4$, which corresponds to $S(Z)=-0.3578$, (dotted and dashed lines in the left plot). See Appendix \ref{App:GPD} for further numerical illustrations.

Differently from the usual GARCH$(p,q)$ process (e.g., see \cite{FraZak2019}), the INGARCH$(p,q)$ is defined as an integer-valued process $\lbrace Z_{t}\rbrace_{t \in \mathbb{Z}}$, where $Z_{t}$ is a series of counts. Let $\mathcal{F}_{t-1}$ be the $\sigma$-field generated by $\lbrace Z_{t-j}\rbrace_{j\geq 1}$, then the GPD-INGARCH$(p,q)$ is defined as 
\begin{equation*}\label{EqG2}
Z_{t}\vert \mathcal{F}_{t-1} \sim GPD(\tilde{\mu}_{t},\tilde{\sigma}^{2}_{t},\lambda)
\end{equation*}
with
\begin{equation}\label{ingarch}
\frac{\tilde{\mu}_{t}}{1-\lambda} =\mu_{t}=\alpha_{0}+\sum_{i=1}^{p}\alpha_{i}Z_{t-i}+\sum_{j=1}^{q} \beta_{j}\mu_{t-j}
\end{equation}
where $\tilde{\mu}_{t-j}=\mu_{t-j}(1-\lambda)$, $\alpha_{0}\in \mathbb{R}$, $\alpha_{i}\geq 0$, $\beta_{j}\geq 0$, $i=1,\ldots,p$, $p\geq 1$, $j=1,\ldots, q$, $q \geq 0$. For $q=0$ the model reduces to a GPD-INARCH$(p)$ and for $\lambda=0$ one obtains a Skellam INGARCH$(p,q)$ which extends to Poisson differences the Poisson INGARCH$(p,q)$ of \cite{Fer2006}. From the properties of the GPD, the conditional mean $\mu_t=E(Z_{t}\vert \mathcal{F}_{t-1})$ and variance $\sigma^2_t=V(Z_{t}\vert \mathcal{F}_{t-1})$ of the process are:
\begin{eqnarray}
\mu_{t}&=&\frac{\tilde{\mu}_{t}}{1-\lambda},\quad \sigma^{2}_{t}=\frac{\tilde{\sigma}^{2}_{t}}{(1-\lambda)^{3}}
\end{eqnarray}
respectively. 
In the application, we assume $\sigma^{2}_{t}=\vert \mu_{t} \vert \phi$. Following the parametrization defined in Remark \ref{RemEquiv}, we need to impose the constrain $\phi >(1-\lambda)^{-2}$, in order to have a well-defined GPD distribution. In Fig. \ref{GARCH11}, we provide some simulated examples of the GPD-INGARCH$(1,1)$ process for different values of the parameters $\alpha_{0}$, $\alpha_{1}$ and $\beta_{1}$.
\begin{figure}[h!]	
\begin{tabular}{cc}
\textbf{Low persistence} & \textbf{High persistence}\\
$(\alpha_{1}=0.23$, $\beta_{1}=0.25)$ & $(\alpha_{1}=0.32,\beta_{1}=0.59)$\\
(a) {\small $\alpha_{0}=1.55$, $\lambda=0.4$, $\phi=3$} & (b) {\small $\alpha_{0}=1.55$, $\lambda=0.4$, $\phi=3$}\\
\includegraphics[height=4cm,width=7.5cm]{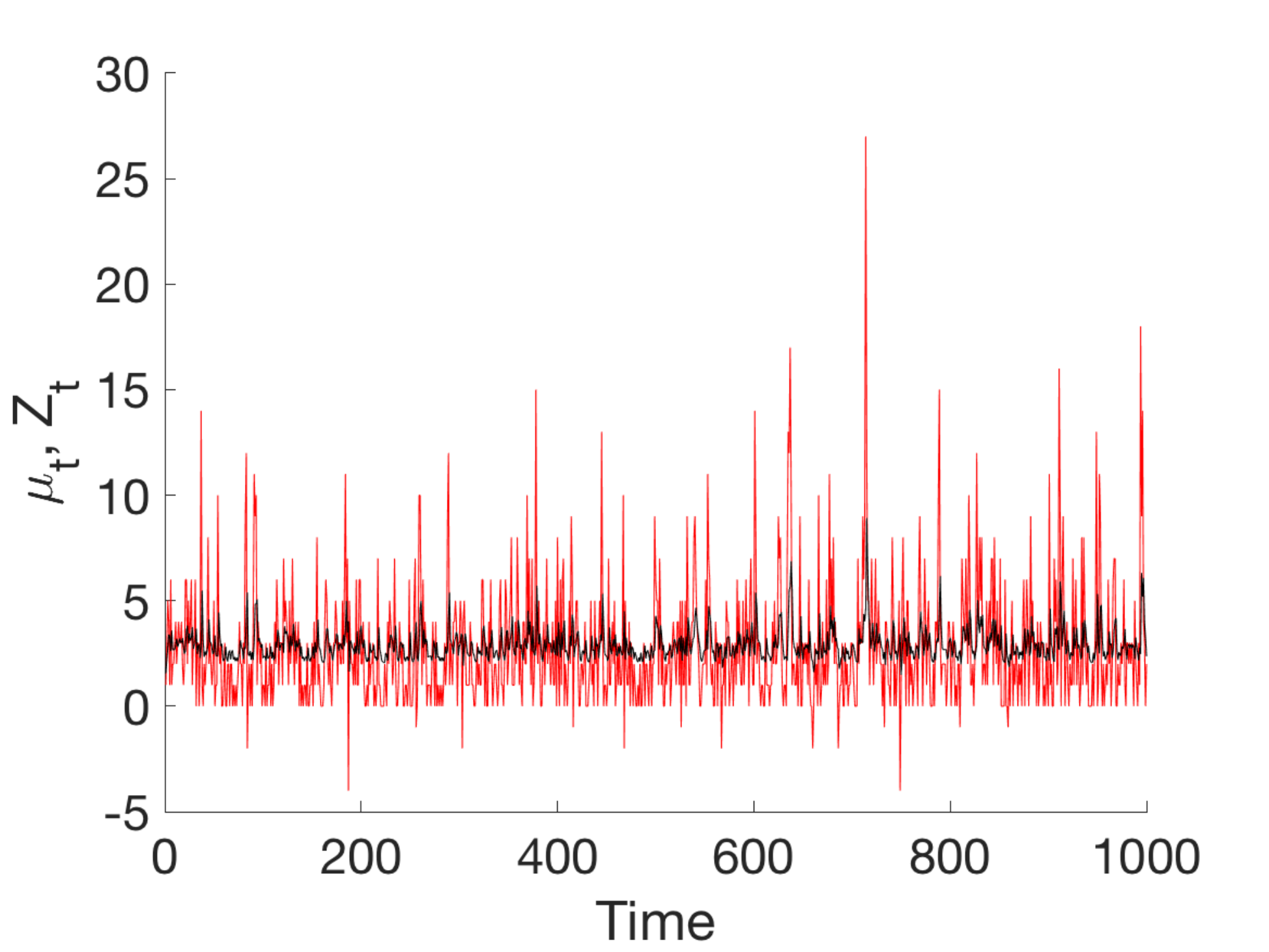}&
\includegraphics[height=4cm,width=7.5cm]{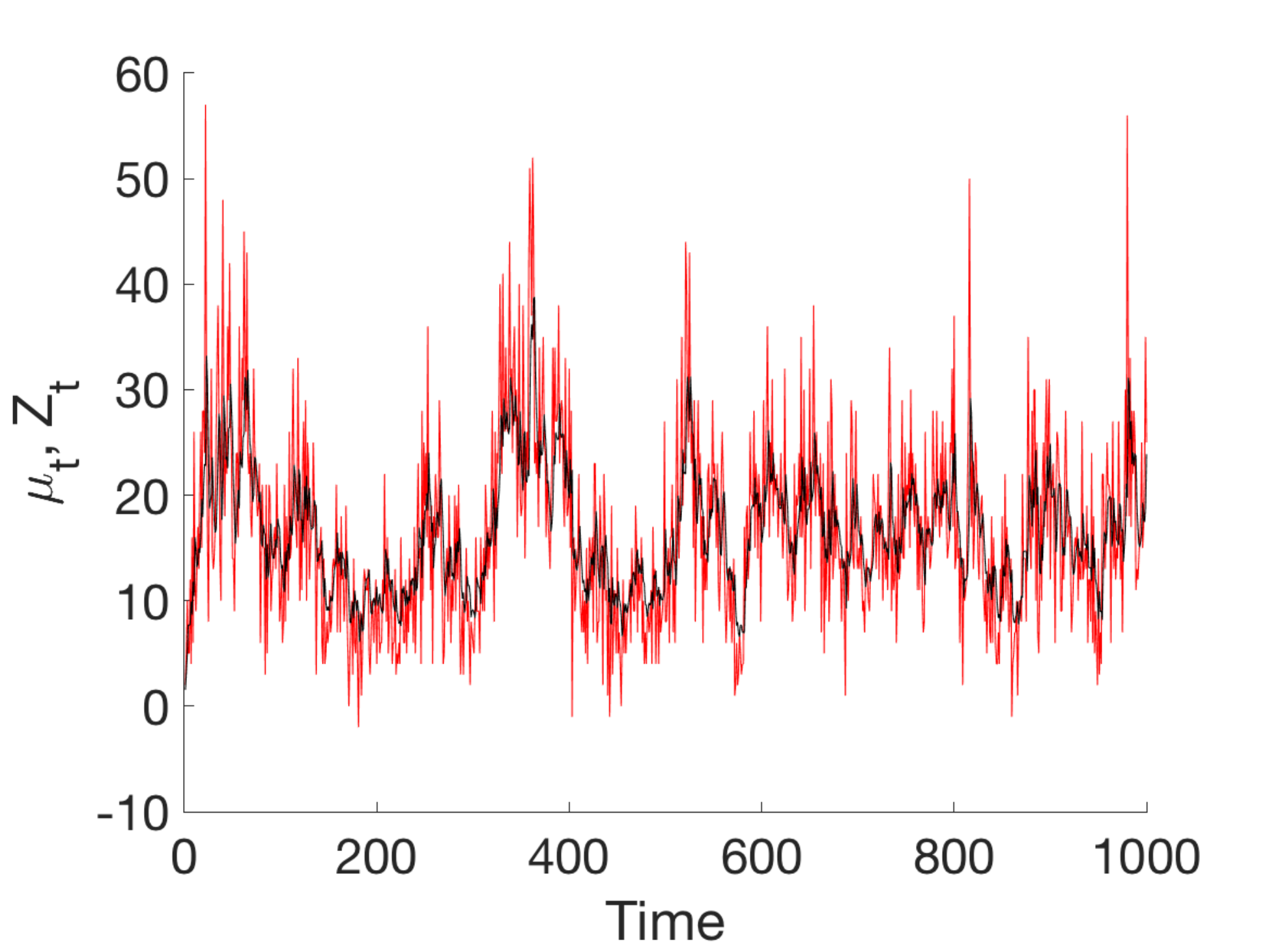}\\
(c) {\small $\alpha_{0}=-1.55$, $\lambda=0.4$, $\phi=3$} & (d) {\small $\alpha_{0}=-1.55$, $\lambda=0.4$, $\phi=3$} \\
\includegraphics[height=4cm,width=7.5cm]{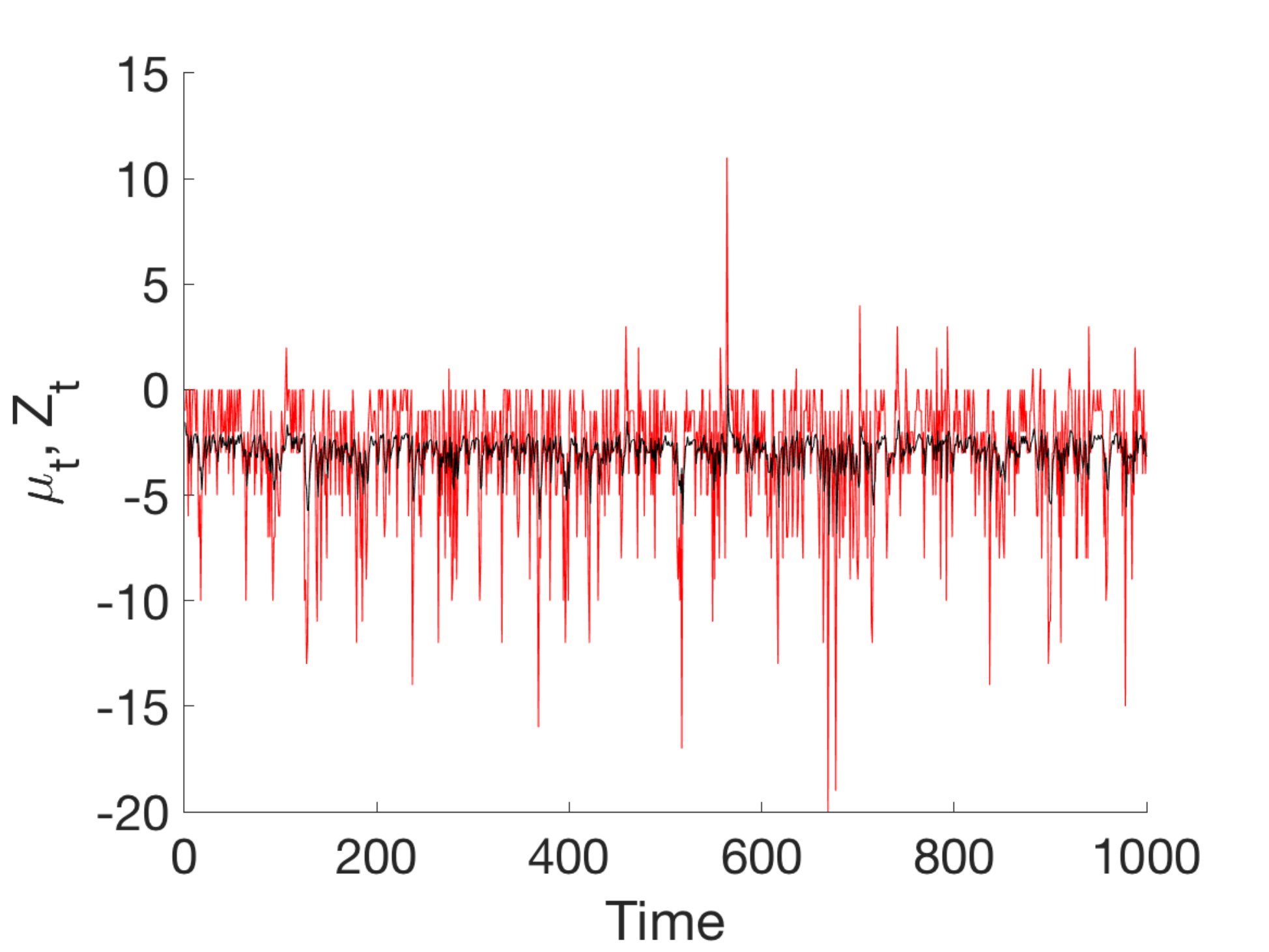}&
\includegraphics[height=4cm,width=7.5cm]{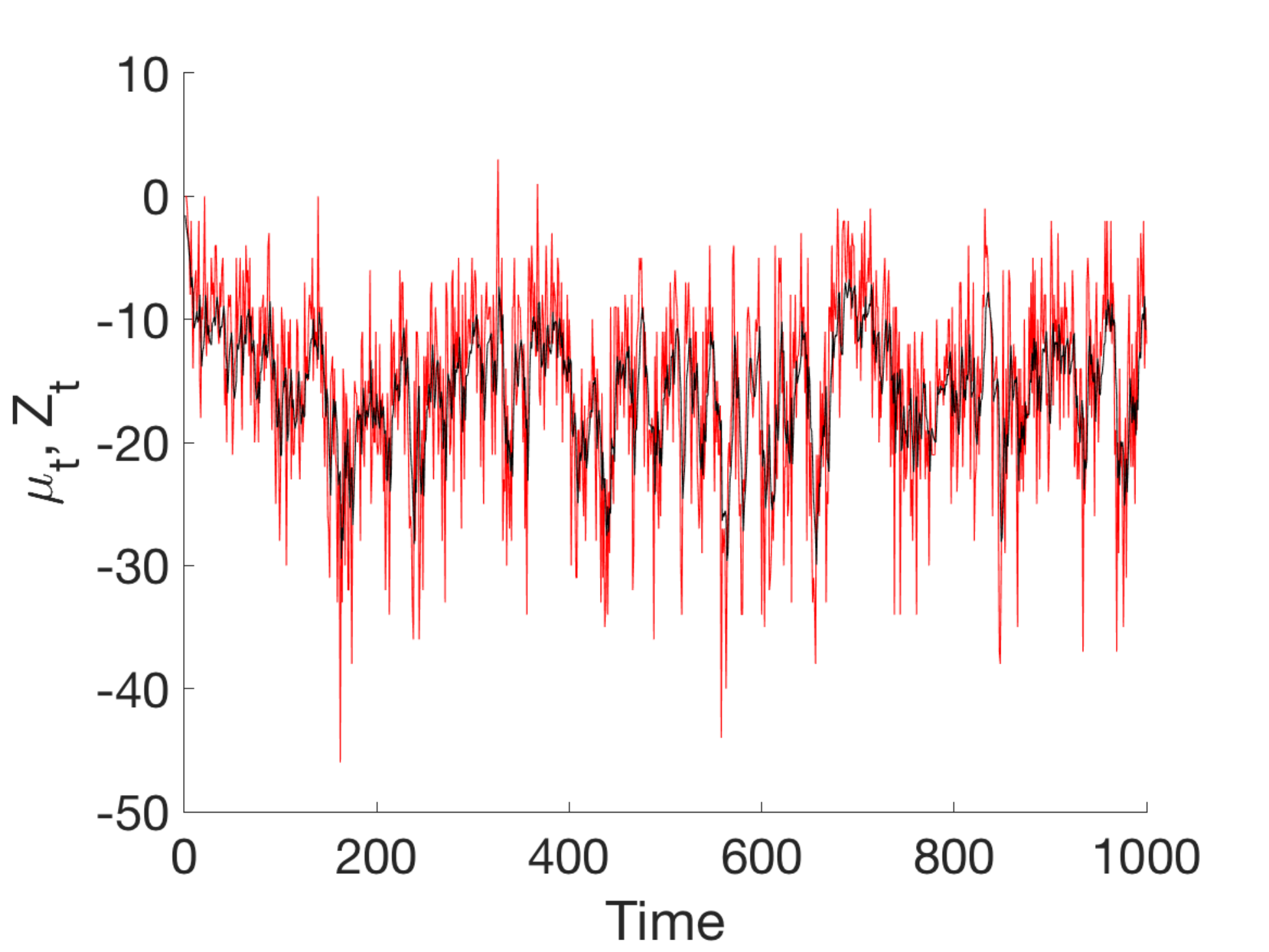}\\
\\
(e) {\small $\alpha_{0}=1.55$, $\lambda=0.1$, $\phi=3$} & (f) {\small $\alpha_{0}=1.55$, $\lambda=0.7$, $\phi=3$} \\
\includegraphics[height=4cm,width=7.5cm]{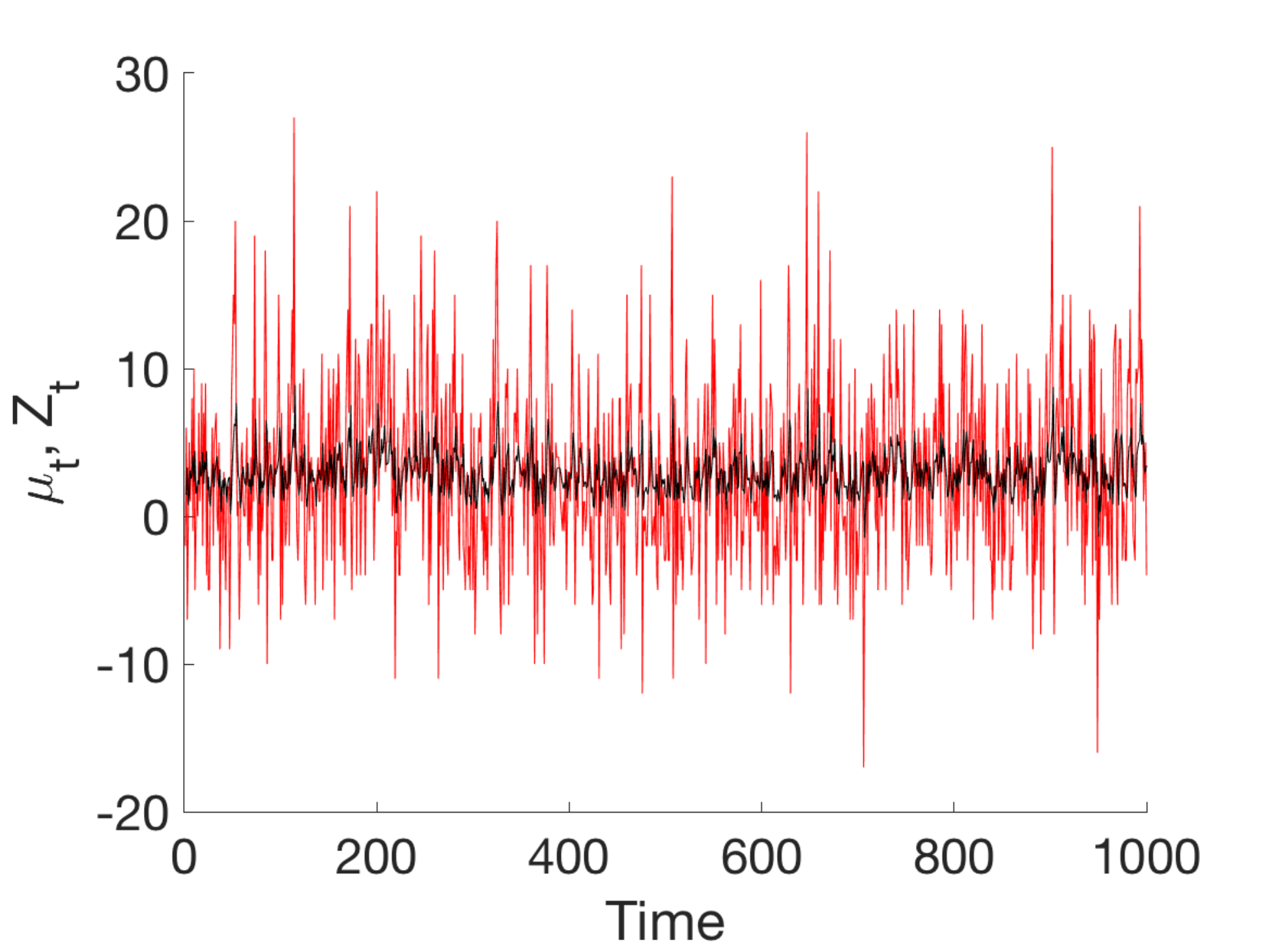}&
\includegraphics[height=4cm,width=7.5cm]{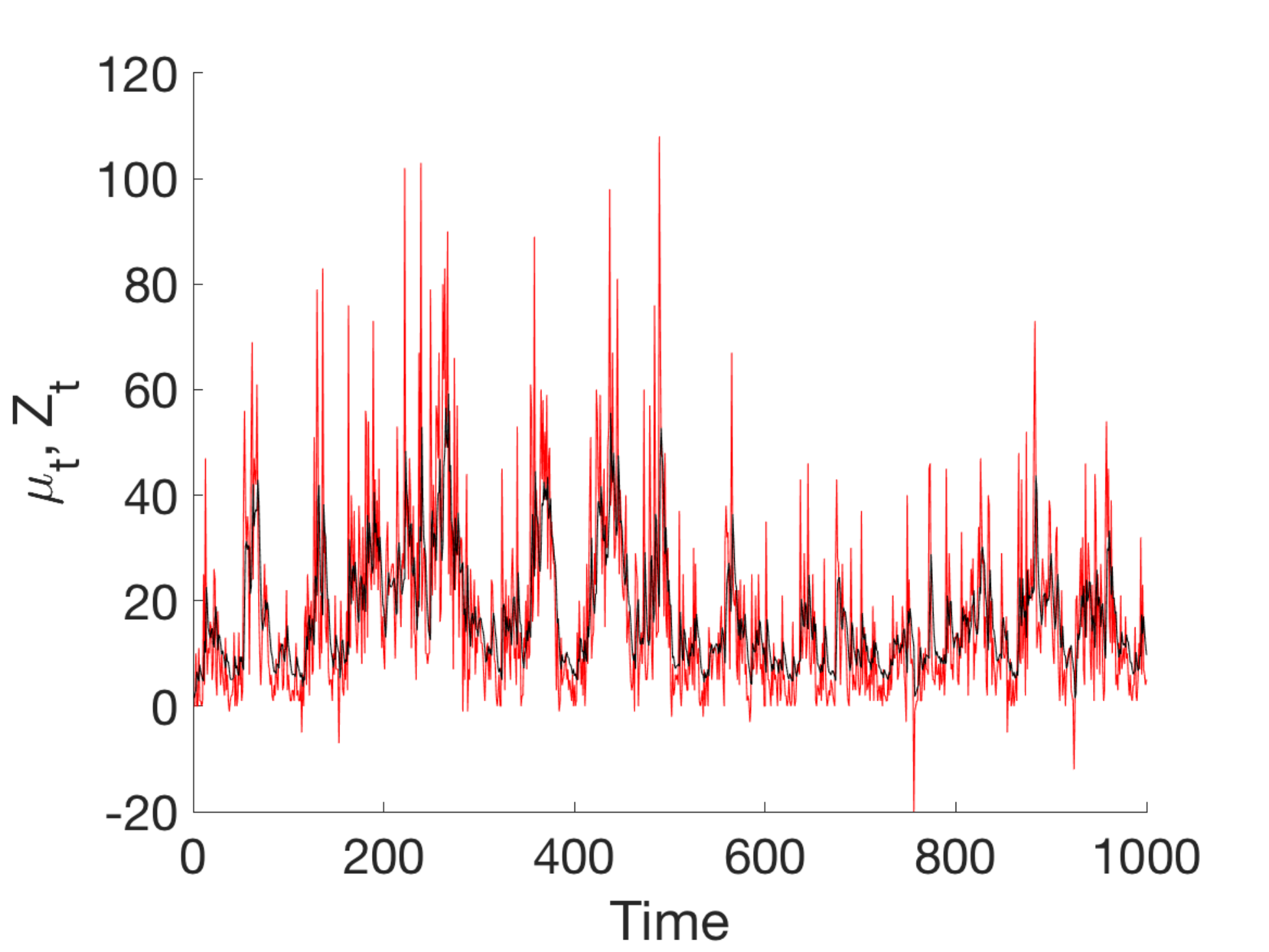}\\
\end{tabular}
\caption{Simulated INGARCH$(1,1)$ paths for different values of the parameters $\alpha_{0}$, $\alpha_{1}$ and $\beta_{1}$. In Panels from (a) to (d) the effect of $\alpha_{0}$ ($\alpha_{0}>0$ in the first line and $\alpha_{0}<0$ in the second line) with $\lambda=0.4$ and $\phi=3$. In Panels (e) and (f) the effect of lambda ($\lambda=0.1$ left and $\lambda=0.7$ right) in the two settings.}
\label{GARCH11}
\end{figure}

Simulations from a GPD-INGARCH are obtained as differences of GP sequences
\begin{eqnarray}\label{StRep}
Z_{t}&=&X_t-Y_t,\, X_{t} \sim GP(\theta_{1t}, \lambda),\,Y_{t} \sim GP(\theta_{2t}, \lambda)\nonumber
\end{eqnarray}
where 
\begin{equation}
\theta_{1t} = \frac{\sigma^{2}_{t}+\mu_{t}}{2}, \quad \theta_{2t} = \frac{\sigma^{2}_{t}-\mu_{t}}{2}.
\end{equation} 
Each random sequence is generated by the branching method in \cite{Fam1997}, which performs faster than the inversion method for large values of $\theta_{1t}$ and $\theta_{2t}$. We considered two parameter settings: low persistence, that is $\alpha_{1}+\beta_{1}$ much less than 1 (first column in Fig. \ref{GARCH11}) and high persistence, that is $\alpha_{1}+\beta_{1}$ close to 1 (second column in Fig. \ref{GARCH11}). The first and second line show paths for positive and negative value of the intercept $\alpha_{0}$, respectively. The last line illustrates the effect of $\lambda$ on the trajectories with respect to the baselines in Panels (a) and (b). Comparing (I.a) and (I.b) in Fig. \ref{GARCH12} one can see that increasing $\beta_{1}$ increases serial correlation and the kurtosis level (compare (II.a) and (II.b)).

\begin{figure}[t]
\centering
\def\arraystretch{0.2}
\begin{tabular}{cc}
\vspace{0.5cm}
\small{(I) Autocorrelation function} & \small{(II) Unconditional histograms}\\

{\small (I.a) \hspace{24 mm} (I.b)} & {\small (II.a) \hspace{24 mm}(II.b)}\\
\begin{minipage}{.5\textwidth}
\includegraphics[width=3.5cm, height=3cm]{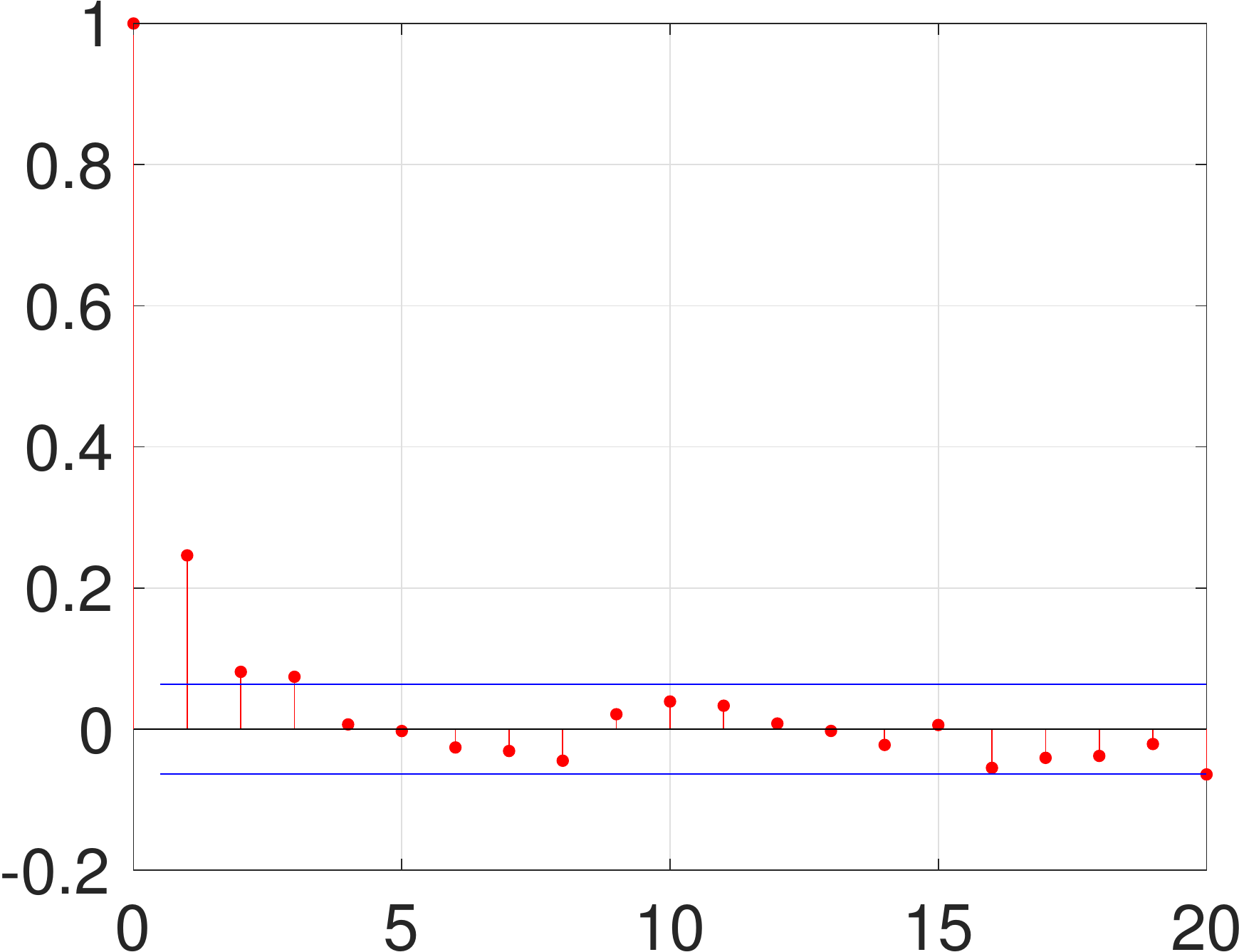}
\includegraphics[width=3.5cm, height=3cm]{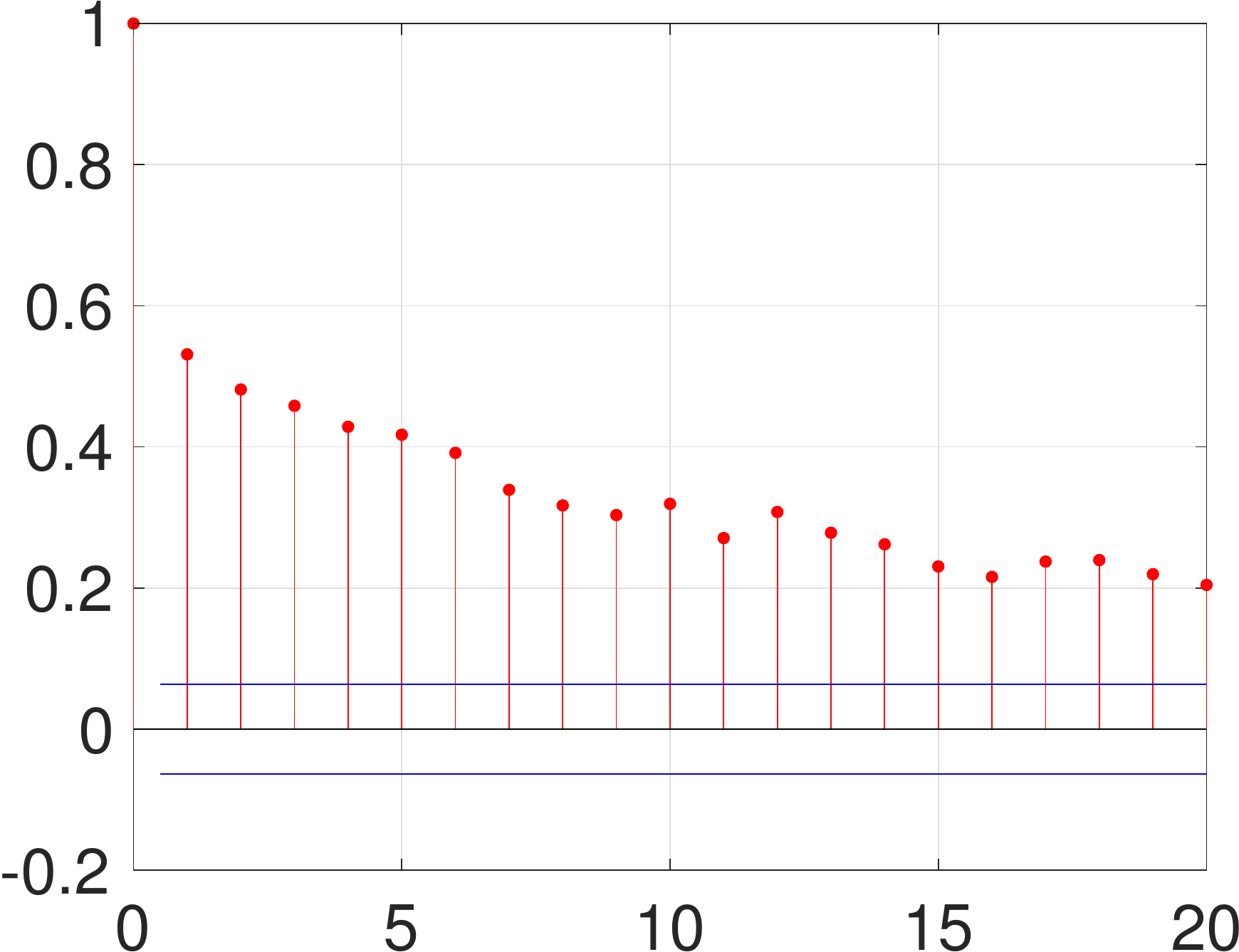}
\end{minipage} & \begin{minipage}{.5\textwidth}
\includegraphics[width=3.5cm, height=3cm]{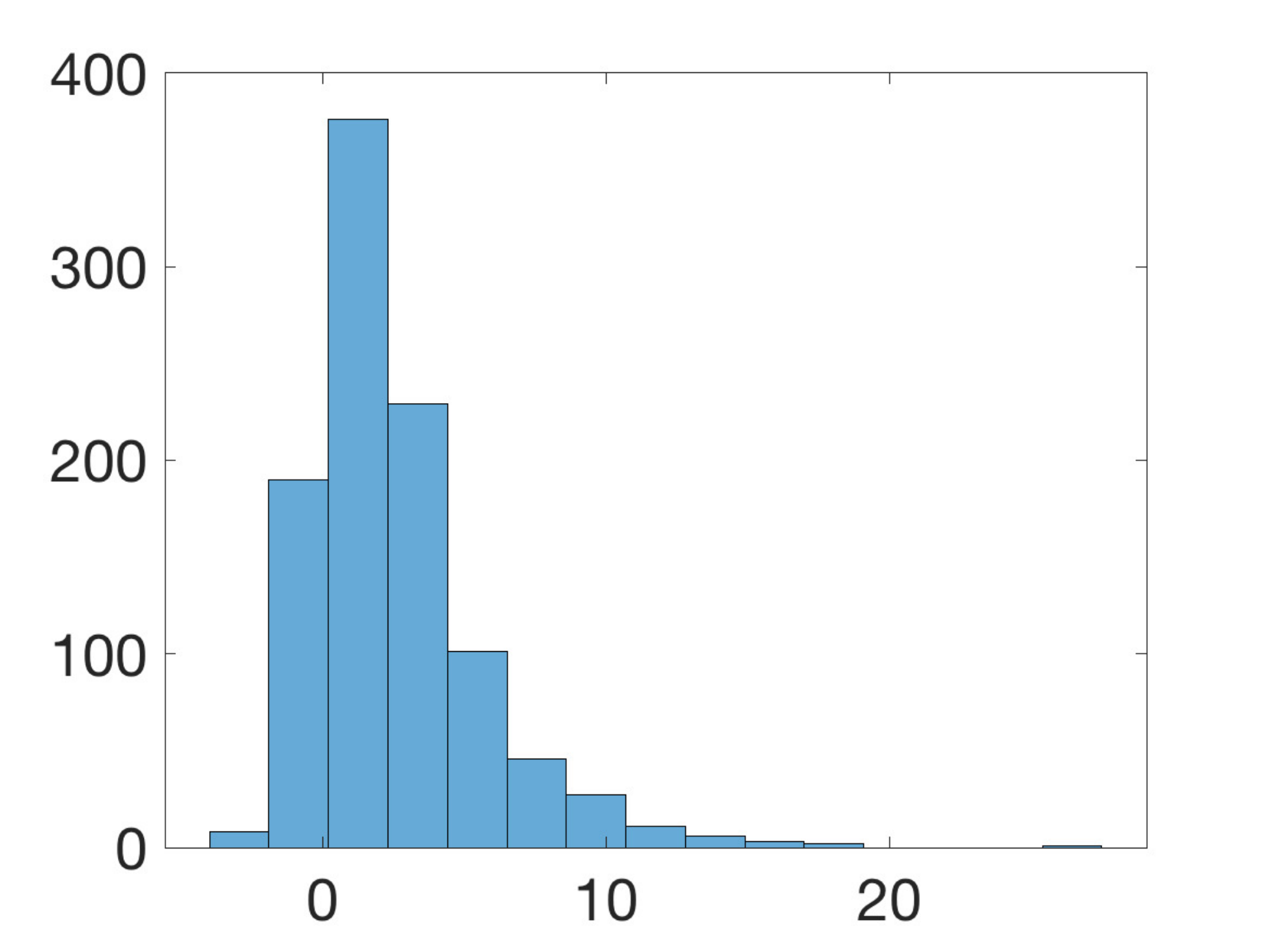} 
\includegraphics[width=3.5cm, height=3cm]{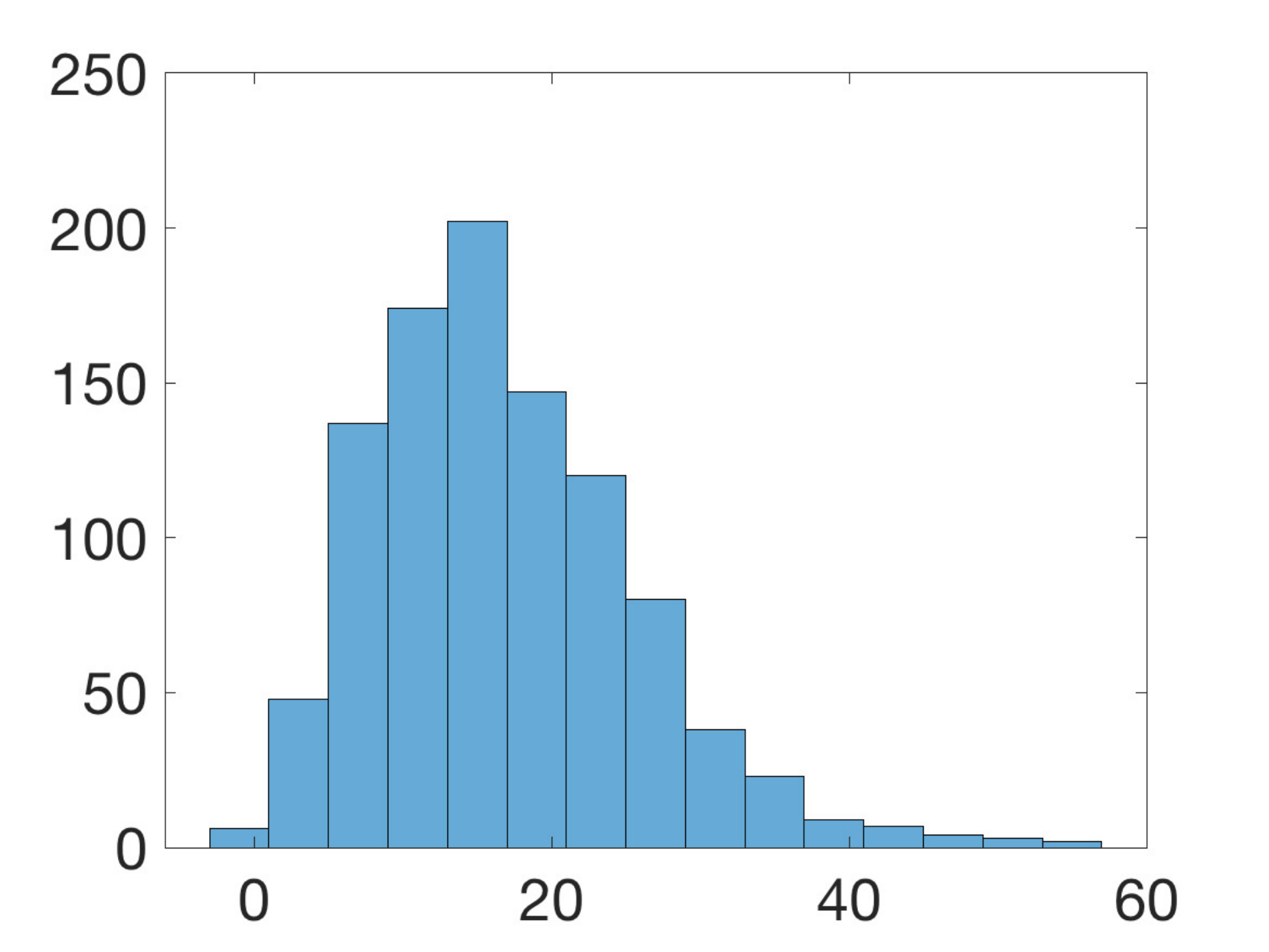}   
\end{minipage} \\

{\small (I.c) \hspace{24 mm} (I.d)} & {\small (II.c) \hspace{24 mm} (II.d)}\\
\begin{minipage}{.5\textwidth}
\includegraphics[width=3.5cm, height=3cm]{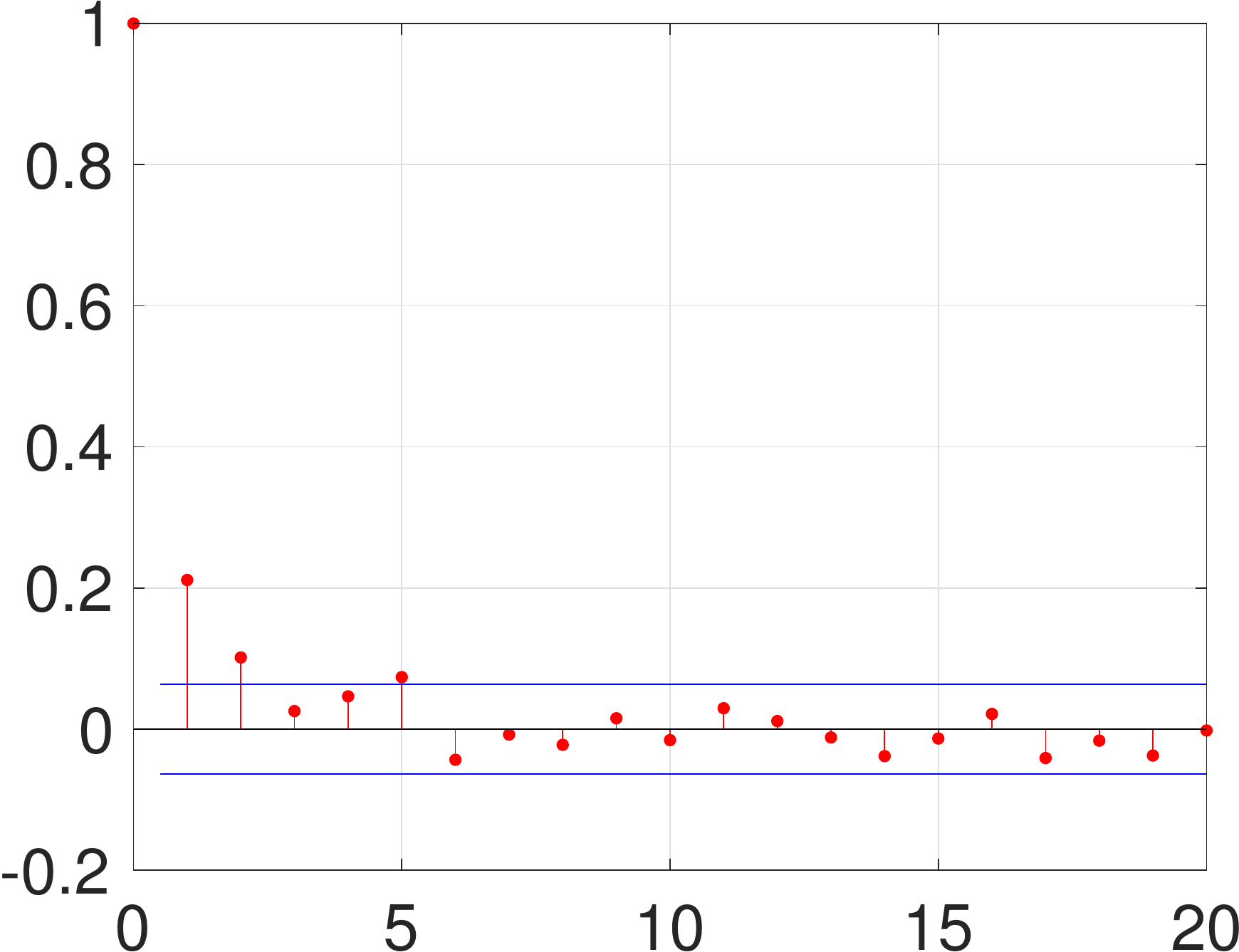} 
\includegraphics[width=3.5cm, height=3cm]{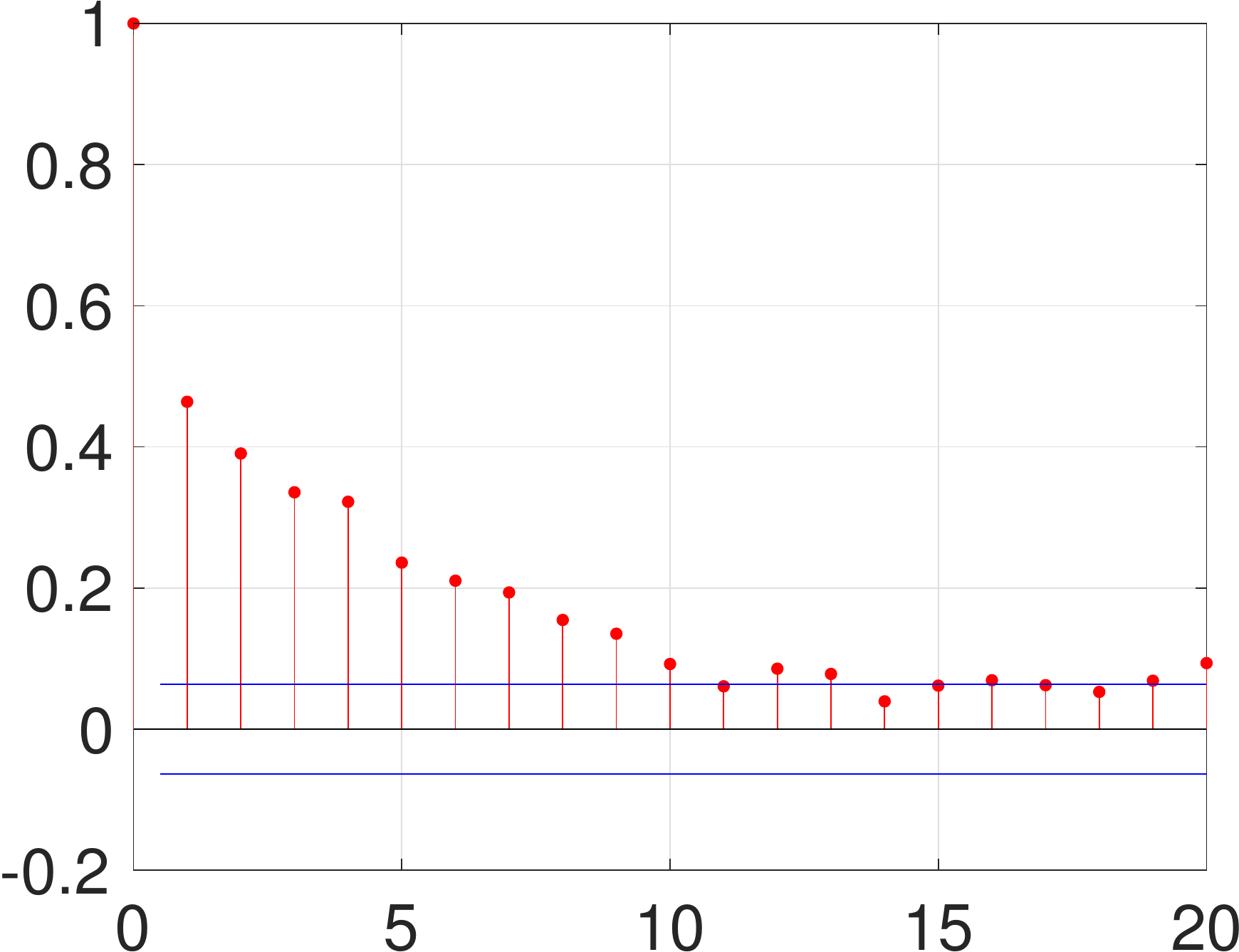} 
\end{minipage} & \begin{minipage}{.5\textwidth}
\includegraphics[width=3.5cm, height=3cm]{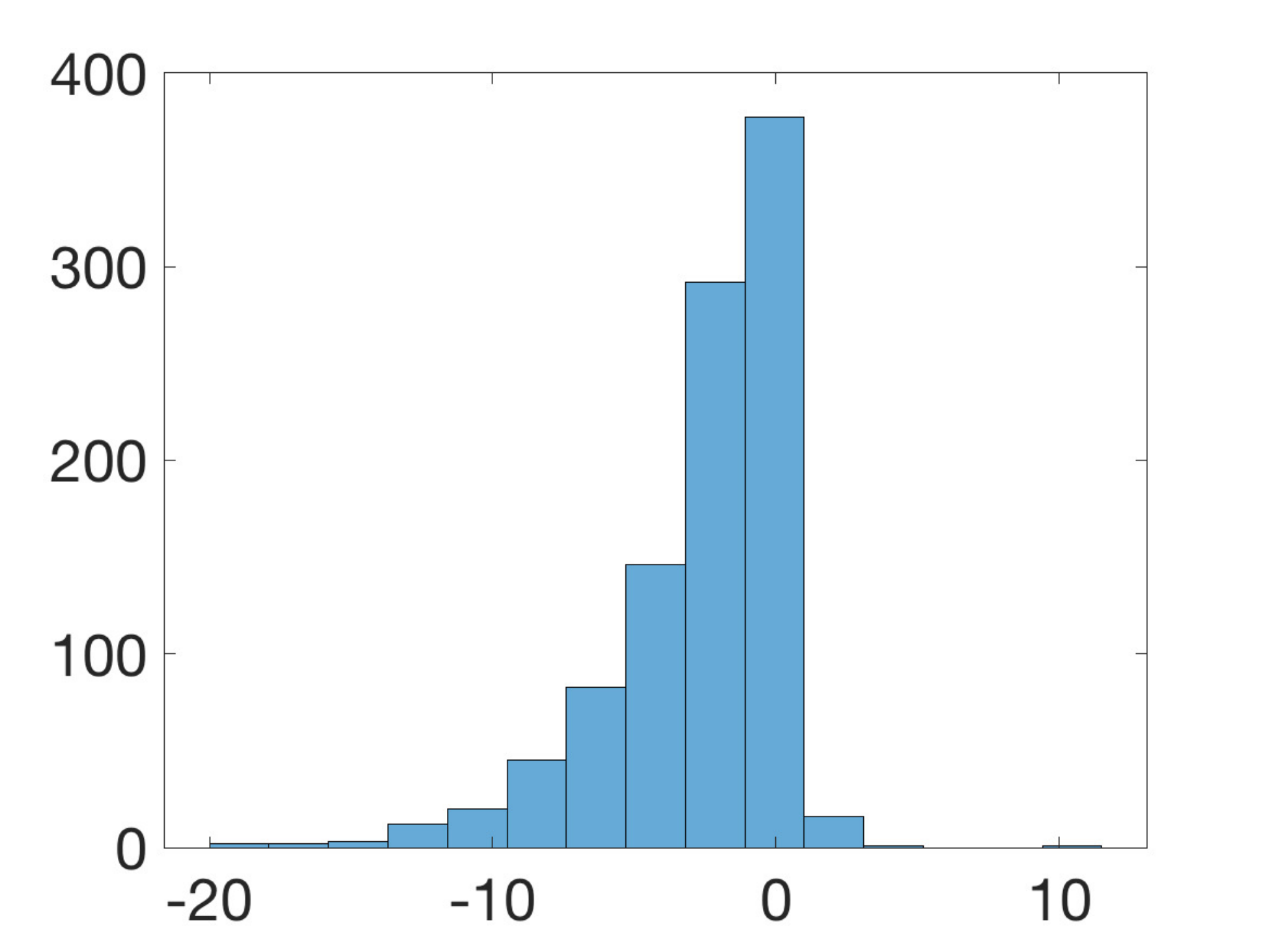} 
\includegraphics[width=3.5cm, height=3cm]{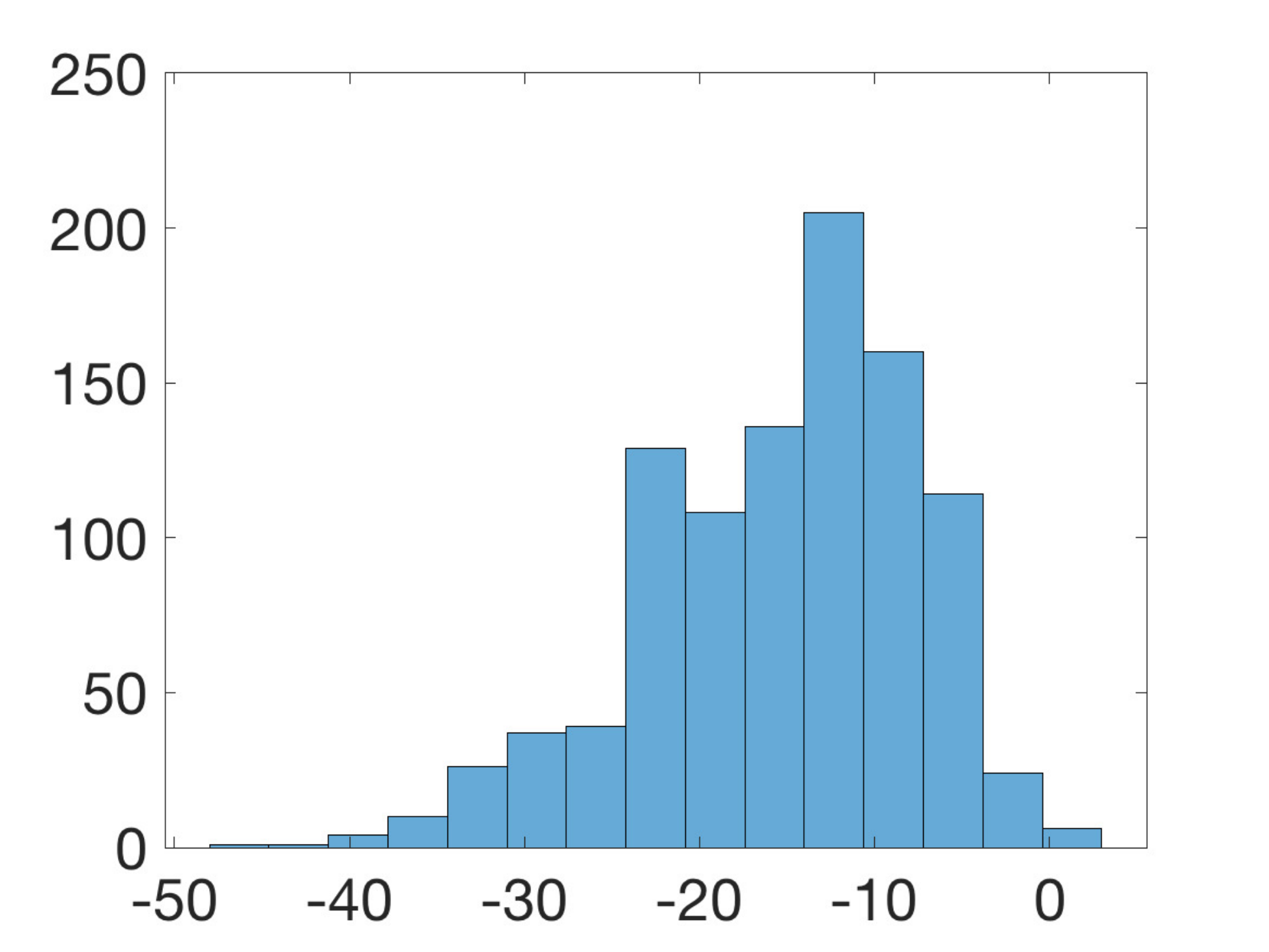} 
 \end{minipage}\\
 
{\small (I.e) \hspace{24 mm} (I.f)} & {\small (II.e) \hspace{24 mm} (II.f)} \\
\begin{minipage}{.5\textwidth}
\includegraphics[width=3.5cm, height=3cm]{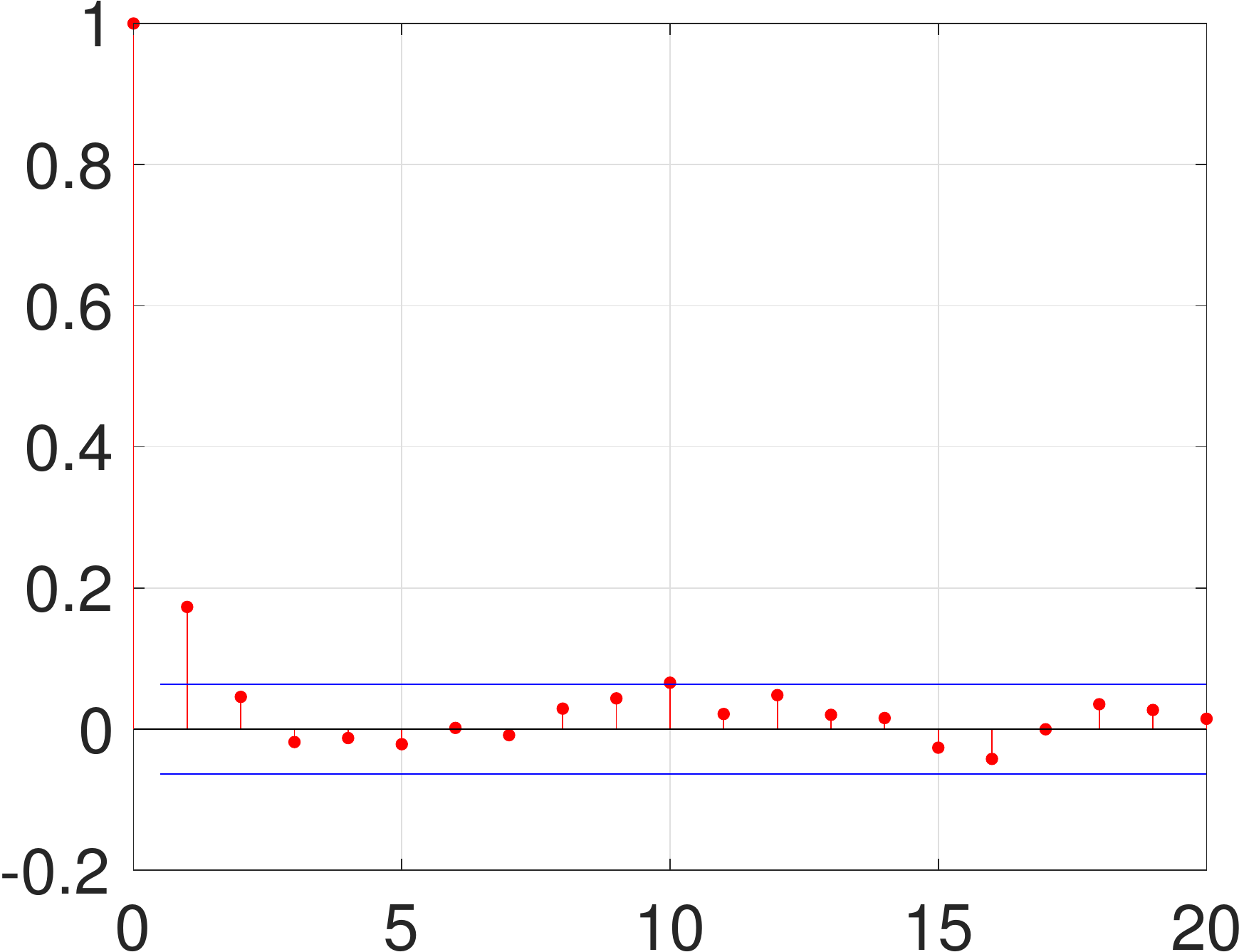} 
\includegraphics[width=3.5cm, height=3cm]{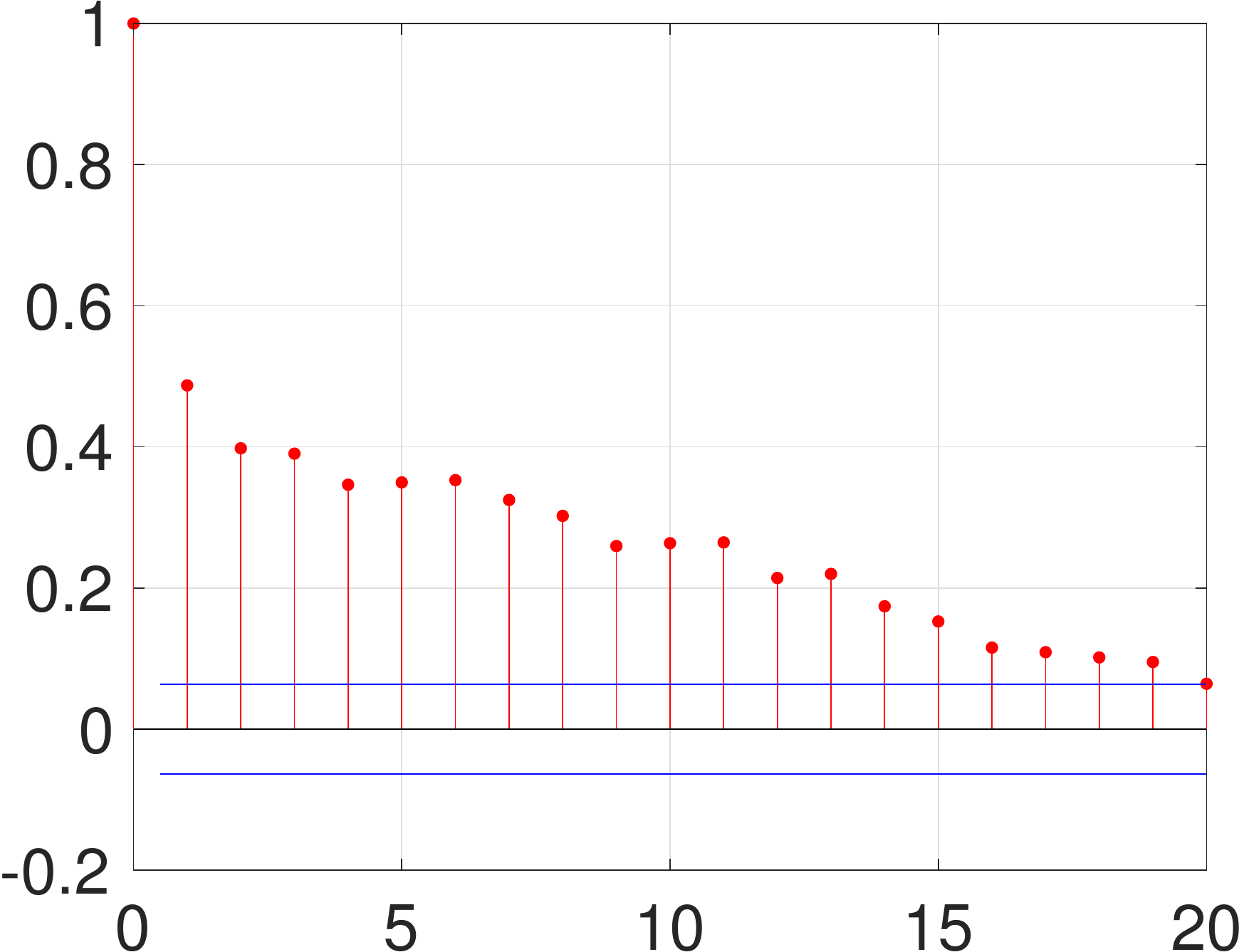} 
\end{minipage} & \begin{minipage}{.5\textwidth}
\includegraphics[width=3.5cm, height=3cm]{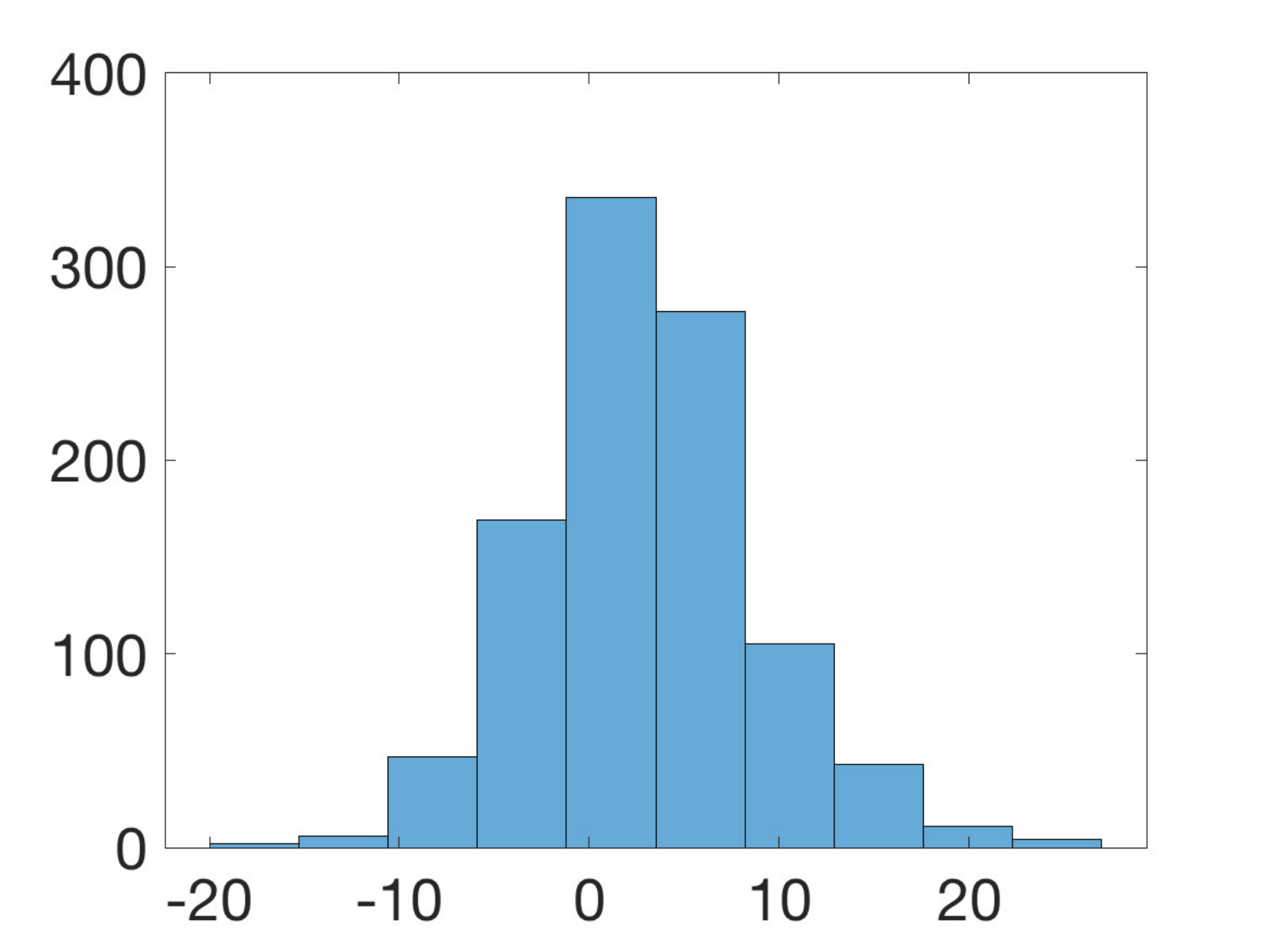} 
\includegraphics[width=3.5cm, height=3cm]{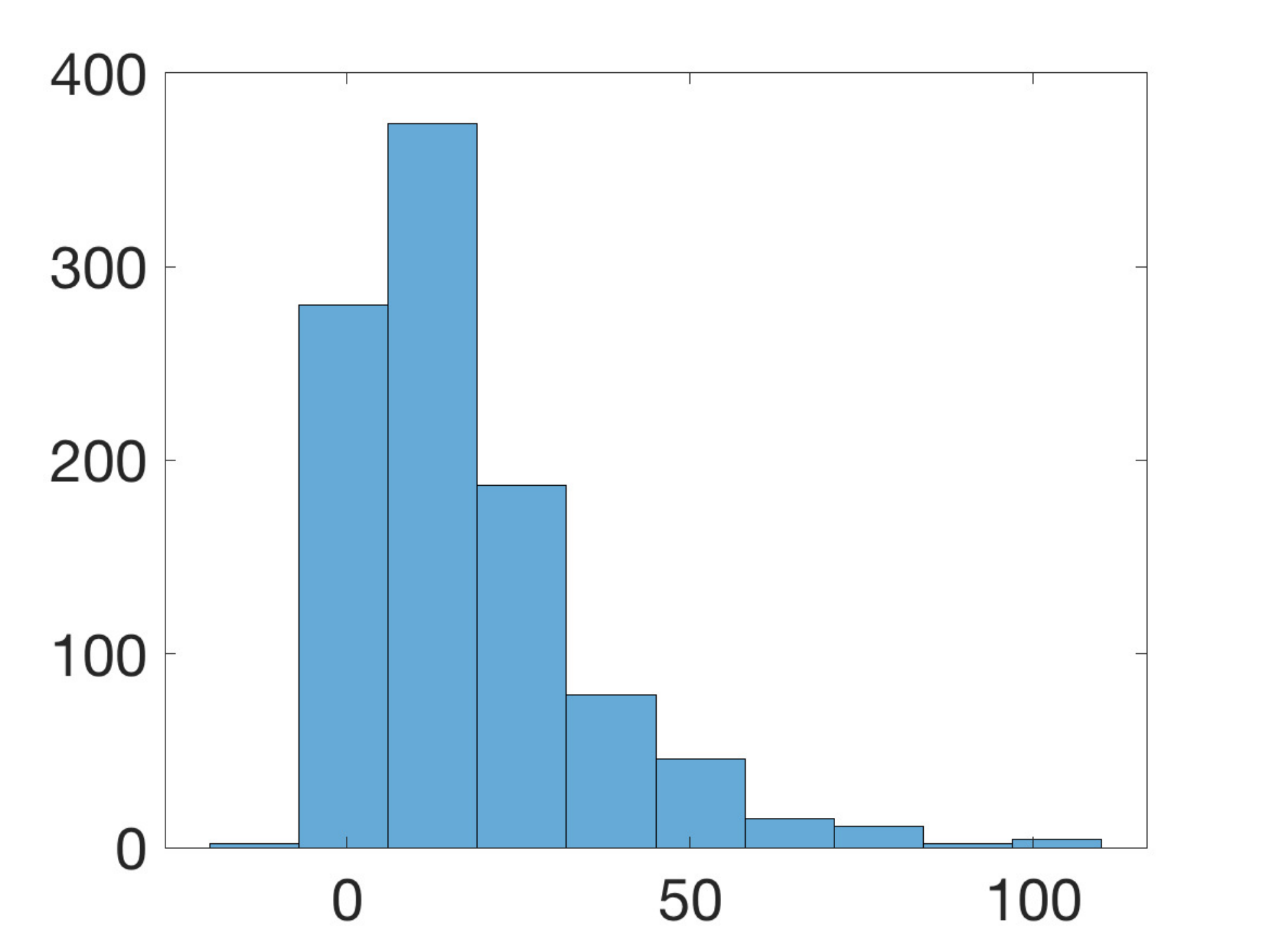}  \end{minipage}\\
\end{tabular}
\caption{Autocorrelation functions (Panel I) and unconditional distributions (Panel II) of $\lbrace Z_{t} \rbrace_{t \in \mathcal{Z}}$ for the different cases presented in Fig. \ref{GARCH11} (different columns in each panel).}
\label{GARCH12}
\end{figure}

We provide a necessary condition on the parameters $\alpha_{i}$ and $\beta_{j}$ that will ensure that a second-order stationary process has an INGARCH representation. First define the two following polynomials: $D(B)=1-\beta_{1}B-\ldots -\beta_{q}B^{q}$ and $G(B)=\alpha_{1}B+\ldots+\alpha_{p}B^{p}$, where B is the backshift operator. Assume the roots of $D(z)$ lie outside the unit circle. For non-negative $\beta_{j}$ this is equivalent to assume $D(1)=\sum_{j=1}^{q}\beta_{j}<1$. Then, the operator $D(B)$ has inverse $D^{-1}(B)$ and it is possible to write
\begin{equation}
\mu_{t}=D^{-1}(B)(\alpha_{0}+G(B)Z_{t})=\alpha_{0}D^{-1}(1)+H(B)Z_{t}
\end{equation}
where $H(B)=G(B)D^{-1}(B)=\sum_{j=1}^{\infty} \psi_{j}B^{j}$ and $\psi_{j}$ are given by the power expansion of the rational function $G(z)/D(z)$ in the neighbourhood of zero. If we denote $K(B)=D(B)-G(B)$ we can write the necessary condition as in the following proposition.
\begin{prop}\label{P1}
A necessary condition for a second-order stationary process $\lbrace Z_{t} \rbrace_{t \in \mathbb{Z}}$ to satisfy Eq. \ref{ingarch} is that $K(1)=D(1)-G(1) > 0$ or equivalently $\sum_{i=1}^{p}\alpha_{i} +\sum_{j=1}^{q}\beta_{j}<1$.
\end{prop}

\begin{proof}
See Appendix \ref{App:proofs}
\end{proof}

\section{Properties of the GPD-INGARCH}\label{Sec:PropModel}
We study the properties of the process by exploiting a suitable thinning representation following the strategy in \cite{Fer2006} and \cite{Zhu2012} for Poisson and Generalized Poisson INGARCH, respectively. We use the quasi-binomial thinning as defined in \cite{Wei2008} and the thinning difference (\cite{CunVasBou2018}) operators.

\subsection{Thinning representation}
We show that the INGARCH process can be obtained as a limit of successive approximations. Let us define: 
\begin{equation}\label{seqX}
X_{t}^{(n)}=\begin{cases} 0, & n<0 \\
(1-\lambda) U_{1t}, & n=0\\
(1-\lambda) U_{1t}+(1-\lambda)\sum_{i=1}^{n} \sum_{j=1}^{\frac{X_{t-i}^{(n-i)}}{(1-\lambda)}} V_{1t-i,i,j}, & n>0
\end{cases}
\end{equation}
and 
\begin{equation}\label{seqY}
Y_{t}^{(n)}=\begin{cases} 0, & n<0 \\
(1-\lambda) U_{2t}, & n=0\\
(1-\lambda) U_{2t}+(1-\lambda)\sum_{i=1}^{n} \sum_{j=1}^{\frac{Y_{t-i}^{(n-i)}}{(1-\lambda)}} V_{2t-i,i,j}, & n>0
\end{cases}
\end{equation}
where $\lbrace U_{1t} \rbrace_{t \in \mathbb{Z}}$ and $\lbrace U_{2t} \rbrace_{t \in \mathbb{Z}}$ are sequences of independent GP random variables and for each $t \in \mathbb{Z}$ and $i \in \mathbb{N}$, $\lbrace V_{1t,i,j} \rbrace_{j \in \mathbb{N}}$ and $\lbrace V_{2t,i,j} \rbrace_{j \in \mathbb{N}}$ represent two sequences of independent integer random variables. Moreover, assume that all the variables $U_{s}$, $V_{t,i,j}$, with $s \in \mathbb{Z}$, $t \in \mathbb{Z}$, $i \in \mathbb{N}$ and $j \in \mathbb{N}$, are mutually independent.

It is possible to show that $X_{t}^{(n)}$  and $Y_{t}^{(n)}$ have a thinning representation. We define a suitable thinning operation, first used by \cite{AlzAlO1993} and follow the notation in \cite{Wei2008}, let $\rho_{\theta,\lambda} \circ$ be the quasi-binomial thinning operator, such that it follows a QB($\rho$,$\theta /\lambda$,$x$).

\begin{prop}
If $X$ follows a GP($\lambda$,$\theta$) distribution and the quasi-binomial thinning is performed independently on $X$, then $\rho_{\theta,\lambda} \circ X$ has a GP($\rho \lambda$,$\theta$) distribution.
\end{prop}
\begin{proof}
See \cite{AlzAlO1993}.
\end{proof}

Both $X_{t}^{(n)}$ and $Y_{t}^{(n)}$ in Eq. \ref{seqX} and \ref{seqY} admit the representation
\begin{equation}
X_{t}^{(n)}= (1-\lambda) U_{1t}+(1-\lambda)\sum_{i=1}^{n} \varphi_{1i}^{(t-i)} \circ \left( \frac{X_{t-i}^{(n-i)}}{1-\lambda}\right) , \; n>0
\end{equation}
and
\begin{equation}
Y_{t}^{(n)}= (1-\lambda) U_{2t}+(1-\lambda)\sum_{i=1}^{n} \varphi_{2i}^{(t-i)} \circ \left( \frac{Y_{t-i}^{(n-i)}}{1-\lambda}\right) , \; n>0
\end{equation}
where $\varphi \circ X$ is the quasi-binomial thinning operation. See Appendix \ref{App:distrib} for a definition.

In the following we introduce the thinning difference operator and show that $Z_{t}^{(n)}=X_{t}^{(n)}-Y_{t}^{(n)}$ has a thinning representation.
\begin{defn}
Let $X\sim GP(\theta_{1},\lambda)$ and $Y\sim GP(\theta_{2},\lambda)$ be two independent random variables and $Z=X-Y$, then $Z\sim GPD(\mu,\sigma^{2},\lambda)$, with $\mu=\theta_{1}-\theta_{2}$ and $\sigma^{2}=\theta_{1}+\theta_{2}$. We define the new operator $\diamond$ as:
\begin{equation}\label{thin}
\rho \diamond Z\vert Z  \overset{\text{d}}{=} (\rho_{\theta_{1},\lambda}\circ X)-(\rho_{\theta_{2},\lambda}\circ Y)\vert (X-Y)
\end{equation}
where $(\rho_{\theta_{1},\lambda}\circ X)$ and $(\rho_{\theta_{2},\lambda}\circ Y)$ are the quasi-binomial thinning operations such that $(\rho_{\theta_{1},\lambda}\circ X)\vert X=x \sim QB(p,\lambda / \theta_{1},x)$ and $(\rho_{\theta_{2},\lambda}\circ Y)\vert Y=y \sim QB(p,\lambda / \theta_{2},y)$. The symbol ``$A\overset{\text{d}}{=} B$" means that the random variables $A$ and $B$ have the same distribution.
\end{defn}

See \cite{CunVasBou2018} for an application of the thinning operation to GPD-INAR processes and Appendix \ref{App:QBth}  for further details. Using the new operator as defined in Eq. \ref{thin}, we can represent $Z_{t}^{(n)}$ as follows.
\begin{prop}\label{Zrepr}
The process $Z_{t}^{(n)}=X_{t}^{(n)}-Y_{t}^{(n)}$ has the representation:
\begin{equation}\label{Zthin}
Z_{t}^{(n)}= (1-\lambda) U_{t}+(1-\lambda)^{2}\sum_{i=1}^{n} \varphi_{i}^{(t-i)} \diamond \left( \frac{Z_{t-i}^{(n-i)}}{1-\lambda}\right), \; n>0
\end{equation}
where $\varphi_{i}^{(\tau)} \diamond$ indicates the sequence of random variables with mean $\psi_{i}/(1-\lambda),$ involved in the thinning operator at time $\tau$ and $\lbrace U_{t} \rbrace_{t \in \mathbb{Z}}$ is a sequence of independent GPD random variables with mean $\psi_{0}/(1-\lambda)$ with $\psi_{0}=\alpha_{0}/D(1)$.
\end{prop}
\begin{proof}
See Appendix \ref{App:proofs}
\end{proof}

The proposition above shows that $Z_{t}^{(n)}$ is obtained through a cascade of thinning operations along the sequence $\lbrace U_{t} \rbrace_{t \in \mathbb{Z}}$. For example:
\begin{equation*}
Z_{t}^{(0)} = (1-\lambda)U_{t}
\end{equation*}
\begin{equation*}
\begin{split}
Z_{t}^{(1)} &= (1-\lambda) U_{t} +(1-\lambda)^{2}\left[ \varphi_{1}^{(t-1)} \diamond (Z_{t-1}^{(1-1)}/(1-\lambda))\right] \\
&= (1-\lambda)U_{t}+(1-\lambda)^{2}( \varphi_{1}^{(t-1)} \diamond U_{t-1})
\end{split}
\end{equation*}
\begin{equation*}
\begin{split}
Z_{t}^{(2)} &= (1-\lambda)U_{t} +(1-\lambda)^{2} \left[  \varphi_{1}^{(t-1)} \diamond (Z_{t-1}^{(2-1)}/(1-\lambda))+ \varphi_{2}^{(t-2)} \diamond (Z_{t-2}^{(2-2)}/(1-\lambda)) \right] \\
&= (1-\lambda)U_{t}+(1-\lambda)^{2} \left[ \varphi_{1}^{(t-1)}\diamond U_{t-1} + \varphi_{1}^{(t-1)} \diamond (\varphi_{1}^{(t-2)} \diamond U_{t-2})+\varphi_{2}^{(t-2)} \diamond U_{t-2} \right]
\end{split}
\end{equation*}
\begin{equation*}
\begin{split}
Z_{t}^{(3)} &= (1-\lambda)U_{t} + (1-\lambda)^{2} [ \varphi_{1}^{(t-1)} \diamond (Z_{t-1}^{(3-1)}/(1-\lambda))+ \varphi_{2}^{(t-2)} \diamond (Z_{t-2}^{(3-2)}/(1-\lambda))+\\
&+  \varphi_{3}^{(t-3)} \diamond (Z_{t-3}^{(3-3)}/(1-\lambda)) ]\\
&= (1-\lambda)U_{t}+(1-\lambda)^{2} [ \varphi_{1}^{(t-1)}\diamond U_{t-1} + \varphi_{1}^{(t-1)} \diamond (\varphi_{1}^{(t-2)} \diamond U_{t-2})+\\
&+ \varphi_{2}^{(t-2)} \diamond U_{t-2}+\varphi_{1}^{(t-1)} \diamond (\varphi_{1}^{
(t-2)}\diamond(\varphi_{1}^{(t-3)}\diamond U_{t-3}))+\\
&+\varphi_{1}^{(t-1)} \diamond (\varphi_{2}^{(t-3)} \diamond U_{t-3}) + \varphi_{2}^{(t-2)} \diamond (\varphi_{1}^{(t-3)} \diamond U_{t-3}) + \varphi_{3}^{(t-3)} \diamond U_{t-3} ] .
\end{split}
\end{equation*}

 Since $Z_{t}^{(n)}$ is a finite weighted sum of independent GPD random variables, the expected value and the variance of $Z_{t}^{(n)}$ are well defined. Moreover, it can be seen that $E[Z_{t}^{(n)}]$ does not depend on $t$ but only on $n$, hence it can be denoted as $\mu_{n}$. Using Proposition \ref{Zrepr} and $\mu_{k}=0$ if $k<0$, it is possible to write $\mu_{n}$ as follows
\begin{equation}
\begin{split}
\mu_{n}&=(1-\lambda)E[U_{t}]+(1-\lambda)^{2}\sum_{i=1}^{n} E\left[ \varphi_{i}^{(t-i)}\diamond \left( \frac{Z_{t-i}^{(n-i)}}{1-\lambda}\right)\right] \\
&= \psi_{0}+\sum_{j=1}^{\infty} \psi_{j}\mu_{n-j} = D^{-1}(B)\alpha_{0}+H(B)\mu_{n}
\end{split} 
\end{equation}
from which it follows $D(B)\mu_{n}=G(B)\mu_{n} + \alpha_{0}\Leftrightarrow K(B)\mu_{n}=\alpha_{0}$, where $K(B)=D(B)-G(B)$.
From the last equation it can be seen that the sequence $\lbrace \mu_{n}\rbrace$ satisfies a finite difference equation with constant coefficients. The characteristic polynomial is $K(z)$ and all its roots lie outside the unit circle if $K(1)>0$. Under the assumption $K(1)>0$, the following holds true.
\begin{prop}\label{asprop}
If $K(1)>0$ then the sequence $\lbrace Z_{t}^{(n)} \rbrace_{n \in \mathbb{N}}$ has an almost sure limit.
\end{prop}
\begin{proof}
See Appendix \ref{App:proofs}.
\end{proof}
\begin{prop}\label{msprop}
If $K(1) > 0$ then the sequence $\lbrace Z_{t}^{(n)} \rbrace$ has a mean-square limit.
\end{prop}
\begin{proof}
See Appendix \ref{App:proofs}.
\end{proof}
\subsection{Stationarity}

Given Proposition \ref{asprop}, if we can show that $\lbrace Z_{t}^{(n)} \rbrace$ is a strictly stationary process, for any given $n$, then also its almost sure limit $\lbrace Z_{t} \rbrace_{t \in \mathbb{Z}}$ will be a strictly stationary process. In order to show stationarity for $\lbrace Z_{t}^{(n)} \rbrace$, we follow a procedure similar to the one in \cite{Fer2006}. Let us define the probability generating function (pgf) $g_{\mathbf{W}}(\mathbf{t})$ of the random vector $\mathbf{W}=(W_{1},\ldots,W_{k})$
\begin{equation}
g_{\mathbf{W}}(\mathbf{t}) = E\left[ \prod_{i=1}^k t_{i}^{W_{i}}\right] = \sum_{\mathbf{W}\in \mathbb{N}^{k}} p(\mathbf{w})\prod_{i=1}^k t_{i}^{W_{i}}
\end{equation}
where $p(\mathbf{W})=Pr(\mathbf{W}=(W_{1}, \ldots, W_{k})')$ and $\mathbf{t}=(t_{1}, \ldots, t_{k})' \in \mathbb{C}^{k}$.
The probability generating function has the following properties.

\begin{prop}\label{prop:pgf}
Let $\mathbf{Z}_{1\ldots k}^{(n)}=(Z_{1}^{(n)},\ldots, Z_{k}^{(n)})$ be a subsequence of $\lbrace Z_{t}^{(n)}\rbrace_{t \in \mathbb{Z}}$ where, without loss of generality, we choose the first $k$ periods. Let $\mathbf{X}_{1\ldots k}^{(n)}=(X_{1}^{(n)},\ldots, X_{k}^{(n)})$ and $\mathbf{Y}_{1\ldots k}^{(n)}=(Y_{1}^{(n)},\ldots, Y_{k}^{(n)})$ be such that $\mathbf{Z}_{1\ldots k}^{(n)} = (\mathbf{X}_{1\ldots k}^{(n)} - \mathbf{Y}_{1\ldots k}^{(n)})'$ then
\begin{equation}\label{Eq:diffpgf}
g_{\mathbf{Z}_{1\ldots k}}(\mathbf{t})=g_{\mathbf{X}_{1\ldots k}}(\mathbf{t})g_{\mathbf{Y}_{1\ldots k}}(\mathbf{t^{-1}})
\end{equation}
\end{prop}
\begin{proof} 
See Appendix \ref{App:proofs}
\end{proof}
Using the probability generating function, in the following we know the stationarity of the process.
\begin{prop}
$\lbrace Z_{t}^{(n)} \rbrace_{t \in \mathbb{Z}}$ is a strictly stationary process, for any fixed value of n.
\end{prop}

\begin{proof}
Let $k$ and $h$ be two positive integers. As pointed out by \cite{Fer2006}, \cite{BroDav1991} show that to prove strictly stationarity we only need to show that
\begin{equation}\label{eqjd}
\mathbf{Z}_{1+h\ldots k+h}^{(n)} = (Z_{1+h}^{(n)},\ldots ,Z_{k+h}^{(n)})' \;\mbox{   and   } \; \mathbf{Z}_{1\ldots k}^{(n)} = (Z_{1}^{(n)},\ldots ,Z_{k}^{(n)})'
\end{equation}
have the same joint distribution, where we can rewrite both vectors in Eq. \ref{eqjd} as
\begin{equation}
\begin{split}
\mathbf{Z}_{1+h\ldots k+h}^{(n)} &= (\mathbf{X}_{1+h\ldots k+h}^{(n)} - \mathbf{Y}_{1+h\ldots k+h}^{(n)})'\\
&= ((X_{1+h}^{(n)}- Y_{1+h}^{(n)}),\ldots ,(X_{k+h}^{(n)}-Y_{k+h}^{(n)}))'
\end{split}
\end{equation}
and
\begin{equation}
\begin{split}
\mathbf{Z}_{1\ldots k}^{(n)} &= (\mathbf{X}_{1\ldots k}^{(n)} - \mathbf{Y}_{1\ldots k}^{(n)})'\\
&= ((X_{1}^{(n)}- Y_{1}^{(n)}),\ldots ,(X_{k}^{(n)}-Y_{k}^{(n)}))'
\end{split}
\end{equation}
To show that the two vectors have the same probability generating function, we first write the pgfs of $X$, $Y$ and $Z$ as shown above.
\begin{equation}
\begin{split}
g_{\mathbf{X}_{1\ldots k}^{(n)}}(\mathbf{t}) &= E\left[ \prod_{j=1}^{k}t_{j}^{X_{j}^{(n)}}\right]= E \left[ E_{\mathbf{X}_{1 \ldots k}^{(n)} \vert \mathbf{U}_{1,1-n \ldots k}}\left[ \prod_{j=1}^{k}t_{j}^{X_{j}^{(n)}}\right] \right] \\
&= \sum_{\mathbf{v_{1}} \in \mathbb{N}^{(k+n)}} E_{\mathbf{X}_{1 \ldots k}^{(n)} \vert \mathbf{U}_{1,1-n \ldots k}=\mathbf{v}_{1}}\left[\prod_{j=1}^{k}t_{j}^{X_{j}^{(n)}}\right] Pr\left( \mathbf{U}_{1,1-n \ldots k}=\mathbf{v}_{1} \right) 
\end{split}
\end{equation}
\begin{equation}
\begin{split}
g_{\mathbf{Y}_{1\ldots k}^{(n)}}(\mathbf{t})
= \sum_{\mathbf{v_{2}} \in \mathbb{N}^{(k+n)}} E_{\mathbf{Y}_{1 \ldots k}^{(n)} \vert \mathbf{U}_{2,1-n \ldots k}=\mathbf{v}_{2}}\left[\prod_{j=1}^{k}t_{j}^{Y_{j}^{(n)}}\right] Pr\left( \mathbf{U}_{2,1-n \ldots k}=\mathbf{v}_{2} \right) 
\end{split}
\end{equation}

\begin{equation}
G_{\mathbf{Z}_{1\ldots k}^{(n)}}(\mathbf{t}) 
= \sum_{\mathbf{v} \in \mathbb{N}^{(k+n)}} E_{\mathbf{Z}_{1 \ldots k}^{(n)} \vert \mathbf{U}_{1-n \ldots k}=\mathbf{v}}\left[ \prod_{j=1}^{k}t_{j}^{(X_{j}^{(n)}-Y_{j}^{(n)})}\right] Pr\left( \mathbf{U}_{1-n \ldots k}=\mathbf{v} \right) 
\end{equation}

By the thinning representation, for a any given value $\mathbf{u}_{1,t-n\ldots t+k} = (u_{1,t-n}, \ldots , u_{1,t+k})'$ of the vector $\mathbf{U}_{1,t-n \ldots t+k}= (U_{1,t-n}, \ldots , U_{1,t+k})'$ and $\mathbf{u}_{2,t-n\ldots t+k} = (u_{2,t-n}, \ldots , u_{2,t+k})'$ of the vector $\mathbf{U}_{2,t-n \ldots t+k}= (U_{2,t-n}, \ldots , U_{2,t+k})'$, the components of the vectors $(X_{1}^{(n)},\ldots ,X_{k}^{(n)})'$ and $(Y_{1}^{(n)},\ldots ,Y_{k}^{(n)})'$ are computed using a set of well-determined variables from the sequences $\mathcal{V}_{1,\tau,\eta}$ and $\mathcal{V}_{2,\tau,\eta}$, where $\tau=t-n, \ldots , t+k-1$  and $\eta = 1, \ldots , n$. Therefore, if $\mathbf{U}_{1,t-n \ldots t+k}$ and $\mathbf{U}_{1,t-n+h \ldots t+k+h}$ are both fixed to the same value $\mathbf{v_{1}}$ and $\mathbf{U}_{2,t-n \ldots t+k}$ and $\mathbf{U}_{2,t-n+h \ldots t+k+h}$ are both fixed to the same value $\mathbf{v_{2}}$, it follows that the conditional distribution of
\begin{equation*}
\mathbf{Z}_{1+h\ldots k+h}^{(n)} = ((X_{1+h}^{(n)}- Y_{1+h}^{(n)}),\ldots ,(X_{k+h}^{(n)}-Y_{k+h}^{(n)}))'
\end{equation*}
and 
\begin{equation*}
\mathbf{Z}_{1\ldots k}^{(n)} = ((X_{1}^{(n)}- Y_{1}^{(n)}),\ldots ,(X_{k}^{(n)}-Y_{k}^{(n)}))'
\end{equation*}
given $\mathbf{U}_{t-n \ldots t+k}$ and $\mathbf{U}_{t-n+h \ldots t+k+h}$, are the same. Accordingly, 
\begin{equation*}
E_{\mathbf{Z}_{1+h \ldots k+h}^{(n)} \vert \mathbf{U}_{1-n+h \ldots k+h}=\mathbf{v}}\left[\prod_{j=1}^{k}t_{j}^{Z_{j+h}^{(n)}}\right] = E_{\mathbf{Z}_{1 \ldots k}^{(n)} \vert \mathbf{U}_{1-n \ldots k}=\mathbf{v}}\left[\prod_{j=1}^{k}t_{j}^{Z_{j}^{(n)}}\right]
\end{equation*}
and, since 
\begin{equation*}
Pr\left( \mathbf{U}_{1-n+h \ldots k+h}=\mathbf{v} \right) =Pr\left( \mathbf{U}_{1-n \ldots k}=\mathbf{v} \right),
\end{equation*}
it is possible to write
\begin{equation*}
\begin{split}
g_{\mathbf{Z}_{1\ldots k}^{(n)}}(\mathbf{t}) &=\sum_{\mathbf{v} \in \mathbb{Z}^{(k+n)}}  E_{\mathbf{Z}_{1+h \ldots k+h}^{(n)} \vert \mathbf{U}_{1-n+h \ldots k+h}=\mathbf{v}}\left[\prod_{j=1}^{k}t_{j}^{Z_{j+h}^{(n)}}\right] Pr\left( \mathbf{U}_{1-n+h \ldots k+h}=\mathbf{v} \right)\\
&= g_{\mathbf{Z}_{1+h \ldots k+h}^{(n)}}(\mathbf{t})
\end{split}
\end{equation*}
and claim that $\mathbf{Z}_{1+h\ldots k+h}^{(n)}$ and $\mathbf{Z}_{1\ldots k}^{(n)}$ have the same joint distribution.
\end{proof}

\begin{prop}
The process $\lbrace Z_{t}\rbrace_{t \in \mathbb{Z}}$ is a strictly stationary process.
\end{prop}

\begin{prop}\label{prop:moments}
The first two moments of $\lbrace Z_{t}\rbrace_{t \in \mathbb{Z}}$ are finite.
\end{prop}
\begin{proof}
See Appendix \ref{App:proofs}.
\end{proof}

\subsection{Conditional law of $\lbrace Z_{t}^{(n)}\rbrace_{t \in \mathbb{Z}}$ given $\mathcal{F}_{t-1}$}
To verify that the distributional properties of the sequence are satisfied, we will follow the same arguments in \cite{Fer2006} adjusted for our sequence. Given $\mathcal{F}_{t-1}=\sigma (\lbrace Z_{u}\rbrace_{u\leq t-1})$, for $t \in \mathbb{Z}$, let
\begin{equation*}
\mu_{t}=\alpha_{0}D^{-1}(1)+\sum_{j=1}^{n} \psi_{j}Z_{t-j}.
\end{equation*}
The sequence $\lbrace \mu_{t} \rbrace$ satisfies
\begin{equation}
\mu_{t}=\alpha_{0} + \sum_{i=1}^{p} \alpha_{i}Z_{t-i}+\sum_{j=1}^{q} b_{j}\mu_{t-j}.
\end{equation}
Moreover, recalling that $Z_{t}=X_{t}-Y_{t}$, for a fixed t, we can consider three sequences, $\lbrace r_{1t}^{(n)} \rbrace_{n \in \mathbb{N}}$, $\lbrace r_{2t}^{(n)} \rbrace_{n \in \mathbb{N}}$ and $\lbrace r_{t}^{(n)} \rbrace_{n \in \mathbb{N}}$, defined by 
\begin{equation}
r_{1t}^{(n)}=(1-\lambda)U_{1t}+(1-\lambda)\sum_{i=1}^{n} \sum_{j=1}^{X_{t-i}} V_{1t-i,i,k}
\end{equation}
\begin{equation}
r_{2t}^{(n)}=(1-\lambda)U_{2t}+(1-\lambda)\sum_{i=1}^{n} \sum_{j=1}^{Y_{t-i}} V_{2t-i,i,k}.
\end{equation}
and
\begin{equation}
r_{t}^{(n)}=r_{1t}^{(n)}-r_{2t}^{(n)}.
\end{equation}
As claimed by \cite{Fer2006}, there is a subsequence $\lbrace n_{k} \rbrace$ such that $r_{t}^{(n_{k})}$ converges almost surely to $Z_{t}$. We know that
\begin{equation}\label{r1}
X_{t} - r_{1t}^{(n)} = (X_{t} - X_{t}^{(n)} ) + (X_{t}^{(n)} - r_{1t}^{(n)})
\end{equation}
and 
\begin{equation}\label{r2}
Y_{t} - r_{2t}^{(n)} = (Y_{t} - Y_{t}^{(n)} ) + (Y_{t}^{(n)} - r_{2t}^{(n)}).
\end{equation}
Since $X_{t}^{(n)} \overset{a.s.}{\longrightarrow} X_{t}$ and $Y_{t}^{(n)} \overset{a.s.}{\longrightarrow} Y_{t}$, we know that the first term in both Eq. \ref{r1} and \ref{r2} goes to zero. Therefore, we can write
\begin{equation}\label{rz}
\begin{split}
Z_{t} - r_{t}^{(n)} &= (X_{t} - Y_{t}) - (r_{1t}^{(n)} - r_{2t}^{(n)})\\
&= \left[ (X_{t} - X_{t}^{(n)})- (Y_{t} - Y_{t}^{(n)})\right] +  \left[ (X_{t}^{(n)} - r_{1t}^{(n)})- (Y_{t}^{(n)} - r_{2t}^{(n)})\right]\\
&= (Z_{t} - Z_{t}^{(n)})+ \left[ (X_{t}^{(n)} - Y_{t}^{(n)})- (r_{1t}^{(n)}- r_{2t}^{(n)})\right]\\
&=  (Z_{t} - Z_{t}^{(n)})+ \left[ Z_{t}^{(n)} - (r_{1t}^{(n)}- r_{2t}^{(n)})\right],
\end{split}
\end{equation}
and, as before, $(Z_{t} - Z_{t}^{(n)})$ goes to zero since we have proven almost sure convergence.\\
We have now to show that the second term in the last line of Eq. \ref{rz} goes to zero, for this purpose we need to find a sequence 
\begin{equation*}
W_{t}^{(n)}= (r_{1t}^{(n)}- r_{2t}^{(n)}) - Z_{t}^{(n)}
\end{equation*}
that converges almost surely to zero. For this reason it is more suitable to rewrite the previous sequence as follows
\begin{equation}
\begin{split}
W_{t}^{(n)} &= (r_{1t}^{(n)} - r_{2t}^{(n)}) - (X_{t} - Y_{t})\\
&= (r_{1t}^{(n)} - X_{t} ) - (r_{2t}^{(n)} - Y_{t})
\end{split}
\end{equation}
\cite{Fer2006} show that 
\begin{equation*}
\lim_{n \to \infty} E  \left[ (r_{1t}^{(n)} - X_{t}) \right]  = 0
\end{equation*}
\begin{equation}
\lim_{n \to \infty} E  \left[ (r_{2t}^{(n)} - Y_{t}) \right]  = 0
\end{equation}
therefore, we can conclude that also 
\begin{equation}\label{limZ}
\lim_{n \to \infty} E  \left[ (r_{t}^{(n)} - Z_{t}) \right]  = 0.
\end{equation}
Equation \ref{limZ} implies that $W_{t}^{(n)}$ converges to zero in $L^{1}$, therefore there exist a subsequence $W_{t}^{(n_{k})}$ converging almost surely to the same limit. From this it follows directly that the distributional properties of $X_{t}$ are satisfied.

Since $r_{1t}^{(n_{k})} \overset{a.s.}{\longrightarrow} X_{t}$ and $r_{2t}^{(n_{k})} \overset{a.s.}{\longrightarrow} Y_{t}$, it is also true $r_{t}^{(n_{k})} \overset{a.s.}{\longrightarrow} Z_{t}$. Hence, 
\begin{equation*}
r_{t}^{(n)}\vert \mathcal{F}_{t-1} \overset{a.s.}{\longrightarrow} Z_{t} \vert \mathcal{F}_{t-1}.
\end{equation*}
However, 
\begin{equation*}
r_{t}^{(n)}\vert \mathcal{F}_{t-1} = (r_{1t}^{(n)} - r_{2t}^{(n)})\vert \mathcal{F}_{t-1}
\end{equation*}
and from \cite{Zhu2012} we know that both $r_{1t}^{(n)}$ and $r_{2t}^{(n)}$ have a Generalized Poisson distribution. Since the difference of two GP distributed random variables is GPD distributed, we can write
\begin{equation}
r_{t}^{(n)}\vert \mathcal{F}_{t-1} \sim GPD \left( \alpha_{0}D^{-1}(1) + \sum_{j=1}^{n} \psi_{j} Z_{t-j}\right)
\end{equation}
and conclude that
\begin{equation}
Z_{t}\vert \mathcal{F}_{t-1} \sim GPD(\tilde{\mu}_{t},\tilde{\sigma}^{2}_{t},\lambda).
\end{equation}
\subsection{Moments of the GPD-INGARCH}
The conditional mean and variance of the process $Z_{t}$ are
\begin{equation*}
E(Z_{t}\vert  \mathcal{F}_{t-1}) = \frac{\tilde{\mu}_{t}}{1-\lambda} = \mu_{t}
\end{equation*}
\begin{equation}
V(Z_{t}\vert  \mathcal{F}_{t-1}) = \frac{\tilde{\sigma}^{2}_{t}}{1-\lambda}  = \phi^{3}\tilde{\sigma}^{2}_{t}
\end{equation}
where $\phi = \frac{1}{1-\lambda}$.

The unconditional mean and variance of the process are
\begin{equation*}
E(Z_{t})=\mu_{t} = \frac{\alpha_{0}}{1- \sum_{i=1}^{p} \alpha_{i} -  \sum_{j=1}^{q} \beta_{j}}
\end{equation*}
\begin{equation}\label{varZ}
\begin{split}
V(Z_{t}) &= E \left[ V \left( Z_{t}\vert  \mathcal{F}_{t-1} \right) \right] + V \left[ E \left( Z_{t}\vert  \mathcal{F}_{t-1} \right) \right]\\
&= E(\phi^{3}\tilde{\sigma}^{2}_{t}) + V(\mu_{t})\\
&=\phi^{3} E(\tilde{\sigma}^{2}_{t}) + V(\mu_{t})
\end{split}
\end{equation}
From Th. 1 in \cite{Wei2009} we know a set of equations from which the variance and autocorrelation function of the process can be obtained.
Suppose $Z_{t}$ follows the INGARCH(p,q) model in Eq. \ref{ingarch} with $\sum_{i=1}^{p} \alpha_{i} + \sum_{j=1}^{q} \beta_{j} < 0$. From Th. 1 part (iii) in \cite{Wei2009}, the autocovariances $\gamma_{Z}(k) = Cov[Z_{t},Z_{t-k}]$ and $\gamma_{\mu}(k) = Cov[\mu_{t},\mu_{t-k}]$ satisfy the linear equations
\begin{equation*}
\gamma_{Z}(k) = \sum_{i=1}^{p} \alpha_{i} \gamma_{Z}(|k-i|) + \sum_{j=1}^{\min(k-1,q)} \beta_{j} \gamma_{Z}(k-j) +\sum_{j=k}^{q} \beta_{j} \gamma_{\mu}(j-k), \;\;\;\; k\geq 1;
\end{equation*}
\begin{equation}\label{autocov}
\gamma_{\mu}(k) = \sum_{i=1}^{min(k,p)} \alpha_{i} \gamma_{\mu}(|k-i|) + \sum_{i=k+1}^{p} \alpha_{i} \gamma_{Z}(i-k) +\sum_{j=1}^{q} \beta_{j} \gamma_{\mu}(|k-j|), \;\;\;\; k\geq 0.
\end{equation}
In order to have an explicit expression for the variance of $\mu_{t}$ and $Z_{t}$ and for the autocorrelations, we consider two special cases as in \cite{Zhu2012} and \cite{Wei2009}. For a proof of the results in these examples, see Section \ref{Proofs:S3}.

\begin{example}[INARCH(1)]\label{Example1}
Consider the INARCH(1) model 
\begin{equation}
\mu_{t} = \alpha_{0} + \alpha_{1} Z_{t-1}
\end{equation}
then the linear equations in  Eq. \ref{autocov}, becomes
\begin{eqnarray*}
\gamma_{Z}(k) &=& \sum_{i=1}^{p} \alpha_{i} \gamma_{Z}(|k-i|) + \delta_{k0} \cdot \mu,\;\;\;\; k\geq 0\\
\gamma_{\mu}(k) &=& \sum_{i=1}^{\min(k,p)} \alpha_{i} \gamma_{\mu}(|k-i|) + \sum_{i=k+1}^{p} \alpha_{i} \gamma_{Z}(i-k), \;\;\;\; k\geq 0.
\end{eqnarray*}
Where the second equation comes from Example 2 in \cite{Wei2009}. We derive the following autocovariances
\begin{equation}
\gamma_{Z}(k)=\begin{cases} \alpha_{1}^{k-1} \gamma_{Z}(1), & \mbox{for } k\geq 2 \\ \alpha_{1}[\phi^{3} E(\tilde{\sigma}^{2}_{t})] + \alpha_{1}V(\mu_{t}), & \mbox{for } k=1
\end{cases}
\end{equation}
\begin{equation}
\gamma_{\mu}(k) =\begin{cases}\alpha_{1}^{k} V(\mu_{t}), & \mbox{for } k\geq 1 \\ \alpha_{1}^{2}[\phi^{3} E(\tilde{\sigma}^{2}_{t})]+\alpha_{1}^{2}V(\mu_{t}), & \mbox{for } k=0
\end{cases}
\end{equation}
Therefore, the variance of $\mu_{t}$ is
\begin{equation}
V(\mu_{t}) = \frac{\alpha_{1}^{2}[\phi^{3} E(\tilde{\sigma}^{2}_{t})]}{1-\alpha_{1}^{2}}
\end{equation}
and the variance of $Z_{t}$ is
\begin{equation}
V(Z_{t}) = \frac{\phi^{3} E(\tilde{\sigma}^{2}_{t})}{1-\alpha_{1}^{2}}
\end{equation}
where $\phi=\frac{1}{1-\lambda}$.\\
Lastly, the autocorrelations are
\begin{equation}
\rho_{\mu}(k) = \alpha_{1}^{k}
\end{equation}
\begin{equation}
\rho_{Z}(k)= \alpha_{1}^{k}
\end{equation}

\end{example}
\bigskip
\begin{example}[INGARCH(1,1)]\label{Example2}
Consider the INGARCH(1,1) model 
\begin{equation}
\mu_{t} = \alpha_{0} + \alpha_{1} Z_{t-1} + \beta_{1} \mu_{t-1}
\end{equation}
From Eq. \ref{autocov},
\begin{equation}
\gamma_{Z}(k)=\begin{cases} (\alpha_{1} + \beta_{1})^{k-1} \gamma_{Z}(1), & \mbox{for } k\geq 2 \\ \alpha_{1}[\phi^{3} E(\tilde{\sigma}^{2}_{t})] + (\alpha_{1}+\beta_{1})V(\mu_{t}), & \mbox{for } k=1
\end{cases}
\end{equation}
We can now determine $V(\mu_{t})$. First note that we have
\begin{equation}\label{covmu}
\gamma_{\mu}(k) =\begin{cases}(\alpha_{1} + \beta_{1})^{k} V(\mu_{t}), & \mbox{for } k\geq 1 \\ \alpha_{1}^{2}[\phi^{3} E(\tilde{\sigma}^{2}_{t})]+(\alpha_{1} + \beta_{1})^{2}V(\mu_{t}), & \mbox{for } k=0
\end{cases}
\end{equation}
where the second equation in Eq. \ref{covmu} is equal to $V(\mu_{t})$. From this latter equation, we can derive the expression for $V(\mu_{t})$
\begin{equation}\label{varmu}
V(\mu_{t}) = \frac{\alpha_{1}^{2}[\phi^{3} E(\tilde{\sigma}^{2}_{t})]}{1-(\alpha_{1} + \beta_{1})^{2}}
\end{equation}
Combining Eq. \ref{varZ} and \ref{varmu}, we can derive a close expression for the variance of $Z_{t}$:
\begin{equation}
V(Z_{t}) = \frac{\phi^{3} E(\tilde{\sigma}^{2}_{t})[1-(\alpha_{1} + \beta_{1})^{2}+\alpha_{1}^{2}]}{1-(\alpha_{1} + \beta_{1})^{2}}
\end{equation}
where $\phi=\frac{1}{1-\lambda}$.

The autocorrelations are given by
\begin{eqnarray}
\rho_{\mu}(k) &=& (\alpha_{1}+\beta_{1})^{k}\\
\rho_{Z}(k)&=& (\alpha_{1}+\beta_{1})^{k-1}\, \frac{\alpha_{1}[1-\beta_{1}(\alpha_{1}+\beta_{1})]}{1-(\alpha_{1} + \beta_{1})^{2}+\alpha_{1}^{2}}
\end{eqnarray}
\end{example}

\section{Bayesian Inference}\label{Sec:BayInf}
We propose a Bayesian approach to inference for GPD-INGARCH, which allows the researcher to include extra-sample information through the prior choice and allows us to exploit the stochastic representation of the GPD and the use of latent variables to make more tractable the likelihood function.

\subsection{Prior assumption}
We assume the following prior distributions. A Dirichlet prior distribution for $\boldsymbol{\varphi} =(\alpha_{1},\ldots, \alpha_{p},\beta_{1},\ldots,\beta_{q})$, $\boldsymbol{\varphi} \sim Dir_{d+1}(c)$, with density:
\begin{equation}
 \pi(\boldsymbol{\varphi})=\frac{\Gamma\left( \sum_{i=0}^{d}c_{i}\right)}{\prod_{i=0}^{d} \Gamma(c_{i})} \prod_{i=1}^{d} \varphi_{i}^{c_{i}-1} \left( 1-\sum_{i=1}^{d}\varphi_{i}\right)^{(c_{0}-1)}
\end{equation}
where $\varphi_{i} \geq 0$ and $\sum_{i=1}^{d} \varphi_{i} \leq 1$. Panel (a) in Fig. \ref{fig:PRIORS} provides the level sets of the joint density function of $\alpha_{1}$ and $\beta_{1}$ with hyper-parameters $c_{0}=3$, $c_{1}=4$ and $c_{2}=3$. We assume a flat prior for $\alpha_{0}$, i.e. $\pi(\alpha_{0}) \propto \mathbb{I}_{\mathbb{R}}(\alpha_{0})$. For $\lambda$ and $\phi$ we assume a joint prior distribution with uniform marginal prior $\lambda \sim \mathcal{U}_{[0,1]}$ and shifted gamma conditional prior $\phi \sim \mathcal{G}a^{*}(a,b,c)$, with density function:
\begin{equation}
\pi(\phi)=\frac{b^{a}}{\Gamma(a)} (\phi-c)^{(a -1)} e^{-b (\phi-c)} \;\;\;\;\; \mbox{for  } \phi > c
\end{equation}
where $c=(1-\lambda)^{-2}$. Panel (b) provides the level sets of the joint density function of $\phi$ and $\lambda$, with hyper-parameters $a=b=5$. The joint prior distribution of the parameters will be denoted by $\pi(\boldsymbol{\theta})=\pi(\boldsymbol{\varphi})\pi(\alpha_{0})\pi(\lambda)\pi(\phi)$.
\begin{figure}[t]
\centering
\begin{tabular}{cc}
$\alpha_{1}$ and $\beta_{1}$ & $\phi$ and $\lambda$ \\
\includegraphics[scale=0.35]{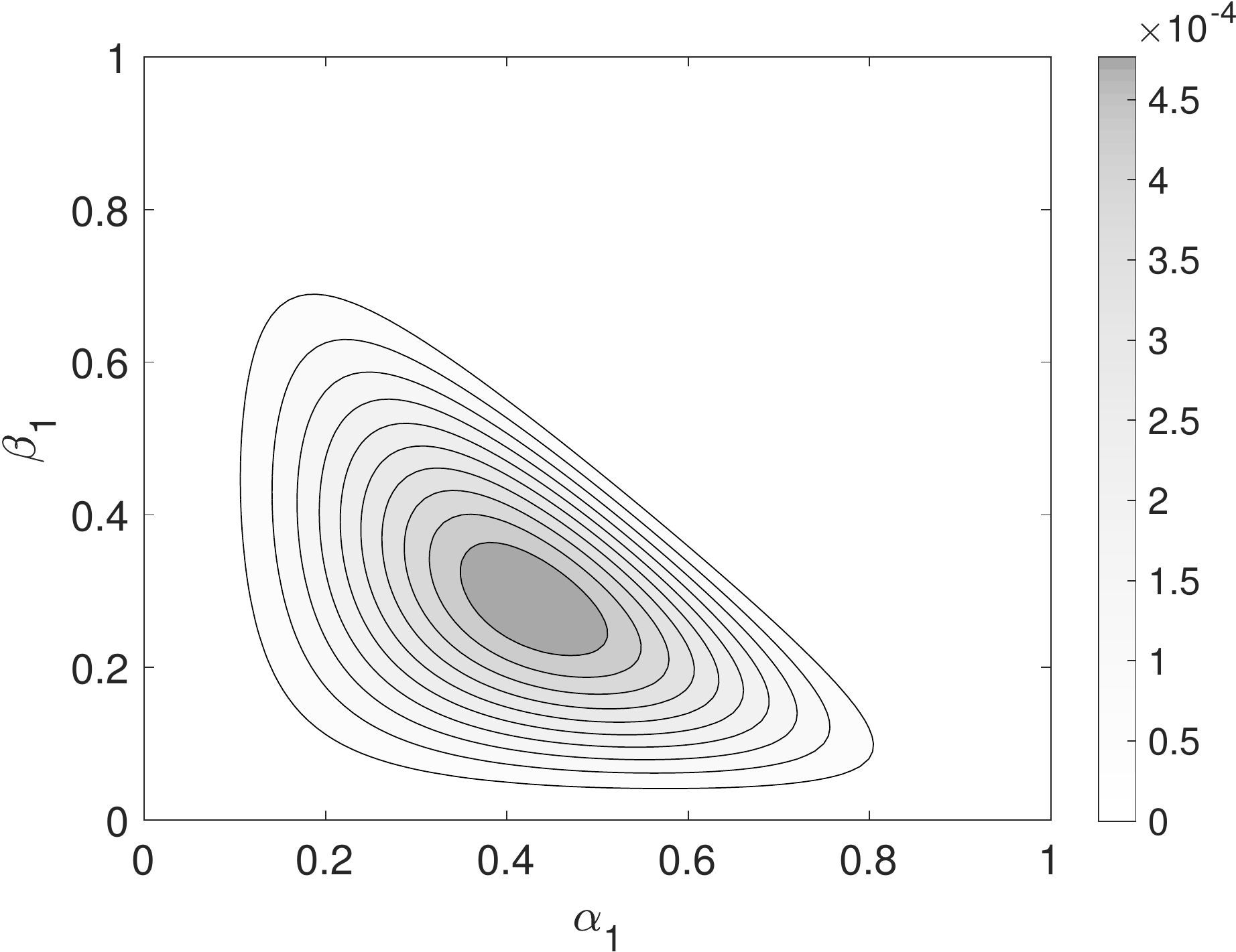} &
\includegraphics[scale=0.35]{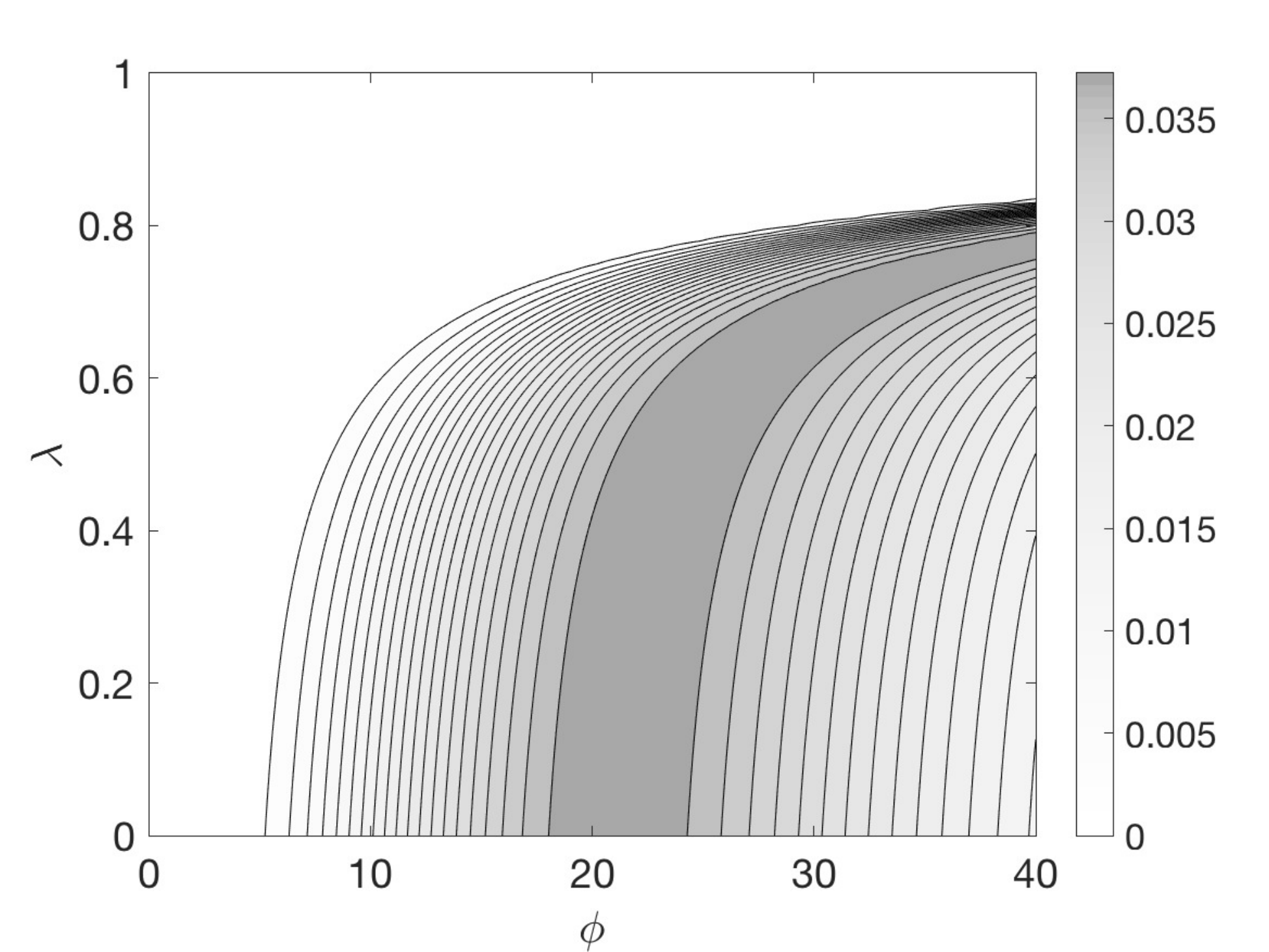} \\
\end{tabular}
\caption{Contour lines of the log-prior density function for $\alpha_{1}$ and $\beta_{1}$ (left) and $\phi$ and $\lambda$ (right). }
\label{fig:PRIORS}
\end{figure}

\subsection{Data augmentation}
Denote the probability distribution of $Z_{t}$ with
\begin{equation}
 f_t(Z_{t}=z \vert \boldsymbol{\theta}) =e^{-\sigma_{t}^{2}-z\lambda} \hspace{-2pt}\sum_{s=\underline{s}}^{+\infty}\!\! \frac{1}{4}\frac{\sigma_{t}^{4}+\mu_{t}^{2}}{s!(s+z)!}\!\!\left[\frac{\sigma^{2}_{t}+\mu_{t}}{2} +(s+z)\lambda\right]^{s+z-1}\!\! \left[ \frac{\sigma^{2}_{t}-\mu_{t}}{2} +s\lambda\right]^{s-1}\hspace{-8pt}e^{-2\lambda s}
\end{equation}
with $\underline{s}=\max(0,-z)$. Since the posterior distribution 
\begin{equation}
\pi(\boldsymbol{\theta} \vert Z_{1:T}) \propto \prod_{t=1}^{T}f_{t}(Z_{t} \vert \boldsymbol{\theta})\pi(\boldsymbol{\theta})
\end{equation}
is not analytically tractable we apply Markov Chain Monte Carlo (MCMC) for posterior approximation in combination with a data-augmentation approach \citep{TanWon1987}. See \cite{RobCas2013} for an introduction to MCMC. As in \cite{KarNtz2006}, we exploit the stochastic representation in Eq. \ref{StRep} and introduce two GPD latent variables $X_{t}$ and $Y_{t}$ with pmfs
\begin{equation}
 f_t(X_{t}=x \vert \theta_{1t},\lambda) = \frac{\theta_{1t}(\theta_{1t} + \lambda x)^{x-1}}{x!} e^{(-\theta_{1t} - \lambda x)}
 \end{equation}
 \begin{equation} f_t(Y_{t}=y \vert \theta_{2t},\lambda) ) = \frac{\theta_{2t}(\theta_{2t} + \lambda y)^{y-1}}{y!} e^{(-\theta_{2t} - \lambda y)}
\end{equation}

Let $Z_{1:T}=(Z_{1},\ldots ,Z_{T})$, $X_{1:T}=(X_{1},\ldots ,X_{T})$ and $Y_{1:T}=(Y_{1},\ldots ,Y_{T})$. The complete-data likelihood becomes
\begin{equation}
\begin{split}
f(Z_{1:T},X_{1:T}Y_{1:T} \vert \boldsymbol{\theta}) &=  \prod_{t=1}^{T} f(Z_{t} \vert X_{t}, Y_{t},\boldsymbol{\theta}) f_{t}(X_{t},Y_{t} \vert \boldsymbol{\theta})\\
&=\prod_{t=1}^{T}  \delta (Z_{t} - X_{t} + Y_{t})f_{t}(X_{t} \vert \boldsymbol{\theta})f_{t}(X_{t}-Z_{t} \vert \boldsymbol{\theta}).
\end{split}
\end{equation}
where $\delta (z-c)$ is the Dirac function which takes value 1 if $z=c$ and 0 otherwise. The joint posterior distribution of the parameters $\boldsymbol{\theta}$ and the two collections of latent variables $X_{1:T}$ and $Y_{1:T}$ is
\begin{equation}
\pi(X_{1:T},Y_{1:T},\boldsymbol{\theta} \vert Z_{1:T}) \propto f(Z_{1:T},X_{1:T}Y_{1:T} \vert \boldsymbol{\theta})\pi(\boldsymbol{\theta})
\end{equation}

\subsection{Posterior approximation}\label{Sec:postapprox}
We apply a Gibbs algorithm  \cite[Ch. 10]{RobCas2013} with a Metropolis-Hastings (MH) steps. In the sampler, we draw the latent variables and the parameters of the model by iterating the following steps:
\begin{enumerate}
\item draw $(X_{t},Y_{t})$ from $f(X_{t},Y_{t} \vert Z_{1:T},\boldsymbol{\theta})$;
\item draw $\boldsymbol{\varphi}$ from $\pi(\boldsymbol{\varphi} \vert Z_{1:T},Y_{1:T},X_{1:T},\boldsymbol{\theta}_{-\varphi})$;
\item draw $\phi$ from $\pi(\phi \vert Z_{1:T},Y_{1:T},X_{1:T},\boldsymbol{\theta}_{-\phi})$;
\item draw $\lambda$ from $\pi(\lambda \vert Z_{1:T},Y_{1:T},X_{1:T},\boldsymbol{\theta}_{-\lambda})$,
\end{enumerate}
where $\boldsymbol{\theta}_{-\eta}$ indicates the collection of parameters excluding the element $\boldsymbol{\eta}$.

The full conditional for the latent variables is
\begin{equation}
(X_{t},Y_{t}) \sim f(Z_{t} \vert X_{t}, Y_{t},\boldsymbol{\theta})f(X_{t},Y_{t} \vert Z_{1:T},\boldsymbol{\theta}).
\end{equation}
We draw from the full conditional distribution by MH. Differently from \cite{KarNtz2006}, we use a mixture proposal distribution which allows for a better mixing of the MCMC chain. At the $j$-th  iteration, we generate a candidate $X_{t}^{*}$  from $GP(\theta_{1t},\lambda)$ with probability $\nu$ and $(X_{t}^{*}-Z_{t})$ from $GP(\theta_{2t},\lambda)$ with probability $1-\nu$, and accept with probability 
\begin{equation}
\varrho = \min \left\{ 1,  \frac{f_t(X^{*}_{t} \vert \theta_{1t},\lambda) f_t(X^{*}_{t}-Z_{t} \vert \theta_{2t},\lambda)}{f_t(X_{t}^{(j-1)} \vert \theta_{1t},\lambda) f_t(X^{(j-1)}_{t}-Z_{t} \vert \theta_{2t},\lambda)} \frac{q(X_{t}^{(j-1)})}{q(X_{t}^{*})}\right\}
\end{equation}
where $q(X_{t})=\nu f(X_{t} \vert \theta_{1t}, \lambda) + (1-\nu) f(X_{t}-Z_{t} \vert \theta_{2t}, \lambda)$ and $X_{t}^{(j-1)}$ is the $(j-1)$-th iteration value of the latent variable $X_{t}$. The method extends to the GPD the technique proposed in \cite{KarNtz2006} for the Poisson differences.

As regards to the parameter $\boldsymbol{\varphi}$, its full conditional distribution is 
\begin{equation}\label{varphiEQ}
\boldsymbol{\varphi} \sim\pi(\boldsymbol{\varphi} \vert Z_{1:T},Y_{1:T},X_{1:T},\boldsymbol{\theta}_{-\varphi}) \propto \pi(\boldsymbol{\varphi})\prod_{t=1}^{T} f_{t}(X_{t},Y_{t}\vert \boldsymbol{\theta}).
\end{equation}
We consider a MH with Dirichlet independent proposal distribution 
\begin{equation}
\boldsymbol{\varphi}^{*} \sim Dir(\textbf{c}^{*})
\end{equation}
where $\textbf{c}^{*} = ({c}_{0}^{*},{c}_{1}^{*}, {c}_{2}^{*})$ and acceptance probability
\begin{equation}
\varrho = 1 \wedge \frac{\pi(\boldsymbol{\varphi}^{*} \vert Z_{1:T},Y_{1:T},X_{1:T},\boldsymbol{\theta}_{-\varphi})}{\pi(\boldsymbol{\varphi}^{j-1} \vert Z_{1:T},Y_{1:T},X_{1:T},\boldsymbol{\theta}_{-\varphi})}.
\end{equation}


The full conditional distribution of $\phi$ is 
\begin{equation}
\pi(\phi \vert Z_{1:T},Y_{1:T},X_{1:T},\boldsymbol{\theta}_{-\phi}) \propto \pi(\phi)\prod_{t=1}^{T} f_{t}(X_{t},Y_{t}\vert \boldsymbol{\theta}).
\end{equation}
We consider the change of variable $\zeta = \log(\phi -c)$ with Jacobian $\exp(\zeta)$ and a MH step with a random walk proposal 
\begin{equation}
\zeta^{*} \sim N(\zeta_{j-1},\gamma^{2})
\end{equation}
where $\zeta_{j-1} = \log (\phi_{j-1}-c)$, $\phi_{j-1}$ is the previous iteration value of the parameter and $c=\frac{1}{(1-\lambda)^{2}}$. The acceptance probability is
\begin{equation}
\varrho = \min \left\{ 1, \frac{\pi(\phi^{*} \vert Z_{1:T},Y_{1:T},X_{1:T},\boldsymbol{\theta}_{-\phi})\exp (\zeta^{*})}{\pi(\phi_{j-1}) \vert Z_{1:T},Y_{1:T},X_{1:T},\boldsymbol{\theta}_{-\phi})\exp(\zeta_{j-1})}  \right\}
\end{equation}
where $\phi^{*}=c+\exp(\zeta^{*})$.

The full conditional distribution of $\lambda$ is
\begin{equation}
\pi(\lambda \vert Z_{1:T},Y_{1:T},X_{1:T},\boldsymbol{\theta}_{-\lambda}) \propto \pi(\lambda)\prod_{t=1}^{T} f_{t}(X_{t},Y_{t}\vert \boldsymbol{\theta}).
\end{equation}
We consider a MH step with Beta random walk proposal 
\begin{equation}
\lambda^{*} \sim Be(s\lambda^{(j-1)},s(1-\lambda^{(j-1)}))
\end{equation}
where $s$ is a precision parameter. The acceptance probability is:
\begin{equation}
\varrho = \min \left\{ 1, \frac{\pi(\lambda^{*} \vert Z_{1:T},Y_{1:T},X_{1:T},\boldsymbol{\theta}_{-\lambda})Be(s\lambda^{*},s(1-\lambda^{*}))}{\pi(\lambda^{(j-1)} \vert Z_{1:T},Y_{1:T},X_{1:T},\boldsymbol{\theta}_{-\lambda})Be(s\lambda^{(j-1)},s(1-\lambda^{(j-1)}))}  \right\}.
\end{equation}

\section{Simulation study}
%
The purpose of our simulation exercises is to study the efficiency of the MCMC algorithm presented in Section \ref{Sec:BayInf}. We evaluated the \cite{Gew92} convergence diagnostic measure (CD), the inefficiency factor (INEFF)\footnote{The inefficiency factor is defined as \begin{equation*}
INEFF = 1+2 \sum_{k=1}^{\infty} \rho(k) 
\end{equation*} 
where $\rho(k)$ is the sample autocorrelation at lag $k$ for the parameter of interest and are computed to measure how well the MCMC chain mixes. An INEFF equal to $n$ tells us that we need to draw MCMC samples $n$ times as many as uncorrelated samples.} and the Effective Sample Size (ESS). 

\begin{table}[t]
\centering
\begin{small}
\begin{tabular}{c|c|c|c|c|c|c}
&\multicolumn{3}{|c}{Low persistence}&\multicolumn{3}{|c}{High persistence}\\
&\multicolumn{3}{|c}{($\alpha=0.25$, $\beta=0.23$, $\lambda=0.4$)}&\multicolumn{3}{|c}{($\alpha=0.53$, $\beta=0.25$, $\lambda=0.6$)}\\
\hline
&$\alpha$&$\beta$&$\lambda$&$\alpha$&$\beta$&$\lambda$\\
\hline
$ACF(1)_{BT}$ & 0.96& 0.97 & 0.97& 0.91& 0.88& 0.98\\
$ACF(10)_{BT}$ &0.86 & 0.83& 0.81& 0.70& 0.52&0.83\\
$ACF(30)_{BT}$ & 0.75 &0.69& 0.63& 0.52 & 0.37& 0.60\\
$ACF(1)_{AT}$ & 0.43 & 0.39 & 0.27 & 0.21& 0.13 &0.16\\
$ACF(10)_{AT}$ & 0.25& 0.18& 0.12& 0.20&0.06 &0.11\\
$ACF(30)_{AT}$ & 0.18& 0.15& 0.07& 0.15& 0.06&0.09\\
\hline
$ESS_{BT}$ & 0.02 & 0.02 & 0.02& 0.02& 0.03&0.02\\
$ESS_{AT}$ & 0.07& 0.07& 0.09&0.09 & 0.12&0.11\\
\hline
$INEFF_{BT}$ & 50.53& 51.07 & 43.88 & 48.39&43.35 &49.25\\
$INEFF_{AT}$ & 26.36 & 27.29 & 13.99 & 17.21&16.84 &12.59\\
\hline
$CD_{BT}$ & 11.81& -28.69 & 0.78& 0.93& -6.27&2.40\\
 & (0.11)& (0.14) & (0.10) & (0.04)&(0.06) &(0.05)\\
$CD_{AT}$ & 5.72 & -13.18 & 0.2 & 0.74 & -3.84 & 1.17\\
 & (0.23)& (0.23) & (0.23) &(0.13) & (0.15)&(0.11)\\
\bottomrule
\end{tabular}
\end{small}
\caption{Autocorrelation function (ACF), effective sample size (ESS) and inefficiency factor (INEFF) of the posterior MCMC samples for the two settings: low persistence and high persistence. The results are averages over a set of 50 independent MCMC experiments on 50 independent datasets of 400 observations each. We ran the proposed MCMC algorithm for 1,010,000 iterations and evaluate the statistics before (subscript BT) and after (subscript AT) removing the first 10,000 burn-in samples, and applying a thinning procedure with a factor of 250. In parenthesis the p-values of the Geweke's convergence diagnostic.}
\label{Stat}
\end{table}
We simulated 50 independent data-series of 400 observations each. We run the Gibbs sampler for 1,010,000 iterations on each dataset, discard the first 10,000 draws to remove dependence on initial conditions, and finally apply a thinning procedure with a factor of 250, to reduce the dependence between consecutive draws.

As commonly used in the GARCH and stochastic volatility literature \citep[e.g., see][ and references therein]{Chib2002, casarin2009online, MCO12, Casetal2019}, we test the efficiency of the algorithm in two different settings: low persistence and high persistence. The true values of the parameters are: $\alpha=0.25$, $\beta=0.23$, $\lambda=0.4$ in the low persistence setting and $\alpha=0.53$, $\beta=0.25$, $\lambda=0.6$ in the high persistence setting.  Table \ref{Stat} shows, for the parameters $\alpha$, $\beta$ and $\lambda$, the INEFF, ESS and ACF averaged over the 50 replications before (BT subscript) and after thinning (AT subscript). 

The thinning procedure is effective in reducing the autocorrelation levels and in increasing the ESS, especially in the high persistence setting. The p-values of the CD statistics indicate that the null hypothesis that two sub-samples of the MCMC draws have the same distribution is accepted. The efficiency of the MCMC after thinning generally improved. On average, the inefficiency measures (19.05), the p-values of the CD statistics (0.18) and the  acceptance rates (0.35) achieved the values recommended in the literature \citep[e.g., see ][]{Robetal1997}.

\section{Real data examples}
\subsection{Accident data}
\begin{figure}[p]
\centering
\begin{tabular}{c}
\includegraphics[height=100pt,width=300pt]{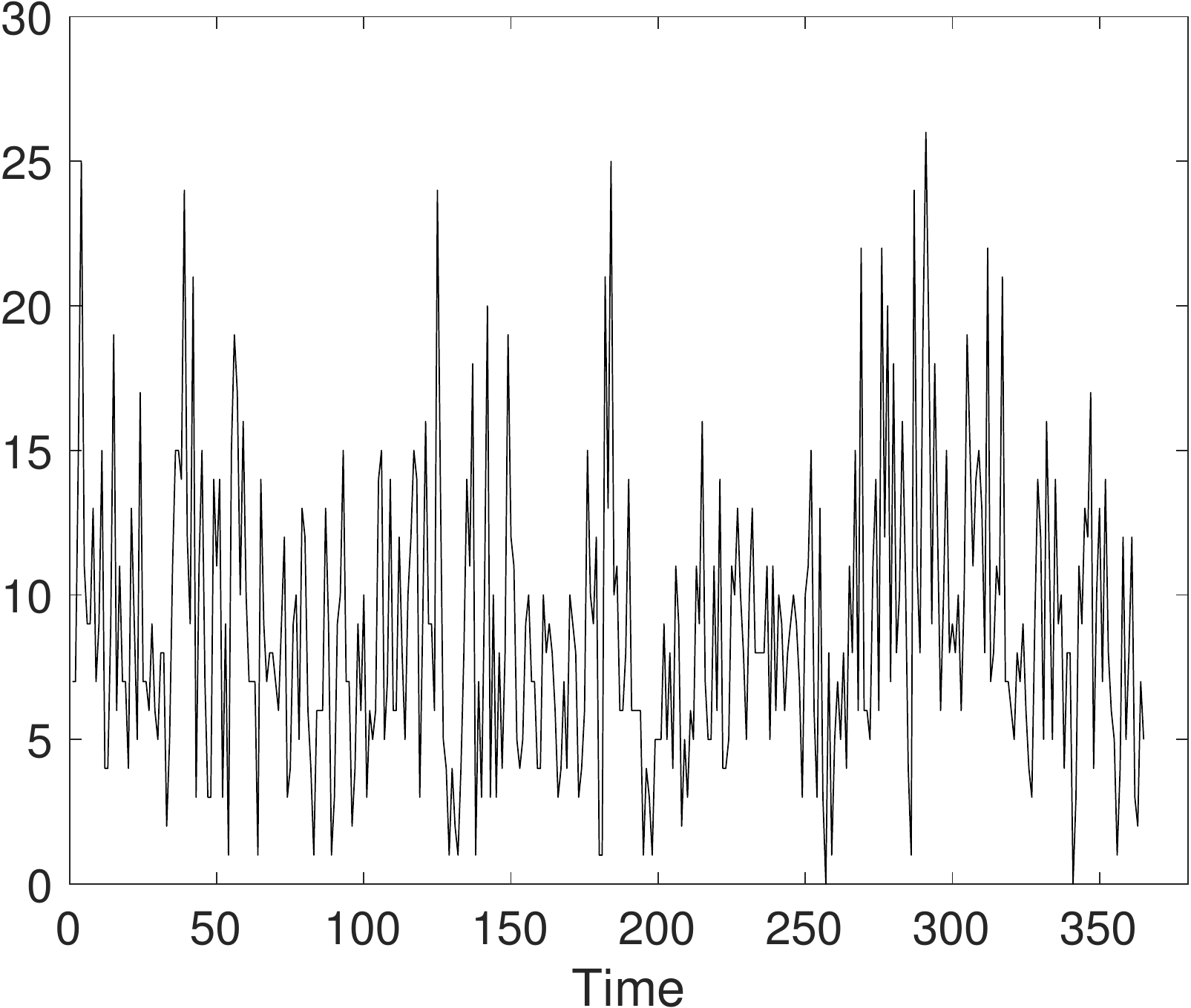} \\
\includegraphics[height=100pt,width=300pt]{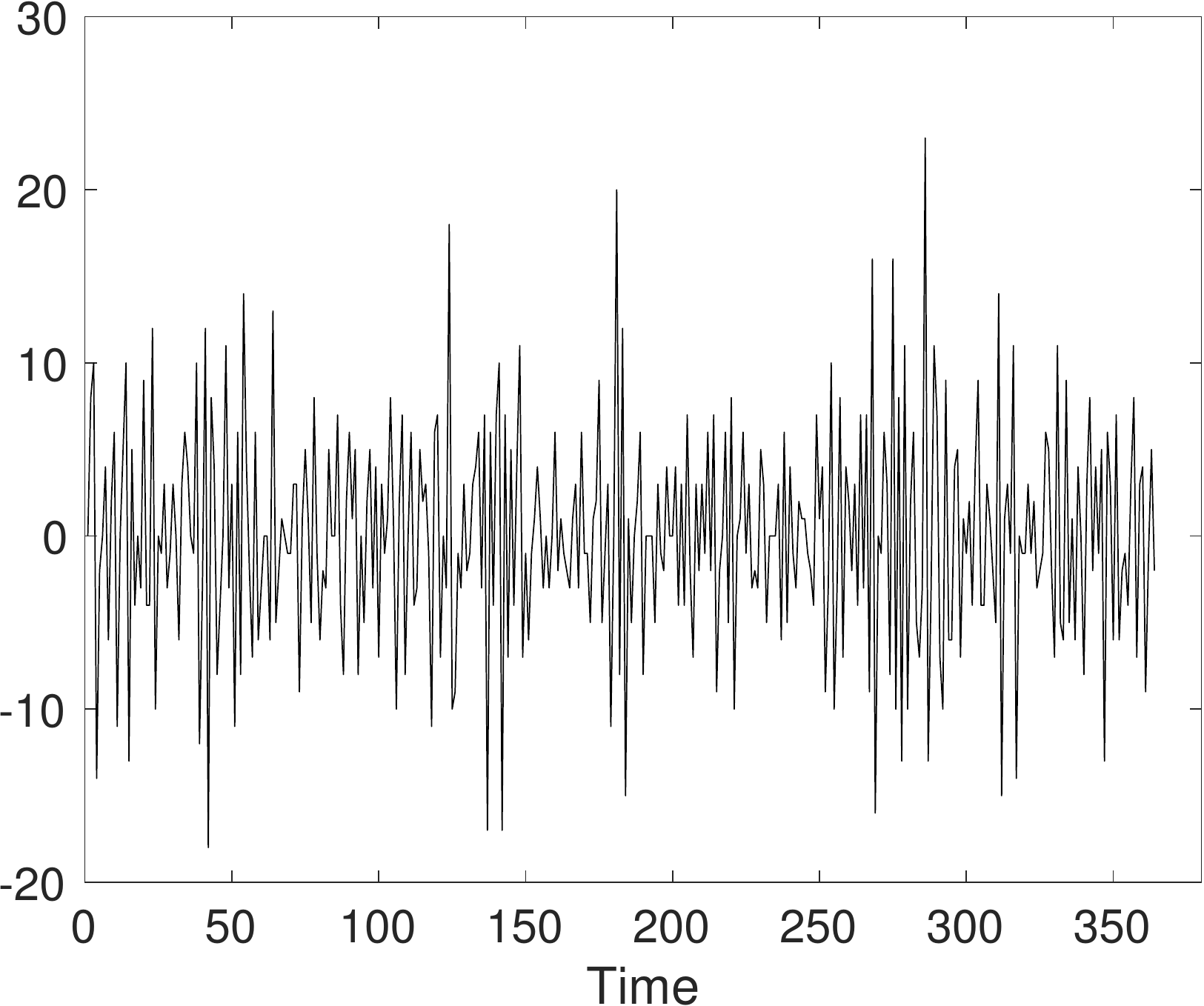}
\end{tabular}
\caption{Frequency (top) and month-on-month changes (bottom) of the accidents at the Schiphol airport in The Netherlands in 2001.}
\label{Dataset2}
\end{figure}
Data in this application are the number of accidents near Schiphol airport in The Netherlands during 2001 (Fig. \ref{Dataset2}). They have been previously considered in \cite{Brietal2008} and \cite{AndKar2014}. The time series of accident counts is non-stationary and should be differentiated \citep{KimPar2008}. We applied our Bayesian estimation procedure, as described in Section \ref{Sec:BayInf}. 
\begin{figure}[p]
\centering
\setlength{\tabcolsep}{-3pt} 
\begin{tabular}{cc}
$\alpha$ & $\beta$\vspace{-2pt}\\
\includegraphics[width=210pt, height=110pt]{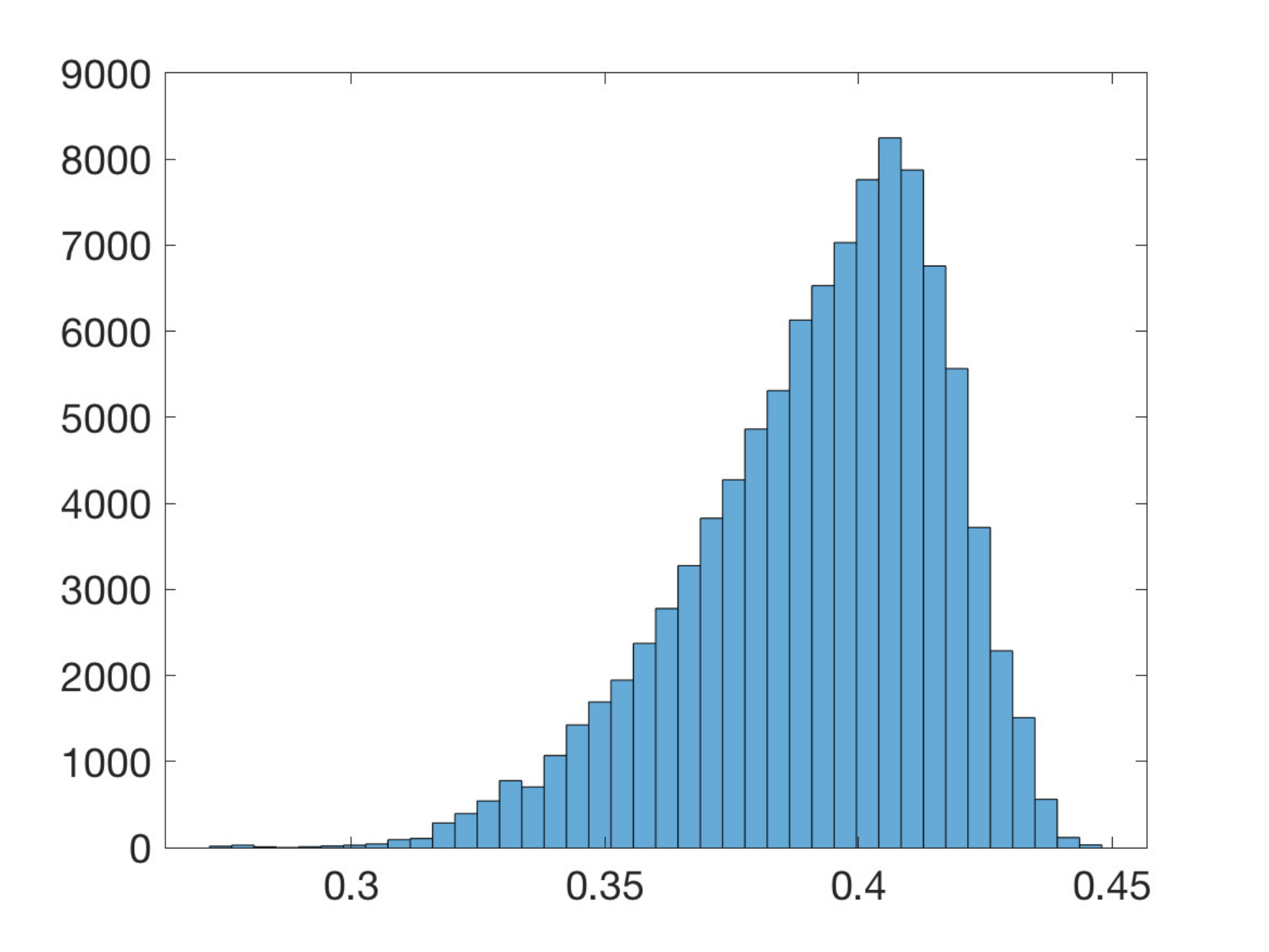} &
\includegraphics[width=210pt, height=110pt]{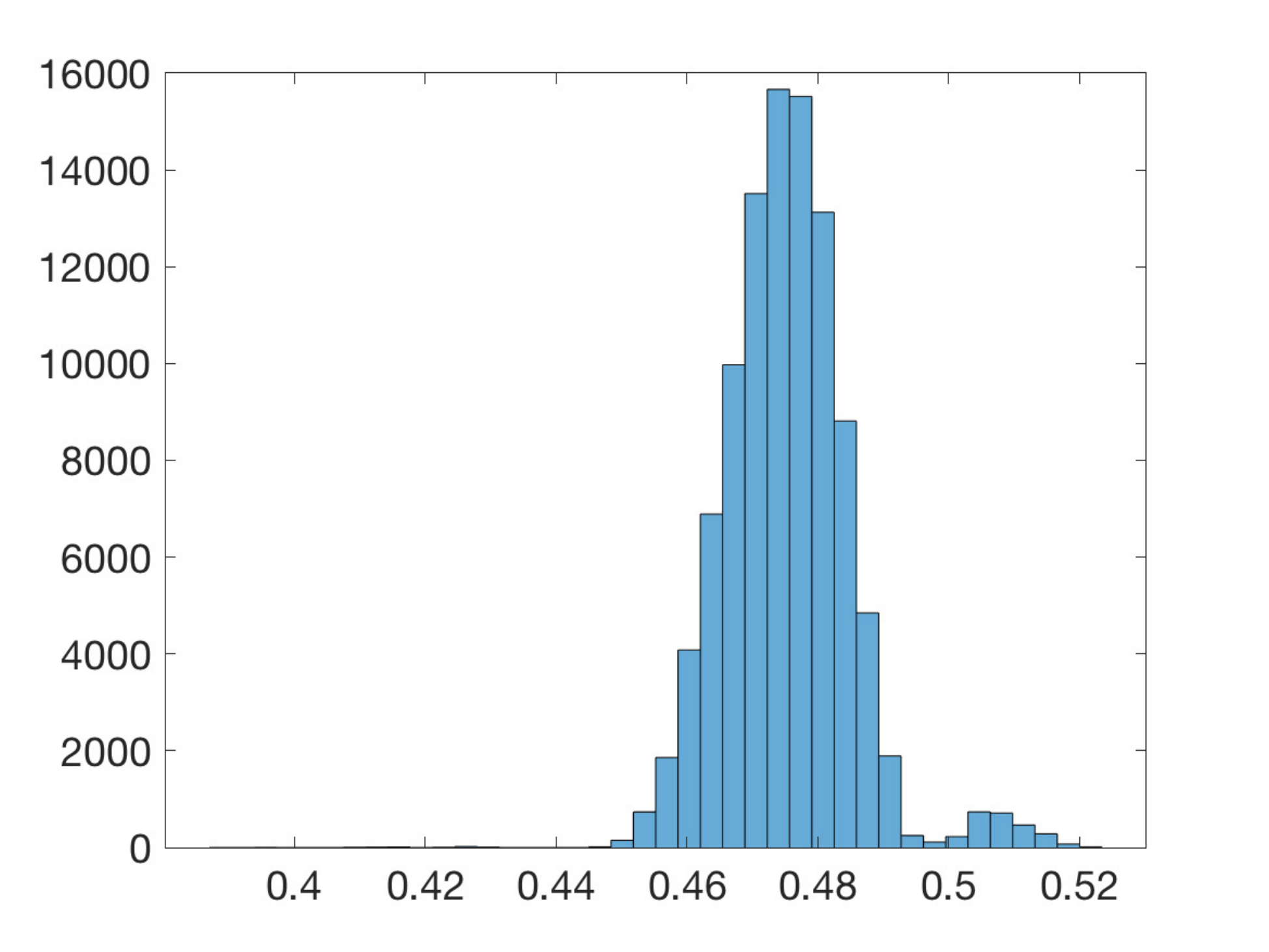} \\
$\lambda$ & $\phi$\vspace{-2pt}\\
\includegraphics[width=210pt, height=110pt]{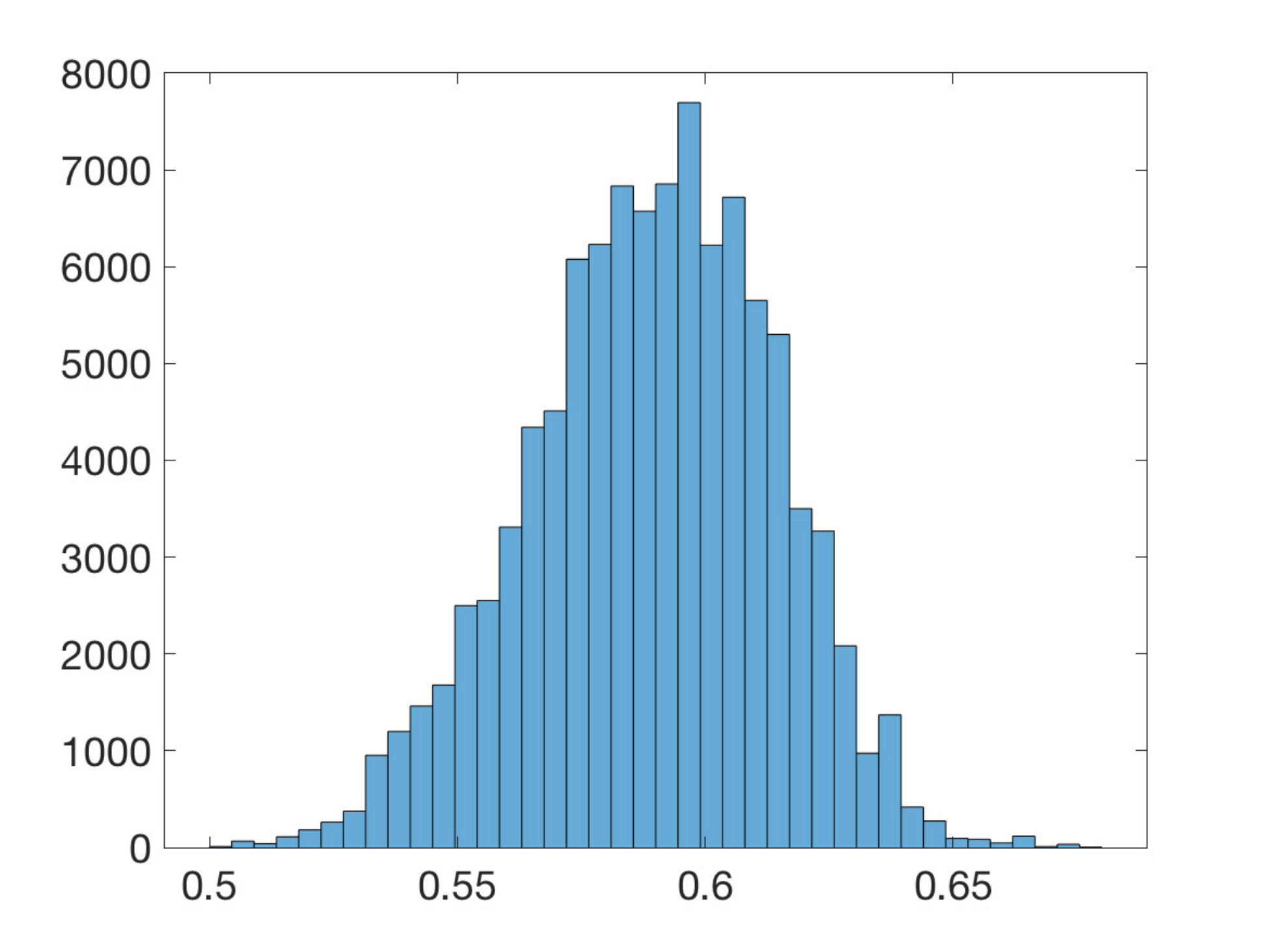} &
\includegraphics[width=210pt, height=110pt]{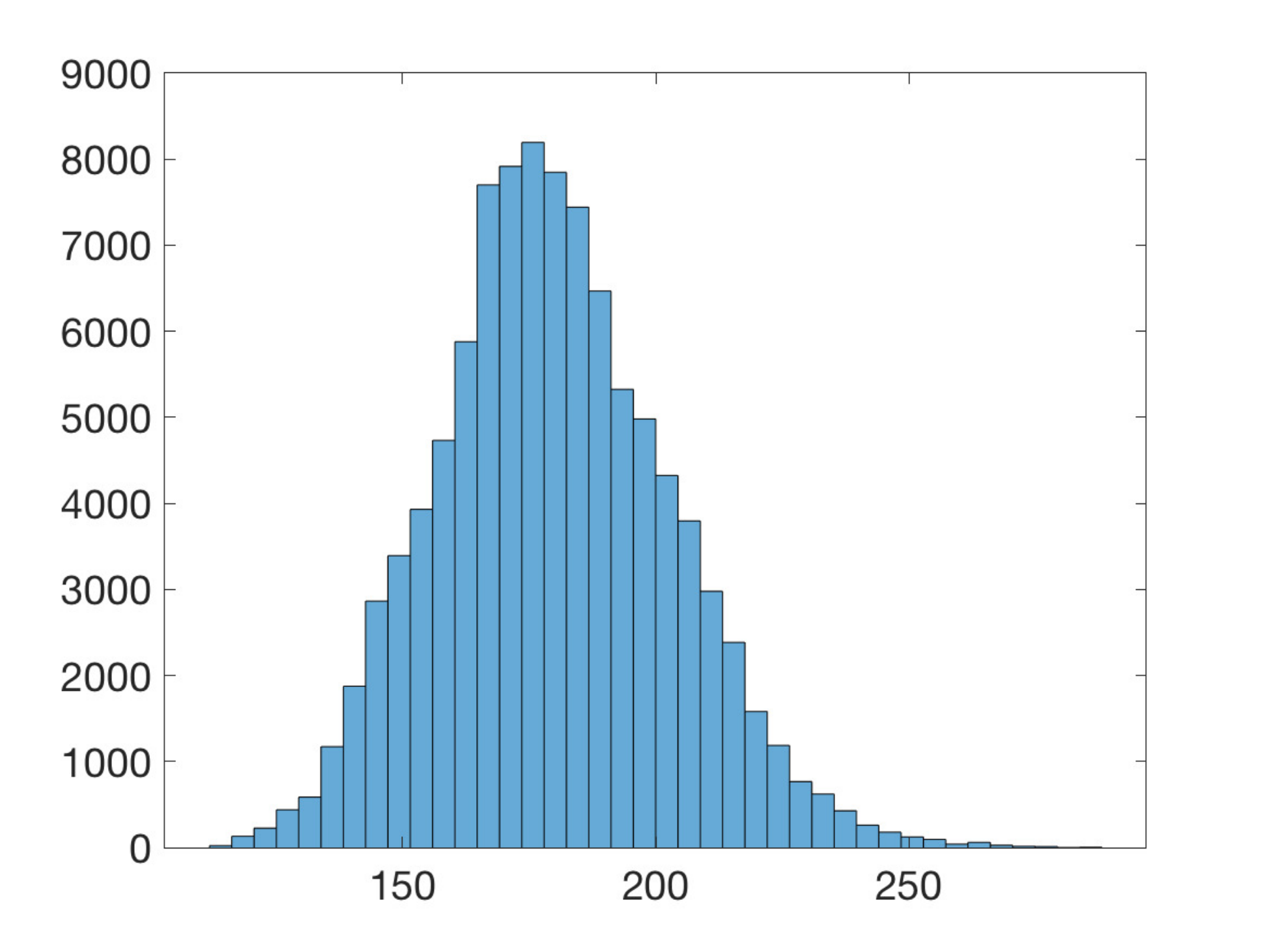} \\
\end{tabular}
\caption{Histograms of the MCMC draws for the parameters of the Schipol's accident data of Fig. \ref{Dataset2}.}
\label{fig:GibbsDrawsReal}
\end{figure}


\begin{table}[t]
\centering
\begin{small}
\begin{tabular}{c|c|c|c}
\specialrule{.1em}{.05em}{.05em} 
\textbf{Parameters} & \textbf{Mean} & \textbf{Std} & \textbf{CI} \\
\hline
\hline
\multicolumn{4}{c}{Model $\mathcal{M}_1$: GPD-INGARCH(1,1)}\\
\hline
$\alpha$ & 0.3920 & 0.0246 & (0.3347, 0.4297)\\
$\beta$   & 0.4753 & 0.0096 & (0.4582, 0.4999)\\
$\lambda$   & 0.5892 & 0.0246 & (0.53833, 0.6349)\\
$\phi$  & 179.7905 & 22.8040 & (138.2406, 226.99)\\
\hline
\multicolumn{4}{c}{Model $\mathcal{M}_2$: PD-INGARCH(1,1) and $\lambda=0$}\\
\hline
$\alpha$ & 0.1121 & 0.0095& (0.1004, 0.1340)\\
$\beta$   & 0.1798 & 0.0101 & (0.1549, 0.1989)\\
$\lambda$   & - & - & -\\
$\phi$  & 94.9340 & 8.6488 & (77.0653, 110.6276)\\
 \hline
\multicolumn{4}{c}{Model $\mathcal{M}_3$: GPD-INARCH(1,0)}\\
\hline
$\alpha$ & 0.2286 & 0.0485& (0.1407, 0.3287)\\
$\beta$   & - & - & -\\
$\lambda$   & 0.5682 & 0.0243 & (0.5195, 0.6166)\\
$\phi$  & 218.6333 & 36.2307 & (155.7151, 297.2252)\\
 \hline
\multicolumn{4}{c}{Model $\mathcal{M}_4$: PD-INARCH(1,0) and $\lambda=0$}\\
\hline
$\alpha$ & 0.1013 & 0.0013 & (0.1000, 0.1050)\\
$\beta$   & - & - & -\\
$\lambda$   & - & - & -\\
$\phi$  & 104.4131 & 7.4362 & (86.4588, 115.8723)\\
\specialrule{.1em}{.05em}{.05em} 
\end{tabular}
\end{small}
\caption{Posterior mean (Mean), 95\% credible intervals (CI), and standard deviation (Std) for different specifications (different panels) of the GPD-INGARCH.}
\label{Tabpar}
\end{table}

In Fig. \ref{fig:GibbsDrawsReal} are presented the histograms for the Gibbs draws for each parameters. Table \ref{Tabpar} presents the parameter posterior mean and standard error and the 95\% credible interval for the unrestricted INGARCH(1,1) model (model $\mathcal{M}_1$). In the data, we found evidence of high persistence in the expected accident arrivals, i.e. $\hat{\alpha}+\hat{\beta}=0.8673$ and heteroskedastic effects, i.e. $\hat{\beta}=0.4753$. Also, there is evidence in favour of overdispersion, $\hat{\lambda}=0.5892$ and overdispersion persistsence $\hat{\phi}=179.7905$. We study the contribution of the heteroskedasticy and persistence by testing some restrictions of the INGARCH(1,1) (models from $\mathcal{M}_2$ to $\mathcal{M}_4$ in Tab. \ref{Tabpar}).

Bayesian inference compares models via the so-called Bayes factor, which is the ratio of normalizing constants of the posterior distributions of two different models (see \cite{CamPet2014} for a review). MCMC methods allows for generating samples from the posterior distributions which can be used to estimate the ratio of normalizing constants.

In this paper we use the method proposed by \cite{Gey1994}. The method consists in deriving the normalizing constants by reverse logistic regression. The idea behind this method is to consider the different estimates as if they were sampled from a mixture of two distributions with probability
\begin{equation}
p_{j}(x,\eta) = \frac{h_{j}(x)\exp(\eta_{j})}{h_{1}(x)\exp(\eta_{1})+h_{2}(x)\exp(\eta_{2})},\,\, j=1,2
\end{equation}
to be generated from the $j$-th distribution of the mixture. \cite{Gey1994} proposed to estimate the log-Bayes factor $\kappa=\eta_{2}-\eta_{1}$ by maximizing the quasi-likelihood function
\begin{equation}
\ell_{n}(\kappa) = \sum_{i=1}^{n}\log p_{1}(X_{i1},\eta_1)+\sum_{i=1}^{n} \log p_{2}(X_{i2},\eta_2)
\end{equation}
where $n$ is the number of MCMC draws for each model and $X_{ij}=\log f(Z_{1:T},X_{1:T}^{(i)},Y_{1:T}^{(i)} \vert \boldsymbol{\theta}^{(i)})$
is the log-likelihood evaluated at the $i$-th MCMC sample for each model of Tab. \ref{Tabpar}.

We performed six reverse logistic regressions, in which we compare pairwise our models. The approximated logarithmic Bayes factors $BF(\mathcal{M}_i,\mathcal{M}_j)$ are given in Tab. \ref{Tab:BF}. It is possible to see that our GPD-INGARCH$(1,1)$, $\mathcal{M}_{1}$, is preferable with respect to the other models. Notice that $\mathcal{M}_{2}$ corresponds to an INGARCH$(1,1)$ where the observations are form a standard Poisson-difference model PD-INGARCH$(1,1)$, $\mathcal{M}_{3}$ corresponds to an autoregressive model, GPD-INARCH$(1,0)$, whereas $\mathcal{M}_{4}$ is a standard Poisson difference augoregressive model, PD-INARCH$(1,0)$.

\begin{table}[t]
\centering 
\begin{small}
\begin{tabular}{l| c c c}
\hline
BF($\mathcal{M}_{i}$,$\mathcal{M}_{j})$ & $\mathcal{M}_{2}$&$\mathcal{M}_{3}$&$\mathcal{M}_{4}$\\
\hline
\multirow{2}*{$\mathcal{M}_{1}$} &  333.45 & 25.19 & 121.44\\ 
  & (5.818) & (0.253)  & (0.521) \\	
\hline
\multirow{2}*{$\mathcal{M}_{2}$} &  &  -226.86 & -300.96\\
&  & (2.024)  & (2.522)\\	
\hline
\multirow{2}*{$\mathcal{M}_{3}$}  &  &   & -73.25\\ 
  &  &   & (0.358) \\	
\hline
\end{tabular}
\end{small}
\caption{Logarithmic Bayes Factor, BF($\mathcal{M}_{i}$,$\mathcal{M}_{j})$, of the model $\mathcal{M}_{i}$ (rows) against model $\mathcal{M}_j$ (columns), with $i<j$. Where $\mathcal{M}_{1}$ is the GPD-INGARCH(1,1), $\mathcal{M}_{2}$ is the PD-INGARCH(1,1) with $\lambda =0$, $\mathcal{M}_{3}$ is the GPD-INARCH(1,0) and $\mathcal{M}_{4}$ is the PD-INARCH(1,0) with $\lambda =0$. Number in parenthesis are standard deviations of the estimated Bayes factors.}
\label{Tab:BF}
\end{table}

\subsection{Cyber threats data}
According to the Financial Stability Board \citep[][pp. 8-9]{FSB1}, a cyber incident is any observable occurrence in an information system that jeopardizes the cyber security of the system, or violates the security policies and procedures or the use policies. Over the past years there have been several discussions on the taxonomy of incidents classification \citep[see, e.g.][]{ENISA}, in this paper we use the classification provided in the \citetalias{Hackweb} dataset. \citetalias{Hackweb} is a well-known cyber-incident website that collects public reports and provides the number of cyber incidents for different categories of threats: crimes, espionage and warfare. 
\begin{figure}[p]
\centering
\begin{tabular}{cc}
Total frequency& Crimes frequency\\
\includegraphics[width=190pt, height=100pt]{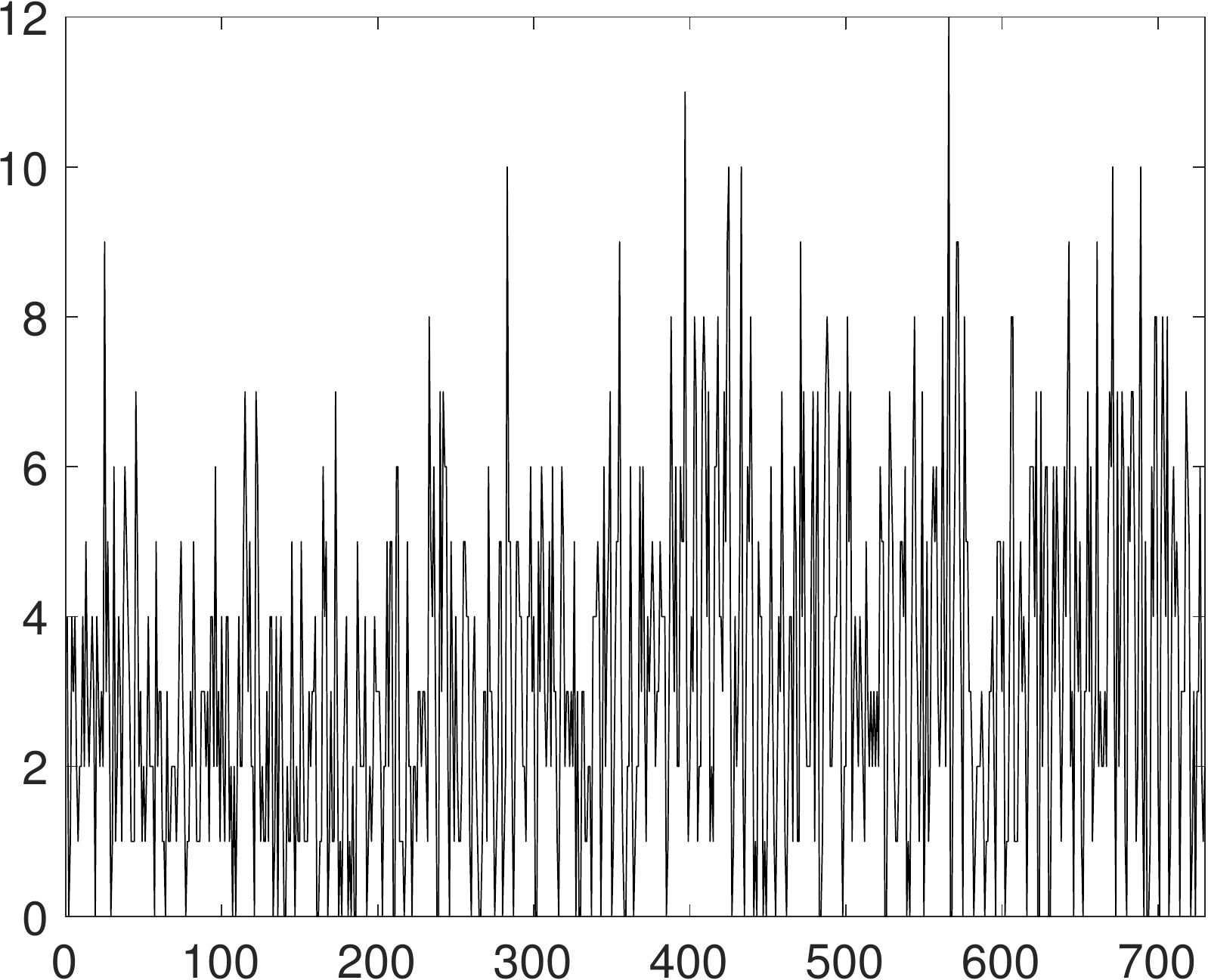}&\includegraphics[width=190pt, height=100pt]{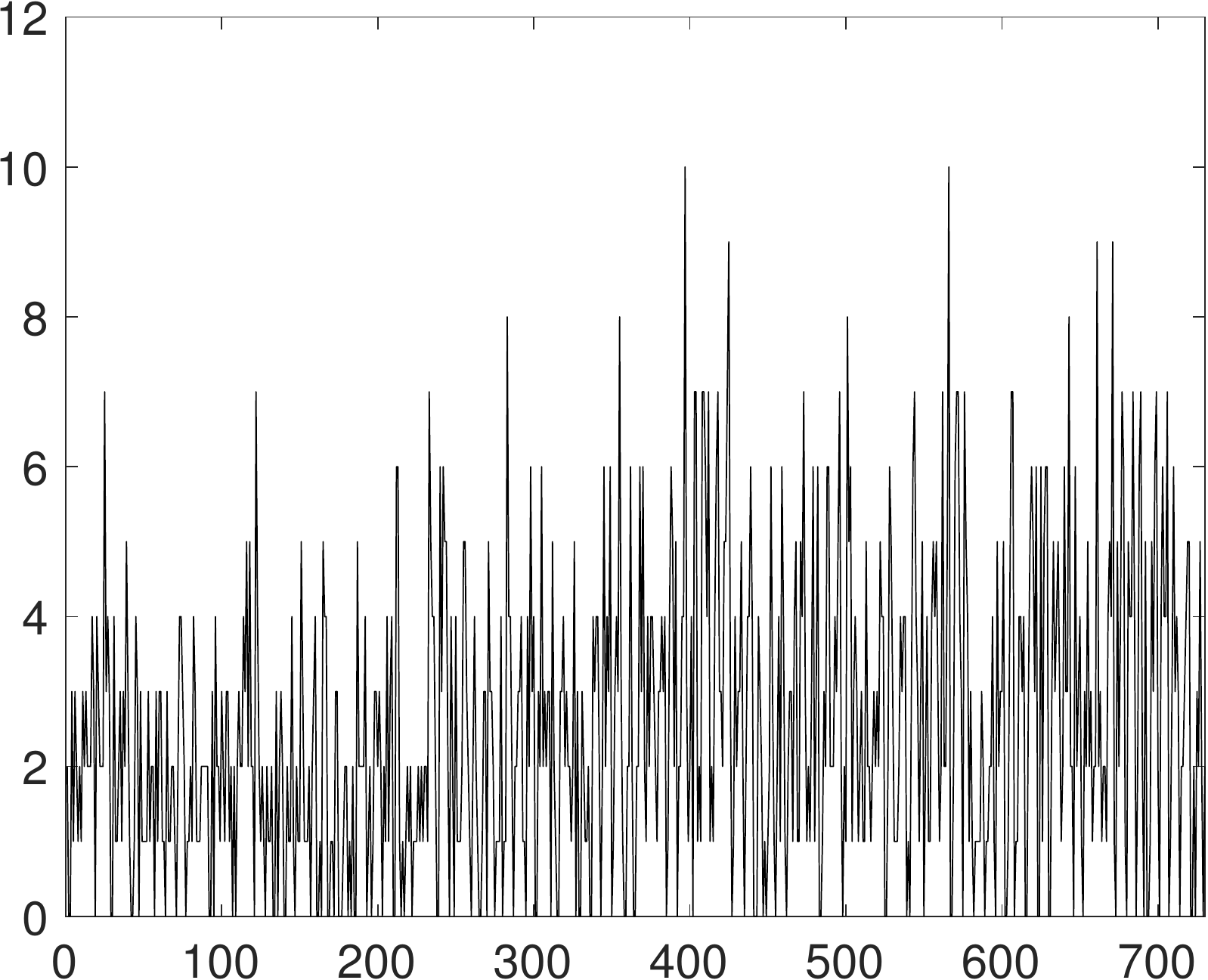}\\
Espionage frequency& Warfare frequency\\
\includegraphics[width=190pt, height=100pt]{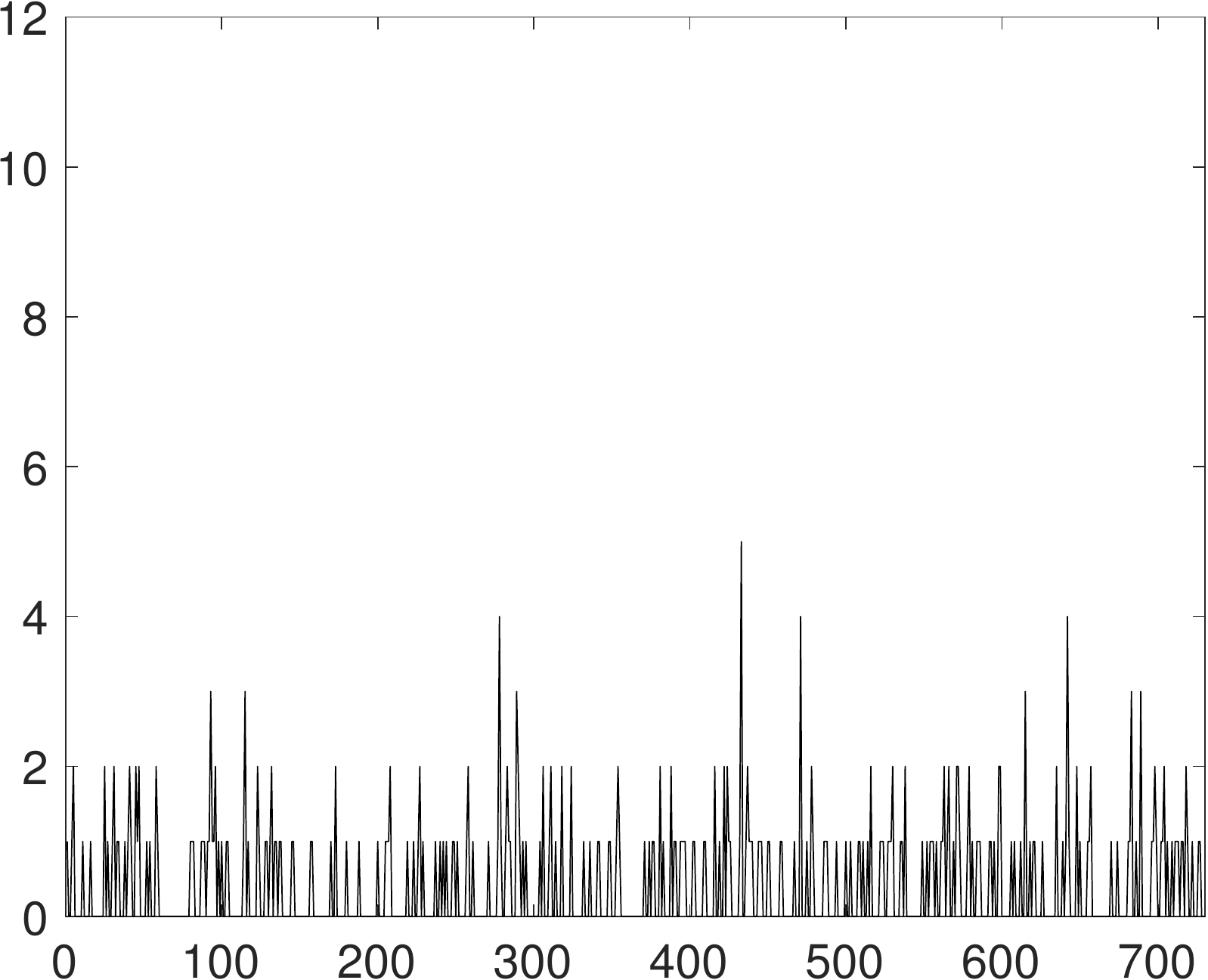}&\includegraphics[width=190pt, height=100pt]{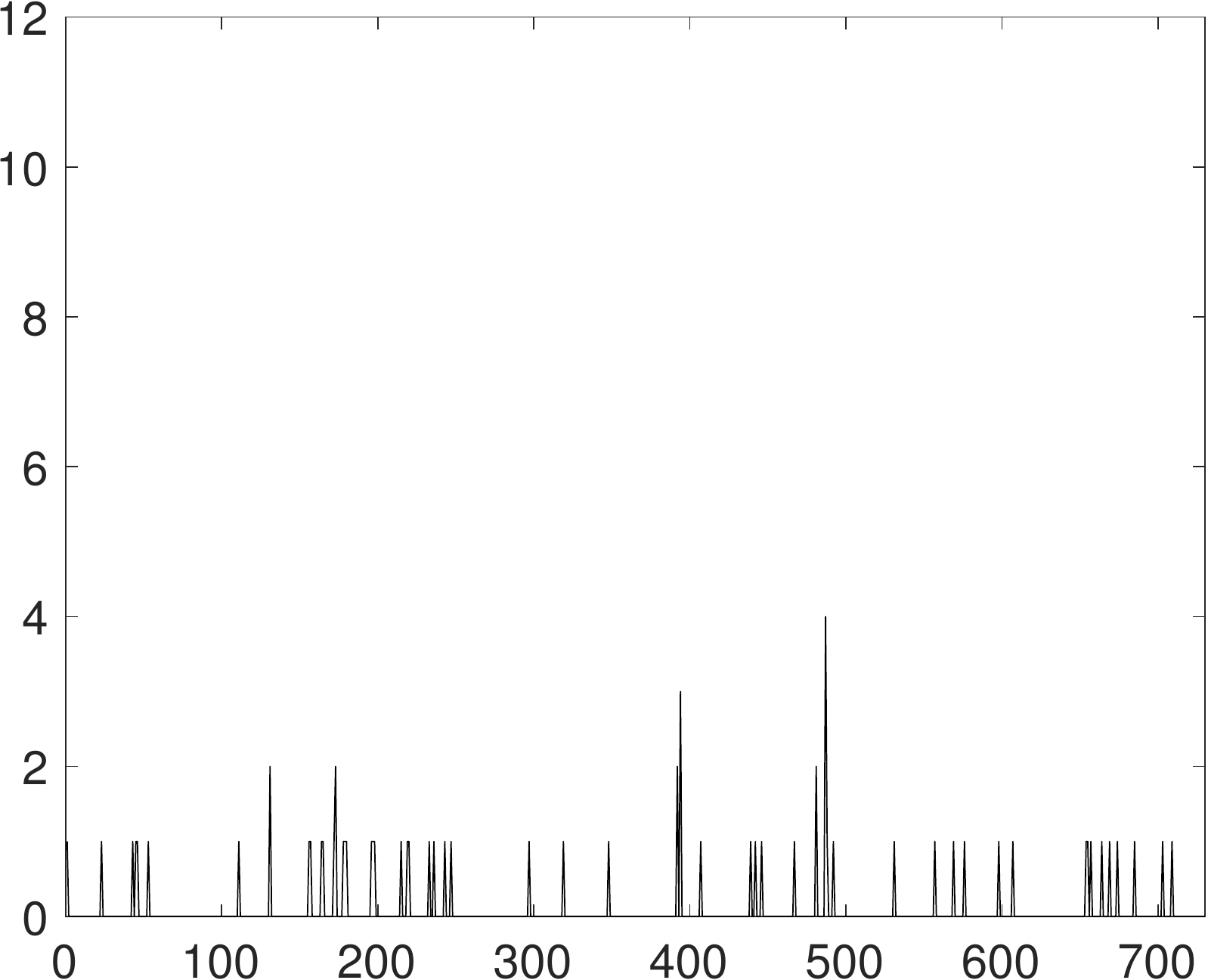}\\
\end{tabular}
\caption{Daily cyber-threats counts between 1st January 2017 and 31st December 2018.}
\label{Fig:Cyb1718levels}
\end{figure}

Figure \ref{Fig:Cyb1718levels} shows the total and category-specific number of cyber attacks at a daily frequency from January 2017 to December 2018. Albeit limited in the variety of cyber attacks the dataset covers some relevant cyber events and is one of the few publicly available datasets \citep{Agra18}. The daily threats frequencies are between 0 and 12 which motivates the use of a discrete distribution. We remove the upward trend by considering the first difference and fit the GPD-INGARCH model proposed in Section \ref{Model}.

We applied our estimation procedure, as described in Section \ref{Sec:BayInf}. As in the previous application, we fix $\alpha_{0}=1.05$ that is coherent with the conditional mean of the time series. We ran the Gibbs sampler for 110000 iterations, where we discarded the first 10000 iterations as burn-in sample. In Fig. \ref{Fig:HistCyb} are presented the histograms for the Gibbs draws for each parameters. 


\begin{figure}[p]
\centering
\setlength{\tabcolsep}{-3pt}
\begin{tabular}{cc}
$\alpha$ & $\beta$\vspace{-2pt}\\
\includegraphics[width=210pt, height=100pt]{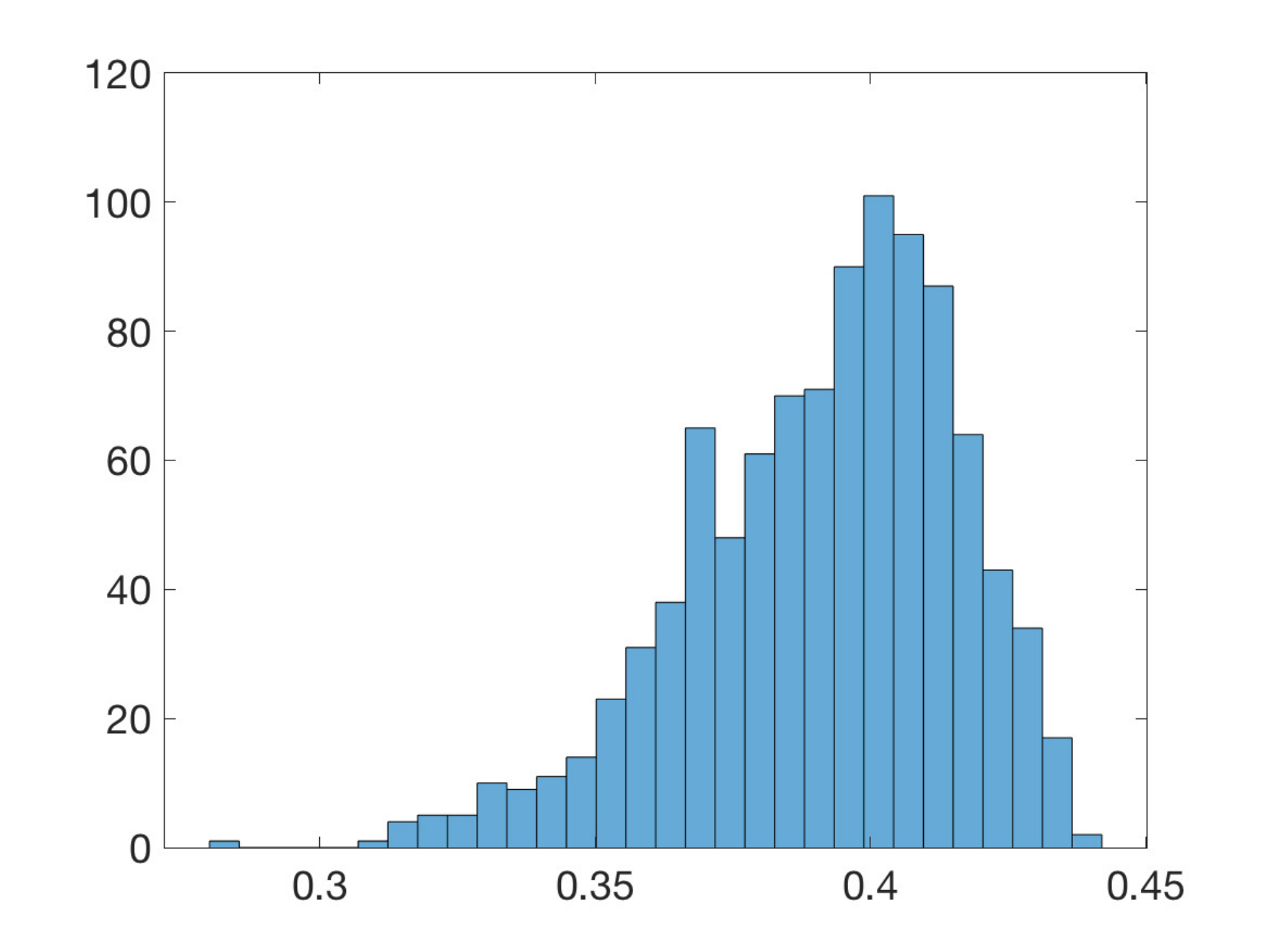} &
\includegraphics[width=210pt, height=100pt]{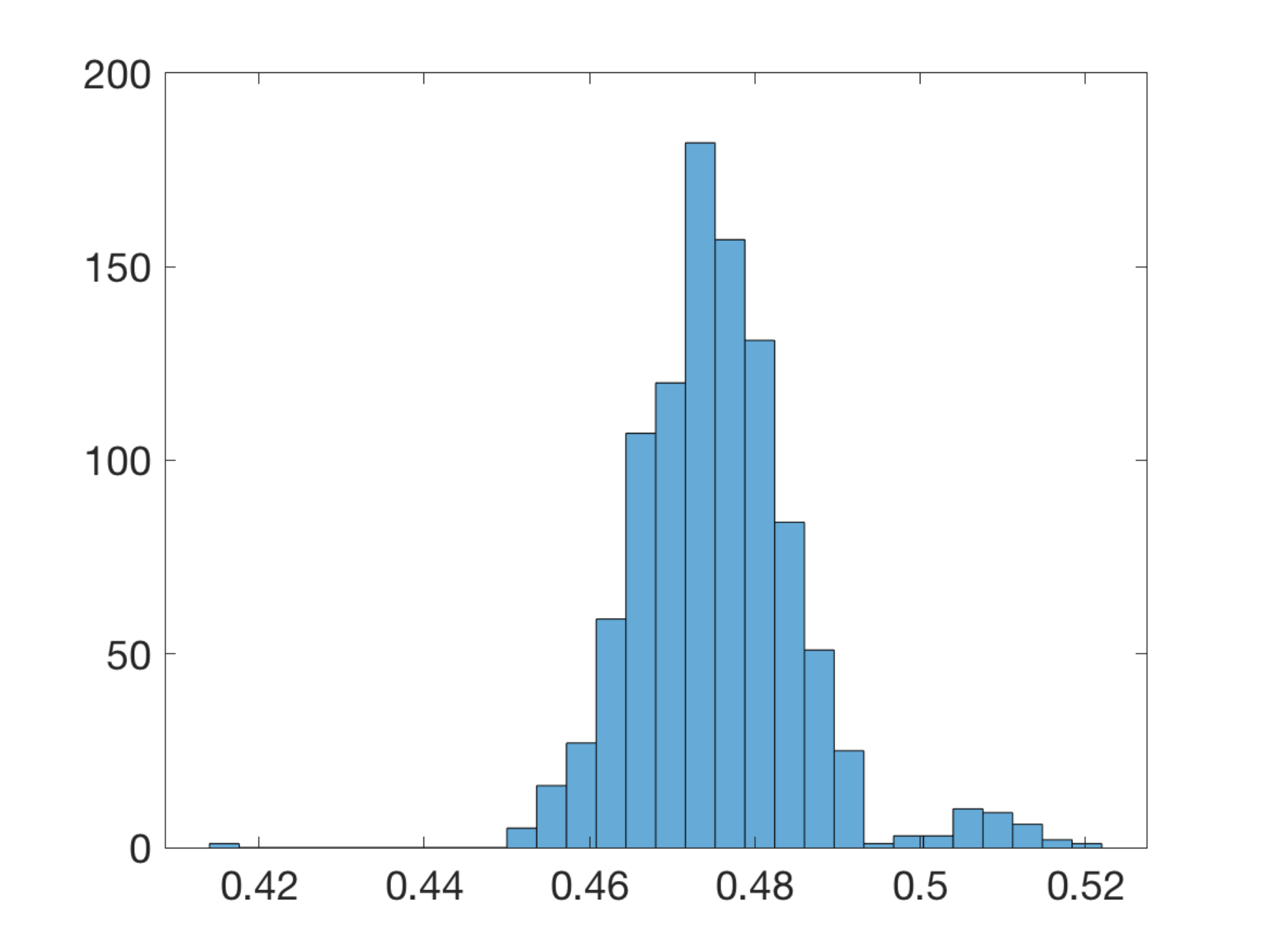} \\
$\lambda$ & $\phi$\vspace{-2pt}\\
\includegraphics[width=210pt, height=100pt]{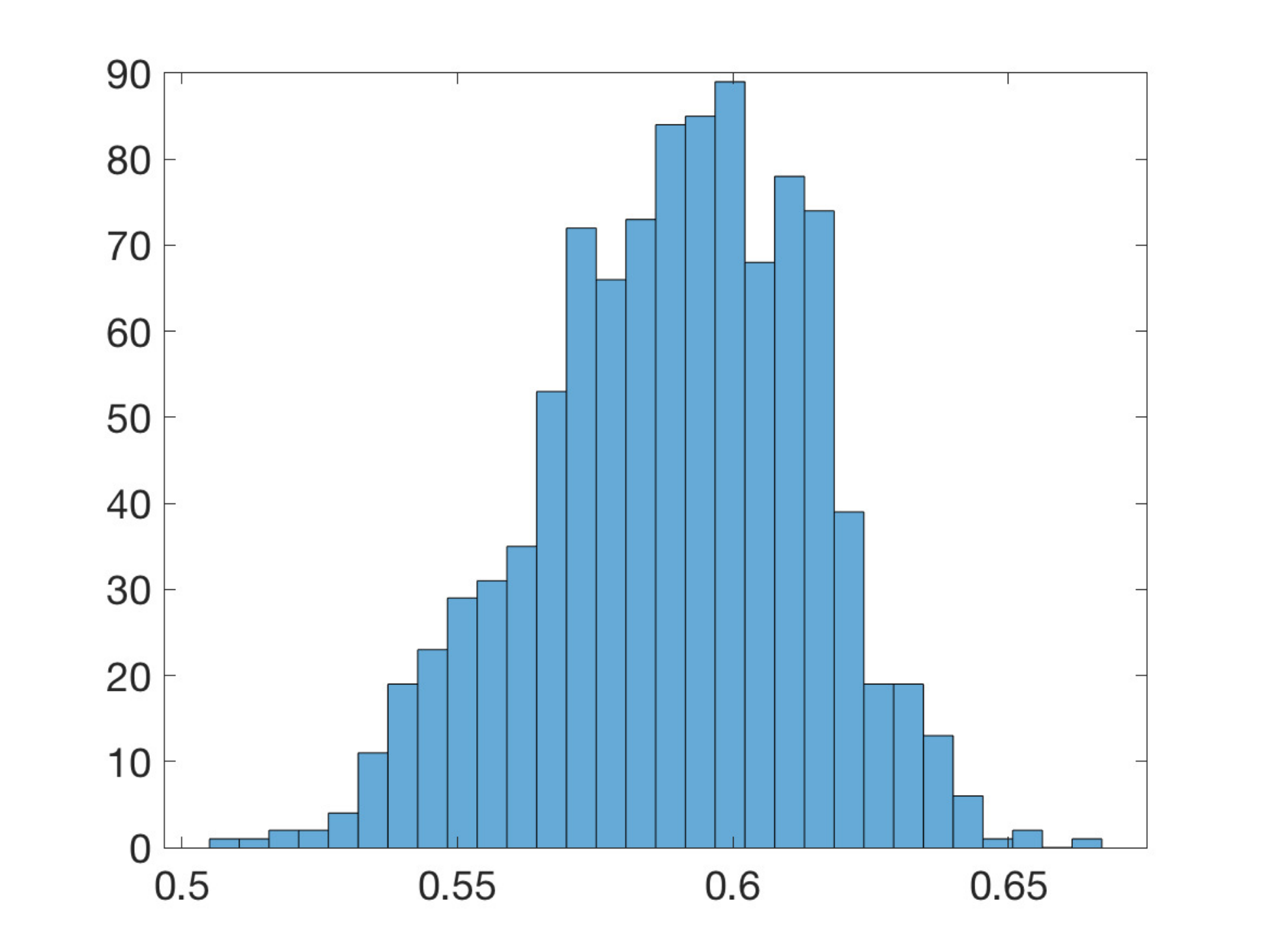} &
 \includegraphics[width=210pt, height=100pt]{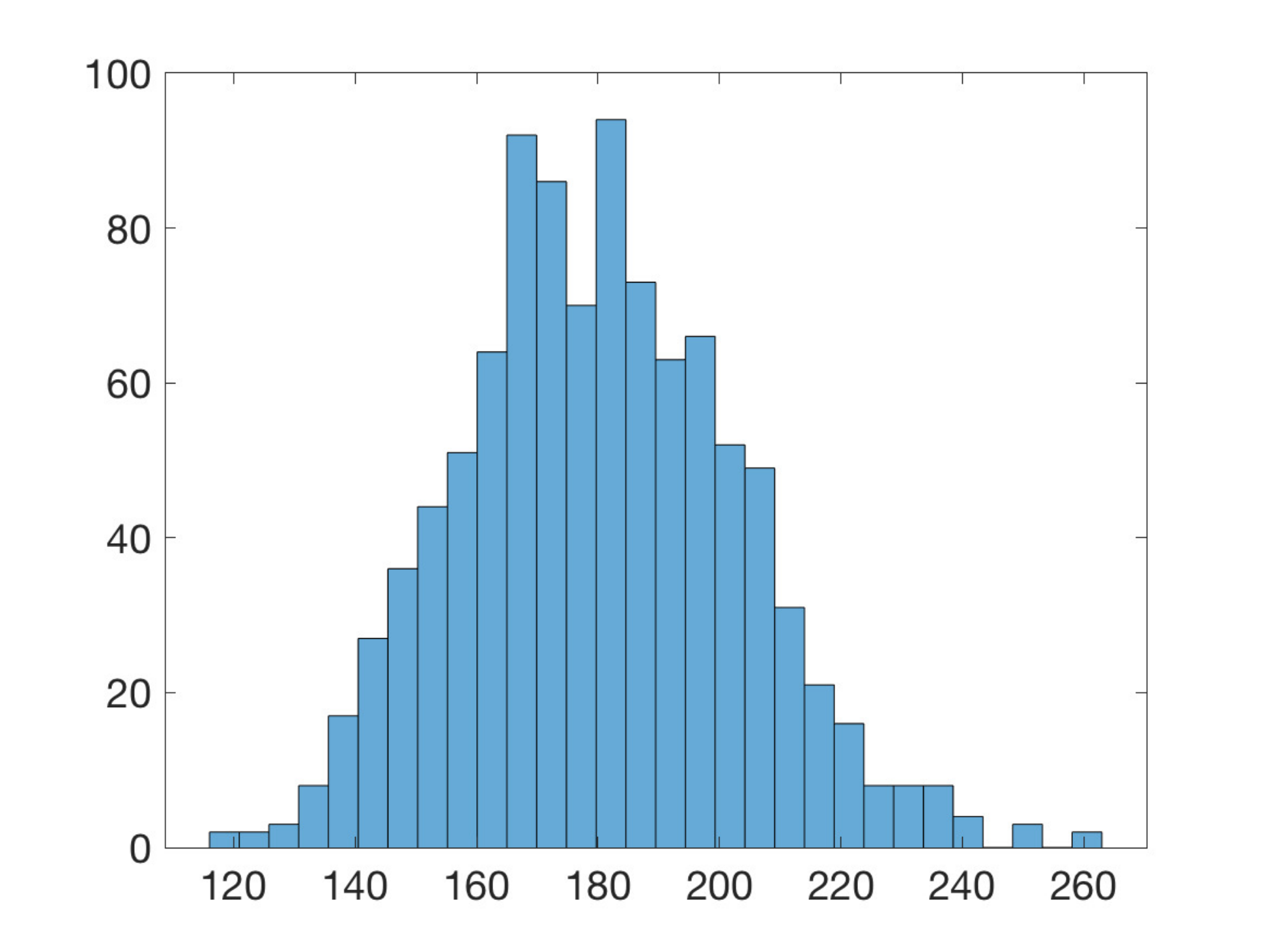} \\
\end{tabular}
\caption{Histograms of the MCMC draws for the parameters.}
\label{Fig:HistCyb}
\end{figure}


Figure \ref{Fig:HistCyb} shows that, as before, it is reasonable to fit a GPD-INGARCH process to the difference of cyber attacks since both the autoregressive parameter $\alpha$ and $\beta$, that represent the heteroskedastic feature of the data, are different from zero. Additionally, the value of $\lambda$ suggest the presence of over-dispersion in the data.

Given the importance of forecasting cyber-attacks, in this section we present the results of one-step-ahead forecasting exercise over a period of 120. We follow an approach based on predictive distributions which quantifies all uncertainty associated with the future number of attacks and is used in a wide range of applications \citep[see, e.g.][ and references therein]{Cab05,Cab11}. We account for parameter uncertainty and approximate the predictive distribution by MCMC. At tht $j$-th MCMC iteration we draw $Z_{T+h}^{(j)}$ from the conditional distribution given past observations and the parameter draw $\boldsymbol{\theta}^{(j)}$
\begin{figure}[t]
\centering
\begin{tabular}{c}
\includegraphics[height=100pt,width=300pt]{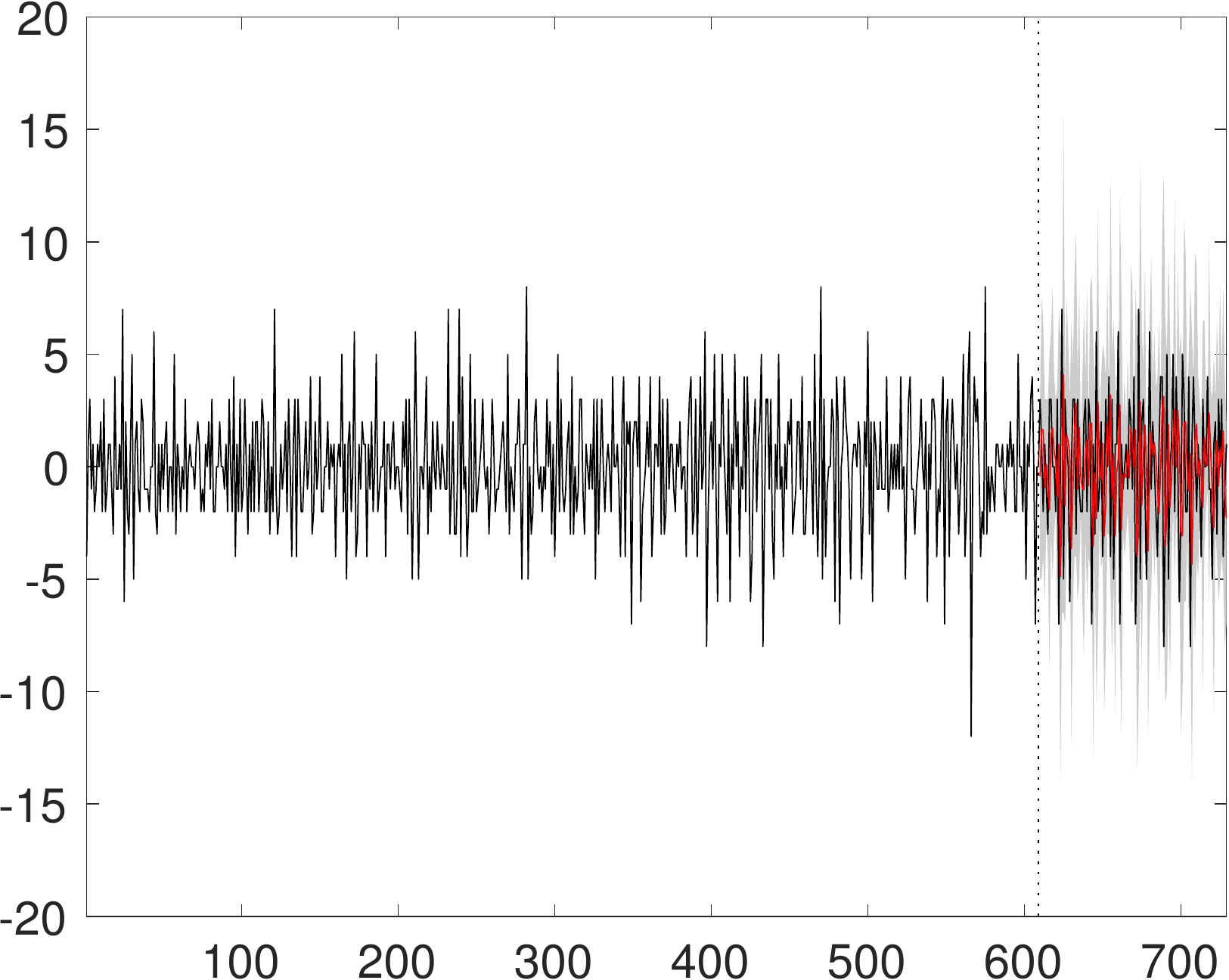} 
\end{tabular}
\caption{Changes in the number of cyber threats and their sequential one-step-ahead forecast (red line) and 95\% HPD region (gray area) between 1st November 2018 (vertical dashed line) and 31st December 2018.}
\label{Fig:ForCyb}
\end{figure}

\begin{equation}
Z_{T+h}^{(j)}\vert \mathcal{F}_{T},\boldsymbol{\theta}^{(j)} \sim GPD \left( \mu_{T+h}^{*(j)}, \sigma^{2^{*}(j)}_{T+h}, \lambda^{(j)}\right) 
\end{equation}
where $j=1,\ldots,J$, denotes the MCMC draw, $h=1,\ldots,H$ the forecasting horizon and
\begin{eqnarray}
\mu_{T+h}^{(j)} &=&\begin{cases} \alpha_{0}^{(j)}+\alpha_{1}^{(j)}Z_{T}+\beta_{1}^{(j)}\mu_{T}^{(j)}, & \mbox{for }h=1 \\ \alpha_{0}^{(j)}+\alpha_{1}^{(j)}Z^{(j)}_{T+h-1}+\beta_{1}^{(j)}\mu_{T+h-1}^{(j)}, & \mbox{for }h=2,\ldots ,H
\end{cases}\\
\sigma^{2\,(j)}_{T+h}&=&|\mu_{T+h}^{(j)}|\phi^{(j)}
\end{eqnarray}

\section{Conclusions}
We introduce a new family of stochastic processes with values in the set of integers with sign. The increments of the process follow a generalized Poisson difference distribution with time-varying parameters. We assume a GARCH-type dynamics, provide a thinning representation and study the properties of the process. We provide a Bayesian inference procedure and an efficient Monte Carlo Markov Chain sampler for posterior approximation. Inference and forecasting exercises on accidents and cyber-threats data show that the proposed GPD-INGARCH model is well suited for capturing persistence in the conditional moments and in the over-dispersion feature of the data.


\bibliographystyle{apalike}
\bibliography{GPD}

\clearpage


\appendix

\renewcommand{\thesection}{A}
\renewcommand{\theequation}{A.\arabic{equation}}
\renewcommand{\thefigure}{A.\arabic{figure}}
\renewcommand{\thetable}{A.\arabic{table}}
\setcounter{table}{0}
\setcounter{figure}{0}
\setcounter{equation}{0}

\section{Distributions used in this paper}\label{App:distrib}
\subsection{Poisson Difference distribution}
The Poisson difference distribution, a.k.a. as Skellam distribution, is a discrete distribution defined as the difference of two independent Poisson random variables $N_{1}-N_{2}$, with parameters $\lambda_{1}$ and $\lambda_{2}$. It has been introduced by \cite{Irw1937} and \cite{Ske1946}.

The probability mass function of the Skellam distribution for the difference $X=N_{1}-N_{2}$ is
\begin{equation}
P(X=x)=e ^{-(\lambda_{1}+\lambda_{2})}\left(\frac{\lambda_{1}}{\lambda_{2}}\right) ^{x/2}I_{\vert x\vert}(2\sqrt{\lambda_{1}\lambda_{2}}), \qquad \mbox{ with } X\in \mathbb{Z}
\end{equation}
where $\mathbb{Z} ={\ldots,-1,0,1,\ldots}$ is the set of positive and negative integer numbers, and $I_{k}(z)$ is the modified Bessel function of the first kind, defined as \citep{AbrSte1965}
\begin{equation}
I_{v}(z)= \left(\frac{z}{2}\right)^{2} \sum_{k=0}^{\infty} \frac{\left( \frac{z^{2}}{4}\right)^{k}}{k!\Gamma (v+k+1)}
\end{equation} 

It can be used, for example, to model the difference in number of events, like accidents, between two different cities or years. Moreover, can be used to model the point spread between different teams in sports, where all scored points are independent and equal, meaning they are single units. Another applications can be found in graphics since it can be used for describing the statistics of the difference between two images with a simple Shot noise, usually modelled as a Poisson process. 

The distribution has the following properties:
\begin{itemize}
\item Parameters: $\lambda_{1}\geq 0$, $\lambda_{2}\geq 0$
\item Support: $\lbrace-\infty, +\infty\rbrace$
\item Moment-generating function: $e^{-(\lambda_{1}+\lambda_{2})+\lambda_{1}e^{t}+\lambda_{2}e^{-t}}$
\item Probability generating function: $e^{-(\lambda_{1}+\lambda_{2})+\lambda_{1}t+\lambda_{2}/t}$
\item Characteristic function: $e^{-(\lambda_{1}+\lambda_{2})+\lambda_{1}e^{it}+\lambda_{2}e^{-it}}$
\item Moments
\begin{enumerate}
\item Mean: $\lambda_{1}-\lambda_{2}$
\item Variance: $\lambda_{1}+\lambda_{2}$
\item Skewness: $\frac{\lambda_{1}-\lambda_{2}}{(\lambda_{1}+\lambda_{2})^{3/2}}$
\item Excess Kurtosis: $\frac{1}{\lambda_{1}+\lambda_{2}}$
\end{enumerate}
\item The Skellam probability mass function is normalized: $\sum_{k=-\infty}^{+\infty} p(k;\lambda_{1},\lambda_{2})=1$ 
\end{itemize}

\subsection{Generalized Poisson distribution}\label{app:GP}
The Generalized Poisson distribution (GP) has been introduced by \cite{ConJai1973} in order to overcome the equality of mean and variance that characterizes the Poisson distribution. In some cases the occurrence of an event, in a population that should be Poissonian, changes with time or dependently on previous occurrences. Therefore, mean and variances are unequal in the data. In different fields a vastness of mixture and compound distribution have been considered, Consul and Jain introduced the GP distribution in order to obtain a unique distribution to be used in the cases said above, by allowing the introduction of an additional parameter.


See \cite{ConFam2006} for some applications of the Generalized Poisson distribution. 
Application of the GP distribution can be find as well in economics and finance. \cite{Con1989} showed that the number of unit of different commodities purchased by consumers in a fixed period of time follows a Generalized Poisson distribution. He gave interpretation of both parameters of the distribution: $\theta$ denote the basic sales potential for the commodity, while $\lambda$ the average rates of liking generated by the product between consumers. \cite{TriEtal1986} provide an application of the GP distribution in textile manufacturing industry. In particular, given the established use of the Poisson distribution in the field, they compare the Poisson and the GP distributions when firms want to increase their profit. They found that the Generalized Poisson, considering different values of the parameters,  always yield larger profits. Moreover, the Generalized Poisson distribution, as studied by \cite{Con1989}, can be used to describe accidents of various kinds, such as: shunting accidents, home injuries and strikes in industries. Another application to accidents has been carried out by \cite{FemCon1995}, where they introduced a bivariate extension to the GP distribution and studied two different estimation methods, i.e. method of moments and MLE, and the goodness of fit of the distribution in accidents statistics. \cite{HubEtal2009} test for the value of the GP distribution extra parameter by means of a Bayesian hypotheses test procedure, namely the Full Bayesian Significance Test. \cite{Fam1997} and \cite{Dem2017} provided different methods of sampling from the Generalized Poisson distribution and algorithms for sampling. 
As regard processes, the GP distribution has been used in different models. For example, \cite{ConFam1992} introduced the GP regression model, while \cite{Fam1993} studied the restricted generalized Poisson regression. \cite{WanFam1997} applied the GP regression model to households' fertility decisions and \cite{FamEtal2004} carried out an application of the GP regression model to accident data.
\cite{ZamIsm2012} develop a functional form of the GP regression model, \cite{Zam2016} introduced a few forms of bivariate GP regression model and different applications using dataset on healthcare, in particular the Australian health survey and the US National Medical Expenditure survey. \cite{Fam2015} provide a multivariate GP regression model, based on a multivariate version of the GP distribution, and two applications: to the healthcare utilizations and to the number of sexual partners.

The Generalized Poisson distribution of a random variable $X$ with parameters $\theta$ and $\lambda$ is given by
\begin{equation}\label{Eq11}
P_{x}(\theta ,\lambda)=\begin{cases} \frac{\theta(\theta + \lambda x)^{x-1}}{x!} e^{(-\theta - \lambda x)}, & x=0,1,2,\ldots \\ 0, &  \mbox{for }x>m\mbox{ if } \lambda<0.\end{cases}
\end{equation}
The GP is part of the class of general Lagrangian distributions. The GP has Generating functions and moments
\begin{itemize}
\item Parameters: 
\begin{enumerate}
\item $\theta >0$
\item $max(-1,-\theta /m)\leq\lambda \leq 1$
\item $m(\geq 4)$ is the largest positive integer for which $\theta + m\lambda > 0$ when $\lambda<0$
\end{enumerate}
\item Moment generating function (mgf): $M_{x}(\beta)=e^{\theta(e^{s}-1)}$, where $z=e^{s}$ and $u=e^{\beta}$
\item Probability generating function (pgf): $G(u)=e^{\theta(z-1)}$, where $z=ue^{\lambda(z-1)}$
\item Moments:
\begin{enumerate}
\item Mean: $\mu = \theta(1- \lambda)^{-1}$
\item Variance: $\sigma^{2}= \theta(1- \lambda)^{-3}$
\item Skewness: $\beta_{1}= \frac{1+2\lambda}{\sqrt{\theta(1-\lambda)}}$
\item Kurtosis: $\beta_{2}=3+ \frac{1+8\lambda+6\lambda^{2}}{\theta(1-\lambda)}$
\end{enumerate}
\end{itemize}
The pgf of the GP is derived by \cite{ConJai1973} by means of the Lagrange expansion, namely:
\begin{equation}\label{Eq22}
z=\sum_{x=1}^{\infty} \frac{u^{x}}{x!}\lbrace D^{x-1}(g(z))^{x} \rbrace_{z=0}
\end{equation}
\begin{equation}\label{Eq23}
\begin{split}
f(z)&=\sum_{x=0}^{\infty}\frac{u^{x}}{x!}D^{x-1}\lbrace (g(z))^{x} f'(z)\rbrace\vert_{z=0}\\
&=f(0)+\sum_{x=1}^{\infty}\frac{u^{x}}{x!}D^{x-1}\lbrace (g(z))^{x} f'(z)\rbrace\vert_{z=0}
\end{split}
\end{equation}
where $D^{x-1}=\frac{d^{x-1}}{dz^{x-1}}$.
In particular, for the GP distribution we have \citep{ConFam2006} :
\begin{equation}
f(z)=e^{\theta(z-1)}\qquad \mbox{and} \qquad g(z)=e^{\lambda(z-1)}
\end{equation}
Now, by setting $G(u)=f(z)$ we have the expression above for the pgf. (see proof in 
\paragraph{Properties}
\cite{ConJai1973}, \cite{Con1989}, \cite{ConFam1986} and \cite{ConFam2006} derived some interesting properties of the Generalized Poisson distribution. 
\begin{theorem}[\textbf{Convolution Property}]\label{Thm1}
The sum of two independent random Generalized Poisson variates, $X+Y$, with parameters $(\theta_{1},\lambda)$ and $(\theta_{2},\lambda)$ is a Generalized Poisson variate with parameters $(\theta_{1}+\theta_{2},\lambda)$.
\end{theorem}
For a proof of Th. \ref{Thm1} see \cite{ConJai1973}.
\begin{theorem}[\textbf{Unimodality}]\label{Thm2}
The GP distribution models are unimodal for all values of $\theta$ and $\lambda$ and the mode is at $x=0$ if $\theta e^{-\lambda} <1$ and at the dual points $x=0$ and $x=1$ when $\theta e^{-\lambda}=1$ and for $\theta e^{-\lambda} >1$ the mode is at some point $x=M$ such that:
\begin{equation}
(\theta-e^{-\lambda})(e^{\lambda}-2\lambda)^{-1}<M<a
\end{equation}
where $a$ is the smallest value of M satisfying the inequality 
\begin{equation}
\lambda^{2}M^{2}+M[2\lambda \theta -(\theta+2\lambda)e^{\lambda}]>0
\end{equation}
\end{theorem}
For a proof of Th. \ref{Thm2} see \cite{ConFam1986}.

\cite{ConShe1975} and \cite{Con1989} derived some recurrence relations between noncentral moments $\mu_{k}'$ and the cumulants $K_{k}$:
\begin{equation}
(1-\lambda)\mu_{k+1}'=\theta \mu_{k}'+\theta \dfrac{\partial \mu_{k}'}{\partial \theta}+\lambda \dfrac{\partial \mu_{k}'}{\partial \lambda},\qquad\qquad k=0,1,2,\ldots
\end{equation}
\begin{equation}
(1-\lambda)K_{k+1}=\lambda \dfrac{\partial K_{k}}{\partial \lambda}+\theta \dfrac{\partial K_{k}}{\partial \theta}+,\qquad\qquad k=1,2,3,\ldots .
\end{equation}

Moreover, a recurrence relation between the central moments of the GP distribution has been derived:
\begin{equation}
\mu_{k+1}=\frac{\theta k}{(1-\lambda)^{3}}\mu_{k-1}+\frac{1}{1- \lambda}\left\lbrace \dfrac{d\,\mu_{k}(t)}{dt} \right\rbrace_{t=1},\qquad\qquad k=1,2,3,\ldots
\end{equation}
where $\mu_{k}(t)$ is the central moment $\mu_{k}$ with $\theta t$ and $\lambda t$ in place of $\theta$ and $\lambda$.

\subsection{Generalized Poisson Difference distribution}\label{App:GPD}
The random variable $X$ follows a Generalized Poisson distribution (GP) $P_{x}(\theta, \lambda)$ if and only if
\begin{equation}
P_{x}(\theta, \lambda)=\frac{\theta(\theta +x\lambda)^{x-1}}{x!} e^{-\theta-x\lambda}\;\;\;\;\;\;\;\;\; x=0, 1, 2, \ldots
\end{equation}
with $\theta >0$ and $0\leq \lambda < 1$.

Let $X\sim P_{x}(\theta_{1}, \lambda)$ and $Y\sim P_{y}(\theta_{2}, \lambda)$ be two independent random variables, \cite{Con1986} showed that the probability distribution of $D$, $D=X-Y$, is:

\begin{equation}\label{Eq31}
P(D=d)=e^{-\theta_{1}-\theta_{2}-d\lambda} \sum_{y=0}^{\infty} \frac{\theta_{2}(\theta_{2}+y\lambda)^{y-1}}{y!} \frac{\theta_{1}(\theta_{1}+(y+d)\lambda)^{y+d-1}}{(y+d)!} e^{-2y\lambda}
\end{equation}
where d takes all integer values in the interval $(-\infty, +\infty)$. As for the GP distribution, we need to set lower limits for $\lambda$ in order to ensure that there are at least five classes with non-zero probability when $\lambda$ is negative. Hence, we set
\begin{equation*}
\max (-1,-\theta_{1}/m_{1}) < \lambda < 1
\end{equation*}
\begin{equation*}
\max (-1,-\theta_{2}/m_{2}) < \lambda < 1
\end{equation*}
where, $m_{1},m_{2}\geq 4$ are the largest positive integers such that $\theta_{1}+m_{1}\lambda >0$ and $\theta_{2}+m_{2}\lambda >0$.
\begin{prop}
The probability distribution in (\ref{Eq31}) can be written as
\begin{equation}\label{Eq32}
P(D=d)=e^{-(\theta_{1}+\theta_{2})}\cdot e^{-d\lambda} \sum_{y=0}^{\infty} \frac{\theta_{1}\theta_{2}}{y!(y+d)!} (\theta_{2}+y\lambda)^{y-1} (\theta_{1}+(y+d)\lambda)^{y+d-1} \cdot e^{-2y\lambda}=C_{d}(\theta_{1},\theta_{2},\lambda).
\end{equation} 
\end{prop}

Therefore, equation (\ref{Eq32}) is the pgf of the difference of two GP variates, from which is possible to obtain the following particular cases:
\begin{equation}\label{eq33}
\begin{cases}C_{d}(\theta_{1},0,\lambda)=\frac{\theta_{1}(\theta_{1}+d\lambda)^{d-1}}{d!} e^{-\theta_{1}-d\lambda}, & \mbox{for }d=0, 1, 2, \ldots \\ C_{d}(0,\theta_{2},\lambda)=\frac{\theta_{2}(\theta_{2}-d\lambda)^{-d-1}}{(-d)!} e^{-\theta_{2}+d\lambda}, & \mbox{for } d=0, -1, -2, \ldots \\ C_{d}(\theta_{1},\theta_{2},0)=e^{-\theta_{1}-\theta_{2}} (\theta_{1})^{d/2}(\theta_{2})^{-d/2}\mathit{I}_{d}(2\sqrt{\theta_{1}\theta_{2}}), & \\
\end{cases}
\end{equation}
where d is any integer (positive, 0 or negative) and $\mathit{I}_{d}(z)$ is the modified Bessel function of the first kind, of order d and argument $z$.

The last result in equation (\ref{eq33}) is the Skellam distribution \citep{Ske1946}. Therefore, the Skellam distribution is a particular case of the difference of two GP variates.

By \cite{ConShe1973}:
\begin{equation*}
G_{1}(u)=e^{\theta_{1}(t_{1}-1)} \mbox{ , where } t_{1}=ue^{\lambda(t_{1}-1)}
\end{equation*}
and
\begin{equation*}
G_{2}(u)=e^{\theta_{2}(t_{2}-1} \mbox{ , where } t_{2}=u^{-1}e^{\lambda(t_{2}-1)}.
\end{equation*}
Therefore, given that $G(u)=G_{1}(u)G_{2}(u)$, the probability generating function (pgf) of the random variable $D=X-Y$ is
\begin{equation}
G(u)=\exp[(\theta_{1}(t_{1}-1)+(\theta_{2}(t_{2}-1)].
\end{equation}

From the cumulant generating function
\begin{equation*}
\psi(\beta)=\frac{(T_{1}-\beta)\theta_{1}}{\lambda}+\frac{(T_{2}+\beta)\theta_{2}}{\lambda}
\end{equation*}
where
$T_{1}=\beta+\lambda(e^{T_{1}}-1)$ and $T_{2}=-\beta+\lambda(e^{T_{2}}-1)$, it is possible to define the mean,variance, skewness and kurtosis of the distribution. 
\begin{equation}\label{Eq:meanGPD}
L_{1}=\frac{(\theta_{1}-\theta_{2})}{1-\lambda} \mbox{   is the first cumulant and the mean.}
\end{equation}

\begin{equation}\label{Eq:VarGPD}
L_{2}=\frac{(\theta_{1}+\theta_{2})}{(1-\lambda)^{3}} \mbox{   is the second cumulant and the variance.}
\end{equation}

\begin{equation}\label{Eq:SkwGPD}
\beta_{1}=\frac{(\theta_{1}-\theta_{2})^{2}}{(\theta_{1}+\theta_{2})^{3}}\frac{(1+2\lambda)^{2}}{1-\lambda} \mbox{ is the skewness.}
\end{equation}

\begin{equation}\label{Eq:KurGPD}
\beta_{2}=3+\frac{1+8\lambda+6\lambda^{2}}{(\theta_{1}+\theta_{2})(1-\lambda)} \mbox{ is the kurtosis.}
\end{equation}
\\

In Fig. \ref{fig:lambda}-\ref{fig:theta}, we show the GPD for various setting of $\lambda$, $\sigma^2$ and $\mu$.

\begin{figure}[H]
\centering
\begin{tabular}{cc}
(a) GPD for $\mu=-4$ and $\sigma^{2}=8$ &(b) GPD for $\mu=4$ and $\sigma^{2}=8$\\
\\
\includegraphics[scale=0.4]{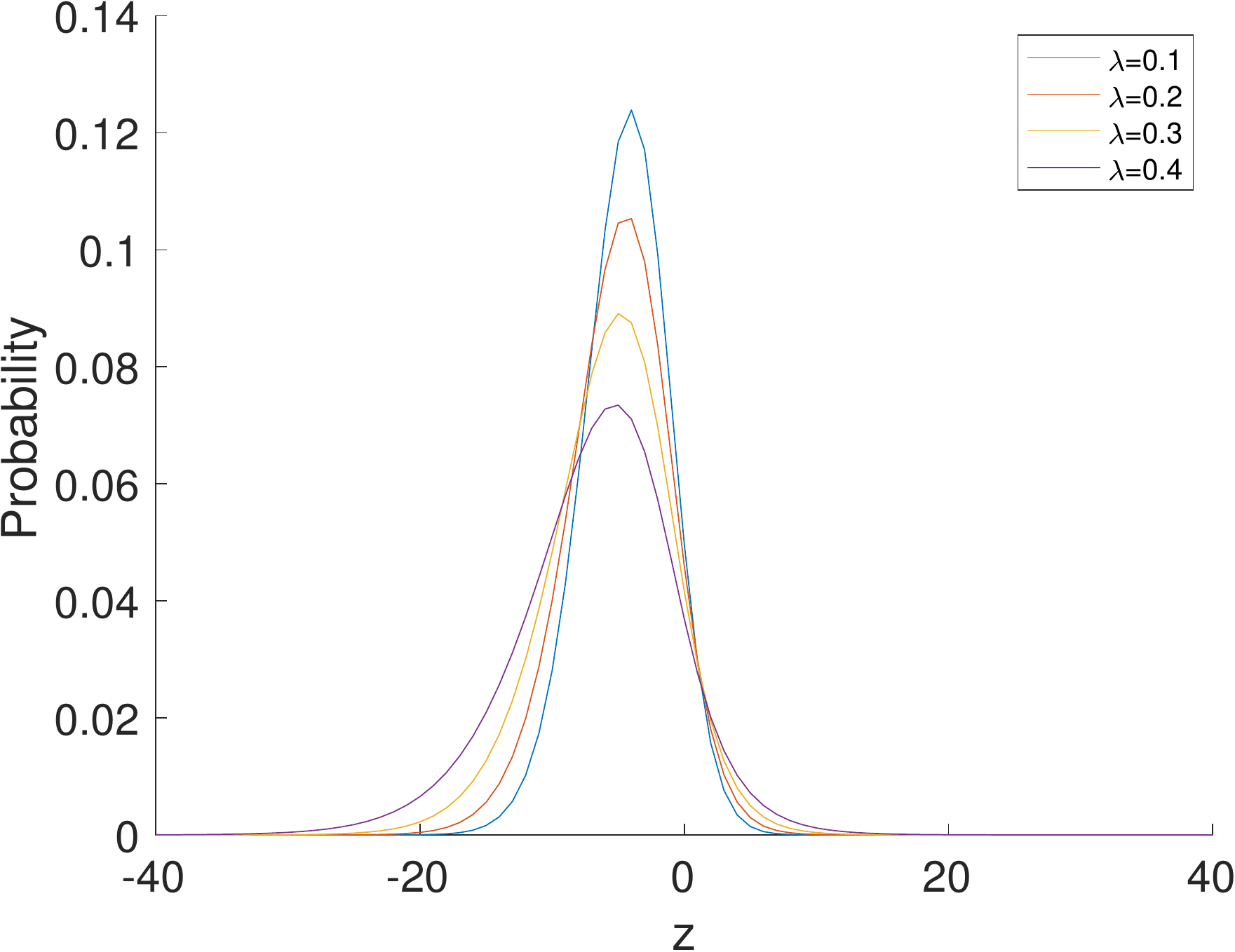}&
\includegraphics[scale=0.4]{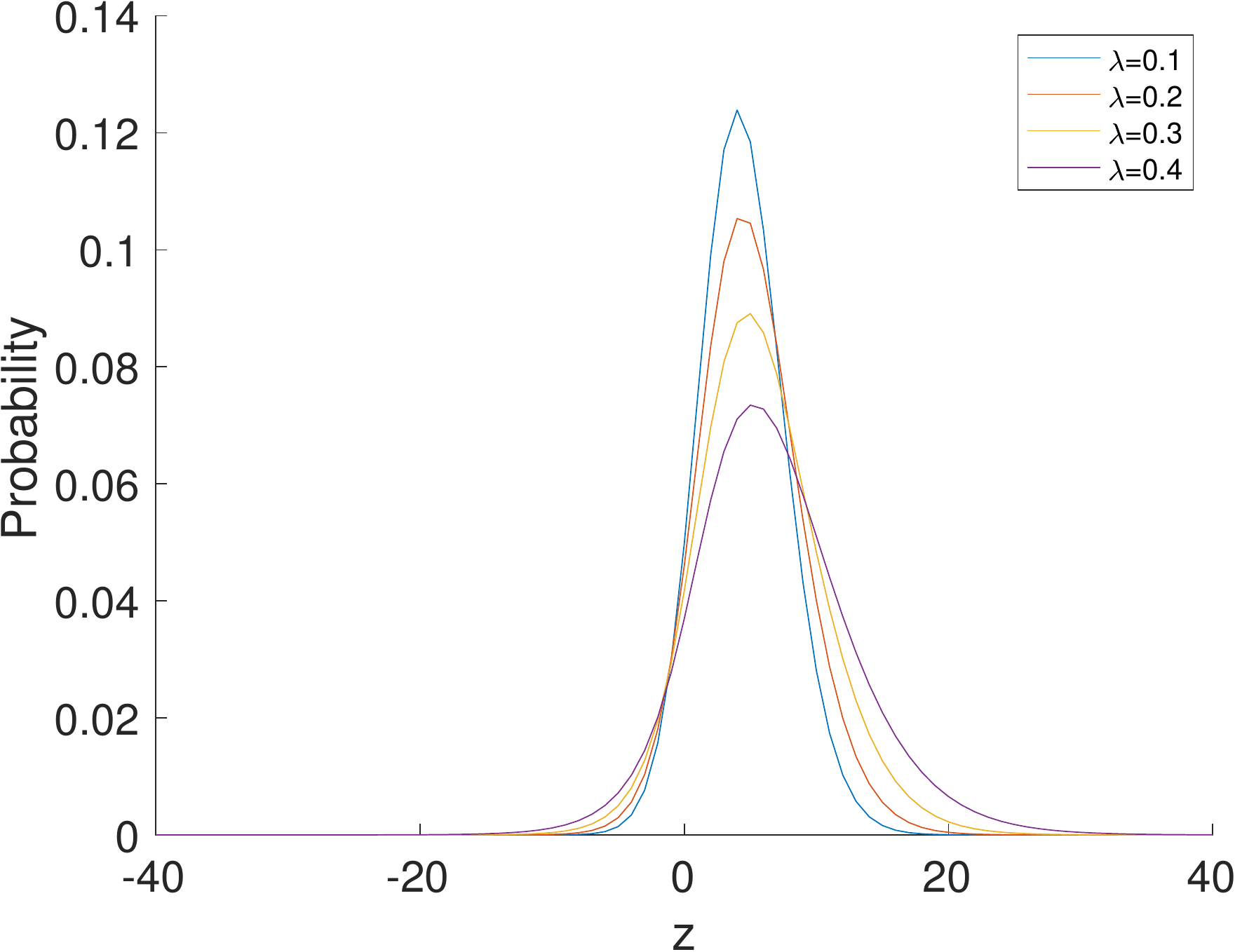}\\
\\
(c) GPD for $\mu=2$ and $\sigma^{2}=5$ &(d) GPD for $\mu=2$ and $\sigma^{2}=15$\\
\\
\includegraphics[scale=0.4]{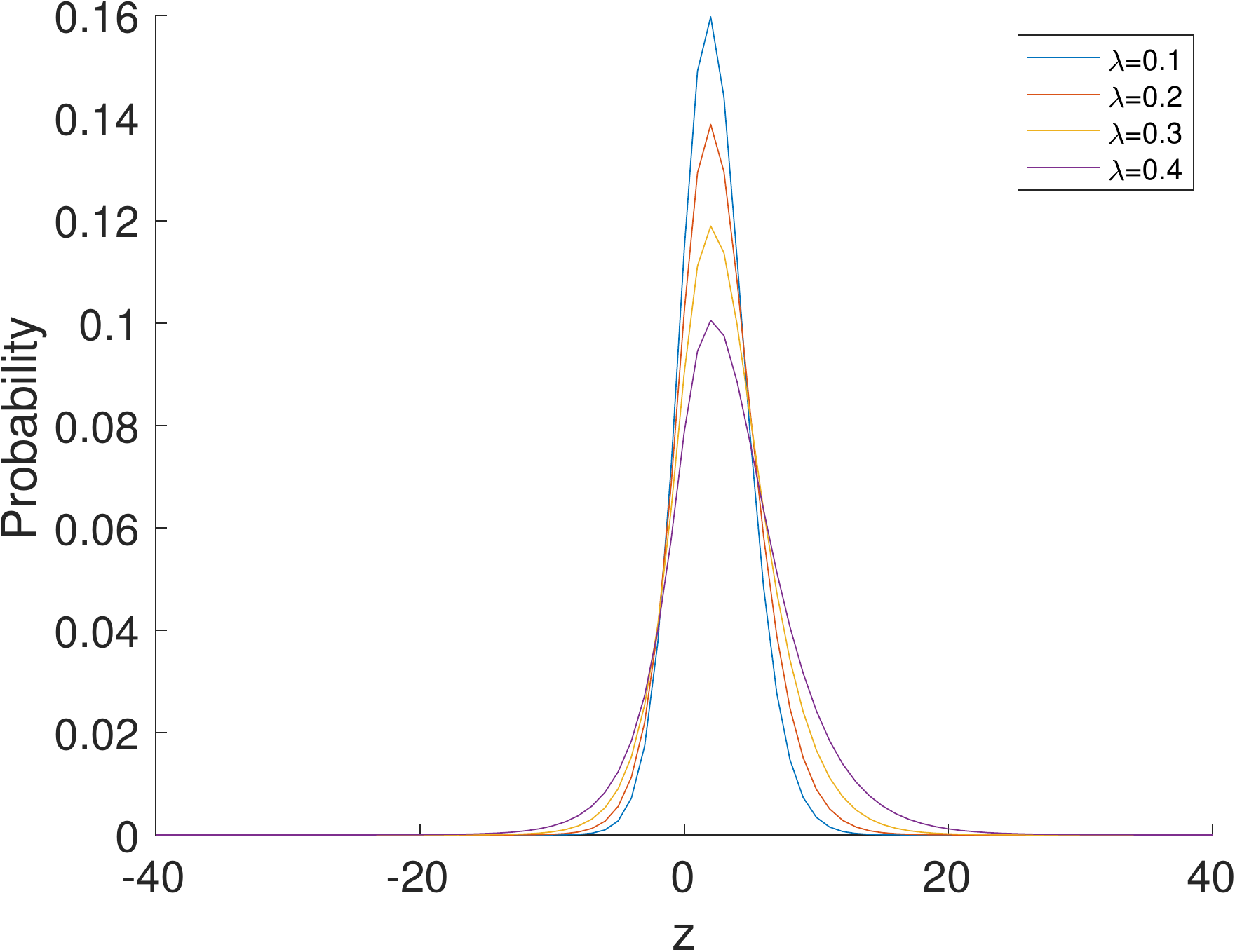}&
\includegraphics[scale=0.4]{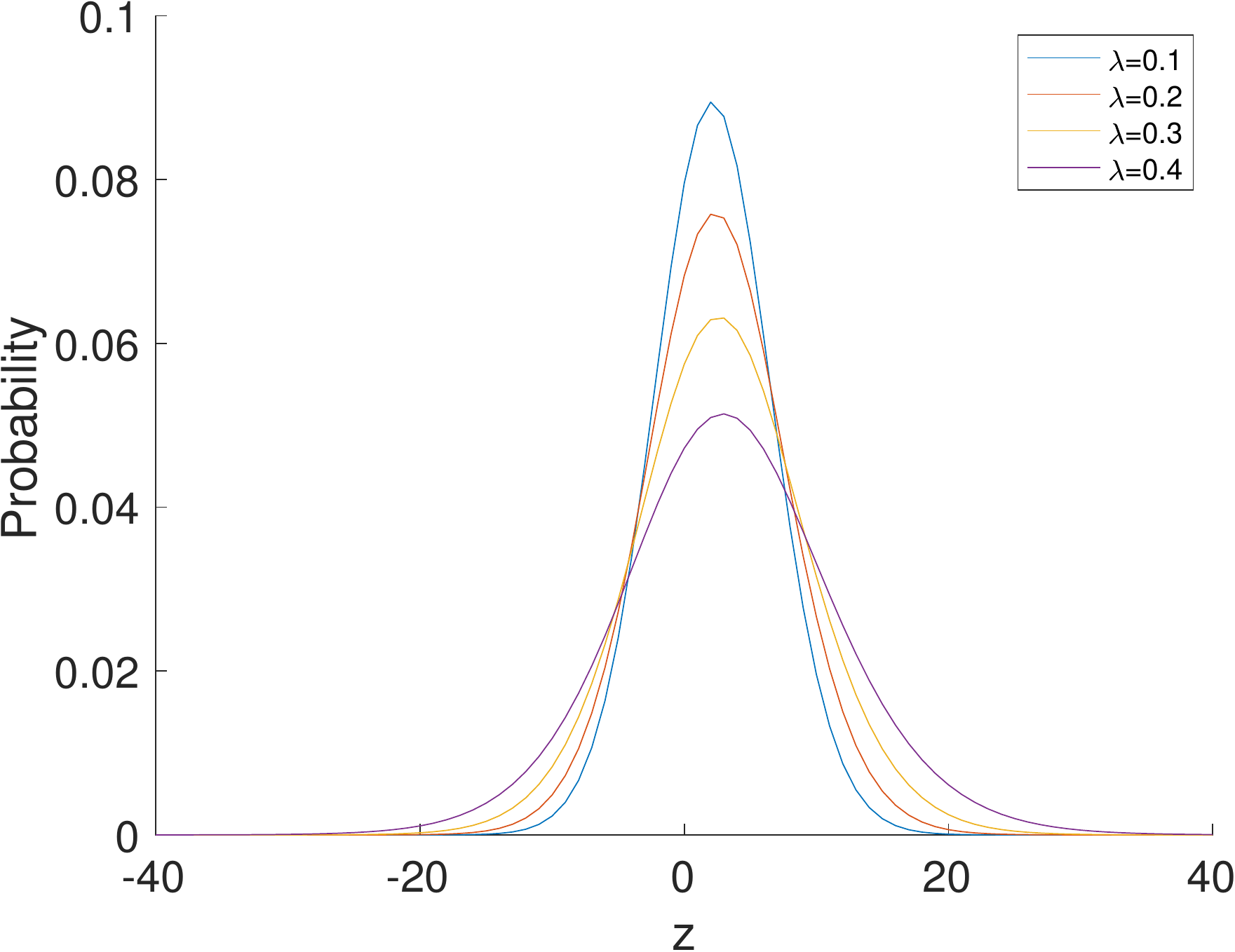}
\end{tabular}
\caption{Generalized Poisson difference distribution for different values of $\lambda$. Panels (a) and (b) show the GPD when $\lambda$ varies, for a fixed value of $\sigma^{2}=8$ and two different values of $\mu$. Panels (c) and (d) show the GPD when $\lambda$ varies, for a fixed value of $\mu=2$ and two different value of $\sigma^{2}$.}
\label{fig:lambda}
\end{figure}

\begin{figure}[H]
\centering
\begin{tabular}{cc}
(a) GPD for different values of $\mu$ &(b) GPD for different values of $\sigma^{2}$\\
\\
\includegraphics[scale=0.4]{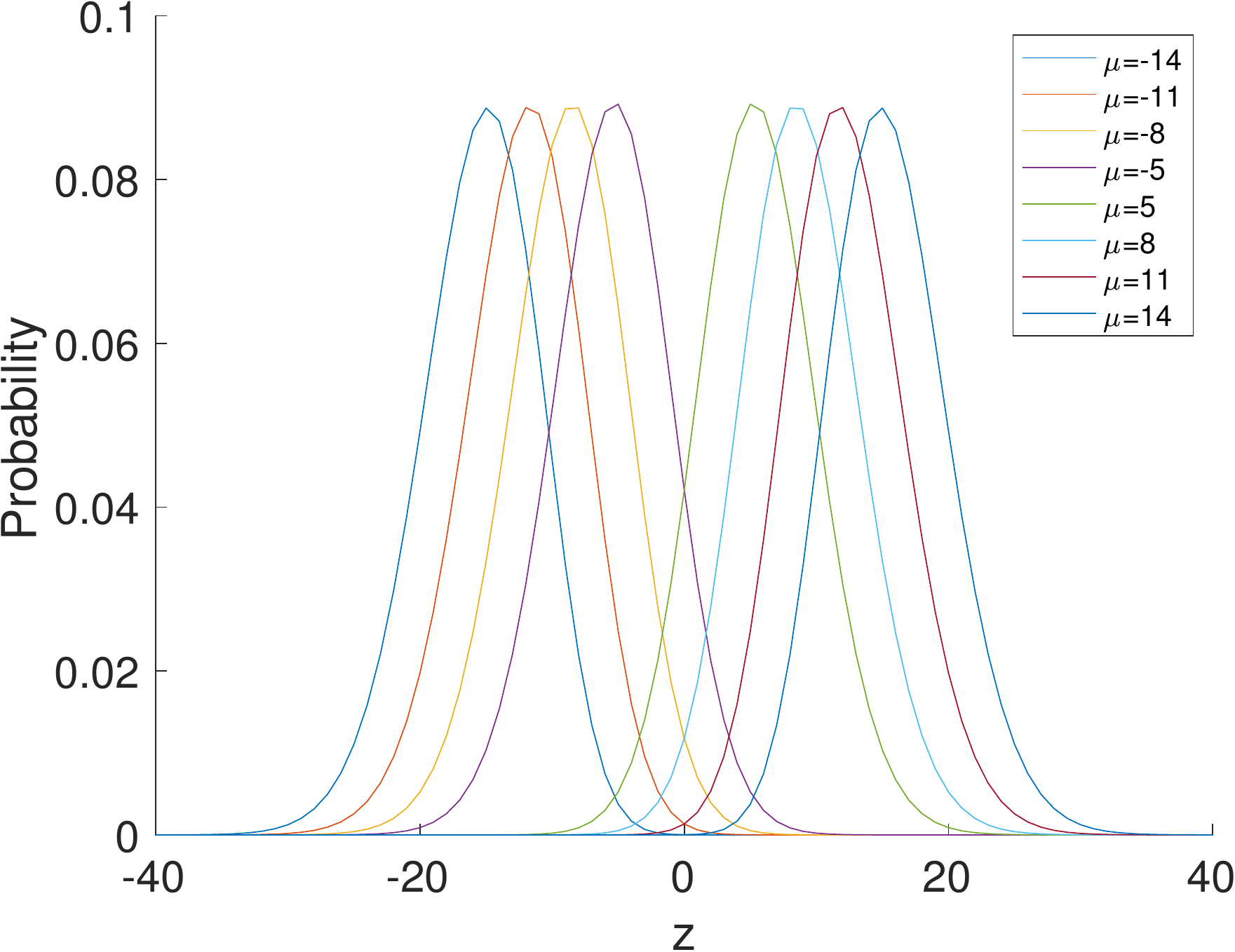}&
\includegraphics[scale=0.4]{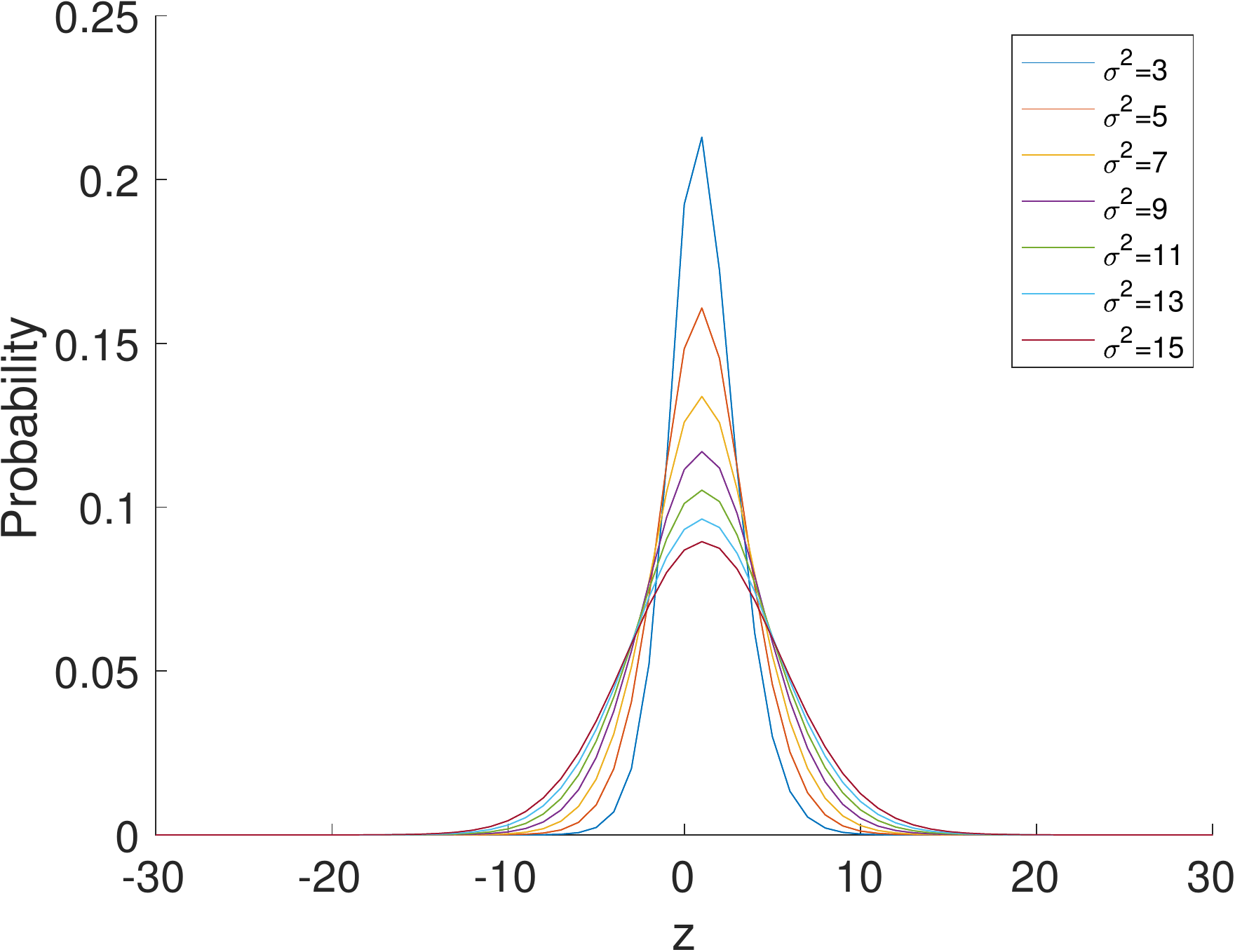}\\ 
\end{tabular}
\caption{Generalized poisson difference distribution for different values of $\mu$ and $\sigma$ and fixed $\lambda$. In panel (a) $\mu$ varies, while $\sigma^{2}=15$ and $\lambda=0.1$ are fixed. In panel (b) $\sigma^{2}$ varies and $\mu=1$ and $\lambda=0.1$ are fixed.}
\label{fig:theta}
\end{figure}

Figure \ref{fig:lambda} shows how the GPD distribution varies when $\lambda$ varies. Given the constraints imposed to the parameters (see section \ref{Model}) here $\lambda=(0,0.1,0.2,0.3,0.4)$ since smaller values and, possibly, negative values do not met the conditions for the selected values of $\theta_{1}$ and $\theta_{2}$. From panel (a) and especially in panel (b) can be seen that when $\lambda$ increases the distribution becomes longer tailed.
From panel (c) and (d) we can see that for fixed values $\mu$, when $\lambda$ decreases, the GPD is more skewed respectively to the right for $\mu>0$ ($\theta_{1}>\theta_{2}$) and to the left for $\mu<0$ ($\theta_{1}<\theta_{2}$). Therefore, the sign of $\mu$ determines the skewness of the GPD. \\

From figure \ref{fig:theta} we can see again, that for positive values of $\mu$ the distribution becomes more right-skewed, panels (a) and (b), and more left-skewed for negative values of $\mu$ in panels (c) and (d). Moreover, here can be seen better that has $\theta_{1}$ increases the distribution has longer tails. 
\subsection{Quasi-Binomial distribution}\label{App:QBth}
A first version of the Quasi-Binomial distribution, defined as QB-I by \cite{ConFam2006}, was investigated by \cite{Con1990} as an urn model. In their definition of the QB thinning, however, \cite{AlzAlO1993} used the QB distribution introduced in the literature by \cite{ConMit1975} and defined by\cite{ConFam2006} as QB-II.
\begin{equation}
P(X=x)=	\binom{n}{x} \frac{ab}{a+b} \frac{(a+x\theta)^{x-1}(b+n\theta - x\theta)^{n-x-1}}{(a+b+n\theta)^{n-1}}
\end{equation}
for $x=0,1,2,\ldots,n$ and zero otherwise. where $a>0$, $b>0$ and $\theta >-a/n$.
However, \cite{AlzAlO1993}, when defining the QB thinning operator, used a particular case of the QB-II distribution:
\begin{equation}
P(X=x)=pq \binom{n}{x} \frac{(p+x\theta)^{x-1}(q+(n-x)\theta)^{n-x-1}}{(1+n\theta)^{n-1}}
\end{equation}
where $a=p$, $q=b$, $0 < q=1-p <1$ and $\theta$ is such that $n\theta < \min(p,q)$. We denote the QB-II with QB(p,$\theta$,n). For large $n$, such that $np \rightarrow \lambda$, the QB distribution tends to the Generalized Poisson distribution.\\

The quasi-binomial (QB) thinning has been introduced by \cite{AlzAlO1993} as a generalization of the binomial thinning to model processes with GP distribution marginals. Unlike the binomial thinning, the QB thinning is able to obtain the distribution of the corresponding innovation. In particular, the authors argued in many counting process is more suitable to consider that the probability of retaining an element may depend on the time and/or the number of already retained elements. They assumed that, at time $t$, the number of retained elements follows a QB distribution. Using the notation in \cite{Wei2008}, the QB thinning is defined as follows:

\begin{prop}[\textbf{Quasi-Binomial Thinning}]
Let $\rho_{\theta,\lambda} \circ$ be the quasi-binomial thinning operator such that $\rho_{\theta,\lambda} \circ$ follows a QB($\rho$,$\theta /\lambda$,$x$). If $X$ follows a GP($\lambda$,$\theta$) distribution and the quasi-binomial thinning is performed independently on $X$, then $\rho_{\theta,\lambda} \circ X$ has a GP($\rho \lambda$,$\theta$) distribution.
\end{prop}
The thinned variable, $\rho_{\theta,\lambda} \circ X$, can be interpreted as the number of survivors from a population described by $X$.
\paragraph{Properties}
\begin{itemize}
\item Expected value: $E[\rho_{\theta,\lambda} \circ X]=\rho\cdot \mu_{X}$
\item Covariance: $E[\rho_{\theta,\lambda} \circ X,X]=\rho\cdot \sigma_{X}^{2}$
\end{itemize}

\medskip

\renewcommand{\thesection}{B}
\renewcommand{\theequation}{B.\arabic{equation}}
\renewcommand{\thefigure}{B.\arabic{figure}}
\renewcommand{\thetable}{B.\arabic{table}}
\setcounter{table}{0}
\setcounter{figure}{0}
\setcounter{equation}{0}

\section{Proofs of the results of the paper}\label{App:proofs}
\subsection{Convolution property of the Generalized Poisson distribution in Section \ref{App:distrib}} \label{Proof:convolution}
\begin{proof}[Proof of Lemma \ref{Lemma:convo}]
Let $X\sim GPD(\theta_{1},\theta_{2},\lambda)$ and $Y\sim GPD(\theta_{3},\theta_{4},\lambda)$. We can write each r.v. as 
\begin{equation}
X=X_{1}-X_{2}
\mbox{   and   }
Y=Y_{1}-Y_{2}
\end{equation}
where $X_{1} \sim GP(\theta_{1},\lambda)$, $X_{2} \sim GP(\theta_{2},\lambda)$, $Y_{1} \sim GP(\theta_{3},\lambda)$ and $Y_{2} \sim GP(\theta_{4},\lambda)$. Therefore we can write
\begin{equation}
\begin{split}
X+Y&=X_{1}-X_{2}+Y_{1}-Y_{2}\\
&= (X_{1}+Y_{1})-(X_{2}+Y_{2})\\
&= GP(\theta_{1}+\theta_{3},\lambda)-GP(\theta_{2}+\theta_{4},\lambda) \sim GPD(\theta_{1}+\theta_{3},\theta_{2}+\theta_{4},\lambda).
\end{split}
\end{equation}
We can generalized the result as follows. Let $Z_{i} \overset{\text{i.i.d}}{\sim} GPD(\theta_{1},\theta_{2},\lambda)$. Then we have
\begin{equation}
\sum_{i=1}^{n} Z_{i} \sim GPD(\sum_{i=1}^{n}\theta_{1i},\sum_{i=1}^{n}\theta_{2i},\lambda).
\end{equation}

In the same way we can prove that the difference of two r.v. GPD distributed, is again a GPD.\\

Let $X\sim GPD(\theta_{1},\theta_{2},\lambda)$ and $Y\sim GPD(\theta_{3},\theta_{4},\lambda)$. We can write each r.v. as 
\begin{equation}
X=X_{1}-X_{2}
\mbox{   and   }
Y=Y_{1}-Y_{2}
\end{equation}
where $X_{1} \sim GP(\theta_{1},\lambda)$, $X_{2} \sim GP(\theta_{2},\lambda)$, $Y_{1} \sim GP(\theta_{3},\lambda)$ and $Y_{2} \sim GP(\theta_{4},\lambda)$. Therefore we can write
\begin{equation}
\begin{split}
X-Y&=X_{1}-X_{2}-Y_{1}+Y_{2}\\
&= (X_{1}+Y_{2})-(X_{2}+Y_{1})\\
&= GP(\theta_{1}+\theta_{4},\lambda)-GP(\theta_{2}+\theta_{3},\lambda) \sim GPD(\theta_{1}+\theta_{4},\theta_{2}+\theta_{3},\lambda).
\end{split}
\end{equation}
\end{proof}
\subsection{Proofs of the results in Section \ref{Model}}
\begin{proof}[Proof of Remark \ref{RemEquiv}]
Let $X\sim GP(\theta_{1},\lambda)$ and $Y\sim GP(\theta_{2},\lambda)$, then $Z=(X-Y)$ follows a $GPD(\theta_{1}, \theta_{2}, \lambda)$. We know that $P(Z=z)=P(X=x)\cdot P(Y=y)$, therefore we can name $S=Y$ and we substitute $S$ in $Z=X-Y$, obtaining $X=S+Z$. Now we can write:
\begin{equation}\label{eq35}
P(Z=z)=e^{-(\theta_{1}+\theta_{2})}\cdot e^{-z\lambda} \sum_{s=max(0,-z)}^{+\infty} \frac{\theta_{1}\theta_{2}}{s!(s+z)!} (\theta_{2}+\lambda s)^{s-1} (\theta_{1}+\lambda(s+z))^{s+z-1} \cdot e^{-2\lambda s}
\end{equation}
which is the probability of a $GPD(\theta_{1}, \theta_{2}, \lambda)$. We can now introduce the new parametrization of the probability density function of the GPD. Define
\begin{equation}\label{e214r}
\begin{cases} \mu = \theta_{1}-\theta_{2}\\
\sigma^{2}=\theta_{1}+\theta_{2}\end{cases},
\end{equation}
thus, we can rewrite both parameters $\theta_{i}$, $i=1,2$, with respect to $\mu$ and $\sigma^{2}$:
\begin{equation}\label{eq:25}
\begin{cases} \theta_{1}= \frac{\sigma^{2}+\mu}{2}\\
\theta_{2}=\frac{\sigma^{2}-\mu}{2}\end{cases}
\end{equation}
By substituting  \ref{eq:25} into equation \ref{eq35} we have
\begin{small}
\begin{eqnarray}
P(Z=z)&=& e^{(\frac{\sigma^{2}+\mu}{2}+\frac{\sigma^{2}-\mu}{2})}e^{-z\lambda}\sum_{s=max(0,-z)}^{+\infty} \frac{\left(\frac{\sigma^{2}+\mu}{2}\right)\left(\frac{\sigma^{2}-\mu}{2}\right) }{s!(s+z)!}\left[ \frac{\sigma^{2}+\mu}{2}+(s+z)\lambda\right]^{s+z-1}\nonumber\\
&\quad& \left[ \frac{\sigma^{2}-\mu}{2}+s\lambda\right]^{s-1} e^{-2\lambda s}\label{p214r}
\end{eqnarray}
\end{small}
Carrying out the operations in Eq. \ref{p214r} we obtain Eq. \ref{e214r}.
\end{proof}

\begin{proof}[Proof of Remark \ref{remMom}]
If $Z\sim GPD(\theta_{1},\theta_{2},\lambda)$ in the parametrization of \cite{Con1986}, the moments are given in equations \ref{Eq:meanGPD}-\ref{Eq:KurGPD}.
By using our reparametrization of the GDP
\begin{equation}\label{eq:29}
\theta_{1}= \frac{\sigma^{2}+\mu}{2},\quad \theta_{2}=\frac{\sigma^{2}-\mu}{2}
\end{equation}
in equations \ref{Eq:meanGPD}-\ref{Eq:KurGPD}, we obtain
\begin{equation}
E(Z)=\frac{\theta_{1}-\theta_{2}}{1-\lambda}=\frac{\mu}{1-\lambda} 
\end{equation}
\begin{equation}
V(Z)=\frac{\theta_{1}+\theta_{2}}{(1-\lambda)^{3}}=\frac{\sigma^{2}}{(1-\lambda)^{3}} = \kappa_{2}
\end{equation}
\begin{equation}
\begin{split}
S(Z) &= \frac{\mu^{(3)}}{\sigma^{3}} =  \frac{\kappa_{3}}{\kappa_{2}^{3/2}}\\
&=\frac{\mu(1+2\lambda)}{(1-\lambda)^{5}}\left( \frac{(1-\lambda)^{3}}{\sigma^{2}}\right)^{3/2}  \\
&= \frac{\mu}{\sigma^{3}} \frac{(1+2\lambda)}{\sqrt{1-\lambda}}
\end{split}
\end{equation}
\begin{equation}
\begin{split}
K(Z) &=  \frac{\kappa_{4}+3\kappa_{2}^{2}}{\sigma^{4}}\\
&=\left( \frac{\sigma^{2}(1+8\lambda+6\lambda^{2})}{(1-\lambda)^{7}} +\frac{3\sigma^{4}}{(1-\lambda)^{6}} \right) \left( \frac{(1-\lambda)^{6}}{\sigma^{4}} \right) \\
&=3+\frac{1+8\lambda+6\lambda^{2}}{\sigma^{2}(1-\lambda)}
\end{split}
\end{equation}
where $\kappa_{i}$, $i=2,3,4$ are respectively the second, third and fourth cumulants.
\end{proof}

\begin{proof}[Proof of Proposition \ref{P1}]
Let $\psi_{j}$ be the coefficient of $z^{j}$ in the Taylor expansion of $G(z)D(z)^{-1}$. We have 
\begin{equation}\label{EqG3}
\begin{split} \mu &= E[Z_{t}] \\ 
&= E[E[Z_{t}\vert \mathcal{F}_{t-1}]]\\
&= E\left[ \alpha_{0}D^{-1}(1)+\sum_{j=1}^{\infty} \psi_{j} Z_{t-j}\right] \\
&= \alpha_{0}D^{-1}(1)+\mu D^{-1}(1)G(1) \\
\end{split}
\end{equation}
\begin{equation}\label{EqG4}
\Rightarrow \mu = \frac{\alpha_{0}}{D(1)-G(1)}=\alpha_{0} \left( 1-\sum_{i=1}^{p}\alpha_{i} -\sum_{j=1}^{q}\beta_{j} \right)^{-1}=\alpha_{0} K^{-1}(1).
\end{equation}
Where we go from line two to line three of eq \ref{EqG3} as follows
\begin{equation}
\begin{split}
E[Z_{t}\vert \mathcal{F}_{t-1}] &= E[\alpha_{0}D^{-1}(1)+H(B)Z_{t}\vert \mathcal{F}_{t-1}]\\
&= E[\alpha_{0}D^{-1}(1)]+E[H(B)Z_{t}\vert \mathcal{F}_{t-1}]\\
&= \alpha_{0}D^{-1}(1) + \sum_{j=1}^{\infty} \psi_{j} Z_{t-j}.
\end{split}
\end{equation}
Following \cite{Fer2006}, to go from line three to line four of \ref{EqG3}:
\begin{equation}
\begin{split}
E[\alpha_{0}D^{-1}(1) + \sum_{j=1}^{\infty} \psi_{j} Z_{t-j}] &=E[\alpha_{0}D^{-1}(1)] + E[\sum_{j=1}^{\infty} \psi_{j} Z_{t-j}]\\
&= \alpha_{0}D^{-1}(1)+ \sum_{j=1}^{\infty} \psi_{j} E[Z_{t-j}] \\
&=  \alpha_{0}D^{-1}(1)+ \mu H(1) \\
&= \alpha_{0}D^{-1}(1)+ \mu D^{-1}(1)G(1).
\end{split}
\end{equation}
From \ref{EqG4}, a necessary condition for the second-order stationarity of the integer-valued process $\lbrace Z_{t} \rbrace$ is: $\left( 1-\sum_{i=1}^{p}\alpha_{i} -\sum_{j=1}^{q}\beta_{j} \right)>0$.
\end{proof}
\subsection{Proof of the results in Section \ref{Sec:PropModel}}\label{Proofs:S3}
\begin{proof}[Proof of Proposition \ref{Zrepr}]
\begin{equation*}
     \begin{split}
Z_{t}^{(n)}&=X_{t}^{(n)}-Y_{t}^{(n)} \\
&= (1-\lambda) U_{1t}+(1-\lambda)\sum_{i=1}^{n} \varphi_{1i}^{(t-i)} \circ X_{t-i}^{(n-i)} - (1-\lambda) U_{2t}-(1-\lambda)\sum_{i=1}^{n} \varphi_{2i}^{(t-i)} \circ Y_{t-i}^{(n-i)} \\
&= (1-\lambda)(U_{1t} -  U_{2t}) + (1-\lambda)\sum_{i=1}^{n} \left[ (\varphi_{1i}^{(t-i)} \circ X_{t-i}^{(n-i)})- (\varphi_{2i}^{(t-i)} \circ Y_{t-i}^{(n-i)}) \right]\\
&= (1-\lambda) U_{t}+(1-\lambda)\sum_{i=1}^{n} \varphi_{i}^{(t-i)} \diamond Z_{t-i}^{(n-i)}.
\end{split}
\end{equation*}
\end{proof}

\begin{proof}[Proof of Proposition \ref{asprop}]
In order to prove the almost sure convergence of $\lbrace Z_{t}^{(n)} \rbrace$ we will prove that the difference of two sequences $\lbrace X_{t}^{(n)} \rbrace$ and $\lbrace Y_{t}^{(n)} \rbrace$ that have an almost sure convergence, will have an almost sure convergence.\\
We know that $Z_{t}^{(n)}=X_{t}^{(n)}-Y_{t}^{(n)}$, where $X_{t}^{(n)}$ and $Y_{t}^{(n)}$ are two sequences of GP random variable. From \cite{Zhu2012} we have 
\begin{equation*}
X_{n}(\omega) \overset{\text{a.s.}}{\longrightarrow} X(\omega)\Longrightarrow \mathbb{P}(\lbrace \omega : \lim_{n \to \infty} X_{n}(\omega) = X(\omega) \rbrace) = 1
\end{equation*}
and 
\begin{equation*}
Y_{n}(\omega) \overset{\text{a.s.}}{\longrightarrow} Y(\omega)\Longrightarrow \mathbb{P}(\lbrace \omega : \lim_{n \to \infty} Y_{n}(\omega) = Y(\omega) \rbrace) = 1.
\end{equation*}
Let 
\begin{equation*}
A=\lbrace \omega : \lim_{n \to \infty} X_{n}(\omega) = X(\omega) \rbrace
\end{equation*}
and
\begin{equation*}
B=\lbrace \omega \in \Omega\times\Omega : \lim_{n \to \infty} (aX_{n}(\omega) + bY_{n}(\omega))= aX(\omega)+bY(\omega) \rbrace ,\;\; \forall a,b \in \mathbb{R}.
\end{equation*}
Now we show the almost sure convergence of the sum $(aX_{n}(\omega) + bY_{n}(\omega))$.

\begin{equation}\label{asp}
\begin{split}
	\int_{\Omega} \mathbb{I}_{B}(\omega) d\mathbb{P}(\omega) 
	&= \int_{(\Omega\cap A) \cup (\Omega\cap A^{C})} \mathbb{I}_{B}(\omega) d\mathbb{P}(\omega) \\
	&= 	\int_{\Omega} \mathbb{I}_{B}(\omega) \mathbb{I}_{A}(\omega) d\mathbb{P}(\omega) + \int_{\Omega} \mathbb{I}_{B}(\omega) \mathbb{I}_{A^{C}}(\omega) d\mathbb{P}(\omega)\\
	&= \int_{\Omega} \mathbb{I}_{B}(\omega) \mathbb{I}_{A}(\omega) d\mathbb{P}(\omega) + \int_{\Omega } \mathbb{I}_{A^{C}}(\omega) \left( \int_{\Omega} \mathbb{I}_{B\cap A^{C}}(\omega) d\mathbb{P}(\omega \vert \omega ') \right) d\mathbb{P}(\omega ') \\
	&=  \mathbb{P}(B\vert A)\mathbb{P}(A) + \mathbb{P}(B\vert A^{C})\underbrace{ \mathbb{P}(A^{C})}_{=0}\\
	&=\mathbb{P}(\lbrace \omega : \lim_{n \to \infty} (aX_{n}(\omega) + bY_{n}(\omega))= aX(\omega)+bY(\omega) \rbrace \vert A)\underbrace{\mathbb{P}(A)}_{=1}\\
	&= \mathbb{P}(\lbrace \omega : a \lim_{n \to \infty} X_{n}(\omega) = aX(\omega) - bY(\omega) + bY(\omega) \rbrace) =1
	\end{split}
\end{equation}
Therefore, if $X_{n}^{(\omega)} \overset{\text{a.s.}}{\longrightarrow} X(\omega)$ and $Y_{n}^{(\omega)} \overset{\text{a.s.}}{\longrightarrow} Y(\omega)$
\begin{equation*}
\Rightarrow  aX_{n}(\omega) + bY_{n}(\omega) \overset{\text{a.s.}}{\longrightarrow}  aX(\omega)+bY(\omega), \;\; \forall a,b \in \mathbb{R}.
\end{equation*}
$\forall a,b \in \mathbb{R}$. Hence, for $a=1$ and $b=-1$, this is true for the difference $Z_{t}^{(n)}=X_{t}^{(n)}-Y_{t}^{(n)}$.
\end{proof}

\begin{proof}[Proof of Proposition \ref{msprop}]
We use again the fact that $Z_{t}^{(n)}=X_{t}^{(n)}-Y_{t}^{(n)}$ and the following lemma.
\begin{lemma}\label{lem1}
If $X_{n}(\omega)$ and $Y_{n}(\omega)$ have a mean-square limit 
\begin{equation*}
X_{n}(\omega)  \overset{L^{2}}{\longrightarrow} X(\omega)
\end{equation*}
\begin{equation}
Y_{n}(\omega)  \overset{L^{2}}{\longrightarrow} Y(\omega)
\end{equation}
also their sum will have a mean-square limit.
\begin{equation}
aX_{n}(\omega) + bY_{n}(\omega) \overset{L^{2}}{\longrightarrow}  a X(\omega) + bY(\omega), \;\;\;\; \forall a,b\in \mathbb{R}
\end{equation}
\end{lemma}
Hence, by setting $a=1$ and $b=-1$ we will obtain that Lemma \ref{lem1} will be valid also for the difference of two sequences
\begin{equation}
X_{n}(\omega) - Y_{n}(\omega)  \overset{L^{2}}{\longrightarrow} X(\omega) - Y(\omega)
\end{equation}
and we can say that $Z_{t}^{(n)}$ converges to $Z_{t}$ in $L^{2}(\Omega, \mathcal{F}, \mathbb{P})$.
\end{proof}

\begin{proof}[Proof of Proposition \ref{prop:pgf}]
\begin{equation}
\begin{split}
g_{\mathbf{Z}}(\mathbf{t}) &=\mathbb{E}\left[ \prod_{i=1}^k t_{i}^{Z_{i}}\right]= \mathbb{E}\left[ \prod_{i=1}^k t_{i}^{X_{i}}\right]\mathbb{E}\left[ \prod_{i=1}^k \frac{1}{t_{i}^{Y_{i}}}\right]= g_{\mathbf{X}}(\mathbf{t})g_{\mathbf{Y}}(\mathbf{t^{-1}})
\end{split}
\end{equation}
\end{proof}

\begin{proof}[Proof of Proposition \ref{prop:moments}]
As said before, $Z_{t}^{(n)}=X_{t}^{(n)}-Y_{t}^{(n)}$. Where $X_{t}^{(n)}$ and $Y_{t}^{(n)}$ are finite sums of independent Generalized Poisson variables and it follows that $Z_{t}^{(n)}$ is a finite sum of Generelized Poisson difference variables. As shown by \cite{Zhu2012}, the first two moments of $X_{t}$ and $Y_{t}$ are finite: $E[X_{t}]=\mu_{X}\leq C_{1}$, $E[Y_{t}]=\mu_{Y}\leq C_{1}'$, $V[X_{t}]  = \sigma^{2}_{X}\leq C_{2}$, $V[Y_{t}]=\sigma^{2}_{Y} \leq C_{2}'$, therefore, 
\begin{equation}
E[Z_{t}] = E[X_{t}] - E[Y_{t}] = \mu_{X} - \mu_{Y}\leq \mu_{X} + \mu_{Y} \leq C_{1}+C_{1}'
\end{equation}
is finite and
\begin{equation}
V[X_{t}-Y_{t}]=V[X_{t}] + V[X_{t}] = \sigma^{2}_{X} + \sigma^{2}_{Y}\leq C_{2}+C_{2}'
\end{equation}
is also finite, where $Cov(X_{t},Y_{t})=0$ given that $X_{t}$ and $Y_{t}$ are independent and where $C_{i}$ and $C_{i}'$, with $i=1,2$ are constants.
\end{proof}

\begin{proof}[Proof of the results in Example \ref{Example1}]
For $k \geq 2$:
\begin{equation}
\gamma_{Z}(k) = \alpha_{1} \gamma_{Z}(k-1) = \alpha_{1}^{k-1} \gamma_{Z}(1)
\end{equation}
For $k=1$:
\begin{equation}
\gamma_{Z}(1) = Cov(Z_{t},Z_{t-1} = \alpha_{1} \gamma_{Z}(0) = \alpha_{1} V(Z_{t} = \alpha_{1}[\phi^{3} E(\sigma^{2^{*}}_{t})] + \alpha_{1}V(\mu_{t}).
\end{equation}
For $k \geq 1$ we have
\begin{equation}
\gamma_{\mu}(k) = \alpha_{1} \gamma_{\mu}(k-1) = \alpha_{1}^{k} V(\mu_{t}).
\end{equation}
For $k=0$:
\begin{equation}
\begin{split}
\gamma_{\mu}(k) &= V(\mu_{t})= \alpha_{1} \gamma_{Z}(1) \\
&= \alpha_{1}\left\lbrace \alpha_{1}[\phi^{3} E(\sigma^{2^{*}}_{t})+ V(\mu_{t})]\right\rbrace = \alpha_{1}^{2}[\phi^{3} E(\sigma^{2^{*}}_{t})]+\alpha_{1}^{2}V(\mu_{t})
\end{split}
\end{equation}
Therefore,
\begin{equation}
V(\mu_{t}) = \frac{\alpha_{1}^{2}[\phi^{3} E(\sigma^{2^{*}}_{t})]}{1-\alpha_{1}^{2}}
\end{equation}
and
\begin{equation}
\begin{split}
V(Z_{t}) &= \phi^{3} E(\sigma^{2^{*}}_{t}) + V(\mu_{t}) \\
&= \phi^{3} E(\sigma^{2^{*}}_{t}) +\frac{\alpha_{1}^{2}[\phi^{3} E(\sigma^{2^{*}}_{t})]}{1-\alpha_{1}^{2}} = \frac{\phi^{3} E(\sigma^{2^{*}}_{t})}{1-\alpha_{1}^{2}}
\end{split}
\end{equation}
where $\phi=\frac{1}{1-\lambda}$.
Finally, the autocorrelations are derived as follows:
\begin{equation}
\rho_{\mu}(k) = \frac{\gamma_{\mu}(k)}{V(\mu_{t})} = \frac{\alpha_{1}^{k}\,V(\mu_{t})}{V(\mu_{t})} = \alpha_{1}^{k}
\end{equation}
\begin{equation}
\begin{split}
\rho_{Z}(k)&= \frac{\gamma_{Z}(k)}{V(Z_{t})} = \alpha_{1}^{k-1}\, \gamma_{Z}\,(1) \frac{1-\alpha_{1}^{2}}{\phi^{3} E(\sigma^{2^{*}}_{t})}\\
&= \alpha_{1}^{k-1}\,\frac{\alpha_{1}(1-\alpha_{1}^{2})\phi^{3} E(\sigma^{2^{*}}_{t})+\alpha_{1}^{3}\phi^{3} E(\sigma^{2^{*}}_{t})}{1-\alpha_{1}^{2}} \frac{1-\alpha_{1}^{2}}{\phi^{3} E(\sigma^{2^{*}}_{t})} = \alpha_{1}^{k}
\end{split}
\end{equation}
\end{proof}

\begin{proof}[Proof of the results in Example \ref{Example2}]
For $k \geq 2$:
\begin{equation}
\gamma_{Z}(k) = \alpha_{1} \gamma_{Z}(k-1) + \beta_{1} \gamma_{Z}(k-1) = (\alpha_{1} + \beta_{1})^{k-1} \gamma_{Z}(1)
\end{equation}

For $k=1$:
\begin{equation}
\begin{split}
\gamma_{Z}(1) &= Cov(Z_{t},Z_{t-1}) = \alpha_{1} \gamma_{Z}(0) + \beta_{1} \gamma_{\mu}(0) \\
&= \alpha_{1} V(Z_{t})+ \beta_{1} V(\mu_{t}) = \alpha_{1}[\phi^{3} E(\sigma^{2^{*}}_{t})] + (\alpha_{1}+\beta_{1})V(\mu_{t}).
\end{split}
\end{equation}
For $k \geq 1$ we have
\begin{equation}
\gamma_{\mu}(k) = \alpha_{1} \gamma_{\mu}(k-1) + \beta_{1} \gamma_{\mu}(k-1) = (\alpha_{1} + \beta_{1})^{k} V(\mu_{t}).
\end{equation}
For $k=0$:
\begin{equation}
\begin{split}
\gamma_{\mu}(k) &= V(\mu_{t})= \alpha_{1} \gamma_{Z}(1) + \beta_{1} \gamma_{\mu}(1) \\
&= \alpha_{1}\left\lbrace \alpha_{1}[\phi^{3} E(\sigma^{2^{*}}_{t})] + (\alpha_{1}+\beta_{1})V(\mu_{t})\right\rbrace + \beta_{1} \left[ (\alpha_{1} + \beta_{1}) V(\mu_{t}) \right] \\
&= \alpha_{1}^{2}[\phi^{3} E(\sigma^{2^{*}}_{t})]+(\alpha_{1} + \beta_{1})^{2}V(\mu_{t})
\end{split}
\end{equation}
Therefore,
\begin{equation}
V(\mu_{t}) = \frac{\alpha_{1}^{2}[\phi^{3} E(\sigma^{2^{*}}_{t})]}{1-(\alpha_{1} + \beta_{1})^{2}}
\end{equation}
and
\begin{equation}
\begin{split}
V(Z_{t}) &= \phi^{3} E(\sigma^{2^{*}}_{t}) + V(\mu_{t}) = \phi^{3} E(\sigma^{2^{*}}_{t}) +\frac{\alpha_{1}^{2}[\phi^{3} E(\sigma^{2^{*}}_{t})]}{1-(\alpha_{1} + \beta_{1})^{2}}\\
&= \frac{\phi^{3} E(\sigma^{2^{*}}_{t})[1-(\alpha_{1} + \beta_{1})^{2}+\alpha_{1}^{2}]}{1-(\alpha_{1} + \beta_{1})^{2}}
\end{split}
\end{equation}
where $\phi=\frac{1}{1-\lambda}$.
The autocorrelations are derived as follows:
\begin{equation}
\rho_{\mu}(k) = \frac{\gamma_{\mu}(k)}{V(\mu_{t})} = \frac{(\alpha_{1}+\beta_{1})^{k}\,V(\mu_{t})}{V(\mu_{t})} = (\alpha_{1}+\beta_{1})^{k}
\end{equation}
\begin{equation}
\begin{split}
\rho_{Z}(k)&= \frac{\gamma_{Z}(k)}{V(Z_{t})} = (\alpha_{1}+\beta_{1})^{k-1}\, \gamma_{Z}\,(1) \frac{1-(\alpha_{1} + \beta_{1})^{2}}{\phi^{3} E(\sigma^{2^{*}}_{t})[1-(\alpha_{1} + \beta_{1})^{2}+\alpha_{1}^{2}]}\\
&= (\alpha_{1}+\beta_{1})^{k-1}\, \frac{\alpha_{1}[1-\beta_{1}(\alpha_{1}+\beta_{1})]}{1-(\alpha_{1} + \beta_{1})^{2}+\alpha_{1}^{2}}
\end{split}
\end{equation}
\end{proof}

\renewcommand{\thesection}{C}
\renewcommand{\theequation}{C.\arabic{equation}}
\renewcommand{\thefigure}{C.\arabic{figure}}
\renewcommand{\thetable}{C.\arabic{table}}
\setcounter{table}{0}
\setcounter{figure}{0}
\setcounter{equation}{0}

\section{Numerical Illustration}
We consider 400 samples from a two GPD-INGARCH(1,1)	and two simulation settings: one with low persistence and the other with high persistence. The first setting has parameters $\lambda=0.4$, $\alpha_{1}=0.25$, $\beta_{1}=0.23$, $\alpha_{0}=-0.2$ and $\phi = 22.78$, while the second setting with parameters $\lambda=0.6$,$\alpha_{1}=0.53$, $\beta_{1}=0.25$, $\alpha_{0}=-0.2$ and $\phi = 26.25$.  We run the Gibbs sampler for 1,010,000 iterations, discard the first 10,000 draws to avoid dependence from initial conditions, and finally apply a thinning procedure with a factor of 100 to reduce the dependence between consecutive draws.\\

The following figures show the posterior approximation of $\alpha_{1}$,$\beta_{1}$ and $\lambda$. For illustrative purposes we report in Fig. \ref{fig:GibbsPost2}-\ref{fig:GibbsACF} the MCMC output for one MCMC draw before removing the burn-in sample and thinning, while in Fig.\ref{fig:GibbsPostABT}-\ref{fig:GibbsACFABT} the MCMC output after removing the burn-in sample and thinning.\\

\subsection{Before thinning}
\begin{figure}[H]
\centering
\setlength{\tabcolsep}{1pt}
\begin{tabular}{cc}
(a) Low persistence & (b) High persistence\\
\includegraphics[scale=0.40]{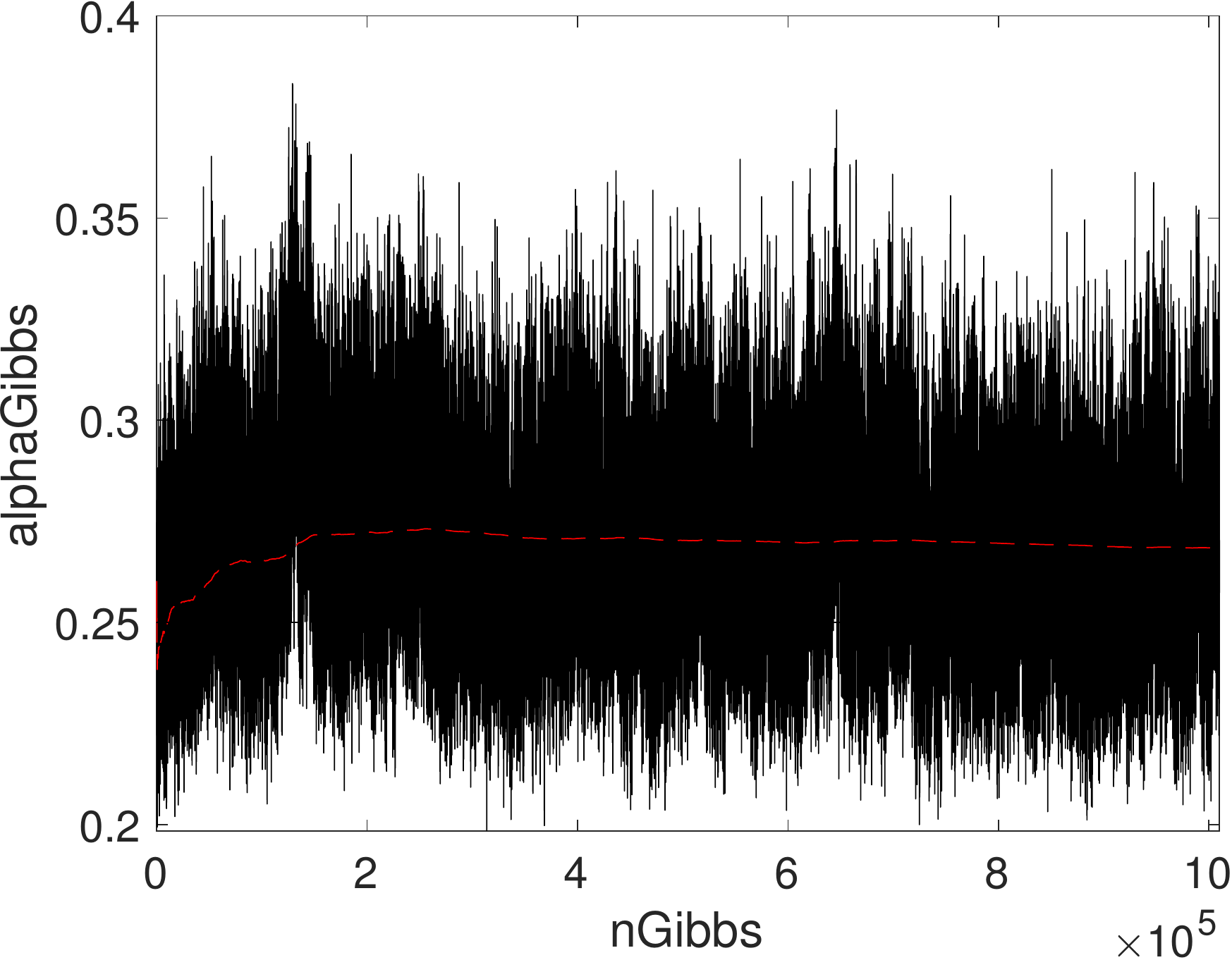} &
\includegraphics[scale=0.40]{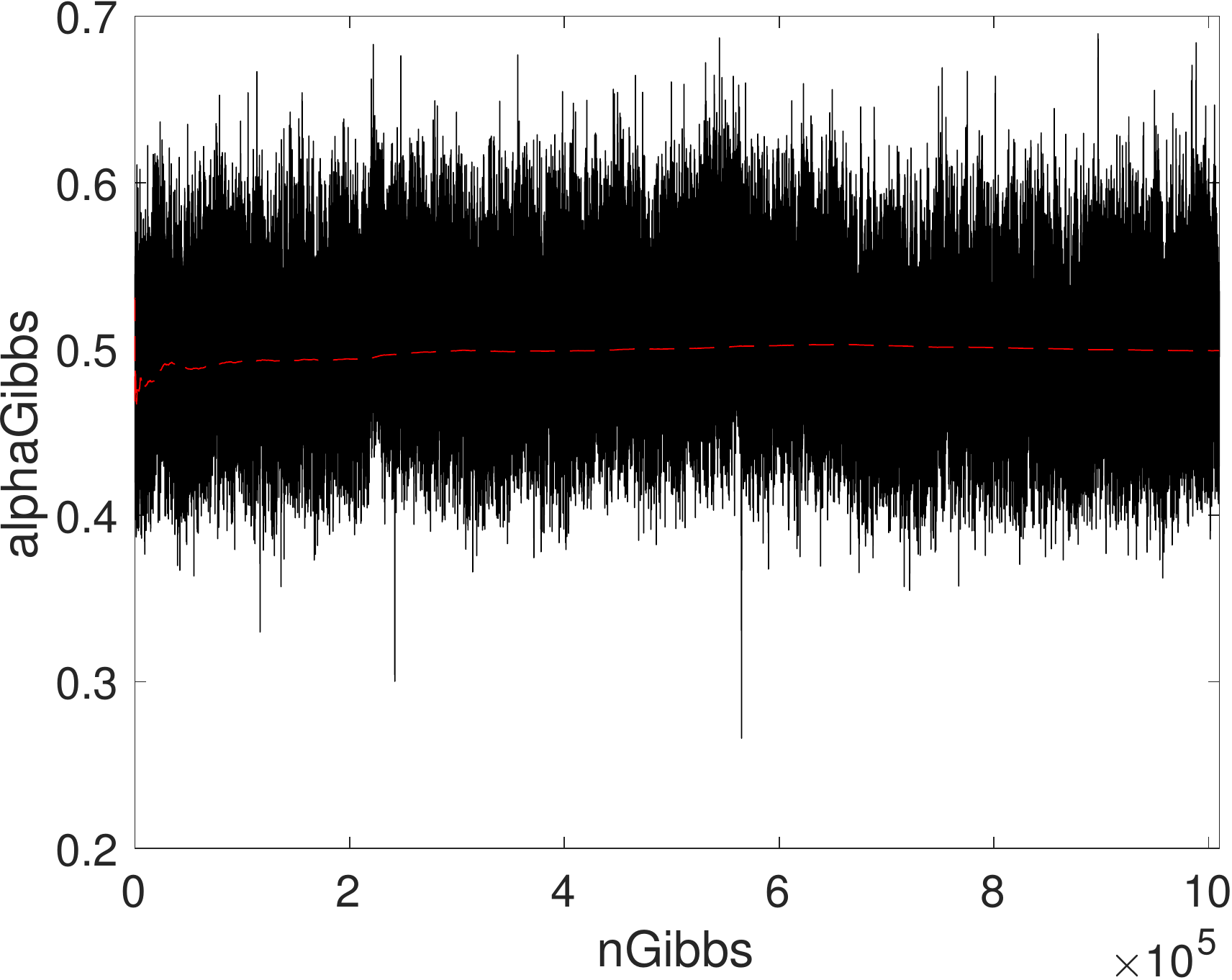}\\
\includegraphics[scale=0.40]{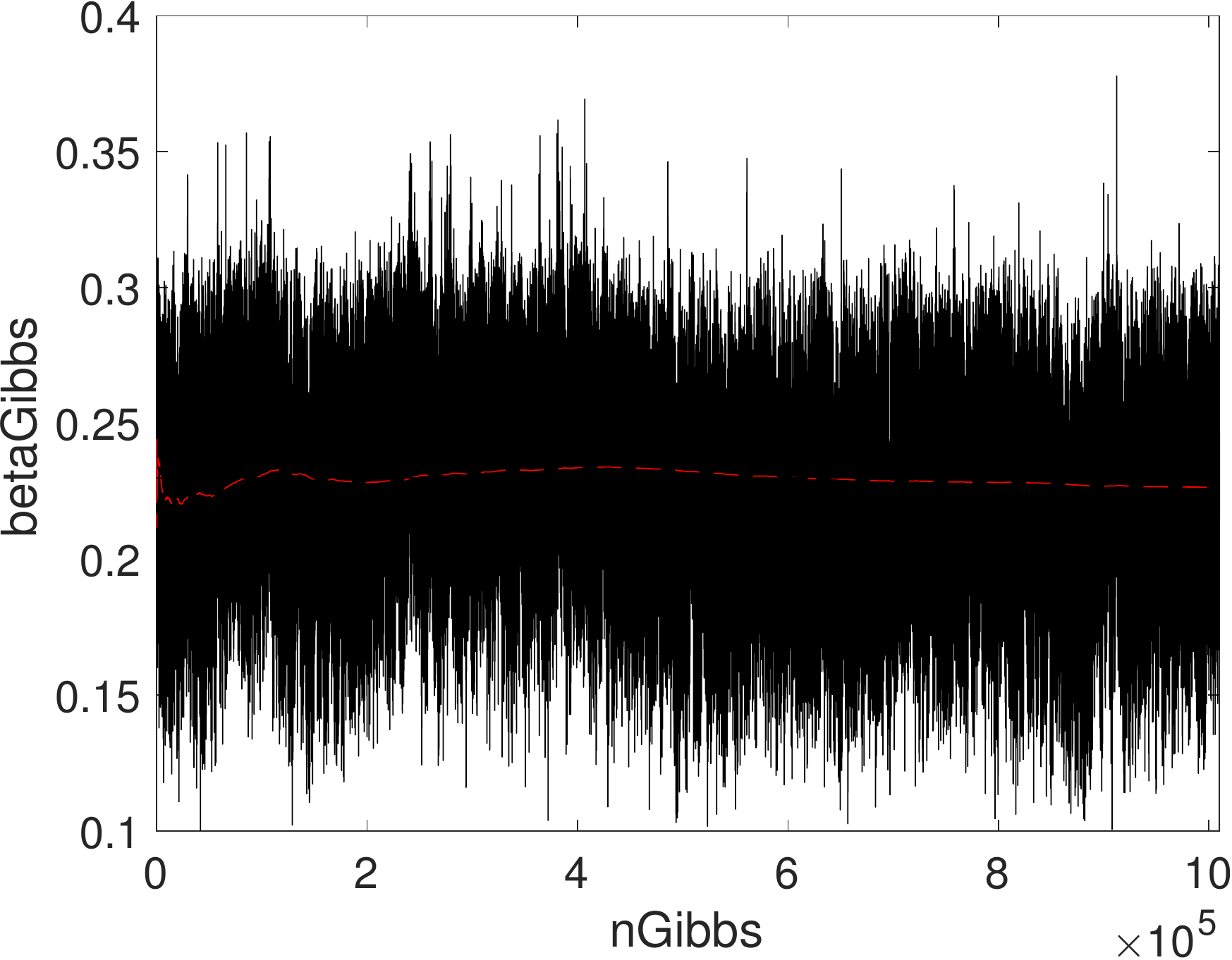} &
\includegraphics[scale=0.40]{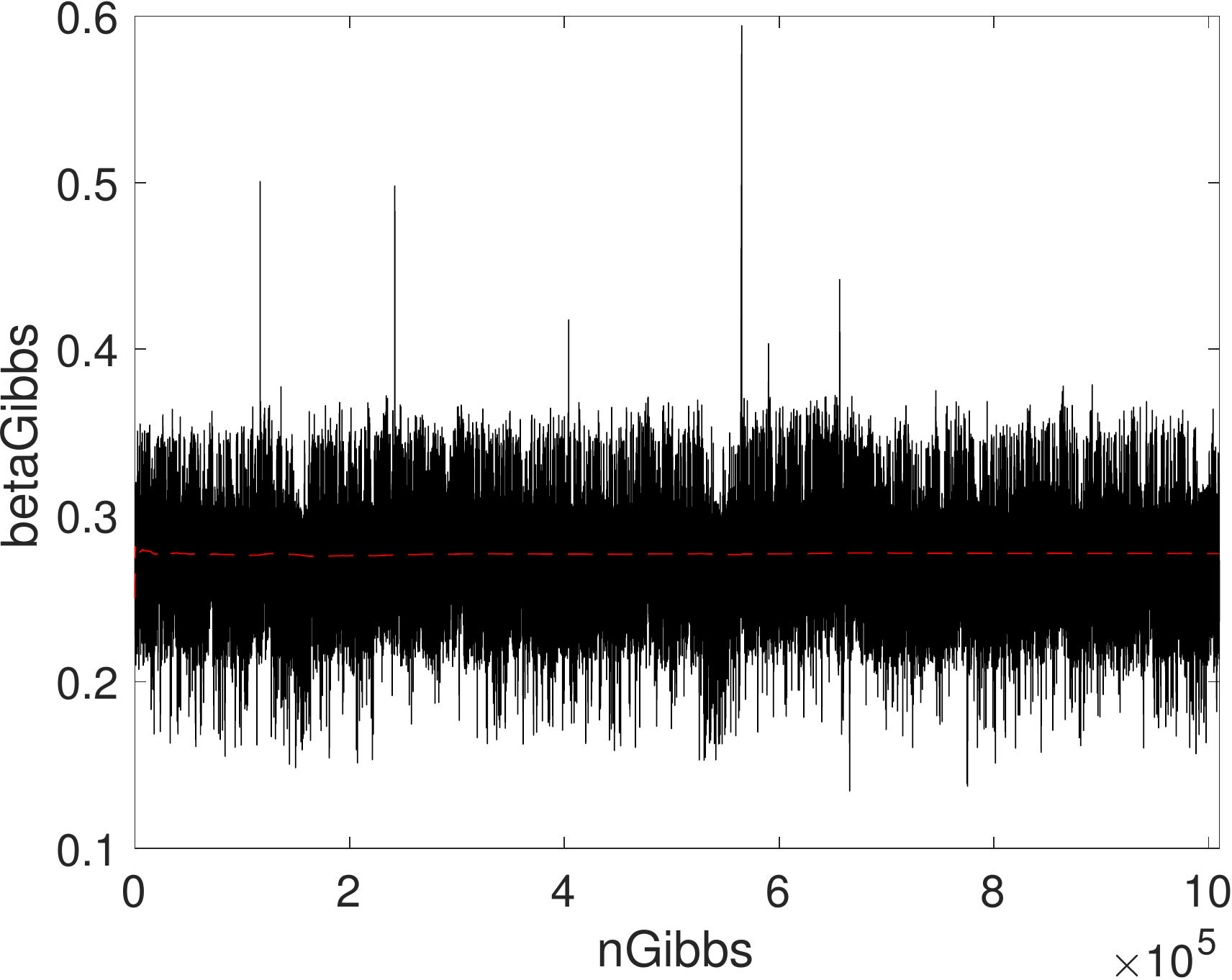}\\
\includegraphics[scale=0.40]{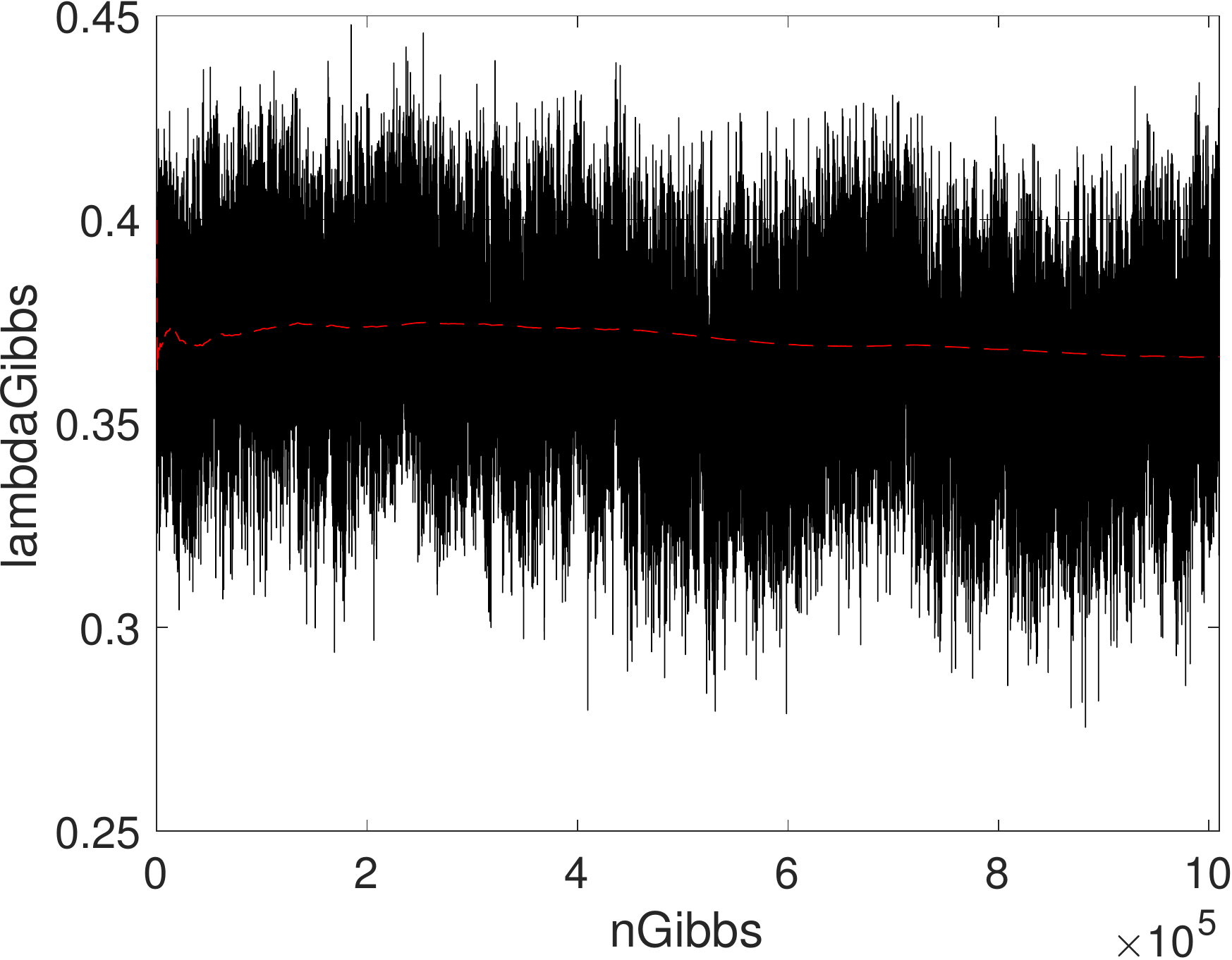} &
\includegraphics[scale=0.40]{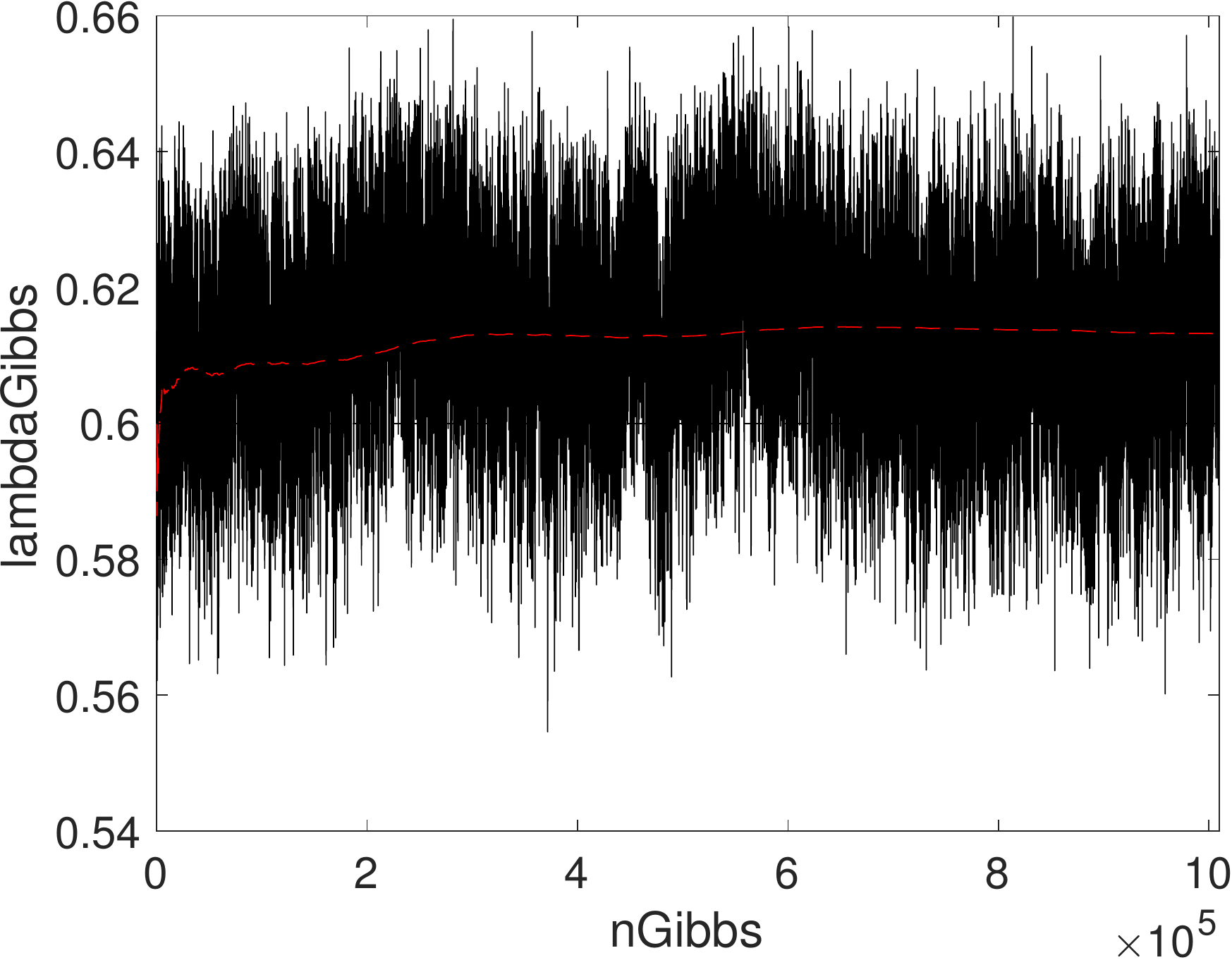} \\
\end{tabular}
\caption{MCMC plot for the parameters in the two setting: low persistence and high persistence.}
\label{fig:GibbsPost2}
\end{figure}

\begin{figure}[H]
\centering
\setlength{\tabcolsep}{-5pt}
\begin{tabular}{cc}
(a) Low persistence & (b) High persistence\\
\includegraphics[scale=0.40]{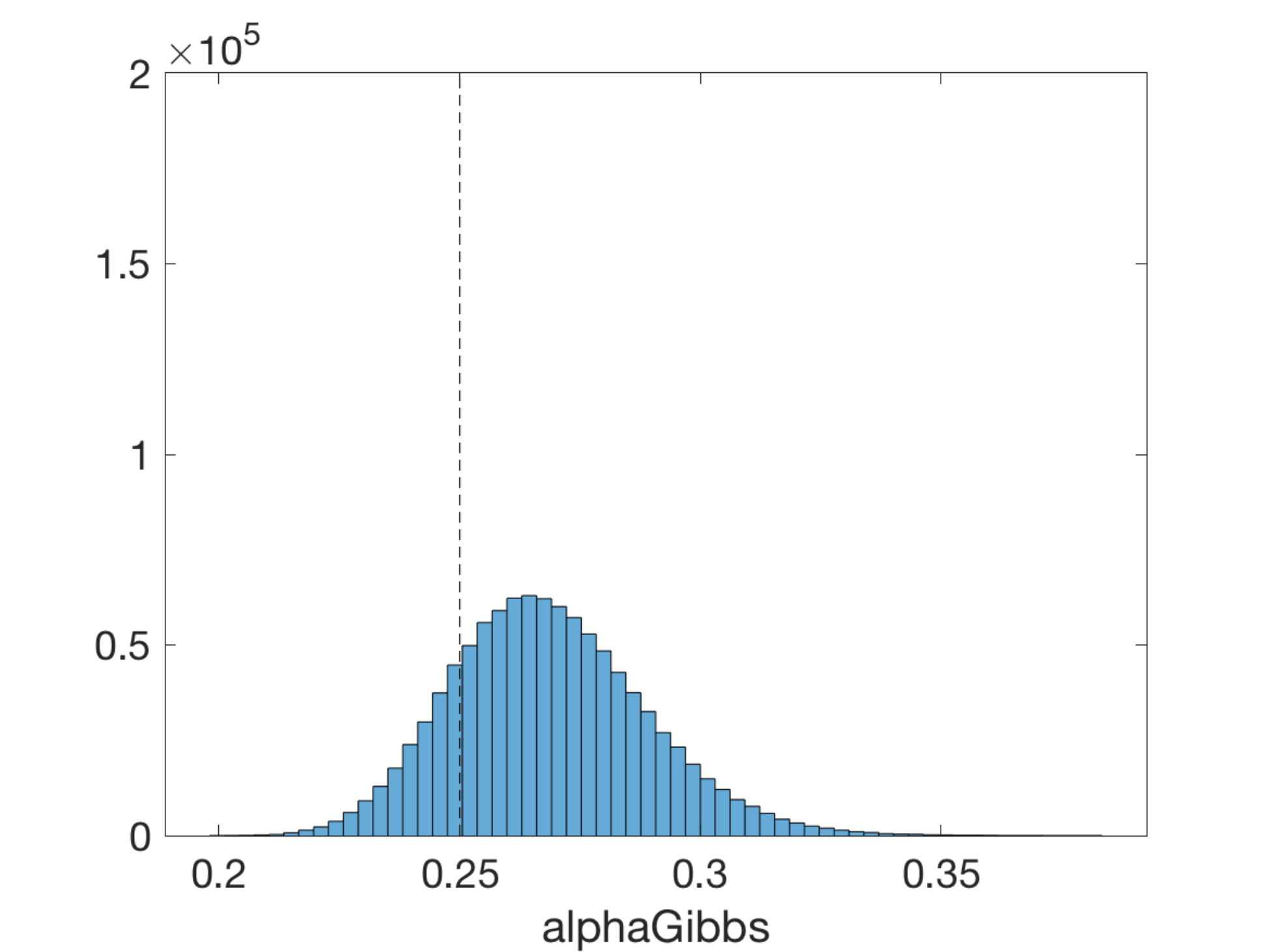} &
\includegraphics[scale=0.40]{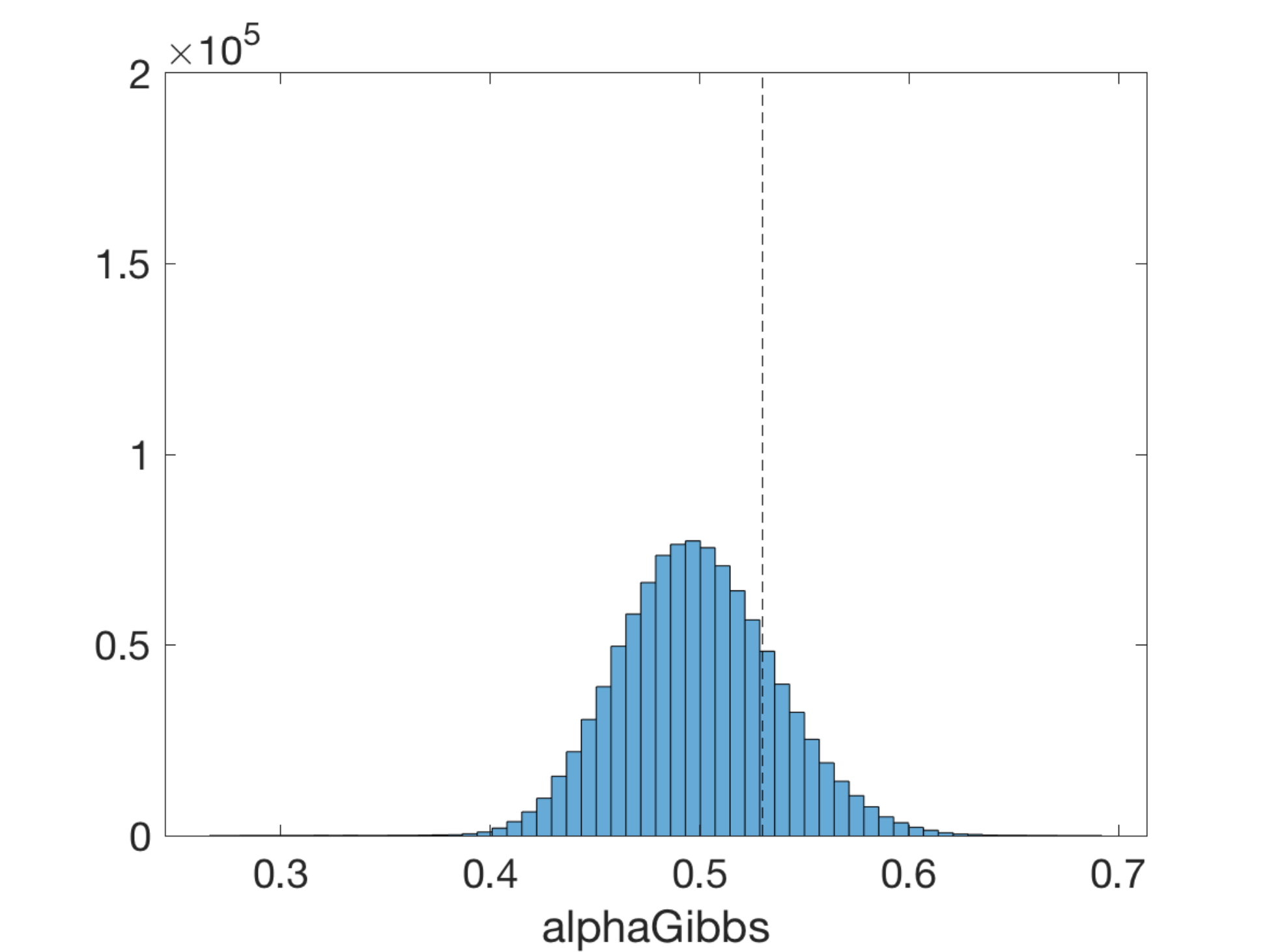}\\
\includegraphics[scale=0.40]{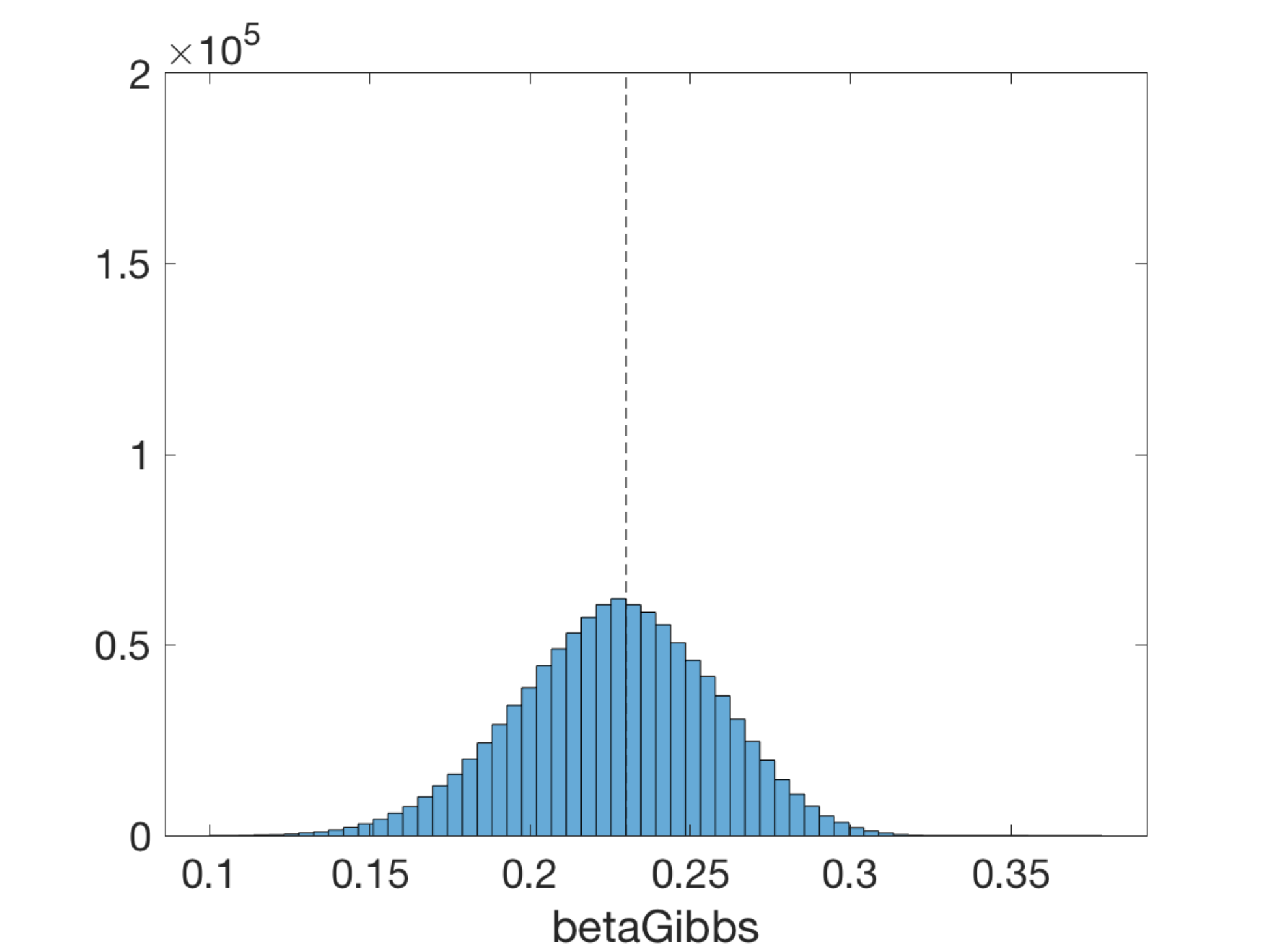} &
\includegraphics[scale=0.40]{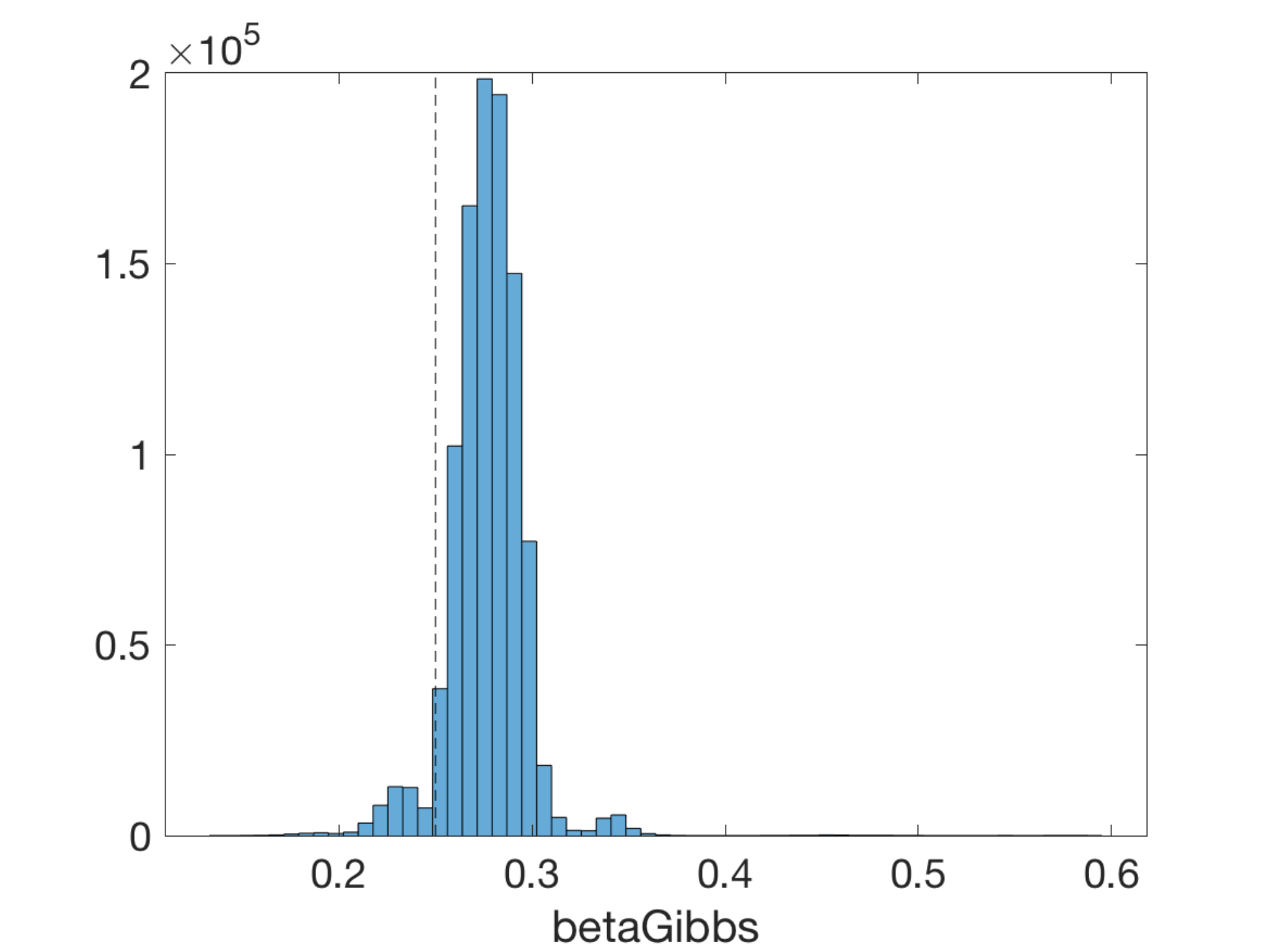}\\
\includegraphics[scale=0.40]{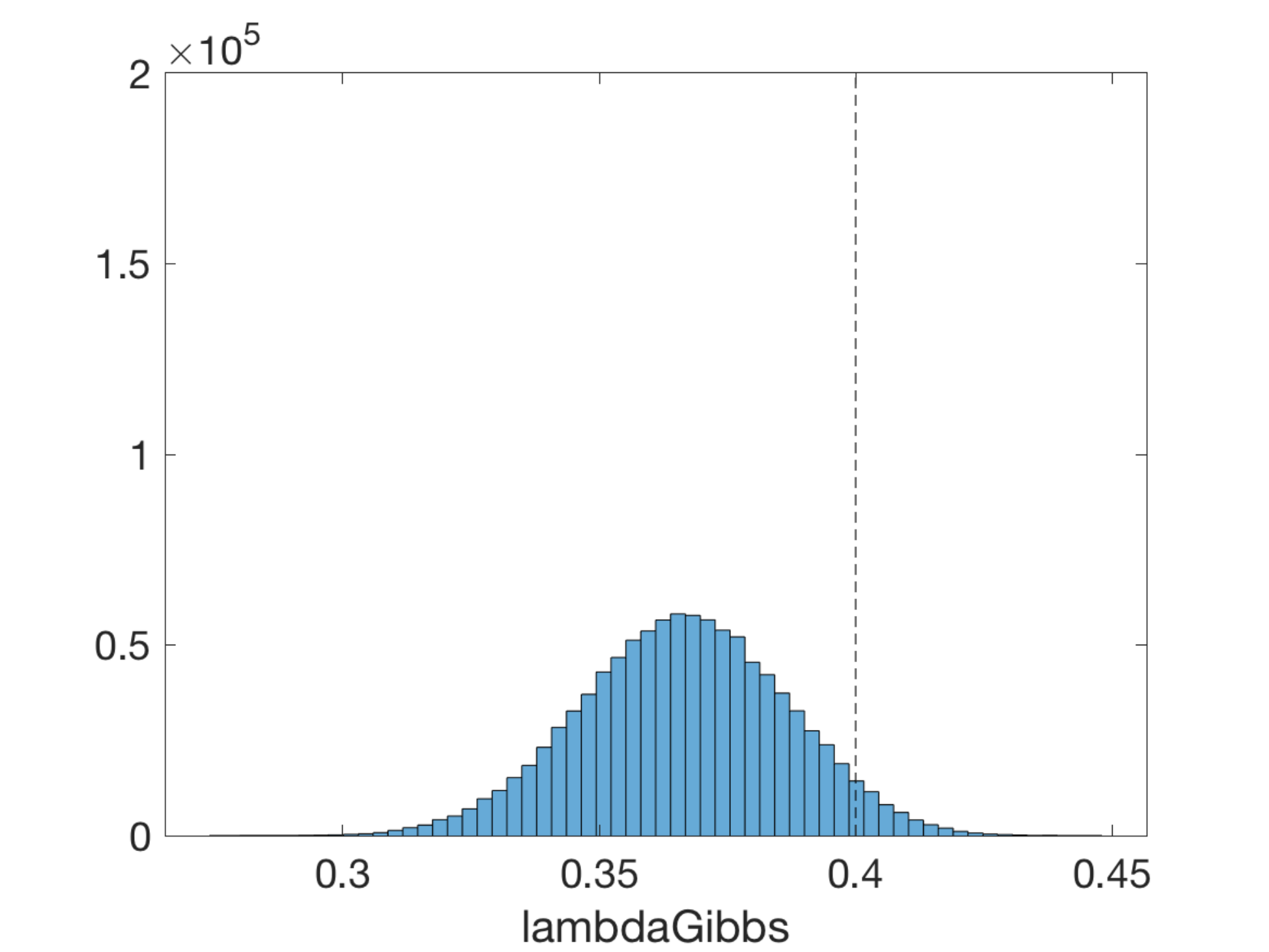} &
\includegraphics[scale=0.40]{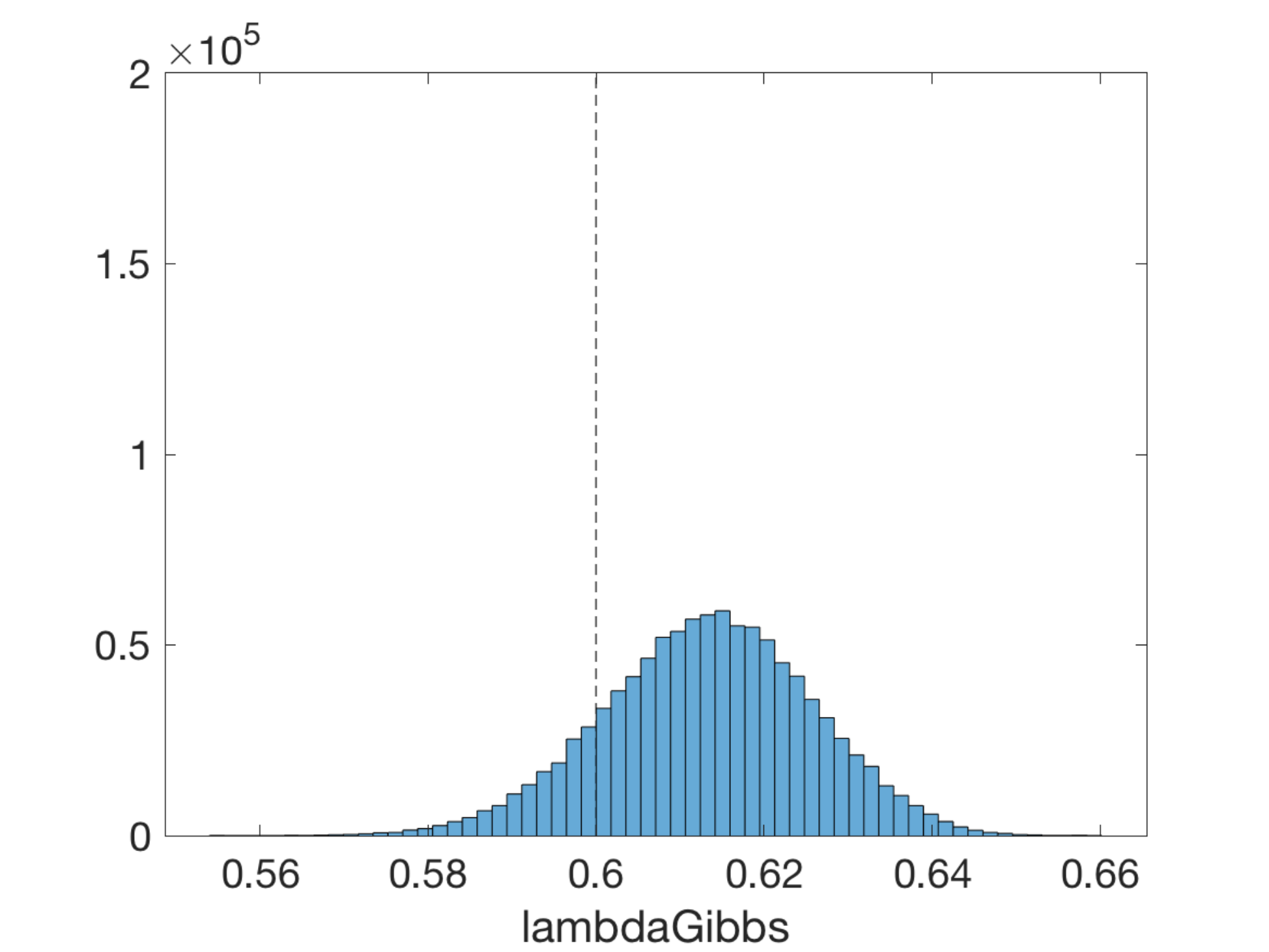} \\
\end{tabular}
\caption{Histograms of the MCMC draws for the parameters in both settings: low persistence and high persistence.}
\label{fig:GibbsDraws}
\end{figure}

\begin{figure}[H]
\centering
\begin{tabular}{cc}
(a) Low persistence & (b) High persistence\\
\begin{scriptsize} $\alpha$ \end{scriptsize} & \begin{scriptsize} $\alpha$ \end{scriptsize}\\
\includegraphics[scale=0.40]{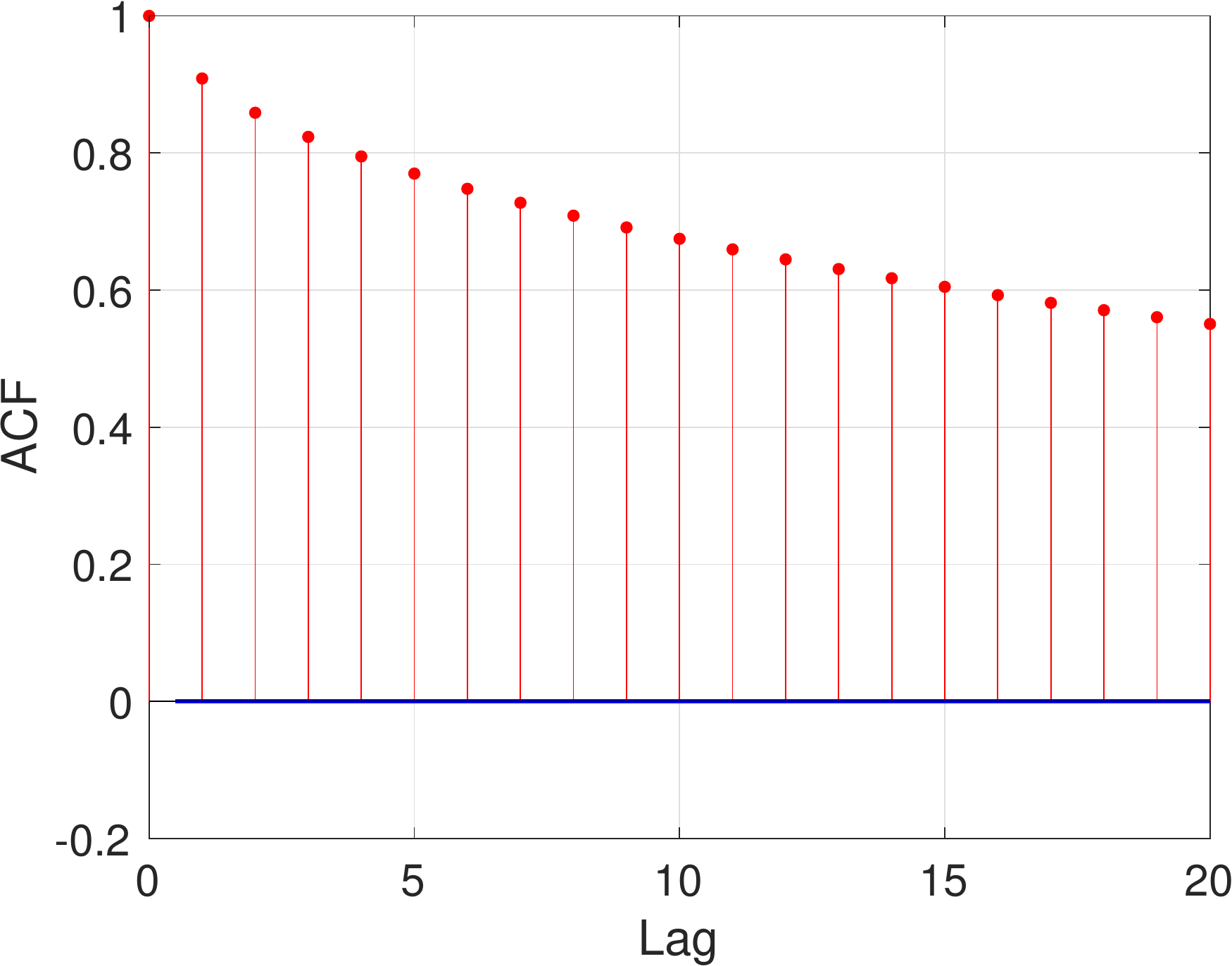} &
\includegraphics[scale=0.40]{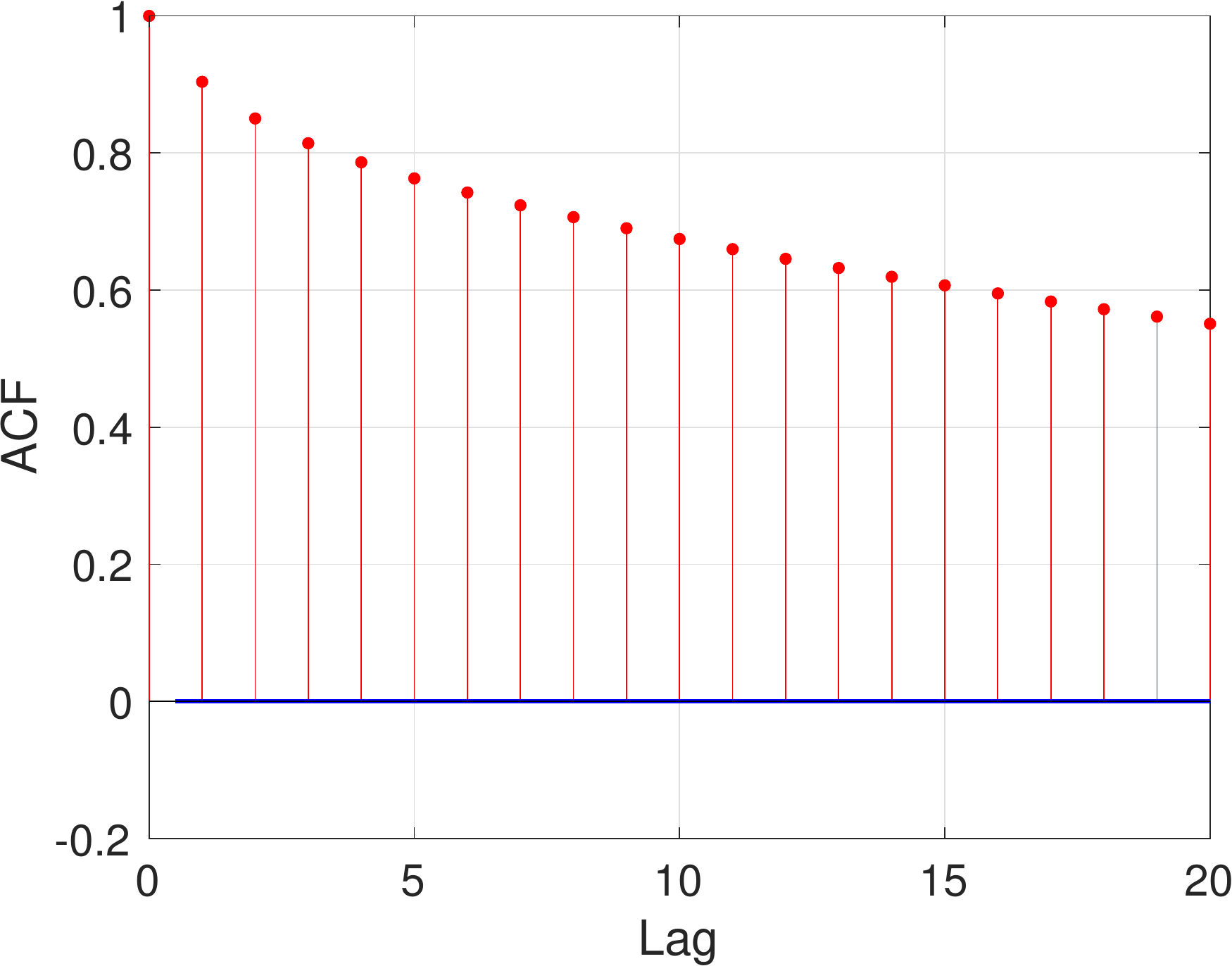}\\
\begin{scriptsize} $\beta$ \end{scriptsize} & \begin{scriptsize} $\beta$ \end{scriptsize}\\
\includegraphics[scale=0.40]{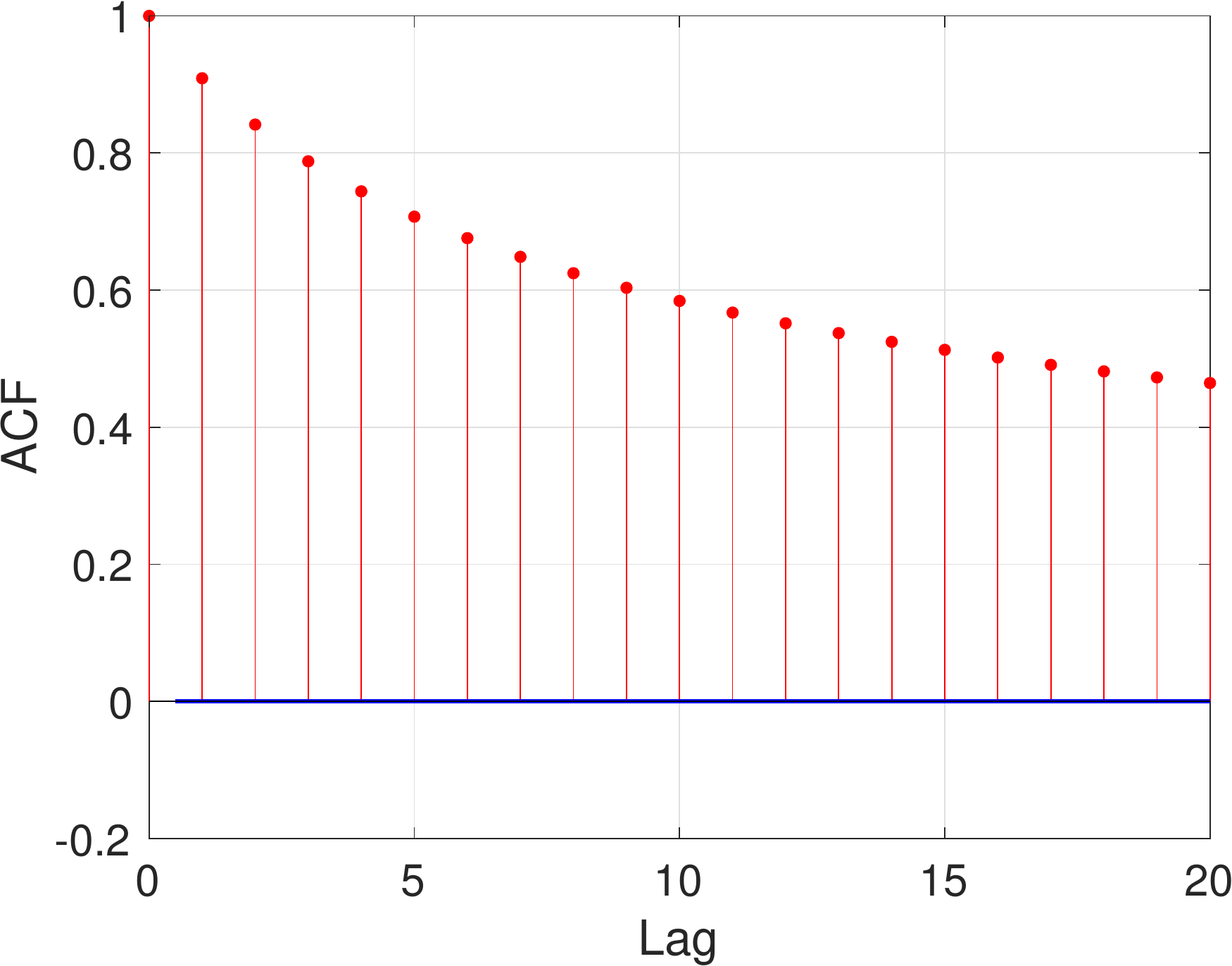} &
\includegraphics[scale=0.40]{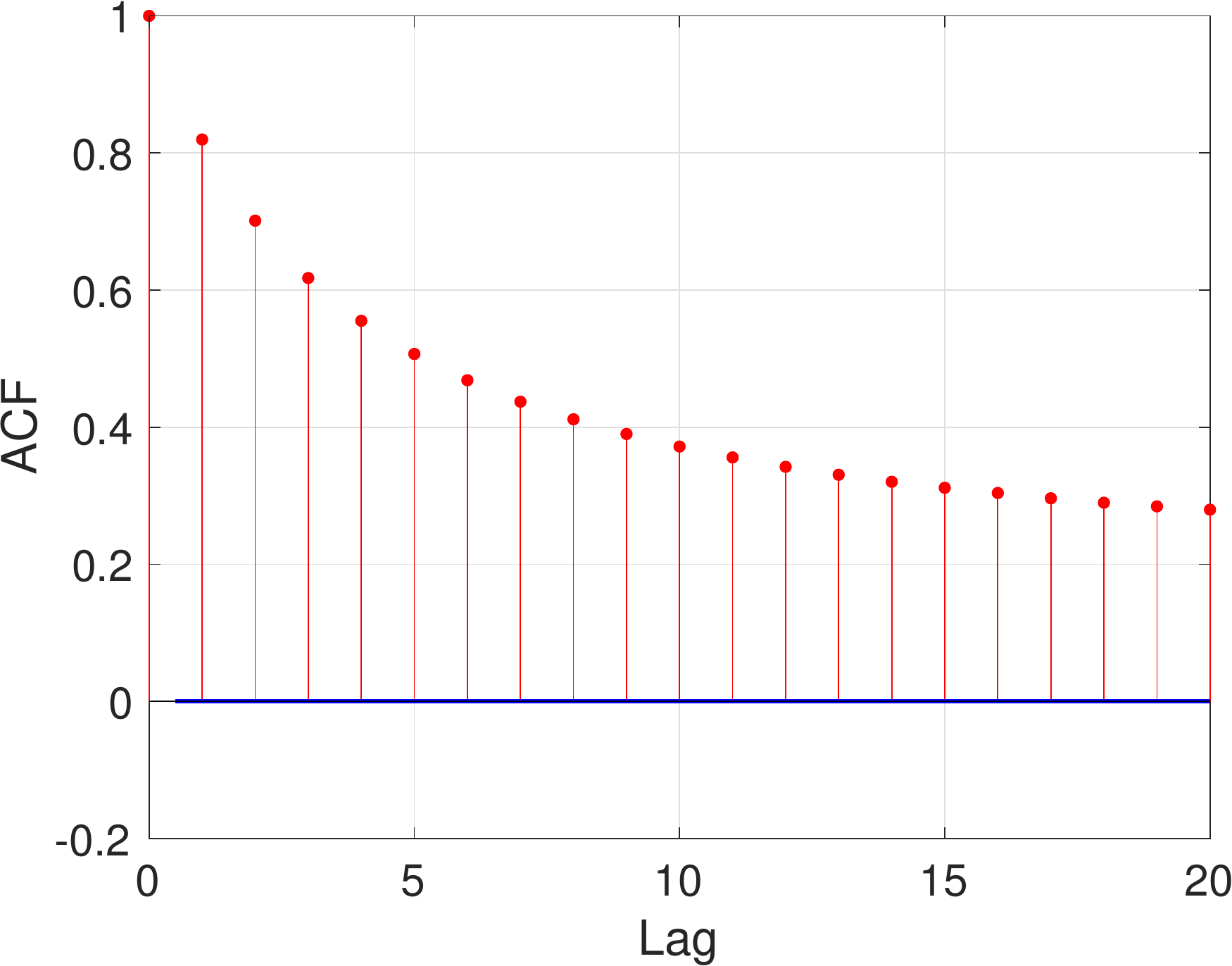}\\
\begin{scriptsize} $\lambda$ \end{scriptsize} & \begin{scriptsize} $\lambda$ \end{scriptsize}\\
\includegraphics[scale=0.40]{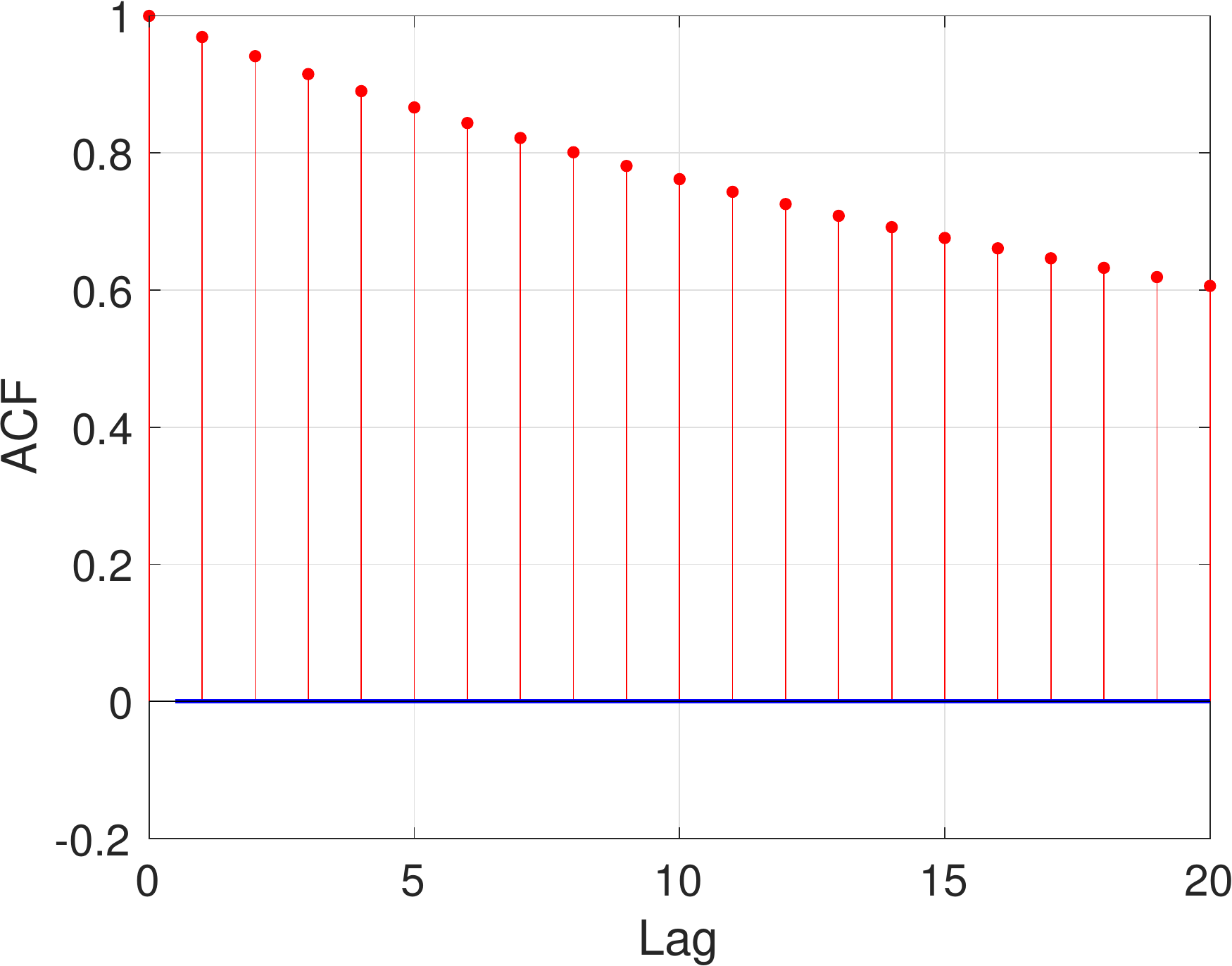} &
\includegraphics[scale=0.40]{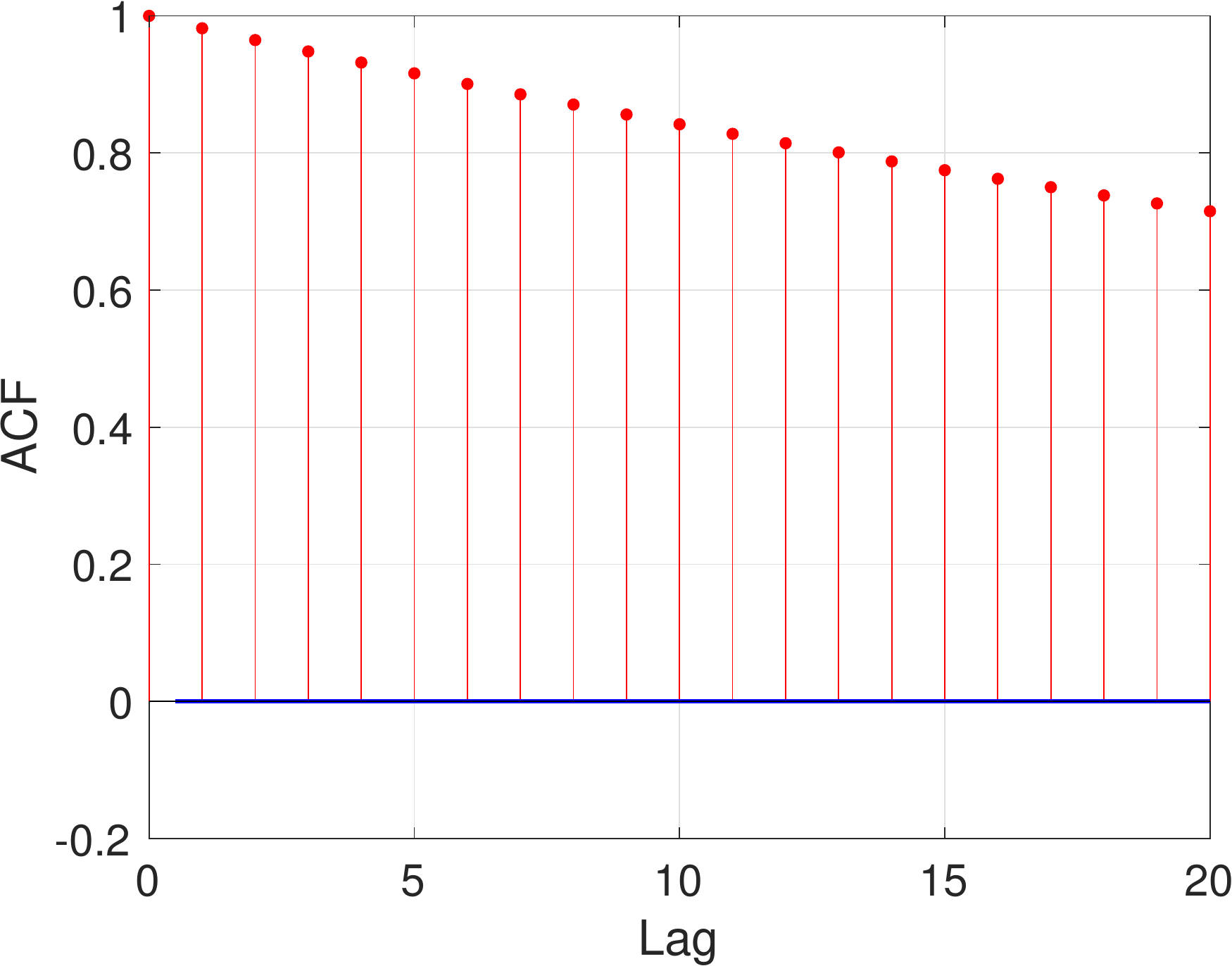} \\
\end{tabular}
\caption{Autocorrelation function for the parameters in both low persistence and high persistence settings.}
\label{fig:GibbsACF}
\end{figure}

\subsection{After thinning}

\begin{figure}[H]
\centering
\setlength{\tabcolsep}{1pt}
\begin{tabular}{cc}
(a) Low persistence & (b) High persistence\\
\includegraphics[scale=0.40]{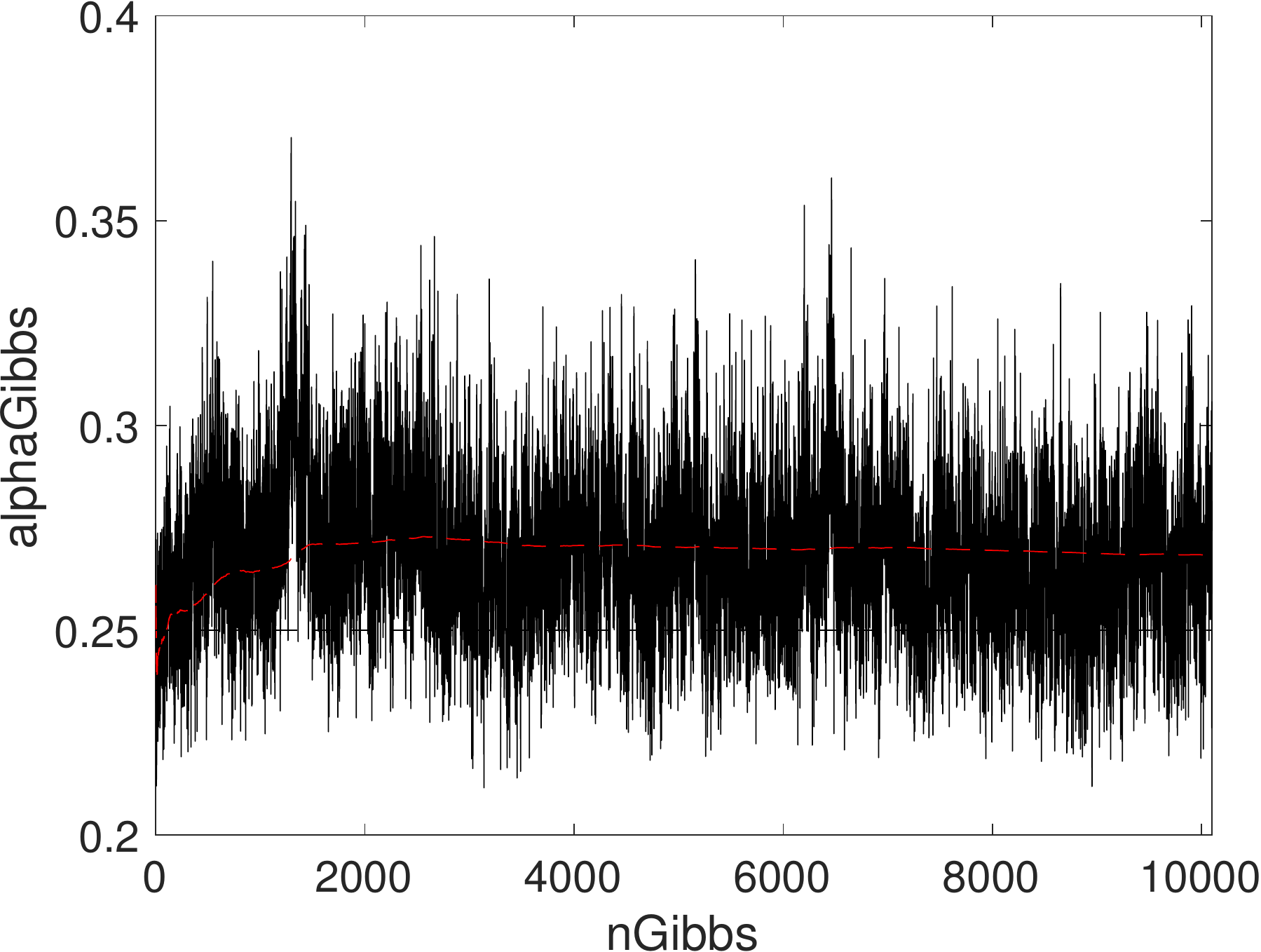} &
\includegraphics[scale=0.40]{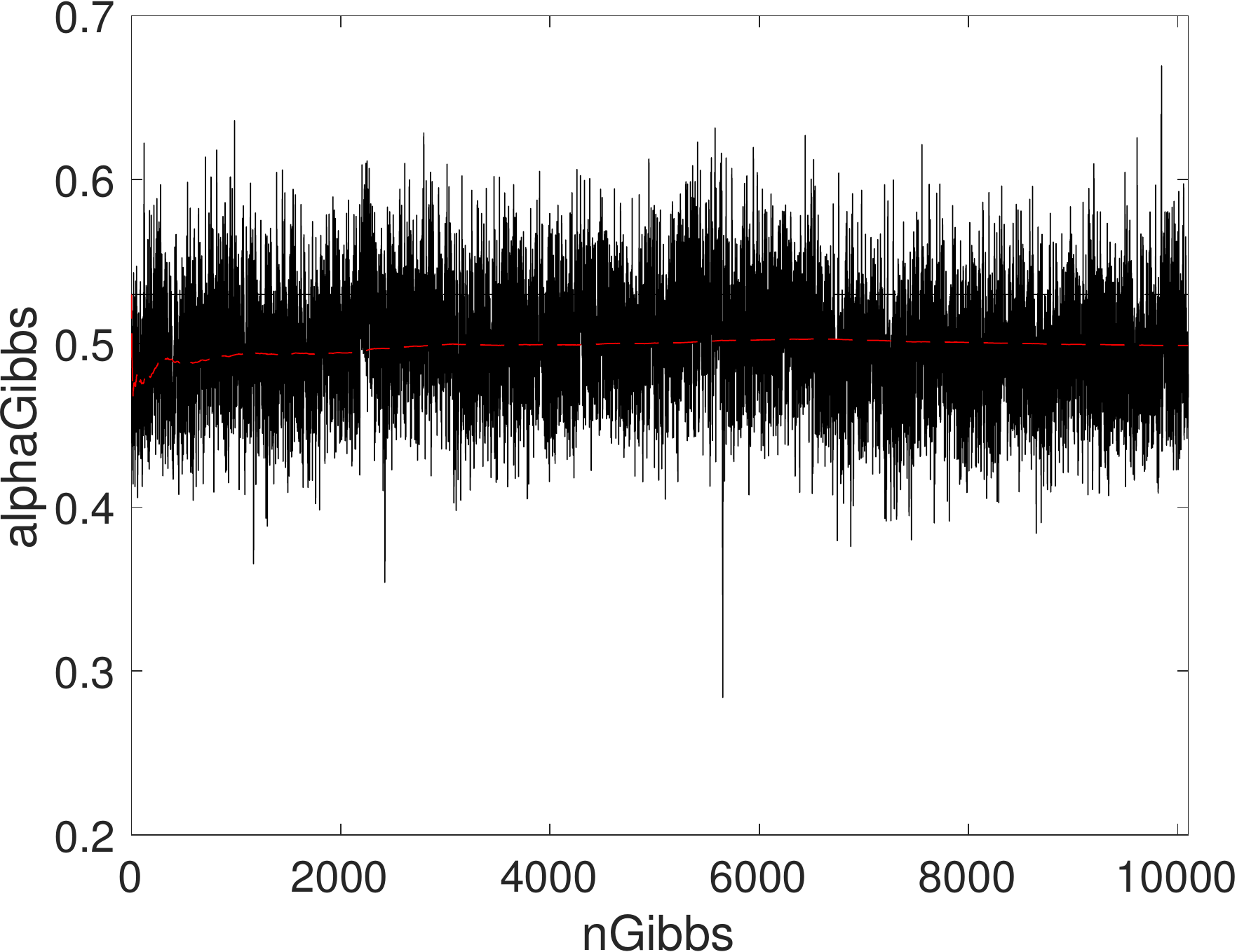}\\
\includegraphics[scale=0.40]{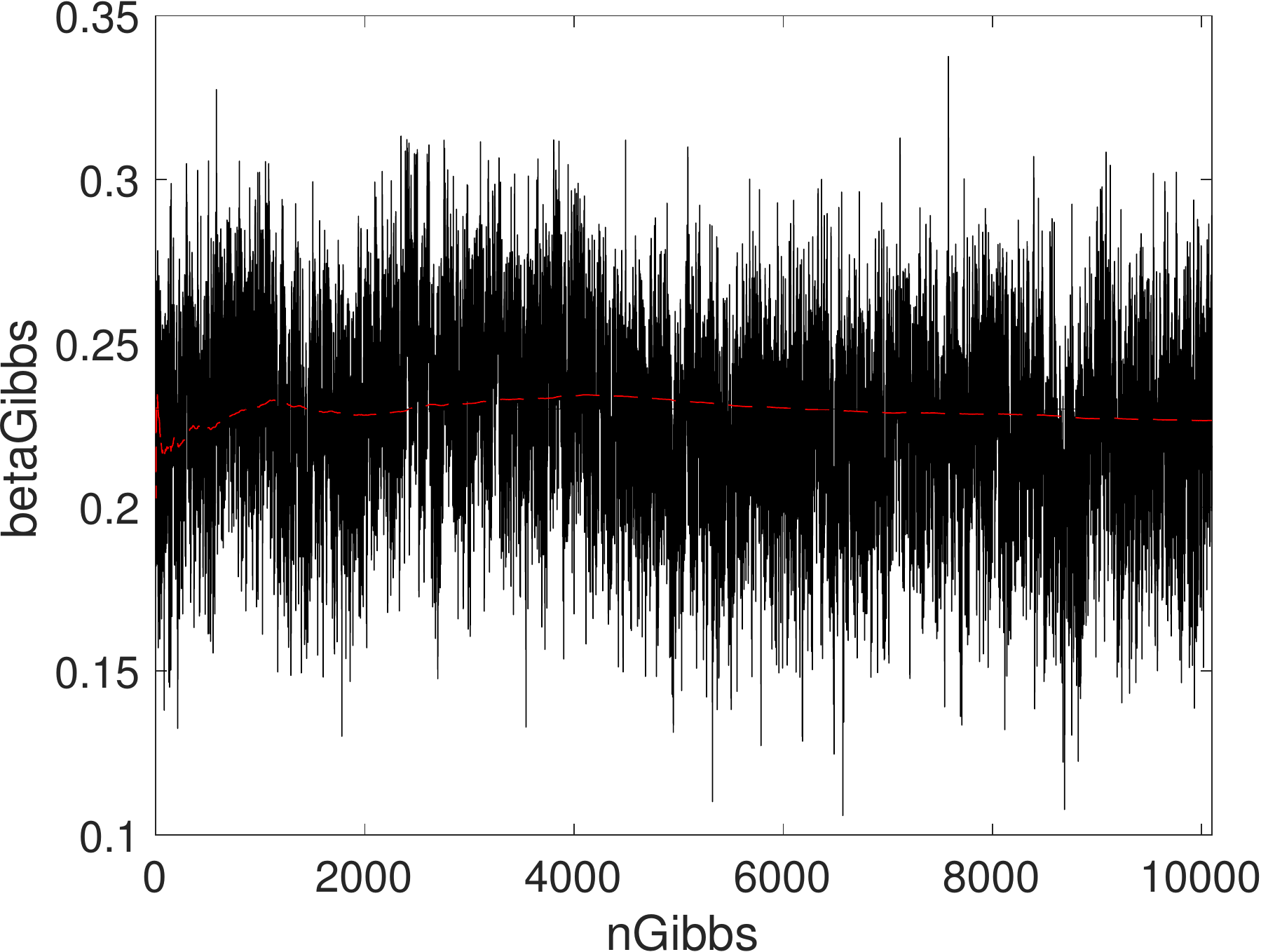} &
\includegraphics[scale=0.40]{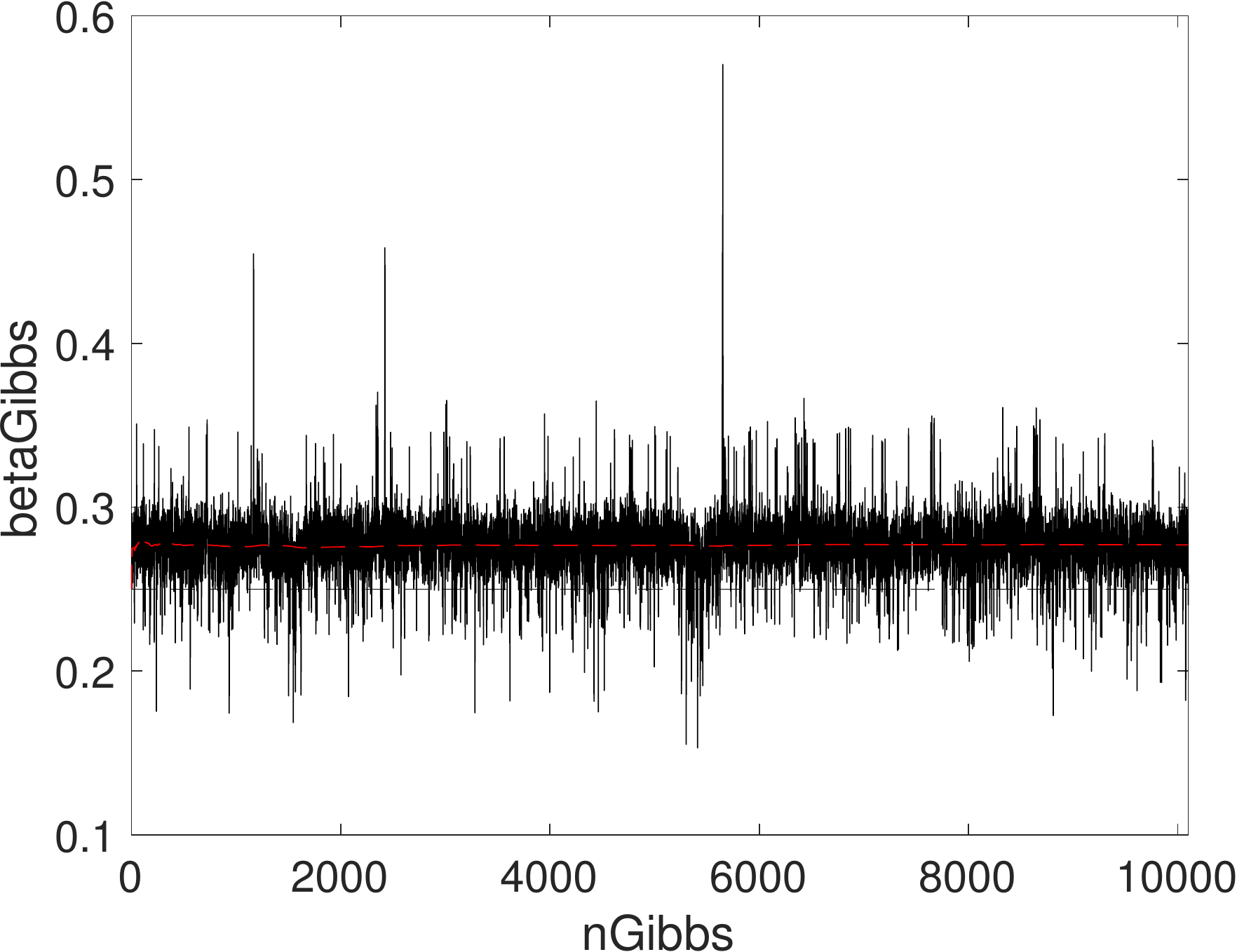}\\
\includegraphics[scale=0.40]{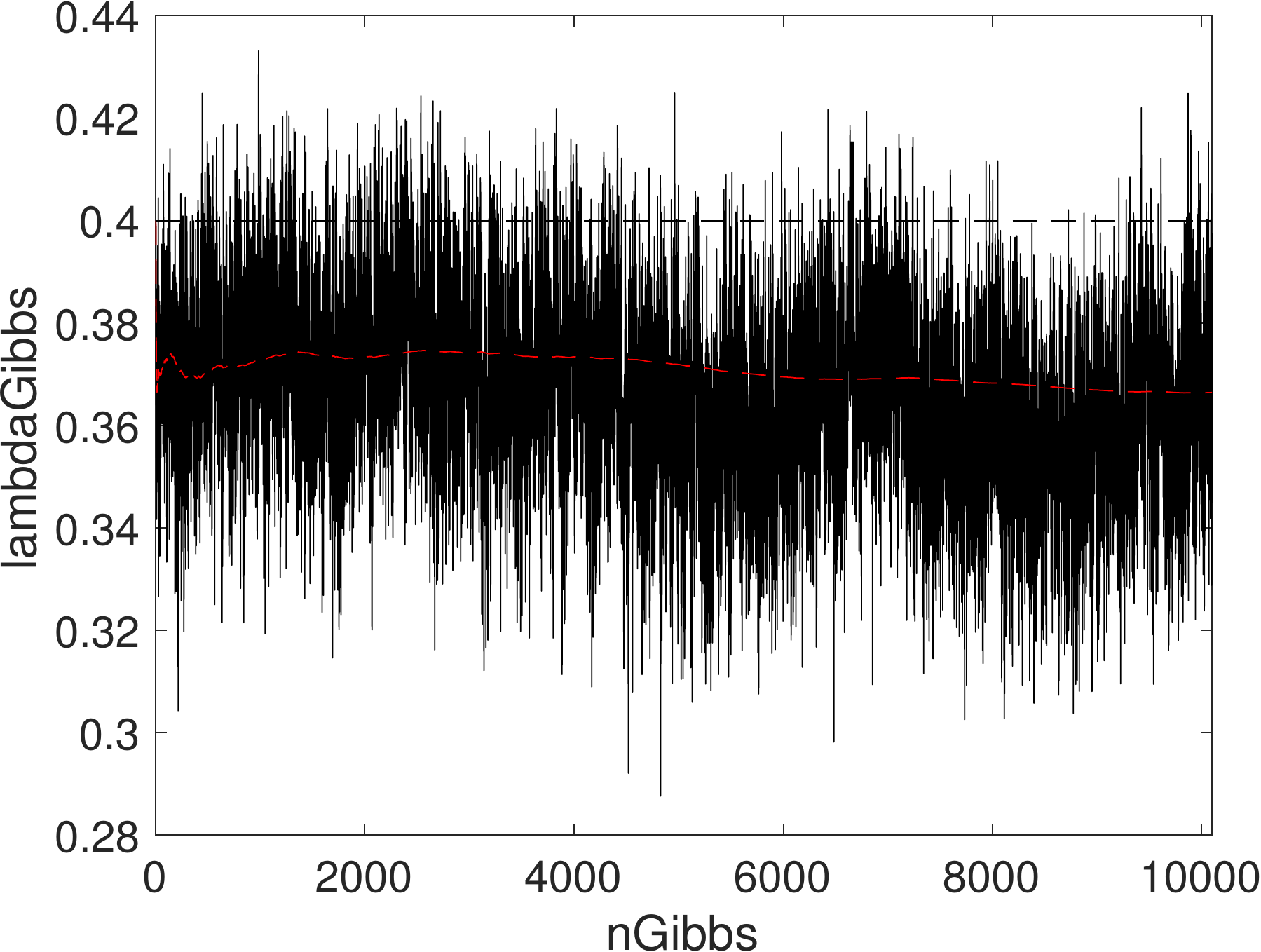} &
\includegraphics[scale=0.40]{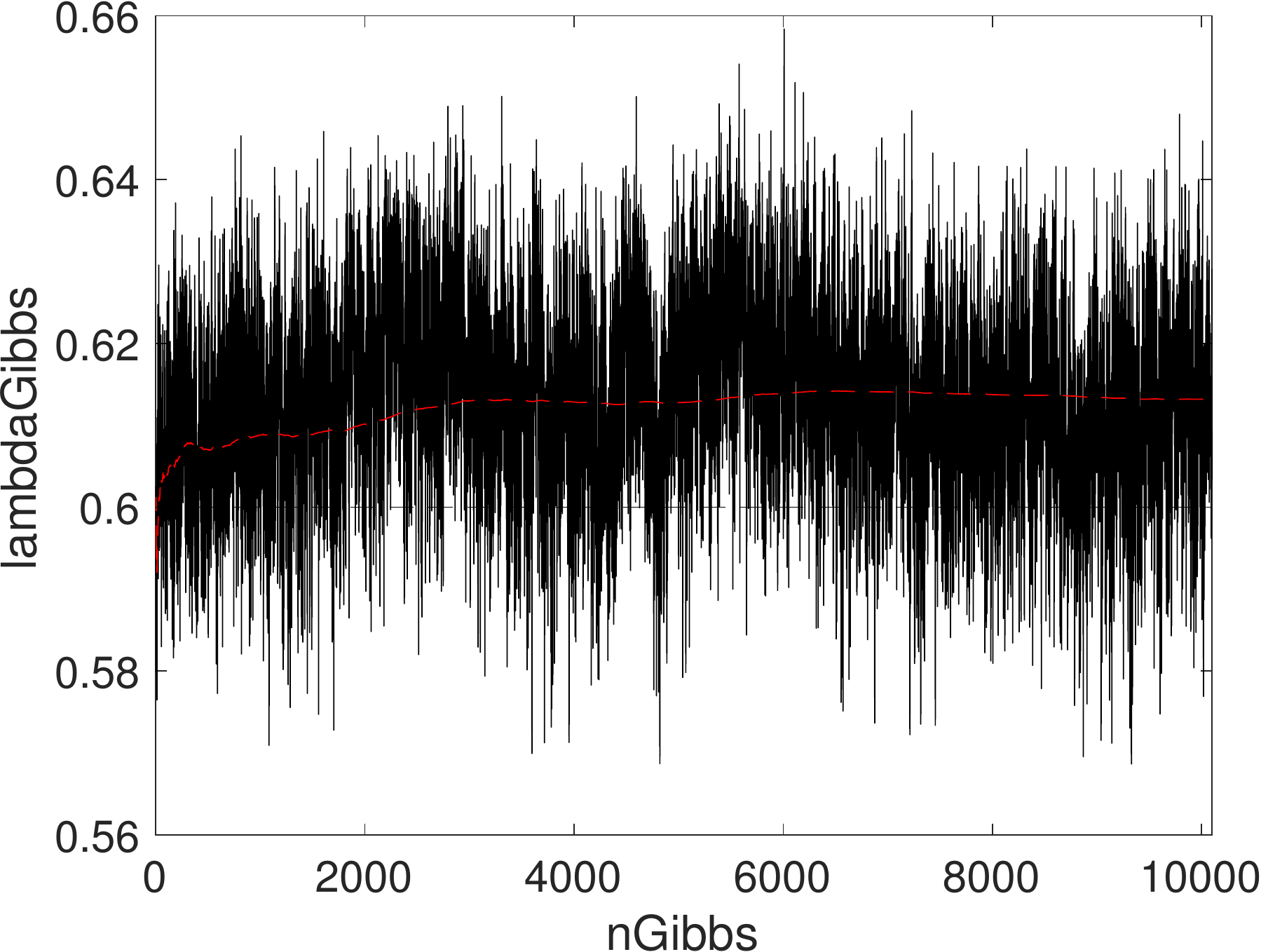} \\
\end{tabular}
\caption{MCMC plot for the parameters in the two setting: low persistence and high persistence, after removing the burn-in sample and thinning.}
\label{fig:GibbsPostABT}
\end{figure}

\begin{figure}[H]
\centering
\setlength{\tabcolsep}{-5pt}
\begin{tabular}{cc}
(a) Low persistence & (b) High persistence\\
\includegraphics[scale=0.40]{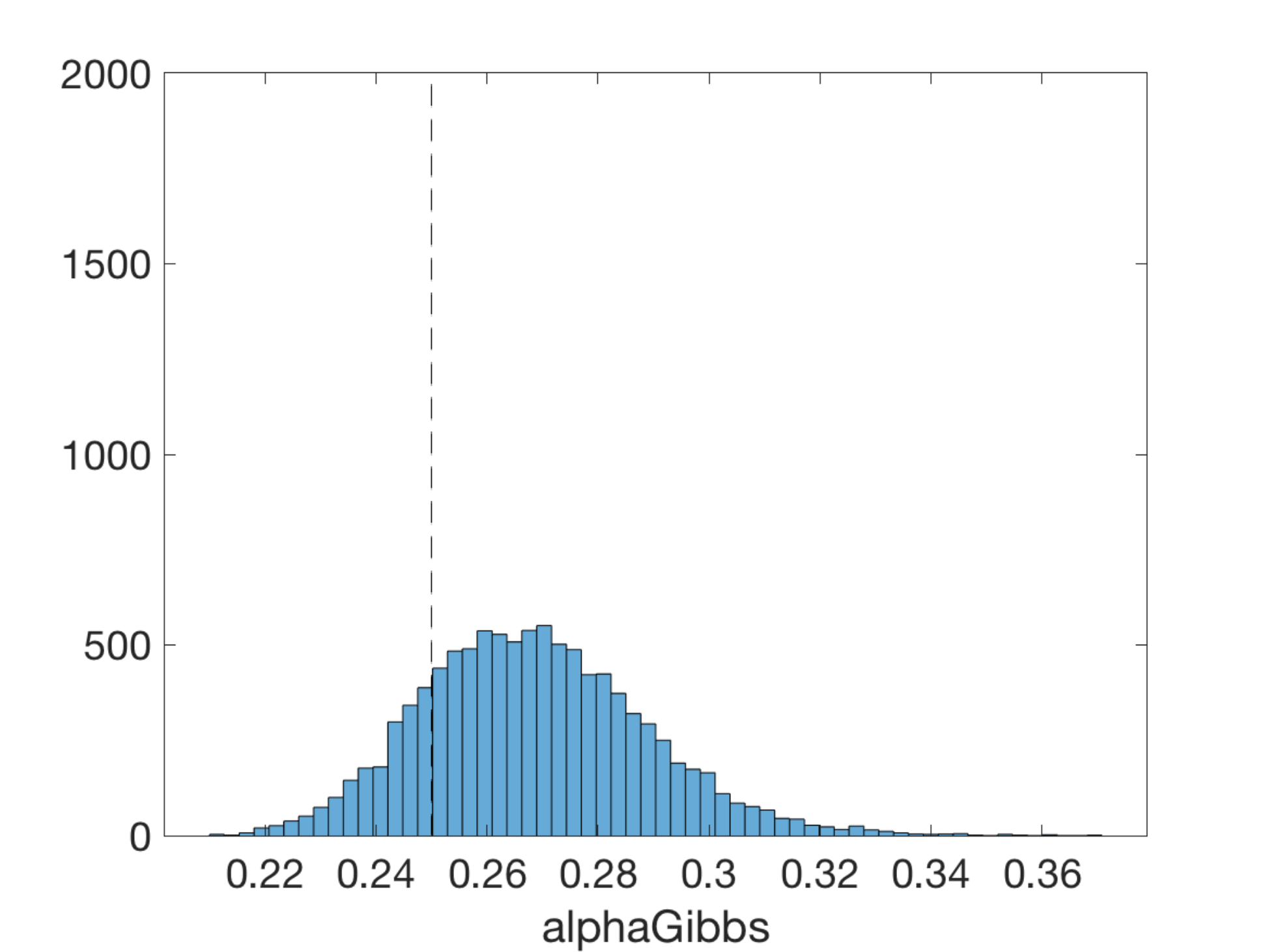} &
\includegraphics[scale=0.40]{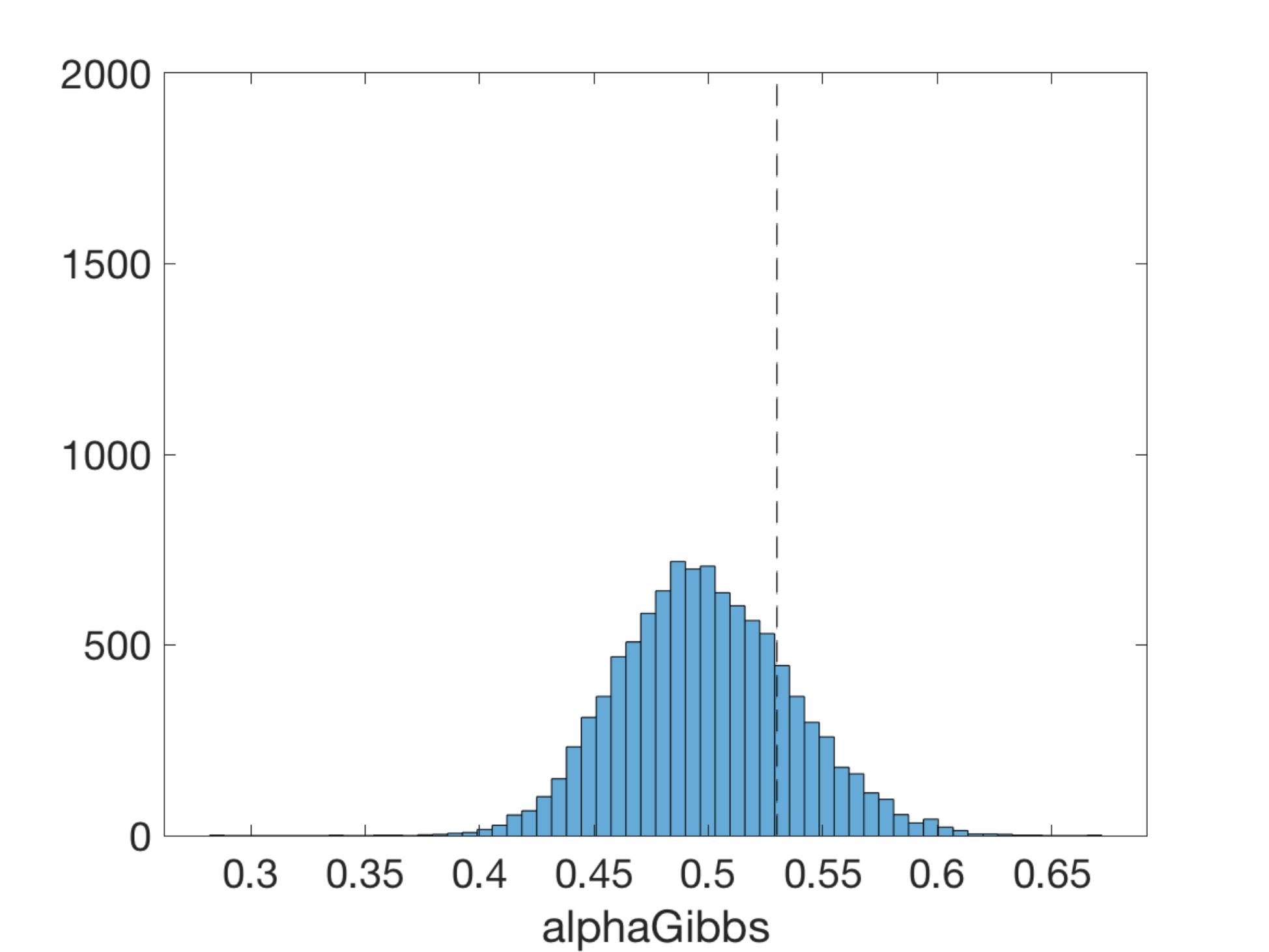}\\
\includegraphics[scale=0.40]{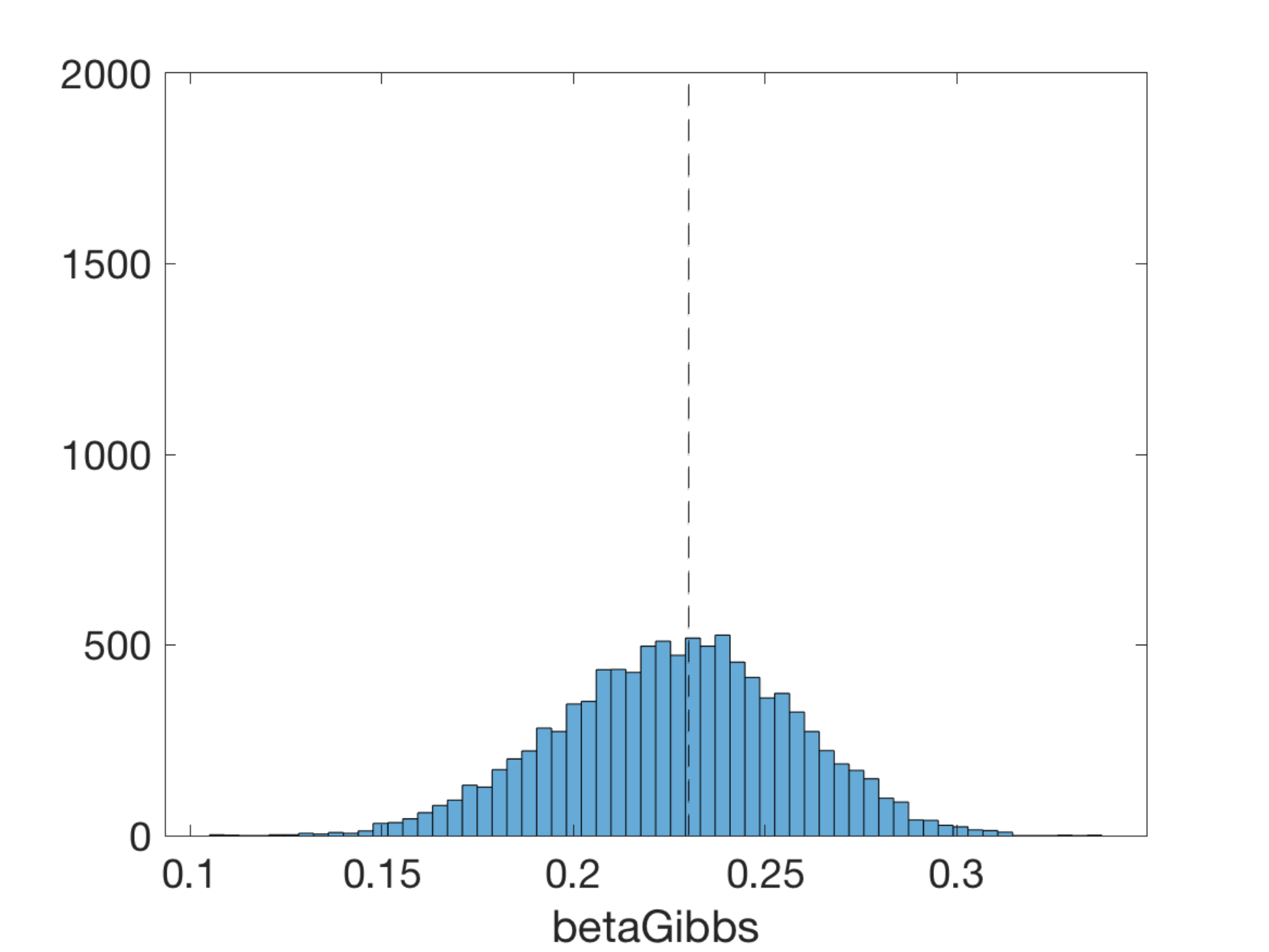} &
\includegraphics[scale=0.40]{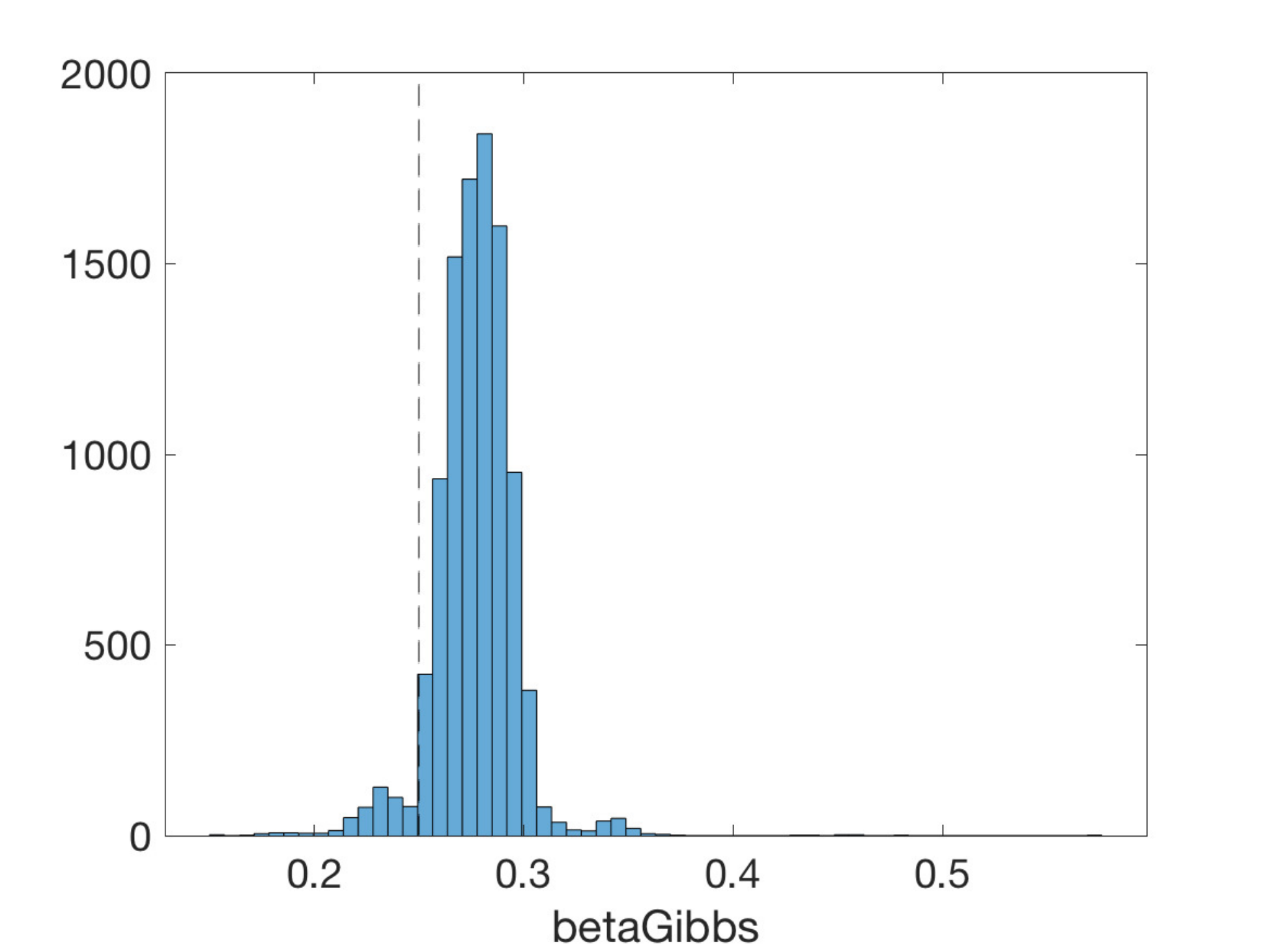}\\
\includegraphics[scale=0.40]{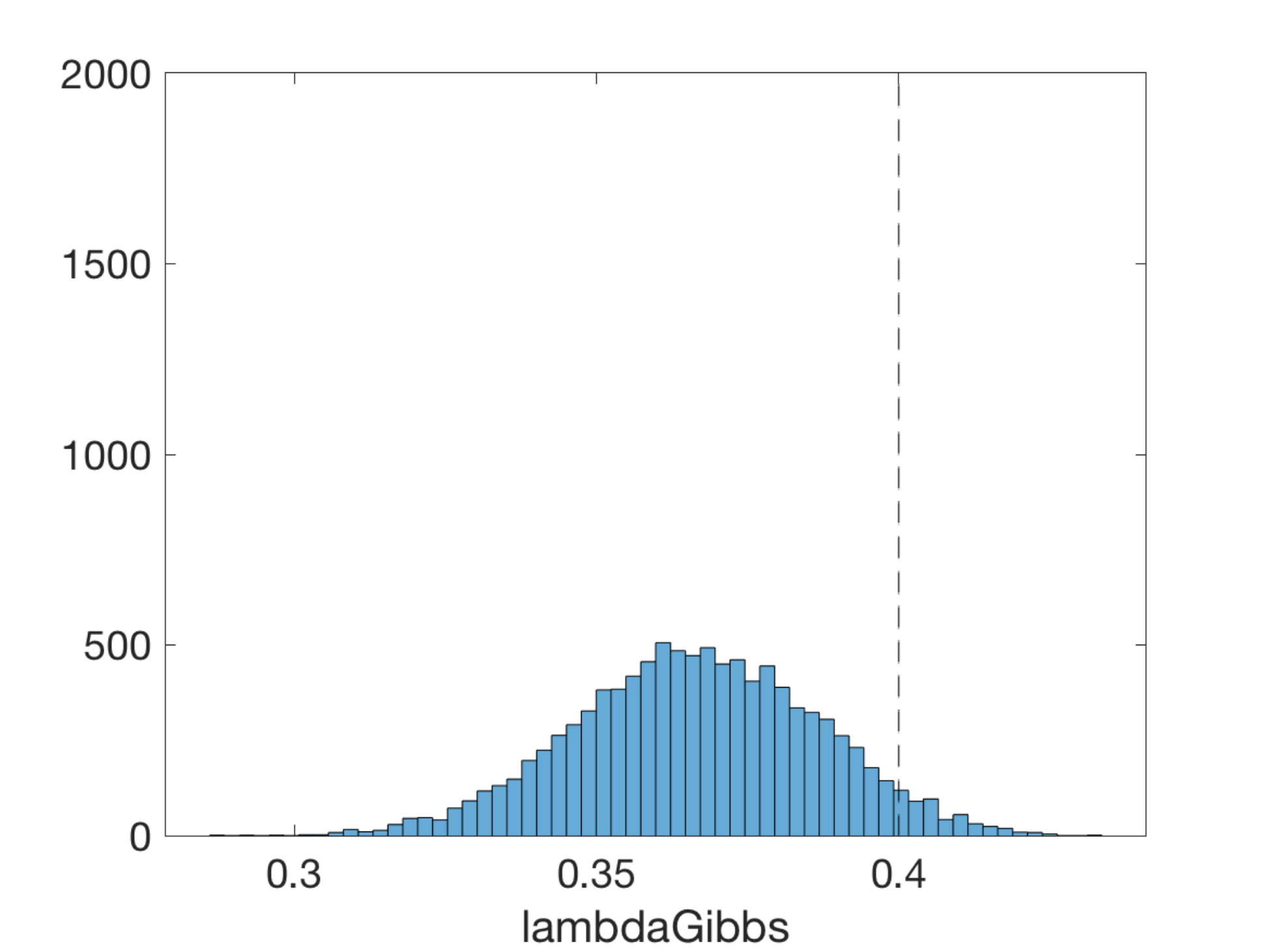} &
\includegraphics[scale=0.40]{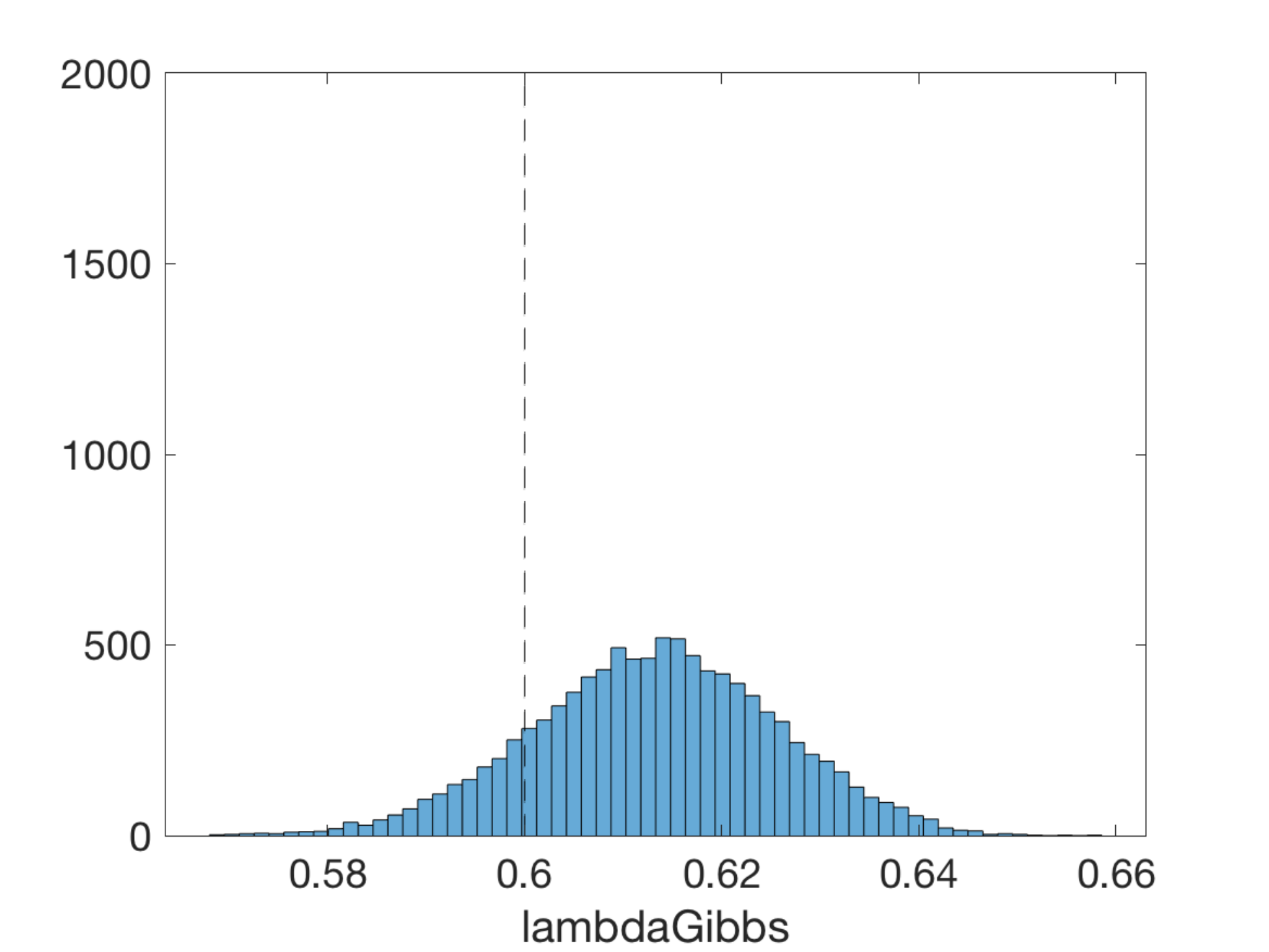} \\
\end{tabular}
\caption{Histograms of the MCMC draws for the parameters in both settings: low persistence and high persistence, after removing the burn-in sample and thinning.}
\label{fig:GibbsHistABT}
\end{figure}

\begin{figure}[H]
\centering
\begin{tabular}{cc}
(a) Low persistence & (b) High persistence\\
\begin{scriptsize} $\alpha$ \end{scriptsize} & \begin{scriptsize} $\alpha$ \end{scriptsize}\\
\includegraphics[scale=0.40]{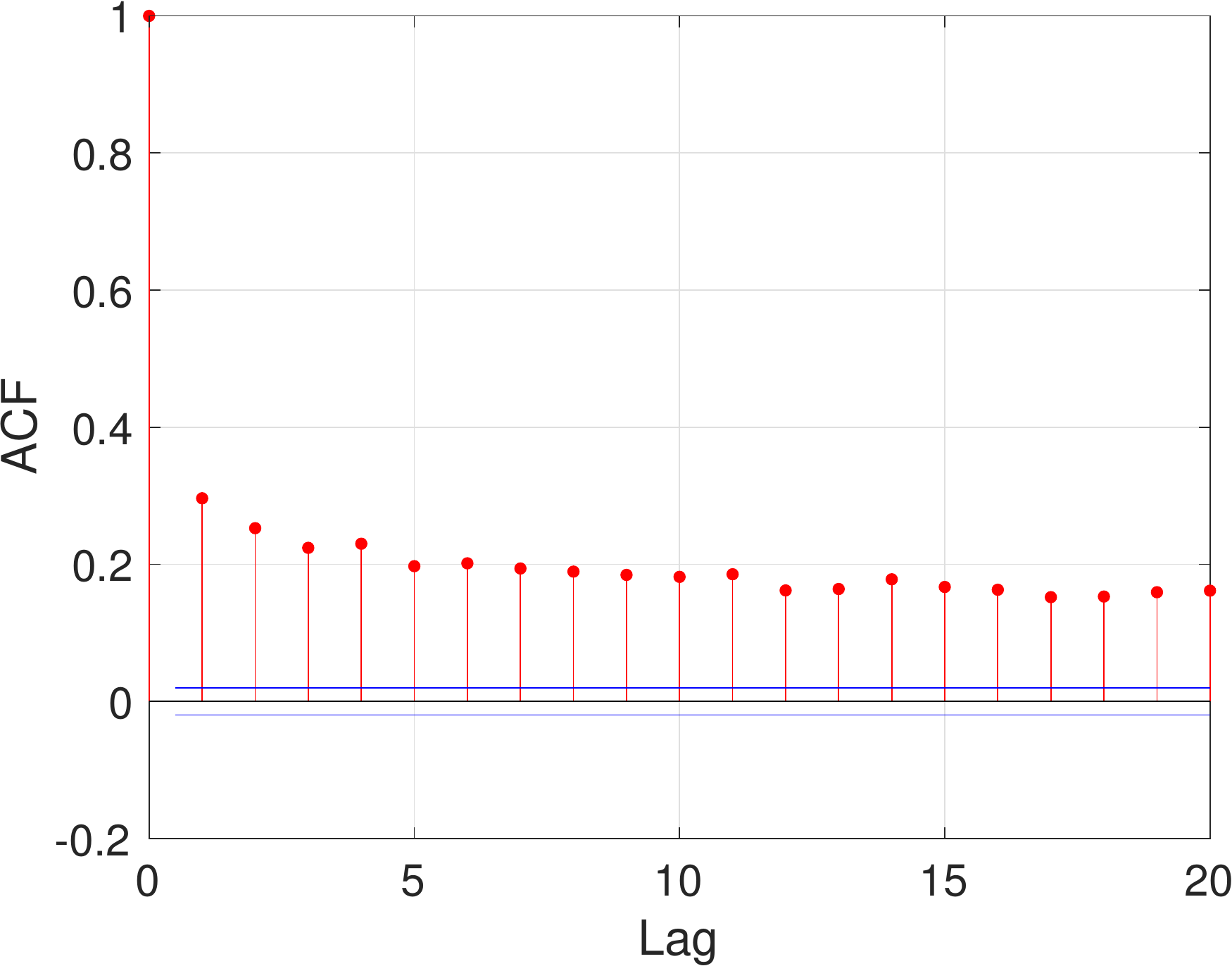} &
\includegraphics[scale=0.40]{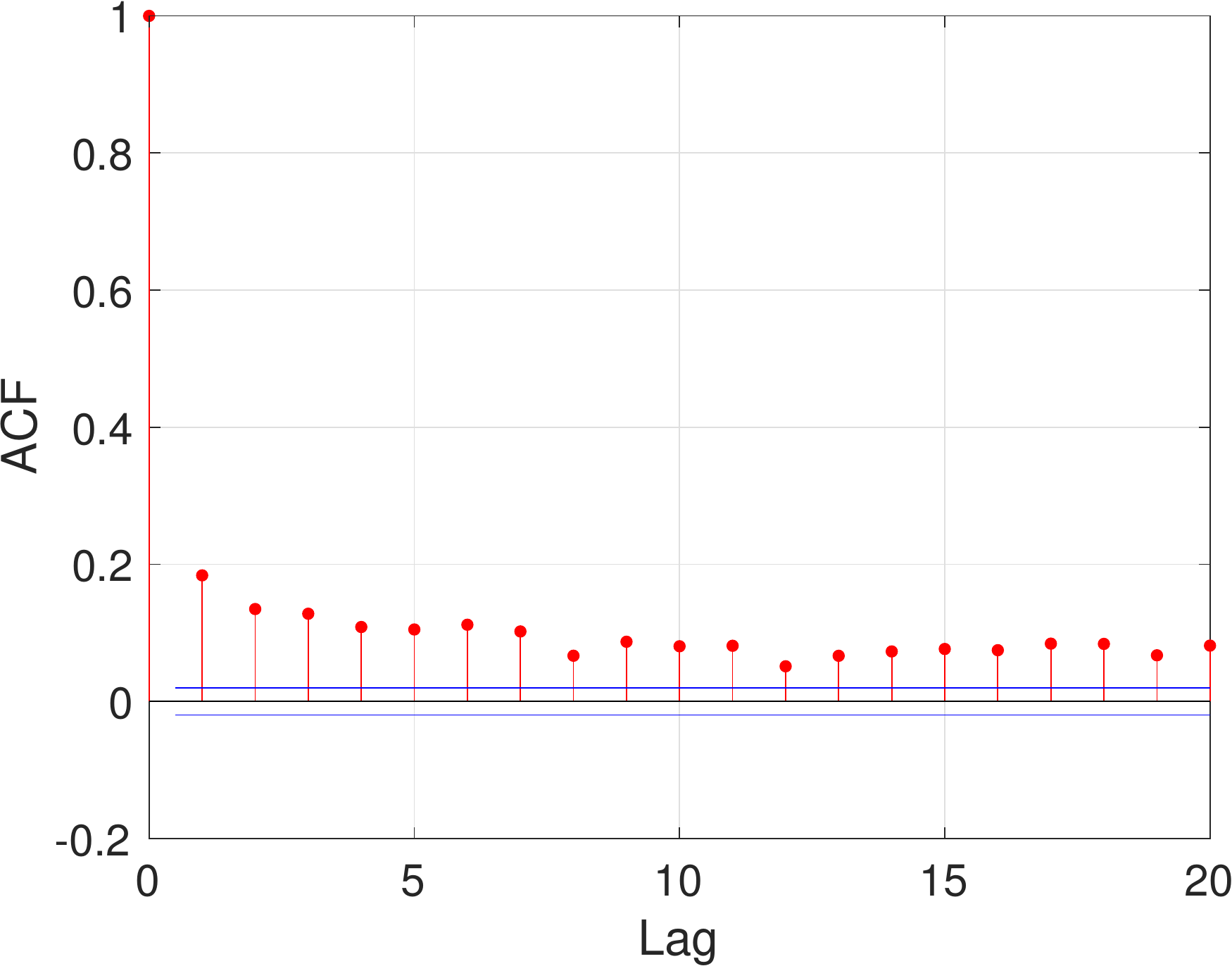}\\
\begin{scriptsize} $\beta$ \end{scriptsize} & \begin{scriptsize} $\beta$ \end{scriptsize}\\
\includegraphics[scale=0.40]{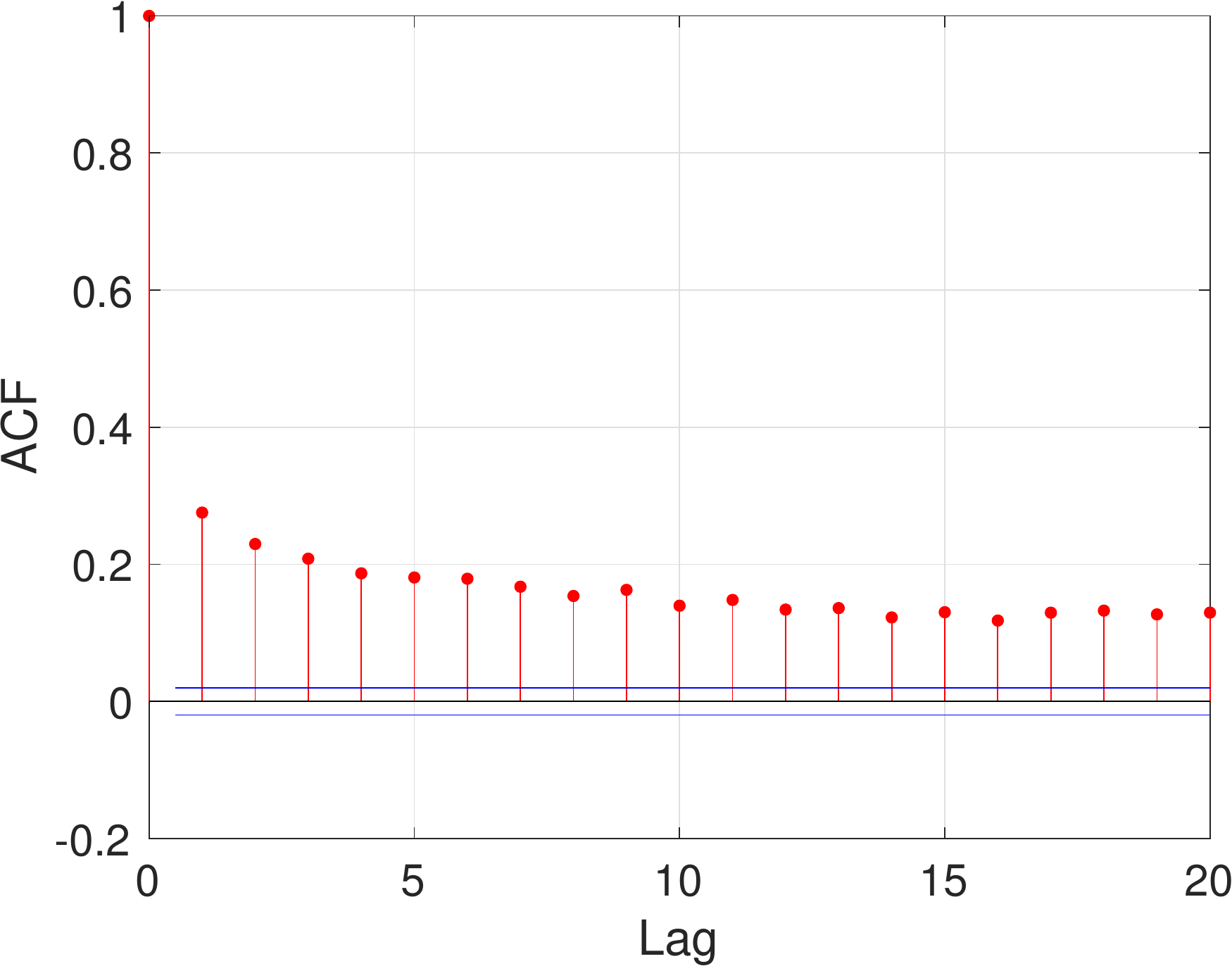} &
\includegraphics[scale=0.40]{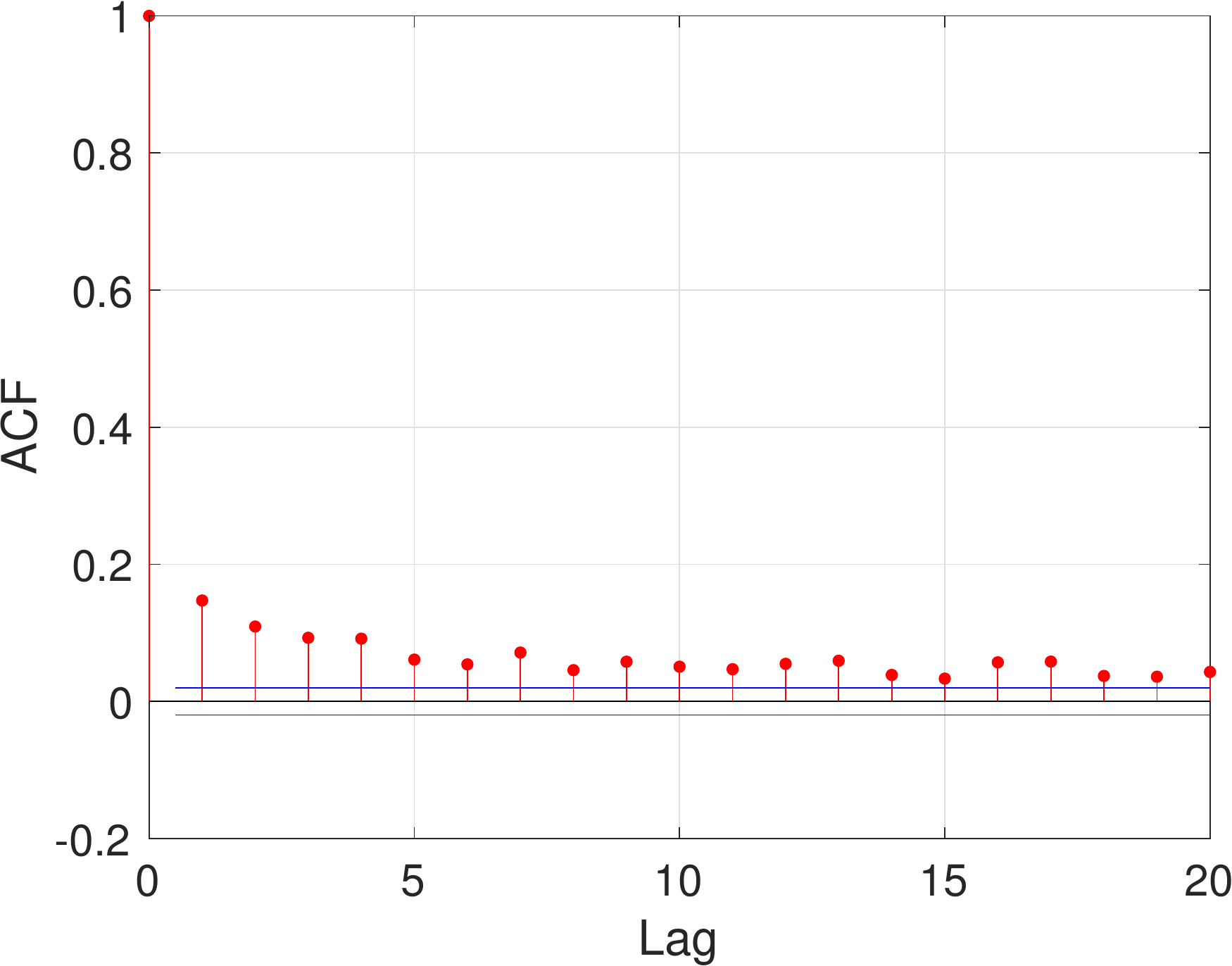}\\
\begin{scriptsize} $\lambda$ \end{scriptsize} & \begin{scriptsize} $\lambda$ \end{scriptsize}\\
\includegraphics[scale=0.40]{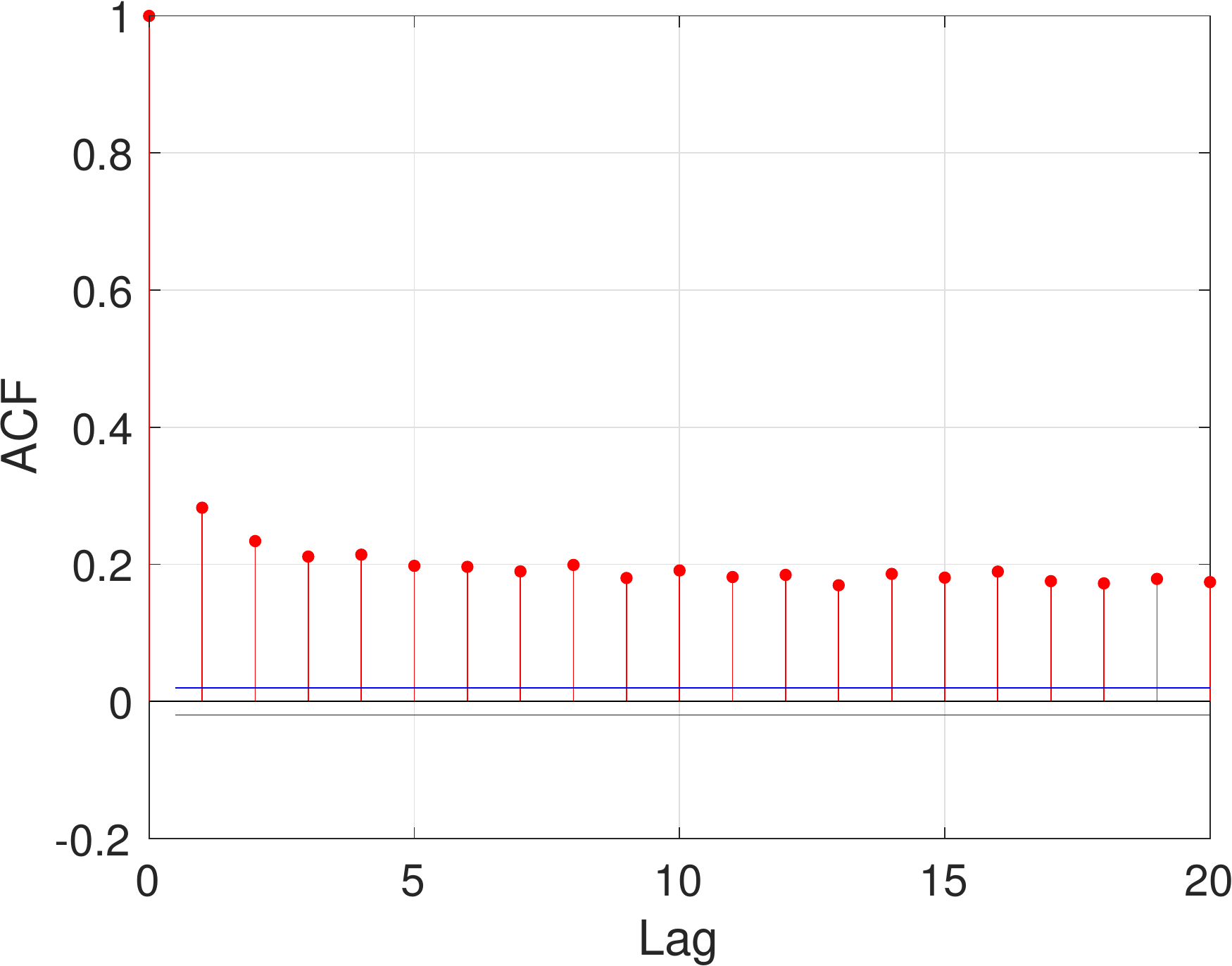} &
\includegraphics[scale=0.40]{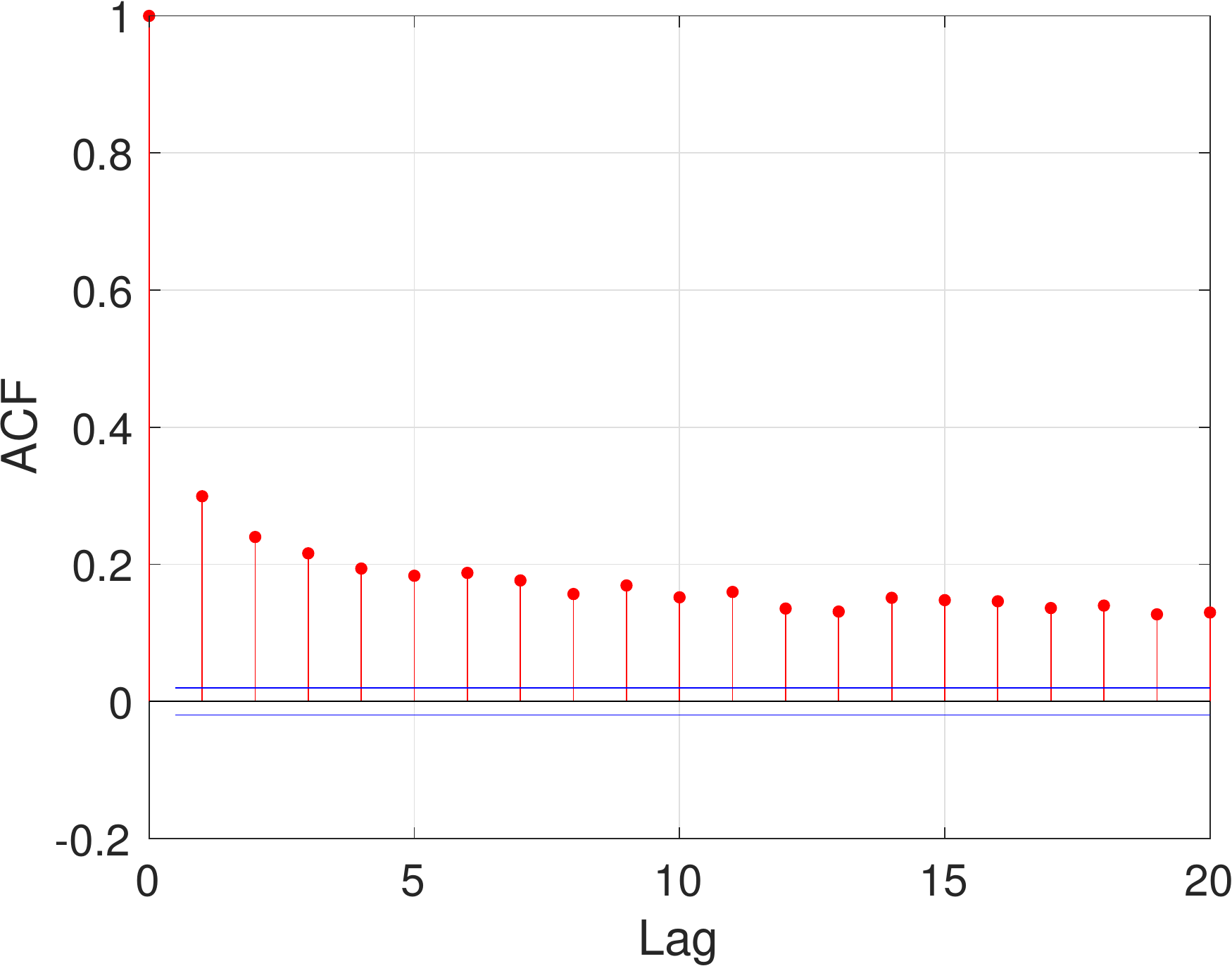} \\
\end{tabular}
\caption{Autocorrelation function for the parameters in both low persistence and high persistence settings, after removing the burn-in sample and thinning.}
\label{fig:GibbsACFABT}
\end{figure}

%

\end{document}